%% file: tesi.tex
\documentclass[final]{thesis-phdNSITunito}

\usepackage{mdframed} 
\usepackage{amsmath} 
\usepackage{amsthm} 
\usepackage{amssymb} 
\usepackage{enumitem} 
\usepackage{mathtools} 
\usepackage{multirow} 
\usepackage{subcaption} 
\usepackage{tikz} 
\usetikzlibrary{patterns}
\usepackage{changebar} 
\usepackage{upgreek} 
\usepackage[hidelinks]{hyperref} 
\usepackage{cleveref} 
\usepackage[font=small,labelfont=bf]{caption} 
\usepackage{lscape} 
\usepackage{appendix} 
\usepackage[final]{listings}
\theoremstyle{plain}
\newtheorem{theorem}{Theorem} 
\newtheorem{lemma}{Lemma} 
\newtheorem{proposition}{Proposition} 
\newtheorem{definition}{Definition} 
\newtheorem{example}{Example} 
\newtheorem{remark}{Remark} 

\crefname{property}{property}{properties}
\crefname{lstlisting}{program}{programs}
\crefname{axiom}{}{}
\newcommand\gref[1]{\nameref{#1} \nameCref{#1}}

\renewcommand\cref\Cref

\input{notation}
\input{gdraft}

\input{gproof}

\includecomment{omissible}

\includecomment{thesis}
\excludecomment{article}

\title{Interactive Realizability, Monads and Witness Extraction (Draft)}
\author{Giovanni Birolo}

\candidate{Giovanni Birolo}
\cycle{XXV} 
\department{Matematica}
\academicyears{2010--2012} 
\submitdate{Aprile 2012}
\tutor{Prof. Stefano Berardi}
\coordinator{Prof. Luigi Rodino}
\scientificfield{MAT/01 Logica Matematica}
\area{Matematica}
\graphicalabstract{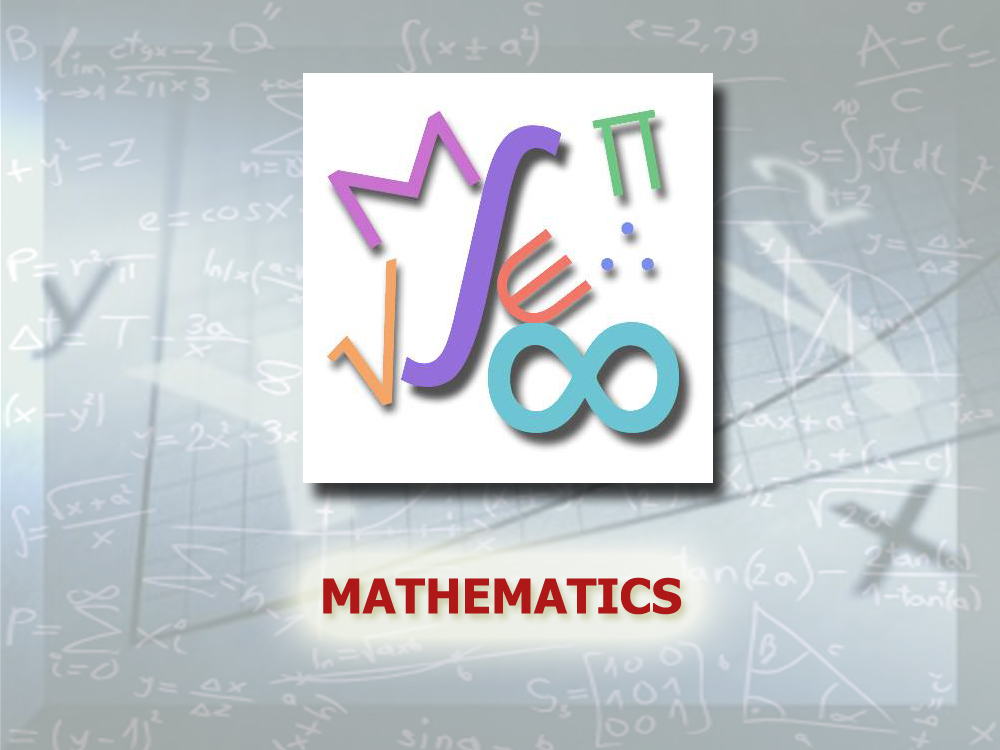}

\newenvironment{legend}{\small \it}{}

\newcommand\HA{\mathsf{HA}} 
\newcommand\PA{\mathsf{PA}} 
\newcommand\IL{\mathsf{IL}} 
\newcommand\CL{\mathsf{CL}} 
\newcommand\EMG{\mathsf{EM}} 
\newcommand\EM{\EMG_1} 
\newcommand\PRA{\mathsf{PRA}} 

\begin{document}

\coverpage
\titlepage

\section*{Abstract}
In this dissertation we collect some results about ``interactive realizability'',  a realizability semantics that extends the Brouwer-Heyting-Kolmogorov interpretation to (sub-)classical logic, more precisely to first-order intuitionistic arithmetic (Heyting Arithmetic, \(\HA\)) extended by the law of the excluded middle restricted to \(\Sigma^0_1\) formulas (\(\EM\)), a system motivated by its interest in proof mining. 
These results are three interconnected works, listed below. 
\begin{itemize}
  \item 
    We describe the interactive interpretation of a classical proof involving real numbers. 
    The statement we prove is a simple but non-trivial fact about points in the real plane. 
    The proof employs \(\EM\) to deduce properties of the ordering on the real numbers, which is undecidable and thus problematic from a constructive point of view. 
  \item
    We present a new set of reductions for derivations in natural deduction that can extract witnesses from closed derivations of simply existential formulas in \(\HA+\EM\).  
    The reduction we present are inspired by the informal idea of learning by making falsifiable hypothesis and checking them, and by the interactive realizability interpretation. We extract the witnesses directly from derivations in \(\HA+\EM\) by reduction, without encoding derivations by a realizability interpretation. 
  \item
    We give a new presentation of interactive realizability with a more explicit syntax. 
    We express interactive realizers by means of an abstract framework that applies the monadic approach used in functional programming to modified realizability, in order to obtain less strict notions of realizability that are suitable to classical logic. 
    In particular we use a combination of the state and exception monads in order to capture the learning-from-mistakes nature of interactive realizers. 
\end{itemize}

\frontmatter
\tableofcontents

\chapter{Preface}

\section{Proofs and Computations}

From the beginning intuitionistic logic has been linked to the idea of computation. 
In hindsight, this is already implicit in the Brouwer-Heyting-Kolmogorov (BHK) interpretation, which is presented in terms of proofs, constructions and transformations thereof (or problems in the case of Kolmogorov). 

The connection becomes more evident with the introduction of recursive realizability by Kleene in \cite{kleene45} and, later, modified realizability by Kreisel in \cite{kreisel59}. 
Realizability semantics can be thought of as formalizations of the BHK interpretation, where the vague notions of proof, construction and transformation are replaced with the notions of computable functionals. 



The full explicitation of this connection is the Curry-Howard correspondence, \cite{howard80}, where the whole proof is seen as a program and the conclusion as the type or the specification of the program. 
While interpreting an intuitionistic proof as a computation is quite natural (in hindsight), this is not the case for classical proofs. 

A computational interpretation of a classical proof can be obtained by first translating a classical proof into an intuitionstic one by means of double-negation translation. 
This approach was used by G\"odel to prove relative consistency results for classical and intuitionistic arithmetic.  
However, the double-negation translation transforms informative statements into non-informative ones, so the computations we can extract in this way yield trivial results. 
Moreover, this approach is indirect, while proofs and computations are almost undistinguishable in intuitionstic logic. 

We quote from \cite{sorensenU06}: 
\begin{quote}
Until around 1990 there was a widespread consensus to the effect that ``there is no Curry-Howard isomorphism for classical logic.'' However, at that time Tim Griffin made a path-breaking discovery which have convinced most critics that classical logics have something to offer the Curry-Howard isomorphism.
\end{quote}

In \cite{griffin90}, Griffin extends the Curry-Howard correspondence to classical proofs, employing functional programs with first-class continuations. 
In Griffin's own words:
\begin{quote}
  The programming language Scheme contains the control construct {\sf call/cc} that allows access to the current continuation (the current control context). This, in effect, provides Scheme with first-class labels and jumps. We show that the well-known formulae-as-types correspondence, which relates a constructive proof of a formula \(\alpha\) to a program of type \(\alpha\), can be extended to a typed Idealized Scheme. What is surprising about this correspondence is that it relates classical proofs to typed programs. 
\end{quote}

After Griffin's discovery, other interpretations extending the Curry-Howard correspondence to classical logic have been put forward. 
In \cite{parigot92}, Parigot introduces the \(\lambda\mu\)-calculus, an extension of lambda calculus with an additional kind of variables for subterms. 

In \cite{krivine94}, Krivine devised a new notion of realizability for classical logic called ``classical realizability''. 
In classical realizability realizers are written in an untyped lambda calculus with save/restore operators for the execution context and they are interpreted by an abstract machine that allows the manipulation of execution contexts, represented as ``stacks'' of arguments. 

Interactive realizability is a more recent proof interpretation for classical logic and the main focus of this dissertation. 

\section{Interactive Realizability}

Introduced by Berardi and de'Liguoro in \cite{berardidL08,berardidL09}, 
interactive realizability is a technique for understanding and extracting the computational content in the case of the sub-classical logic \(\HA+\EM\) (Heyting Arithmetic extended by the law of the excluded middle restricted to \(\Sigma^0_1\) formulas).

The main inspiration sources for interactive realizability are Coquand's game theoretic semantics for classical arithmetic and Gold's idea of learning in the limit. 

Gold original interest is language learnability, for instance a child learning the grammar of a language by repeated exposure to correct sentences. 
We expect that the child will eventually learn the language and stop making mistakes when speaking. 
The interesting point is that we do not know how many sentences he needs to complete the learning, 
In \cite{gold65}, he defines what it means to learn the answer of some question from an unlimited amount of evidence and in a finite time as follows:
\begin{quote}
  The purpose of this paper is to discuss the classes of problems that can be solved by infinitely long decision procedures in the following sense: An algorithm is given which, for any problem of the class, generates an infinitely long sequence of guesses. The problem will be said to be solved in the limit if, after some finite point in the sequence, all the guesses are correct and the same (in case there is more than one correct answer). 
\end{quote}

\newcommand\elo{\(\exists\)loise}
\newcommand\abe{\(\forall\)belard}
In \cite{coquand95}, Coquand presents a novel game theoretic semantics.
As customary in game semantics, each formula defines a game for two players: \elo, trying to show that the formula is true and \abe, trying to show that it is not. 
A formula is then validated by the existence of a winning strategy for \elo. 

Coquand takes the game for intuitionistic logic and extends it to classical logic by allowing \elo\ to retract her moves: 
instead of answering to the last move made by \abe, she can change her mind on her previous moves and go back to any past position. 
Thus a new game with asymmetric backtracking is defined, where \elo\ holds the advantage and the existence of a backtracking strategy validates classical logic. 

In \cite{berardidL08,berardidL09}, Berardi and de'Liguoro 
recast Coquand's idea of backtracking strategy as a strategy for learning the truth of classical statements in the limit in Gold's sense. 
Moreover, they present their proof interpretation 
as a realizability rather than game theoretic semantics
and write backtracking strategies as learning algorithms in a simply typed \(\lambda\)-calculus with primitive recursion. 


The aim of interactive realizability as a proof interpretation for classical logic is to express the computational content of a suitable subset of classical proofs in an \emph{understandable} form. 
This is the motivation for its peculiar features, which we summarize in the following. 
\begin{itemize}
  \item
The interactive interpretation is ``faithful'' to the classical proof, meaning that the computation follows closely the original proof. 
This is possible since interactive realizability interprets the proof directly, without resorting to proof translations. 
We also avoid adding computations that are not explicitly present in the proof, for instance blind searches to realize existential statements. 

\item
A common feature of computational interpretation of classical logic is that they extract programs that manipulate the execution context, that is, they need continuations.
However, the use of continuations can make a program hard to follow. 
Interactive realizability uses the idea of learning to explain the manipulations of the execution contexts that are needed to backtrack. In particular this is accomplished by means of a knowledge state, that is increased during the learning process and that act as a guide in the exectution of the interactive interpretation. 

\item
Interactive realizability is compositional, meaning that the interactive interpretations of different parts of a single proofs can be given independently and then composed to obtain the interactive interpretation of the whole proof. 

\item
  In this dissertation we only consider proofs where the law of the excluded middle is restricted to \(\Sigma^0_1\) formulas. 
  In this case interactive realizers use simpler constructs like states and exceptions instead of continuations in order to handle the backtracking nature of the computational content of classical proofs. 
\end{itemize}

\begin{omitted}
\section{A Philosophical Detour on the Relation Between Classical and Intuitionistic Logic}


  The relative consistency of classical and intuitionistic logic show that they reason about the same world. 
  They differ because they adopt different but consistent viewpoints, as illustrated by the following example. 

  \begin{quote}
  You leave your home in the morning and travel quite far. 
  At midday you are assailed by doubts: did you closed your door before leaving?
  You do not doubt that the door is either open or closed, 
  but you find little reassurance in this. 
  You want something more to decide on your course of action, you want to know which of these is the case.
  So you need to find some way of learning this knowledge, perhaps by trying to remember or going back to check the door. 
  Discovering this knowledge could be easy: you could phone your neighbor, whom you trust, and have him check.
  On the other hand it could be impossible to know: after you left an earthquake could have destroyed you home, erasing any evidence.  
  \end{quote}

  This example show both classical and intuitionistic reasoning in a non mathematical setting. 
  It highlights the difference between the classical and intuitionistic view of disjunction and in particular of the law of the excluded middle. 
  It shows the weakness of classical logic: the fact that you left your door either open or closed is not really helpful.
  In many situations you need more information than that. 
  But it also shows the strength of classical logic: you can wonder about whether you left your door open or closed, but would you even consider the fact that it could be neither open nor closed?

  In the last part we try to sneak in the idea that 
  learning could be the bridge from intuitionistic to classical reasoning. 

  Kripke semantics

  Plato
  Schrodinger cat

  Of course the real difference between classical and intuitionistic logic becomes only apparent when working with infinite or infinitely many objects. \fixme{elaborate}


  Incompleteness theorem guarantees that there are excluded middle instance that are not valid intuitionistically. \fixme{is this true?}

Historically, constructivism and particularly intuitionism have been regarded with suspicion by the mathematical community. 
The following quotation by Hilbert 
\begin{quote}
Taking the principle of excluded middle from the mathematician would be the same, say, as proscribing the telescope to the astronomer or to the boxer the use of his fists. To prohibit existence statements and the principle of excluded middle is tantamount to relinquishing the science of mathematics altogether.
\end{quote}
The Brouwer-Hilbert controversy is the origin of the idea that

The debate became fixed about whether 
mathematics proved by classical means was correct


  \section{History}
  Glivenko proved in 1929 that a propositional formula \(\fa\) is classically provable if and only \(\lnot\lnot\fa\) is intuitionistically provable \cite{troelstra11}.
  This result was generalized by G\"odel in 1933: he showed that any formula is provable in \(\PA\) if and only if its double-negation translation, obtained by adding a \(\lnot\lnot\) in front of each subformula, is provable in \(\HA\), 
  thus proving the relative consistency of classical and intuitionistic arithmetic.
  A further improvement is the fact that \(\PA\) is a conservative extension of \(\HA\) for \(\Pi^0_2\) formulas, that is, an arithmetic function is provably total in \(\PA\) if and only if it is provably total in \(\HA\). \fixme{who proved this? friedman in 78, but he was not the first}

  Classical and intuitionistic logic differ because they answer different questions. 
  They see same world, that is, they never contradict each other:
  \begin{itemize}
    \item if \(\CL \vdash \fa\) then \(\IL \not\vdash \lnot\fa\),
    \item if \(\CL \vdash \lnot\fa\) then \(\IL \not\vdash \fa\),
  \end{itemize}
  For instance, the law of the excluded middle \(\fa \lor \lnot \fa\) is never false in \(\IL\) and for some \(\fa\) it is provably true.

  Classical logic is concerned with truth, while intuitionistic logic is concerned with information: it stands to reason that the bridge between them should be learning. 


\end{omitted}
\begin{wip}
  \section{History and Structure of this Thesis}

  My first interest was understanding inter
  \Cref{chap:monadic_framework} contains the first 
  \Cref{chap:geometric_example}
  \Cref{chap:em_reduction}

\end{wip}

\mainmatter
\chapter{Preliminaries}

In this chapter we introduce the notation and the tools we shall in the rest of the thesis. 
\begin{changed}
  Our analysis is mainly concerned with proofs in first-order arithmetic, both intuitionistic (\(\HA\)) and classical (\(\PA\)).
  We also introduce a simply typed lambda calculus, which we use to give realizability interpretations of proofs. 
\end{changed}

\section{Constructive Arithmetic in Natural Deduction}
\label{sec:prelim_logic}

In this section, we introduce Heyting Arithmetic and the axiom for the law of the excluded middle, which will be the logical setting of the whole dissertation. 
We briefly describe the language of first-order logic, the rules of minimal logic in natural deduction, the axioms and rules of arithmetic and the restricted excluded middle axiom schemes. 

\newcommand\na{n}
\newcommand\LFunc{\mathcal{F}}
\newcommand\LRel{\mathcal{R}}
\newcommand\arity[1]{^{(#1)}}
\newcommand\lfa{f}
\newcommand\lfb{g}
\newcommand\lfc{h}
\newcommand\lra{p}
\newcommand\lrb{q}
\newcommand\lrc{r}
\newcommand\LVar{\mathcal{V}}
\newcommand\lva{x}
\newcommand\lvb{y}
\newcommand\lvc{z}
\newcommand\lta{t} 
\newcommand\ltb{u} 
\newcommand\ltc{v} 
\newcommand\lafa{P} 
\newcommand\lafb{Q} 
\newcommand\lafc{R} 
\newcommand\riA\alpha
\newcommand\riB\beta

\newcommand\lzero{\num 0}
\newcommand\nproj[2]{P^{#1}_{#2}} 
\newcommand\comp{comp} 
\newcommand\prRec{rec} 

\subsection{Primitive Recursive Functions and Predicates}

\begin{changed}
  In the language of arithmetic we include symbols for all the primitive recursive functions and predicates in arithmetic. 
  We briefly recall their definition. 

  We only consider arithmetical functions, that is, functions from the natural numbers to the natural numbers, which we denote with \(\N\). 
  These functions take \(n\) arguments for some natural number \(n\) and are called \(n\)-ary.
  We use the metavariables \( \lfa\arity{n}, \lfb\arity{n}, \lfc\arity{n} \) for \(n\)-ary functions.

  Primitive recursive functions are defined by induction. 
  The basic primitive recursive functions are the following:
  \begin{description}
    \item[the constant function] the 0-ary constant function \(\num 0\),
    \item[the successor function] the 1-ary successor function \(\succ\), which returns the successor of its argument,
    \item[the projection functions] for every \(n \ge 1\) and each \(i\) with \(1 \le i \le n\), the \(n\)-ary projection function \(\nproj{n}{i}\), which returns its \(i\)-th argument: 
      \[ 
        \nproj{n}{i} (\lva_1, \dotsc, \lva_n) = \lva_i. 
      \]
  \end{description} 
  More complex primitive recursive functions can be obtained by combining basic or previously defined primitive recursive function. 
  Given the primitive recursive functions \( \lfb\arity{n} \), \( \lfc_1\arity{m},\dotsc,\lfc_k\arity{m} \) and \( \lfc\arity{k+2} \), we can define new primitive recursive functions in two ways: 
  \begin{description}
    \item[composition] the composition of \(\lfb\) with \( \lfc_1,\dotsc,\lfc_n \), i.e.\ the \(m\)-ary function: 
      \[
        \comp(\lfb, \lfc_1, \dotsc, \lfc_n) = \lfb(\lfc_1(\lva_1,\ldots,\lva_m),\ldots,\lfc_n(\lva_1,\ldots,\lva_m)),
      \] 
      is primitive recursive; 
    \item[primitive recursion] the \((k+1)\)-ary function \(\prRec\lfb\lfc\) is defined as the primitive recursion of \(\lfb\) and \(\lfc\), i.e.\ the function: 
      \begin{align*}
        \prRec\lfb\lfc (0, \lva_1, \ldots, \lva_k) &= \lfb (\lva_1, \ldots, \lva_k), \\ 
        \prRec\lfb\lfc (\succ(\lvb), \lva_1, \ldots, \lva_k) &= \lfc (\lvb, \prRec (\lvb, \lva_1, \ldots, \lva_k), \lva_1, \ldots, \lva_k), 
      \end{align*}
      is primitive recursive. 
  \end{description} 

  Now we can define primitive recursive predicates by saying that they are the predicates whose characteristic function is a primitive recursive function. 
  More precisely an \(n\)-ary predicate \(\lra\) is primitive recursive if and only if there is a primitive recursive \(\lfa\arity{n}\) such that:
  \[ \lra(\lva_1, \dotsc, \lva_n) \text{ if and only if } \lfa(\lva_1, \dotsc, \lva_n) = 1, \]
  for any \(\lva_1, \dotsc, \lva_n \in \N\). 
\end{changed}

\subsection{The Language of Arithmetic}

In this subsection we define the language of first-order arithmetic. 

Let \(\{\LFunc_n\}_{n \in \N}\) and \(\{\LRel_n\}_{n \in \N}\) be two indexed sets of non-logical symbols. 
We assume that \(\LFunc_n\) contains symbols for all and only the primitive recursive (total) functions of arity \(n\), that is, functions in \(\N^n \to \N\). 
\begin{changed}
  Compare our language with the language of Primitive Recursive Arithmetic (PRA) \cite{skolem23}.
  Since we have induction on quantified formulas, unlike in PRA, 
  in principle we only need to define zero constant, the successor function, addition and multiplication (see \cite[p. 155]{mendelson97}). 
  However, we prefer a richer language with symbols for all the recursive functions and relations, because it is simpler to use. 
\end{changed}

Similarly we assume that \(\LRel_n\) contains symbols for all and only the primitive recursive relations of arity \(n\), that is, subsets of \(\N^n\). 
We use the metavariables \( \lfa\arity{n}, \lfb\arity{n}, \lfc\arity{n} \) for function symbols and \(\lra\arity{n}, \lrb\arity{n}, \lrc\arity{n}\) for relation symbols, omitting the superscript when we do not need it. 

The \(0\)-ary symbols are called \emph{constants}. 
We assume that some standard symbols are present: 
\[
  \begin{array}{|c|c|c|c|c|}
    \LFunc_0 & \LFunc_1 & \LFunc_2 & \LRel_0 & \LRel_2 \\ 
    \hline
    \num 0 & \succ & +, \cdot & \ltrue, \lfalse & =, <, \leq \\
    \hline
  \end{array}
\]
For the sake of readability we informally write \(\num n\) instead of \(\succ^n \num 0\) and we shall use the infix notation for binary functions and relations. 

Let \(\LVar\) be an enumerable set of variable symbols.
We use the metavariables \( \lva, \lvb, \lvc \) for variable symbols. 

We use the metavariable \(\lta\) for arithmetic terms, which are defined as: 
\[ \lta \Coloneqq \lva \mid \lfa\arity{n}(\lta_1, \dotsc, \lta_n) \]
We use the metavariables \(\lafa, \lafb, \lafc\) for atomic formulas, defined as: 
\[ \lafa \Coloneqq \lra\arity{n}(\lta_1, \dotsc, \lta_n) \] 
Finally, we use the metavariables \(\fa, \fb, \fc\) for (well formed) formulas, defined as: 
\[ \fa \Coloneqq \lafa \mid \fb \land \fc \mid \fb \lor \fc \mid \fb \limply \fc \mid \qforall\lva \fb \mid \qexists\lva \fb \] 
The entire grammar is given more concisely in \cref{fig:arithmetic_language}.

We write \(\clubsuit\subst{\lva_1, \dotsc, \lva_n}{\lta_1, \dots, \lta_n}\) for the \emph{simultaneous substitution} of the variables \( \lva_1, \dotsc, \lva_n \) with the terms \(\lta_1, \dots, \lta_n\) in the expression \(\clubsuit\) (a term or a formula). 

We use a compact notation for bounded quantification on natural numbers:
\begin{align*}
  \qforall{\lva \leq \lta}. \fa &\text{ stands for } 
  \qforall\lva \lva \leq \lta \limply \fa, \\ 
  \qexists{\lva \leq \lta}. \fa &\text{ stands for } 
  \qexists\lva \lva \leq \lta \land \fa. 
\end{align*}

The language of first-order arithmetic is the language of both Heyting Arithmetic and Peano Arithmetic. 

\begin{table}
  \centering
  \caption{The language of first-order arithmetic.}
  \label{fig:arithmetic_language}
  \begin{tabular}{c|c|c|}
    & Metavariables & Definition \\ 
    \hline
    function symbols & \( \lfa\arity{n}, \lfb\arity{n}, \lfc\arity{n} \) & elements of \(\LFunc_n\) \\
    \hline
    relation symbols & \( \lra\arity{n}, \lrb\arity{n}, \lrc\arity{n} \) & elements of \(\LRel_n\) \\
    \hline
    arithmetic variables  & \( \lva, \lvb, \lvc \) & elements of \(\LVar\) \\
    \hline
    arithmetic terms & \(\lta\) & \( \lta \Coloneqq \lva \mid \lfa \arity{n}(\lta_1, \dotsc, \lta_n) \) \\ 
    \hline
    atomic formulas & \(\lafa, \lafb, \lafc\) & \(\lafa \Coloneqq 
    \lra\arity{n}(\lta_1, \dotsc, \lta_n) \) \\ 
    \hline
    formulas & \(\fa, \fb, \fc\)  &  
    \(\begin{array}{c}\fa \Coloneqq 
      \lafa \mid \fb \land \fc \mid \fb \lor \fc \mid \\
    \mid \fb \limply \fc \mid \qforall\lva \fb \mid \qexists\lva \fb \end{array} \) \\ 
    \hline
  \end{tabular} 
\end{table}

\subsection{Reduction on Arithmetic Terms}

\begin{changed}
  In this subsection we introduce a reduction on arithmetic terms. 
  Since arithmetic terms are build from recursive primitive functions, we can transform the equations defining them into reductions in a natural way. 

  A term \(\lta\) is a \emph{numeral} if it is either \(\num 0\) or \(\succ[\ltb]\) for some numeral \(\ltb\). 
  We consider numerals as the basic arithmetic terms, so we do not reduce them. 

  Consider a term \(\lta\) build by composition, for instance:
  \[ 
    \lfa (\lta_1, \dotsc, \lta_i, \dotsc, \lta_n). 
  \]
  The first option is to reduce one of the arguments: 
  for any \( 0 < i \le n \), if \(\lta_i\) reduces to \(\lta'_i\), then
  \[ 
    \lfa (\lta_1, \dotsc, \lta_i, \dotsc, \lta_n) \reducesto 
    \lfa (\lta_1, \dotsc, \lta'_i, \dotsc, \lta_n). 
  \] 
  The second option is to reduce the whole term. 
  In this case which reduction we use depends on how the primitive recursive function denoted by \(\lfa\) is defined. 
  As we said, if \(\lfa\) denotes zero or the successor function, it does not reduce.
  If \(\lfa\) denotes a projection function \(\nproj{n}{i}\), then it reduces as follows:
  \[
    \lfa (\lta_1, \dotsc, \lta_n) \reducesto \lta_i, 
  \]
  exactly as in the definition of \(\nproj{n}{i}\). 
  If \(\lfa\) denotes the composition of functions whose symbols are \( \lfb, \lfc_1, \dotsc, \lfc_n\), then it reduces as follows: 
  \[
    \comp(\lfb, \lfc_1, \dotsc \lfc_n) = \lfb(\lfc_1(\lta_1,\ldots,\lta_m),\ldots,\lfc_n(\lta_1,\ldots,\lta_m)). 
  \] 
  If \(\lfa\) the primitive recursion of two other functions whose symbols are \(\lfb\) and \(\lfc\), then \(\lfa\) reduces as follows: 
  \begin{align*}
    \lfa (\num 0, \lta_1, \ldots, \lta_k) &\reducesto \lfb (\lta_1, \ldots, \lta_k), \\ 
    \lfa (\succ(\ltb), \lta_1, \ldots, \lta_k) &\reducesto \lfc (\ltb, \prRec (\ltb, \lta_1, \ldots, \lta_k), \lta_1, \ldots, \lta_k), 
  \end{align*}
  depending on the form of the first argument. 

  We basically described how to compute primitive recursive functions. 
  This reduction is strongly normalizing and the normal form is unique. 
  In the case of closed term the normal form is a numeral. 

  We can extend this reduction from terms to formulas, by reducing the terms contained in a formula. 
  Strong normalization and uniqueness are preserved. 
\end{changed}

\subsection{Axioms and Rules of Intuitionistic Logic}

In this subsection we describe natural deduction as a notation for formal proofs and the axioms and the rules of intuitionistic logic with equality. 

A \emph{derivation} is a formal diagram that describes a proof.
We write derivations in natural deduction, that is, as labeled trees of annotated formulas, 
with the requirement that the any subtree conforms to one of a number of patterns called \emph{rules of inference} or simply \emph{rules}.
An annotated formula consists of a formula, the name of the rule it conforms to and an unique label. 
In a proof tree the same formula can and often does appear in more than one place. Similarly a rule can be applied many times. 
In order to distinguish these multiple instance, when reasoning about the structure of a derivation we speak about \emph{formula occurrences} and \emph{rule instances}. 

We draw derivations with the root at the bottom. 
The top formulas of the tree are axiom instances if they are derived from a rule with no premisses, discharged assumption instances if they are discharged by a rule instance below them or open assumption instances if they are not. 

In an elimination rule the premiss containing the logical operator being eliminated is called the \emph{major} premiss; 
the other premisses are called \emph{minor} premisses. 
We follow a common convention in writing rules and place the major premiss in the leftmost position. 
For introduction and atomic rules we consider all premisses to be major premisses. 

We call a rule \emph{atomic} when its assumptions and conclusion are all atomic formulas.
We call \emph{atomic axiom} an atomic rule without premisses. 

There are two notations for natural deduction that differ on how they represent open assumptions. 
For instance, consider the implication introduction rule written in two ways:
\[
  \PrAss\fa\riA
  \PrInf 
  \PrUn\fb 
  \PrLbl[\riA]{\RuleName\limply{I}}
  \PrUn{\fa \limply \fb} 
  \DisplayProof \qquad
  \PrAx{\Gamma, \riA : \fa \seq \fb} 
  \PrLbl[\riA]{\RuleName\limply{I}}
  \PrUn{\Gamma \seq \fa \limply \fb} 
  \DisplayProof 
\]
We say that the leftmost rule is written in Gentzen's style and the rightmost one in sequent style.
In the Gentzen's style rule the open assumption \(\fa\) is inside square brackets with a superscript label \(\riA\) to show that it is discharged by the rule also labeled with \(\riA\). 
Note that labels are essential: if we remove them we may not know which rule instance discharges an occurrence of an open assumption. 

In the sequent style rule all the open assumptions a formula depends on are in the list (metavariable \(\Gamma\)) of couples of labels and formulas that precedes the symbol \(\seq\). 
The fact that the rule discharges an open assumption \(\fa\) is clear from the fact that \(\fa\) is in the open assumption list of the premiss and is not in the list of the conclusion. 

The simplest rule is the identity rule, that shows how to use an open assumption: they say that if we assumed a formula \(\fa\) we can derive \(\fa\), whence the name. When using Gentzen's style it is usually left implicit, but it needs to be written in sequent style. We give both versions:
\[ 
  \PrAss\fa\riA
  \PrLbl{\text{Id}}
  \PrUn\fa
  \DisplayProof, 
  \qquad
  \PrAx{}
  \PrLbl{\text{Id}}
  \PrUn{\Gamma \seq \fa}
  \DisplayProof, 
\]
where \(\Gamma\) contains \(\riA : \fa\) for some label \(\riA\).

We can now present the rules of minimal logic, a subset of intuitionistic and of classical logic. 
It consists of ten rules, one introduction rule and one elimination rule for each of the five logical connectives. \fixme{define introduction and elimination}
They are listed in Gentzen's style in \cref{fig:minimal_logic_gentzen} and in sequent style in \cref{fig:minimal_logic_sequent}. 
The usual restrictions apply to rules \(\RuleName\forall{I}\) and \(\RuleName\exists{E}\): 
in \(\RuleName\forall{I}\) we assume that \(\lva\) is not free in any assumption in \(\Gamma\); 
in \(\RuleName\exists{E}\) we assume that \( \lva \equiv \lvb \) or \(\lvb\) is not free in \(\fa\) and \(\lvb\) is not free in any assumption in \(\Gamma\). 

\begin{figure}[!ht]
  \caption{Rules for minimal logic, in Gentzen's style natural deduction.}
  \label{fig:minimal_logic_gentzen}
  \[ 
    \boxed{
      \begin{array}{ccc} 
        \\ 
        \PrAx\fa 
        \PrAx\fb 
        \PrLbl{\RuleName\land{I}} 
        \PrBin{\fa \land \fb} 
        \DisplayProof & \qquad & 
        \PrAx{\fa \land \fb} 
        \PrLbl{\RuleName[R]\land{E}}
        \PrUn\fa 
        \DisplayProof \quad
        \PrAx{\fa \land \fb} 
        \PrLbl{\RuleName[L]\land{E}}
        \PrUn\fb 
        \DisplayProof 
        \\ [2em]
        \PrAx\fa 
        \PrLbl{\RuleName[R]\lor{I}} 
        \PrUn{\fa \lor \fb} 
        \DisplayProof \quad 
        \PrAx\fb 
        \PrLbl{\RuleName[L]\lor{I}} 
        \PrUn{\fa \lor \fb} 
        \DisplayProof & & 
        \PrAx{\fa \lor \fb }
        \PrAss\fa\riA
        \PrInf
        \PrUn\fc 
        \PrAss\fb\riA
        \PrInf
        \PrUn\fc 
        \PrLbl\riA{\RuleName[L]\lor{E}}
        \PrTri\fc 
        \DisplayProof 
        \\ [2em]
        \PrAss\fa\riA
        \PrInf 
        \PrUn\fb 
        \PrLbl\riA{\RuleName\limply{I}}
        \PrUn{\fa \limply \fb} 
        \DisplayProof & & 
        \PrAx{\fa \limply \fb} 
        \PrAx\fa 
        \PrLbl{\RuleName\limply{E}}
        \PrBin\fb 
        \DisplayProof 
        \\ [2em]
        \PrAx\fa 
        \PrLbl{\RuleName\forall{I}}
        \PrUn{\qforall\lva \fa} 
        \DisplayProof & & 
        \PrAx{\qforall\lva\fa} 
        \PrLbl{\RuleName\forall{E}}
        \PrUn{\fa\subst\lva{t}} 
        \DisplayProof 
        \\ [2em]
        \PrAx{\fa\subst\lva{t}} 
        \PrLbl{\RuleName\exists{I}}
        \PrUn{\qexists\lva\fa} 
        \DisplayProof & & 
        \PrAx{\qexists\lva\fa} 
        \PrAss{\fa\subst\lva\lvb}\riA
        \PrInf
        \PrUn\fc
        \PrLbl\riA{\RuleName\exists{E}}
        \PrBin\fc 
        \DisplayProof 
        \\ [1em]
      \end{array} 
    }
  \]
\end{figure}

\begin{figure}[!ht] 
  \caption{Rules for minimal logic, in sequent style natural deduction.}
  \label{fig:minimal_logic_sequent}
  \[ 
    \boxed{
      \begin{array}{ccc} 
        \\ 
        \PrAx{\Gamma \seq \fa} 
        \PrAx{\Gamma \seq \fb} 
        \PrLbl{\RuleName\land{I}} 
        \PrBin{\Gamma \seq \fa \land \fb} 
        \DisplayProof 
        & &
        \PrAx{\Gamma \seq \fa \land \fb} 
        \PrLbl{\RuleName[R]\land{E}}
        \PrUn{\Gamma \seq \fa} 
        \DisplayProof 
        \qquad 
        \PrAx{\Gamma \seq \fa \land \fb} 
        \PrLbl{\RuleName[L]\land{E}}
        \PrUn{\Gamma \seq \fb} 
        \DisplayProof 
        \\ [2em]
        \PrAx{\Gamma \seq \fa} 
        \PrLbl{\RuleName[R]\lor{I}} 
        \PrUn{\Gamma \seq \fa \lor \fb} 
        \DisplayProof \qquad 
        & & 
        \PrAx{\Gamma \seq \fb} 
        \PrLbl{\RuleName[L]\lor{I}} 
        \PrUn{\Gamma \seq \fa \lor \fb} 
        \DisplayProof 
        \\ [2em]
        \multicolumn{3}{c}{
          \PrAx{\Gamma \seq \fa \lor \fb } 
          \PrAx{\Gamma, \alpha_{k+1} : \fa \seq \fc} 
          \PrAx{\Gamma, \alpha_{k+1} : \fb \seq \fc} 
          \PrLbl{\RuleName[L]\lor{E}}
          \PrTri{\Gamma \seq \fc} 
          \DisplayProof 
        }
        \\ [2em]
        \PrAx{\Gamma, \alpha_{k+1} : \fa \seq \fb} 
        \PrLbl{\RuleName\limply{I}}
        \PrUn{\Gamma \seq \fa \limply \fb} 
        \DisplayProof 
        & & 
        \PrAx{\Gamma \seq \fa \limply \fb} 
        \PrAx{\Gamma \seq \fa} 
        \PrLbl{\RuleName\limply{E}}
        \PrBin{\Gamma \seq \fb} 
        \DisplayProof
        \\ [2em]
        \PrAx{\Gamma \seq \fa} 
        \PrLbl{\RuleName\forall{I}}
        \PrUn{\Gamma \seq \qforall\lva \fa} 
        \DisplayProof 
        & & 
        \PrAx{\Gamma \seq \qforall\lva \fa} 
        \PrLbl{\RuleName\forall{E}}
        \PrUn{\Gamma \seq \fa\subst\lva{t}} 
        \DisplayProof 
        \\ [2em]
        \PrAx{\Gamma \seq \fa\subst\lva{t}} 
        \PrLbl{\RuleName\exists{I}}
        \PrUn{\Gamma \seq \qexists\lva\fa} 
        \DisplayProof 
        & & 
        \PrAx{\Gamma \seq \qexists\lva\fa} 
        \PrAx{\Gamma, \alpha_{k+1} : \fa\subst\lva\lvb \seq \fc} 
        \PrLbl{\RuleName\exists{E}}
        \PrBin{\Gamma \seq \fc} 
        \DisplayProof 
        \\ [1em]
      \end{array} 
    }
  \]
\end{figure}

Intuitionistic logic has all the rules of minimal logic with the addition of the \emph{ex falso quodlibet} rule:
\[ 
  \PrAx\lfalse
  \PrLbl{\RuleName\lfalse{E}}
  \PrUn\fa
  \DisplayProof 
\]
which can be thought of as an elimination rule for the atomic formula \(\lfalse\). 
For technical reason we prefer to have atomic rules when possible, so instead of the \(\RuleName\lfalse{E}\) rule we consider its restricted version, where the conclusion can only be an atomic formula:
\[ 
  \PrAx\lfalse
  \PrLbl{\RuleName[0]\lfalse{E}}
  \PrUn\lafa
  \DisplayProof 
\]
Then \(\RuleName[0]\lfalse{E}\) is an atomic rule. 
The \(\RuleName\lfalse{E}\) rule is admissible given the \(\RuleName[0]\lfalse{E}\) rule. 
It is easy to prove by induction on the structure of the conclusion of the \(\RuleName\lfalse{E}\) rule. 
For instance we can prove \(\fa = \lafa \limply \lafb\) from \(\lfalse\):
\[
  \PrAx\lfalse
  \PrLbl{\RuleName[0]\lfalse{E}}
  \PrUn\lafb
  \PrLbl{\RuleName\limply{I}}
  \PrUn{\lafa \limply \lafb}
  \DisplayProof 
\]

First-order logic always assumes the existence of a binary relation symbol \(=\) and axioms and rules defining it as an equivalence relation compatible with functions and relations. 
These axioms and rules are given in Gentzen's style in \cref{fig:equality_rules_gentzen} and in sequent style in \cref{fig:equality_rules_sequent}. 
Note that they are all atomic. 

\begin{figure}[!ht] 
  \caption{Rules for the equality predicate, in Gentzen's style natural deduction.}
  \label{fig:equality_rules_gentzen}
  \[
    \boxed{
      \begin{array}{ccc} 
        \PrAx{} 
        \PrLbl{\mathrm{Refl}} 
        \PrUn{\lta = \lta} 
        \DisplayProof \qquad
        \PrAx{\lta = \ltb} 
        \PrLbl{\mathrm{Sym}} 
        \PrUn{\ltb = \lta} 
        \DisplayProof \qquad 
        \PrAx{\lta = \ltb} 
        \PrAx{\ltb = \ltc} 
        \PrLbl{\mathrm{Trans}} 
        \PrBin{\lta = \ltc} 
        \DisplayProof \\ [1em]
        \PrAx{\lta_1 = \ltb_1 \quad \dotso \quad \lta_n = \ltb_n} 
        \PrLbl{\RuleName[\LFunc]{}{Sub}}
        \PrUn{\lfa\arity{n}(\lta_1, \dotsc, \lta_n) = \lfa\arity{n}(\ltb_1, \dotsc, \ltb_n)} 
        \DisplayProof \\ [1em]
        \PrAx{\lta_1 = \ltb_1 \quad \dotso \quad \lta_n = \ltb_n} 
        \PrAx{\lra\arity{n}(\lta_1, \dotsc, \lta_n)} 
        \PrLbl{\RuleName[\LRel]{}{Sub}}
        \PrBin{\lra\arity{n}(\ltb_1, \dotsc, \ltb_n)} 
        \DisplayProof \\ 
      \end{array}
    }
  \]
\end{figure}

\begin{figure}[!ht] 
  \caption{Rules for the equality predicate, in sequent style natural deduction.}
  \label{fig:equality_rules_sequent}
  \[
    \boxed{
      \begin{array}{ccc} 
        \PrAx{} 
        \PrLbl{\mathrm{Refl}} 
        \PrUn{\Gamma \seq \lta = \lta} 
        \DisplayProof \qquad
        \PrAx{\Gamma \seq \lta = \ltb} 
        \PrLbl{\mathrm{Sym}} 
        \PrUn{\Gamma \seq \ltb = \lta} 
        \DisplayProof \qquad 
        \PrAx{\Gamma \seq \lta = \ltb} 
        \PrAx{\Gamma \seq \ltb = \ltc} 
        \PrLbl{\mathrm{Trans}} 
        \PrBin{\Gamma \seq \lta = \ltc} 
        \DisplayProof \\ \\ 
        \PrAx{\Gamma \seq \lta_1 = \ltb_1 \quad \dotso \quad \lta_n = \ltb_n} 
        \PrLbl{\RuleName[\LFunc]{}{Sub}}
        \PrUn{\Gamma \seq \lfa\arity{n}(\lta_1, \dotsc, \lta_n) = \lfa\arity{n}(\ltb_1, \dotsc, \ltb_n)} 
        \DisplayProof \\ \\ 
        \PrAx{\Gamma \seq \lta_1 = \ltb_1 \quad \dotso \quad \lta_n = \ltb_n} 
        \PrAx{\Gamma \seq \lra\arity{n}(\lta_1, \dotsc, \lta_n)} 
        \PrLbl{\RuleName[\LRel]{}{Sub}}
        \PrBin{\Gamma \seq \lra\arity{n}(\ltb_1, \dotsc, \ltb_n)} 
        \DisplayProof \\ \\ 
      \end{array}
    }
  \]
\end{figure}

\subsection{Axiom and Rules of Arithmetic}

In this subsection we present the axioms and the rules of Heyting Arithmetic (\(\HA\)).

The rules defining the functions symbols \(\succ, +\) and \(\cdot\) are in \cref{fig:arithmetic_rules}. 

\begin{figure}[!ht] 
  \caption{Axioms and rules for the successor, addition and multiplication, in Gentzen's style natural deduction.}
  \label{fig:arithmetic_rules}
  \[
    \boxed{
      \begin{array}{ccc} 
        \PrAx{\succ \lta = \num 0}
        \PrLbl{\mathrm{Zero}} 
        \PrUn\lfalse 
        \DisplayProof 
        & \qquad &
        \PrAx{\succ \lta = \succ(\ltb)}
        \PrLbl{\mathrm{Succ}} 
        \PrUn{\lta = \ltb} 
        \DisplayProof 
        \\ [2em]
        \PrAx{}
        \PrLbl{\RuleName[0]{}{Add}} 
        \PrUn{\lta + \num 0 = \lta}
        \DisplayProof 
        &&
        \PrAx{}
        \PrLbl{\RuleName[\succ]{}{Add}} 
        \PrUn{\lta + \succ(\ltb) = \succ(\lta) + \ltb}
        \DisplayProof 
        \\ [2em]
        \PrAx{}
        \PrLbl{\RuleName[0]{}{Mult}} 
        \PrUn{\lta \cdot \num 0 = \num 0}
        \DisplayProof 
        &&
        \PrAx{}
        \PrLbl{\RuleName[+]{}{Mult}} 
        \PrUn{\lta \cdot \succ(\ltb) = \lta \cdot \ltb + \lta}
        \DisplayProof 
      \end{array} 
    }
  \]
\end{figure}

\subsubsection{Induction}
The last rule we need to add to have the full \(\HA\) is induction. 
Induction can be thought of as the requirement that any natural number can be written as \(\succ(\dotso( \succ (\num 0)))\). 
In first-order logic is usually expressed by saying that, for any formula \(\fa\), if \(\fa(0)\) holds and \(\fa(x+1)\) holds whenever \(\fa(x)\) holds, then \(\fa\) holds for any natural number. 
Induction has many different but equivalent formulations and it can be written as either an axiom or a rule.
The most common formulation is the following induction axiom schema:
\[ \fa\subst\lva{\num 0} \land (\qforall\lva \fa \limply \fa\subst\lva{\succ(\lva)}) \limply \qforall\lva \fa, \] 
where \(\fa\) is any formula. 
This axiom can be written as the \emph{induction rule} in Gentzen's style:
\[ 
  \PrAx{\fa\subst\lva{\num 0}}
  \PrAss\fa\riA
  \PrInf
  \PrUn{\fa\subst\lva{\succ(\lva)}}
  \PrLbl[\riA]{\text{Ind}}
  \PrBin{\qforall\lva \fa}
  \DisplayProof,
\]
or in sequent style:
\[
  \PrAx{\Gamma \seq \fa\subst\lva{\num 0}}
  \PrAx{\Gamma, \riA : \fa \seq \fa\subst\lva{\succ(\lva)}}
  \PrLbl[\riA]{\text{Ind}}
  \PrBin{\Gamma \seq \qforall\lva \fa}
  \DisplayProof.
\]

An related axiom is \emph{complete} or \emph{course-of-values induction}, which states that if \(\fa(\lva)\) holds whenever \(\fa(\lvb)\) holds for all \(\lvb < \lva\) then \(\fa\) holds for any natural number. 
The axiom for complete induction is written as: 
\[ (\qforall{\lvb < \lva} \fa\subst\lva\lvb) \limply \qforall\lva \fa \]
While this axiom appear to be stronger than the induction axiom we just defined, it is actually equivalent. 
\begin{changed}
  This can be seen by considering the standard induction rule for the formula 
  \[ \fb \equiv \qforall{\lvb<\lva} \fa\subst\lva\lvb. \]
  For more details see \cite[p. ~213]{shoenfield67}. 
\end{changed}

We can write a rule for complete induction in Gentzen's style: 
\[ 
  \PrAss{\qforall{\lvb<\lva} \fa\subst\lva\lvb}\riA
  \PrInf
  \PrUn\fa 
  \PrLbl[\riA]{\mathrm{CInd}}
  \PrUn{\qforall\lva \fa} 
  \DisplayProof, 
\]
or in sequent style: 
\[
  \PrAx{\Gamma, \riA : \qforall{\lvb<\lva} \fa\subst\lva\lvb \seq \fa} 
  \PrLbl[\riA]{\mathrm{CInd}}
  \PrUn{\Gamma \seq \qforall\lva \fa} 
  \DisplayProof. 
\]

Now we have all the ingredients we need in order to define Heyting Arithmetic, the standard intuitionistic theory of arithmetic. 
\begin{definition}[Heyting Arithmetic]\label{def:HA}
  Heyting Arithmetic (\(\HA\)) is defined as the first-order logic theory whose language is the language of arithmetic and whose axioms and rules are the following:
  \begin{itemize}
    \item the identity rule, 
    \item the rules of minimal logic,
    \item the rule for ex falso quodlibet restricted to atomic formulas,
    \item the axiom and rules for equality,
    \item the axiom and rules for successor, addition and multiplication,
    \item any one of the axioms and rules for induction or complete induction.
  \end{itemize}
\end{definition}

\subsection{Axioms for the Law of the Excluded Middle}\label{sec:em}

In this subsection we introduce a hierarchy of axiom schemes that are restrictions of the law of the \emph{excluded middle}, taken from \cite{akamaBHK04}.
\begin{changed}
  We refer to the same work for explanations and proofs of our claims in this subsection.
\end{changed} 

We define a purely syntactical version of the usual classification for formulas in prenex normal form in arithmetic.
\begin{definition}[Syntactical Arithmetical Hierarchy]
  We define the following classes of formulas by induction on \(\na\): 
  \begin{itemize}
    \item \(\Pi^0_0\) and \(\Sigma^0_0\) are the set of the quantifier free formulas,
    \item \(\Pi^0_{\na+1}\) is the set of the formulas \(\qforall\lva \fa\) where \(\fa \in \Sigma^0_\na\), 
    \item \(\Sigma^0_{\na+1}\) is the set of the formulas \(\qforall\lva \fa\) where \(\fa \in \Pi^0_\na\).
  \end{itemize}
\end{definition}
We do not require closure for logical equivalence and thus the definition of \(\Pi^0_\na\) and \(\Sigma^0_\na\) is purely syntactical. 
\newcommand\ldual[1]{{#1}^\lnot}
Since the negation of a formula in prenex normal form is not in prenex normal form we define the following negative conversion, which we call the \emph{dual} of a formula:
\begin{align*}
  \ldual\fa \equiv \begin{cases}
    \qexists\lva \ldual\fb & \text{ if } \fa = \qforall\lva \fb \in \Pi^0_\na, \\
    \qforall\lva \ldual\fb & \text{ if } \fa = \qexists\lva \fb \in \Sigma^0_\na, \\
    \lnot\fa & \text{ if } \fa \in \Pi^0_0.
  \end{cases}
\end{align*}
Note that if \(\fa\) is a \(\Pi^0_\na\) (resp. \(\Sigma^0_\na\)) formula then \(\ldual\fa\) is a \(\Sigma^0_\na\) (resp. \(\Pi^0_\na\)) formula. 
Moreover \(\ldual\fa\) is classically equivalent to \(\lnot\fa\) and intuitionistically stronger than \(\lnot\fa\), except when \(\fa\) is a \(\Pi^0_0\) formula, in which case it is equivalent to \(\lnot\fa\). 

The law of the excluded middle says that any statement is either true or false. 
More precisely, it says that 
either a statement is true or its negation is true. 
This is intuitionistically equivalent to the fact that 
either a statement is true or its dual is true. 
We define a sequence of restricted forms of this law. 
\begin{definition}[Restricted Excluded Middle Axiom Schemas] \label{def:excluded_middle}
  For any \(n \in \Nat\), we define the limited law of the excluded middle \(\EMG_\na\) as the axiom schema:
  \[
    \tag{\(\EMG_\na\)} \ldual\fa \lor \fa, 
  \]
  where \(\fa\) is a \(\Sigma^0_\na\) formula. 
  As a limit case we define \(\EMG_\infty\) where \(\fa\) is a \(\Pi^0_\na\) formula for any \(\na\).  
\end{definition}

By adding \(\EMG_\infty\) to intuitionistic logic we get classical logic, so by adding \(\EMG_\infty\) to Heyting Arithmetic we get Peano Arithmetic, the theory of classical arithmetic.  
We can also produce many intermediate logics by adding \(\EMG_\na\) to Heyting Arithmetic, which we write as \(\HA+\EMG_\na\). 
Note that \(\EMG_0\) is true in Heyting Arithmetic, so \(\HA+\EMG_0\) is simply \(\HA\). 


\section{A Simply Typed $\lambda$-Calculus for Realizability}


\newcommand\ta{X}
\newcommand\tb{Y}
\newcommand\tc{Z}
\newcommand\tta{x} 
\newcommand\ttb{y} 
\newcommand\ttc{z} 
\newcommand\tva{x} 
\newcommand\tvb{y} 
\newcommand\tvc{z} 
\newcommand\ra{p}
\newcommand\rb{q}
\newcommand\rr{r}

In this section we introduce system \(T'\), a simply typed \(\lambda\)-calculus variant of G\"odel's system \(T\) in which we shall write our realizers.
System \(T'\) will be more convenient for our purposes in order to get a more straightforward translation of monads and related concepts from category theory. 
There are two main differences between our system \(T'\) and system \(T\). 
The first one is that we replace the boolean type with the more general sum (or co-product) type. 
The second one is that the recursion operator uses complete recursion instead of standard primitive recursion. 

We begin by defining the types. 
We shall use the metavariables \(\ta,\tb\) and \(\tc\) for types. 
We assume that we have a finite set of atomic types that includes the unit type \(\Unit\) and the type of natural numbers \(\Nat\). 
Moreover we have three type binary type constructors \(\tarrow, \times, +\). 
In other words, for any types \(\ta\) and \(\tb\) we have the arrow (or function) type \(\ta \tarrow \tb\), the product type \(\ta \times \tb\) and the sum (or co-product) type \(\ta + \tb\). 

We can now define the typed terms of the calculus. 
We assume that we have a countable set of typed term constants that includes the constructors and the destructors for the unit, natural, product and sum types (listed in \cref{fig:constant_terms}) and  
a countable set of variables of type \(\ta\) for any type \(\ta\): 
\[ \tva_0 : \ta, \dotsc, \tva_n : \ta, \dotsc. \]
\begin{figure}
  \caption{Constructors and destructors}\label{fig:constant_terms} 
  \begin{gather*} 
    \unit : \Unit, \\ 
    \pair[\ta,\tb] : \ta \tarrow \tb \tarrow \ta \times \tb, \\ 
    \begin{aligned} 
      \prl[\ta,\tb] &: \ta \times \tb \tarrow \ta, \qquad 
      &\prr[\ta,\tb] : \ta \times \tb \tarrow \tb, \\ 
      \inl[\ta,\tb] &: \ta \tarrow \ta + \tb, \qquad
      &\inr[\ta,\tb] : \tb \tarrow \ta + \tb, 
    \end{aligned} \\ 
    \case[\ta,\tb,\tc] : \ta + \tb \tarrow (\ta \tarrow \tc) \tarrow (\tb \tarrow \tc) \tarrow \tc, \\ 
    \zero : \Nat, \qquad \succ : \Nat \tarrow \Nat, \\
    \totRec[\tc]{n}  : (\Nat \tarrow (\Nat \tarrow \tc) \tarrow \tc) \tarrow \Nat \tarrow \tc. 
  \end{gather*} 
  \begin{legend}
    where \(n\) is a natural number or the symbol \(\infty\). 
    In order we have the constant constructor of type \(\Unit\), the constructor and the two destructors of the product types, the two constructors and the destructor of the sum types and the two constructors and the destructor of the natural type.
    Most of those are actually ``parametric polymorphic'' terms, that is, families of constants indexed by the types \(\ta,\tb\) and \(\tc\).
  \end{legend}
\end{figure}
We use the metavariables \(\tta, \ttb, \ttc\) for terms. 
Moreover for any two terms \(\tta : \ta\) and \(\ttb : \ta \tarrow \tb\) we have a term \(\ttb \tta : \tb\) and for any variable \(\tva : \ta\) and term \(\ttb : \tb\) we have a term \(\qlambda\tva \ttb : \ta \tarrow \tb\). 


In order to avoid a parenthesis overflow, we shall follow the usual conventions for writing terms and types. 
For terms this means that application and abstraction are respectively left and right-associative 
and that abstraction binds as many terms as possible on its right; 
for types it means that \(\times\) and \(+\) are left-associative and associate more closely than \(\tarrow\), which is right-associative. 
We also omit outer parenthesis. 
For example: 
\[ 
  \begin{array}{ccc} 
    \ta \tarrow \tb \tarrow \ta \times \tb \times \tc & \text{ stands for } & (\ta \tarrow (\tb \tarrow ((\ta \times \tb) \times \tc))), \\ 
    \abstr{x}\ta \abstr{y}\tb \abstr{z}\tc t_1 t_2 t_3 & \text{ stands for } & (\abstr\lva\ta (\abstr{y}\tb (\abstr{z}\tc ((t_1 t_2) t_3)))). 
  \end{array} 
\] 

We define some reduction relations, that is, binary relations between terms:
\begin{gather*} 
  (\abstr\lva\ta t) a \reducesto_\beta t\subst\lva{a}, \\ 
  \begin{aligned} 
    \prl[\ta,\tb] (\pair[\ta,\tb]{a}{b}) &\reducesto_\times a, 
    &\case[\ta,\tb,\tc] (\inl[\ta,\tb] a) f g &\reducesto_+ f a, \\
    \prr[\ta,\tb] (\pair[\ta,\tb]{a}{b}) &\reducesto_\times b, 
    &\case[\ta,\tb,\tc] (\inr[\ta,\tb] b) f g &\reducesto_+ g b, 
  \end{aligned} \\
  \totRec[\tc]{n} h \num m \reducesto_R \begin{cases} 
    h \num m (\totRec[\tc]{m} h) &\text{if } m < n \text{ or } n = \infty, \\ 
    \dummy[\tc] &\text{otherwise,} 
  \end{cases} 
\end{gather*} 
where \( a : \ta\), \(b : \tb\), \(c : \tc\), \(f : \ta \tarrow \tc\), \(g : \tb \tarrow \tc \) and \( h : \Nat \tarrow (\Nat \tarrow \tc ) \tarrow \tc \). 
Note that we use \(c\) as a dummy term of type \(\tc\)\footnotemark. 
\footnotetext{
  As long as the base types are inhabited, we can define an arbitrary dummy term \(\dummy[\ta]\) for any type \(\ta\):
  \begin{gather*}
    \dummy[\Unit] \equiv \unit, \qquad \dummy[\Nat] \equiv \num 0, \\ 
    \dummy[\ta \tarrow \tb] \equiv \abstr\_\ta \dummy[\tb], \quad
    \dummy[\ta \times \tb] \equiv \pair \dummy[\ta] \dummy[\tb], \quad 
    \dummy[\ta + \tb] \equiv \inl \dummy[\ta]. 
  \end{gather*}
}
Then we set the reduction relation \(\reducesto\) to be the union of \( \reducesto_\beta, \reducesto_\times, \reducesto_+\) and \(\reducesto_R\).
Also let \(\leadsto\) be the transitive and reflexive closure of \(\reducesto\). 

  We explain the reduction given for \(\totRec{}\), since it is not the standard one. The difference is due to the fact that \(\totRec{}\) is meant to realize complete induction instead of standard induction. 
In complete induction, the inductive hypothesis holds not only for the immediate predecessor of the value we are considering, but also  for all the smaller values. 

Similarly, \(\totRec{}\) allows us to recursively define a function \(f\) where the value of \(f(\num m)\) depends not only on the value of \(f(\num{m-1})\) but also on the value of \(f(\num l)\), for any \(l < m \). 
Thus, when computing \(\totRec[\tc]{n} h \num m\), instead of taking the value of \(\totRec[\tc]{n} h \num{(m-1)}\) as an argument, \(h\) takes the whole function \(\totRec[\tc]{n} h\). 
In order to avoid unbounded recursion, we add a guard \(n\) that prevents \(\totRec[\tc]{n} h\) to be computed on arguments greater or equal to \(n\). 
More precisely \(\totRec[\tc]{n} h \num m \) only reduces to \(h \num m (\totRec[\tc]{m} h) \) if \( m < n \); thus, even if \(h\) requires \(\totRec[\tc]{m} h\) to be computed on many values, the height of the computation trees is bound by \(m\)\footnote{Unlike in standard primitive recursion, where the computation always comprises \(m\) steps, in course-of-values primitive recursion the computation can actually be shorter if \(h\) ``skips'' values.}. 
Naturally, a ``good'' \(h\) should not evaluate \(\totRec[\tc]{m} h\) on values bigger than \(m\), but in any case the guard guarantees termination.
The symbol \(\infty\) acts as a dummy guard, which gets replaced with an effective one when \(\totRec[\tc]{\infty} h\) is evaluated the first time. 

We shall also need the following equivalence relations between terms: 
\begin{align*} 
  \tag{\(\alpha\)-conversion} \abstr{x}\ta t &=_\alpha \abstr{y}\ta t[x \coloneqq y], \\ 
  \tag{\(\eta\)-conversion} \abstr{x}\ta t x &=_\eta t, \\ 
  \tag{\(\times\)-conversion} \pair[\ta,\tb](\prl[\ta,\tb] c)(\prr[\ta,\tb] c) &=_\times c, \\ 
  \tag{\(+\)-conversion} \case[\ta,\tb,\ta+\tb] d \inl[\ta,\tb] \inr[\ta,\tb] &=_+ d, 
\end{align*} 
for all terms \(t\), \( c : \ta \times \tb \) and \( d : \ta + \tb \). 
Again we set the equivalence \(=\) to be the union of \( =_\alpha, =_\eta, =_\times \) and \(=_+\). 

It is easy to see that the boolean type and the related terms \(\true,\false\) and \(\ite\) in system \(T\) can be defined in system \(T'\). The reverse is also true, that is, we can define sum types and terms \(\inl,\inr\) and \(\case\) in system \(T\). 
\begin{thesis}
  We show how to interpret the boolean type and related constants of system \(T\) in our system. 
  \begin{gather*}
    \Bool \equiv \Unit + \Unit, \\ 
    \true \equiv \inl \unit \qquad 
    \false \equiv \inr \unit, \\ 
    \ite[\ta] \equiv \abstr{b}\Bool\abstr\tta\ta \abstr\tta\ta \case b (\abstr\_\Unit \tta) (\abstr\_\Unit \tta). 
  \end{gather*}
  We assume that \(\ta\) and \(\tb\) are inhabited, that is, that there exist terms \(\tta_0 : \ta \) and \(\ttb_0 :\tb\). Then we can interpret the sum type and related constants in system \(T\):
  \begin{gather*}
    \ta+\tb \equiv \Bool \times (\ta \times \tb), \\ 
    \inl[\ta,\tb] \equiv \abstr\tta\ta \pair \true (\pair \tta \ttb_0), \\
    \inr[\ta,\tb] \equiv \abstr\ttb\tb \pair \false (\pair \tta_0 \ttb), \\
    \case[\ta,\tb,\tc] \equiv 
    \abstr{b}{\Bool \times (\ta \times \tb)} \abstr{f}{\ta \tarrow \tc} \abstr{g}{\tb \tarrow \tc}
    \ite (\prl b) (f (\prl (\prr b))) (g (\prr (\prr b))).
  \end{gather*}
\end{thesis}

\begin{wip}
  We can easily generalize our recursion operator to work on arbitrary well-founded types. 
  Assume that \(<_\ta : \ta \tarrow \ta \tarrow \Bool\) is a well-founded ordering on an inhabited type \(\ta\) and choose any term \(0_\ta : \ta\) as a default value. 
  \begin{align*} 
    \totRec[\ta,\tb]\lva : (\ta \tarrow (\ta \tarrow \tb) \tarrow \tb) \tarrow \ta \tarrow \tb \\ 
    \totRec[\ta,\tb]\infty h \lva \reducesto \ite <_\ta 
    \totRec[\ta,\tb]\lva h \lvb \reducesto \ite (<_\ta \lvb \lva) (h \lvb (\totRec[\ta,\tb]\lvb) (0_\ta)
  \end{align*} 
  where \( h : \ta \tarrow \tb \). 
\end{wip}

System \(T'\) shares most of the good properties of G\"odel's system \(T\), in particular confluence, strong normalization\footnote{Strong normalization is a consequence of the explicit bound on recursion given by the subscript in the recursion constant.} and a normal form property. 
%

\chapter{A Monadic Framework for Interactive Realizability}
\newcommand\MM{\mon M}
\newcommand\IdM{\mon Id}
\newcommand\ExM{\mon Ex}
\newcommand\IR{\mon IR}

\def\monUnit{}
\def\monStar{}
\def\monMerge{}
\newcommand\setmonad[1]{%
  \def\mm##1{\lVert ##1 \rVert\ifthenelse{\equal{#1}{}}{}{_{#1}}}%
  \def\m##1{\lvert ##1 \rvert\ifthenelse{\equal{#1}{}}{}{_{#1}}}%
  \def\monTrans##1{\llbracket ##1 \rrbracket\optionalSubscript{#1}}%
  \def\monRe{\mathrel{\mon R}\ifthenelse{\equal{#1}{}}{}{_{#1}}}%
  \def\re{\mathrel{\mathtt R}\ifthenelse{\equal{#1}{}}{}{_{#1}}}%
  \def\T{T\optionalSubscript{#1}}%
  \renewcommand\monUnit[1][]{\operatorname{\mon{unit}}\optionalSubscript{#1}\optionalSuperscript{##1}}%
  \renewcommand\monStar[1][]{\operatorname{\mon{star}}\optionalSubscript{#1}\optionalSuperscript{##1}}%
  \renewcommand\monMerge[1][]{\operatorname{\mon{merge}}\optionalSubscript{#1}\optionalSuperscript{##1}}%
  \def\monSeq{\Vdash\optionalSubscript{#1}}
}

In this chapter we give a new presentation of interactive realizability with a more explicit syntax. 

Monads can be used to structure functional programs by providing a clean and modular way to include impure features in purely functional languages. 
We express interactive realizers by means of an abstract framework that applies the monadic approach used in functional programming to modified realizability, in order to obtain more ``relaxed'' realizability notions that are suitable to classical logic. 
In particular we use a combination of the state and exception monads in order to capture the learning-from-mistakes nature of interactive realizers at the syntactic level.

\section{Introduction}




\begin{omitted}
While constructive and in particular intuitionistic proofs have a computational content by design, this is not the case for proofs in classical logic. 
In general it is not possible to give a computational content to any classically provable statement: for instance, 
giving a computational interpretation of the classical law of the excluded middle that satisfies the BHK interpretation means being able to effectively decide for any statement whether it holds or not, which leads to contradictions in the case of many statements. 
On the other hand saying that classical logic has no computational content at all is unreasonable: 
for instance, given a classically provable statement we can say that its computational content is that of an intuitionistic proof of its double-negation translation. 
The problem with this approach is that the double-negation translation 
transforms informative statements into negative, non-informative ones.
Thus, relating the original statement with the computational content of its translation is not usually straightforward. 
\end{omitted}

As we have already remarked in the preface, the Curry-Howard correspondence was originally discovered for intuitionistic proofs. 
This is not coincidental: 
the type systems needed to interpret intuitionistic proofs are usually very simple and natural, as in the case of Heyting Arithmetic and System T (see \cite{girard88}). 
While classical proofs can be transformed into intuitionistic ones by means of the double-negation translation and then translated into typed programs, 
the existence of a direct correspondence was deemed unlikely until Griffin showed otherwise in \cite{griffin90}. 

Starting with Griffin's, other interpretations extending the Curry-Howard correspondence to classical logic have been put forward. 
Griffin uses a ``typed Idealized Scheme'' with the control construct {\sf call/cc}, that allows access to the current continuation. 
In \cite{parigot92}, Parigot introduces the \(\lambda\mu\)-calculus, an extension of lambda calculus with an additional kind of variables for subterms. 
In \cite{krivine94}, Krivine uses lambda calculus with a non-standard semantics, described by an abstract machine that allows the manipulation of ``stacks'', which can be thought of as execution contexts. 

All these different approaches seems to suggest that, in order to interpret classical logic, we need control operators or some syntactical equivalent thereof. 
This could be generalized in the idea that ``impure'' computational constructs are needed in order to interpret non-constructive proofs. 
Monads are a concept from category theory that has been widely used in computer science. 
In particular, they can be used to structure functional programs that mimic the effects of impure features. 

In \cite{moggi91}, Moggi advocates the use of monads as a framework to describe and study many different ``notions of computation'' in the context of categorical semantics of programming languages. 
A different take on the same idea that actually eschews category theory completely is suggested in \cite{wadler92} by Wadler: 
the definition of monad becomes purely syntactic and is used as a framework to structure functional programs by providing a clean and modular way to include impure features in purely functional languages (one noteworthy example is I/O in Haskell). 

The main idea of this work is to use monads as suggested by Wadler in order to structure programs extracted from classical proofs by interactive realizability. 
A program extracted by means interactive realizability, called interactive realizer, can be thought of as a learning process. 
It accumulate information in a knowledge state and use this knowledge in order to ``decide'' the instances of \(\EM\) used in the proof. 
Since these instances are in general undecidable, the realizer actually makes an ``educated guess'' about which side of an \(\EM\) instance is true by looking at the state. 
Such guesses can be wrong. 

This can become apparent later in the proof, when the guessed side of the \(\EM\) instance is used to deduce some decidable statement. 
If this decidable statement turns out to be false, then the guess was wrong and the proof cannot be completed. 
In this case the realizer cannot produce the evidence required for the final statement and fails. 
However failure is due to the fact that we made a wrong guess. 
We can add this information to the state, so that, using this new state, we will be able to guess the \(\EM\) instance correctly. 
At this point we discard the computation that occurred after the wrong guess and we resume from there. 
This time we guess correctly and can proceed until the end or until we fail again because we guessed incorrectly another \(\EM\) instance. 

There are three ``impure'' parts in the behavior we described: the dependency on the knowledge state, the possibility of failure to produce the intended result and the backtracking after the failure.
In this work we use a monadic approach to describe the first two parts which are peculiar to interactive realizability. 
We do not describe the third part, which is common also to the other interpretations of classical logic. 

\begin{omitted}
Interactive realizability is based on the idea of learning by trial and error. 
More precisely: the computation of an interactive realizer depends on a finite knowledge state. 
If the knowledge state contains enough information, the realizer, computed on such state, will behave according to BHK semantics and produce the intended result. 
If the state is too small the realizer will fail and yield the missing piece of knowledge instead. 
Thus we have the following learning process: we begin by computing a realizer starting from the empty state. 
If the realizer succeeds we are done, otherwise we fail and add the piece of information obtained from the failure of the realizer to the state and compute the realizer on the updated state. 
We carry on this procedure until the realizer succeeds. 
The main computational properties of interactive realizers are thus exceptions (since they can fail) and dependency on a state. 
\end{omitted}


This chapter is an account of interactive realizability where interactive realizers are encoded as \(\lambda\)-terms following the monadic approach to structuring functional programms suggested by Wadler. 
We shall prove that our presentation of interactive realizability is a sound semantics for \(\HA+\EM\). 

\subsection{Main Results}

In our presentation, interactive realizer are written in a simply typed \(\lambda\)-calculus with products, coproducts and natural numbers with course-of-value recursion, extended with some abstract terms to represent states and exceptions. 
The peculiar features of interactive realizability, namely the dependency on the knowledge state and the possibility of failure, are explicitly computed by the \(\lambda\)-terms encoding the realizers. 
Thus the computational behavior of interactive realizers is evident at the syntactic level, without the need for special semantics. 

While proving the soundness of \(\HA+\EM\) with respect to our definition of interactive realizability, we observed that the soundness of \(\HA\) did not require any assumption on the specific monad we chose to structure interactive realizers (while the soundness of \(\EM\) requires them as expected). 
Prompted by this discovery, we split the presentation in two parts. 

The former is an abstract monadic framework for producing realizability notions where the realizers are written in monadic style. 
We prove that \(\HA\) is sound with respect to any realizability semantics defined by the framework, for any monad. 

The latter is an application of this abstract framework to interactive realizability. 
We define a specific monad which we use to structure interactive realizers and a class of realizability semantics that

our definition of interactive realizability can be generalized to an abstract realizability notion where the monad is a parameter.
We call this family of realizability notions \emph{monadic realizability} and we show that all instances of monadic realizability are sound semantics for \(\HA\). 
Our interest is in concrete monadic realizability notions that can realize classical principles beyond \(\HA\). 
Another motivation is that in the proof of the soundness of \(\HA+\EM\) with respect to interactive realizability semantics, 
we take advantage of the special properties of the interactive realizability monad just to prove the soundness \(\EM\), while for the rest of \(\HA\) we just need some generic properties on the semantics which require no assumption on the monad we are considering.

\subsection{Related Works}

This work builds on the presentation of interactive realizability given in \cite{aschieriB10} by Aschieri and Berardi. 
The main contributions with respect to \cite{aschieriB10} is a more precise description of the computational behavior of interactive realizer. 
This is explained in more detail at the end of this chaper. 

Monads have first been used to describe interactive realizability by Berardi and de'Liguoro in \cite{berardidL10} and \cite{berardidL11}, where interactive realizers for \(\PRA+\EM\) are given a monadic categorical semantics following Moggi's approach. 
While our idea of using monad to describe interactive realizability was inspired by \cite{berardidL10}, our work is mostly unrelated: our use of monads follows Wadler's syntactical approach and we employ a different monad that emphasizes different aspects of interactive realizability. 

\section{Monadic Realizability} \label[section]{sec:monadic_realizability}
\setmonad\MM

This section contains the abstract part of our work. 
We describe the abstract framework of monadic realizability and show the soundness of \(\HA\) with respect to the semantics induced by a generic monad.

\subsection{The Monadic Realizability Semantics}

In this subsection we define monadic realizability. 
We state the properties that a suitable relation must satisfy in order to be called a monadic realizability relation and 
we show how such a relation induces a (monadic) realizability semantics.
Then we describe the proof decoration procedure to extract monadic realizers from proofs in \(\HA\). 
Here we are only concerned with proofs in \(\HA\), for a non-trivial example of a monadic realizability notion 
see interactive realizability in \cref{sec:monadic_interactive_realizability}. 

We start by introducing a syntactic translation of the concept of monad from category theory. 
Informally, a monad is an operator \(\T\) ``extending'' a type, with a canonical embedding from \(\ta\) to \(\T(\ta)\), a canonical way to lift\fixme{or extend?} a map from \(\ta\) to \(\T(\tb)\) to a map from \(\T(\ta)\) to \(\T(\tb)\), a canonical way of merging an element of \(\T(\ta)\) and an element of \(\T(\tb)\) into an element of \(\T(\ta \times \tb)\). We also requires some equations relating these canonical maps, equations which are often satisfied in the practice of programming. 
\begin{definition}[Syntactic Monad] \label[definition]{def:syntactic_monad} A \emph{syntactic monad} \(\MM\) is a tuple \( (\T\), \(\monUnit\), \(\monStar\), \(\monMerge) \) where \(\T\) is a type constructor, that is, a map from types to types, and, for any types \( \ta, \tb \), 
  \( \monUnit, \monStar \) and \( \monMerge \) are families (indexed by \(\ta\) and \(\tb\)) of closed terms: 
  \begin{align*} 
    \monUnit[\ta] &: \ta \tarrow \T\ta, \\ 
    \monStar[\ta,\tb] &: (\ta \tarrow \T\tb) \tarrow (\T\ta \tarrow \T\tb), \\ 
    \monMerge[\ta,\tb] &: \T\ta \tarrow \T\tb \tarrow \T(\ta \times \tb), 
  \end{align*} 
  satisfying the following properties: 
  \begin{align}
    \tag{M1} \label[property]{mon1} 
    \monStar[\ta,\ta] \monUnit[\ta] \mon\tta &\leadsto \mon\tta, \\ 
    \tag{M2} \label[property]{mon2} 
    \monStar[\ta,\tb] f (\monUnit[\ta] \tta) &\leadsto f \tta, \\ 
    \tag{M3} \label[property]{mon_merge1} 
    \monMerge[\ta,\tb] (\monUnit[\ta] \tta) (\monUnit[\ta] \ttb) &\leadsto \monUnit[\ta\times\tb] (\pair[\ta,\tb]\tta\ttb), 
  \end{align}
  for any \( \mon\tta : \T\ta \), \( f : \ta \tarrow \T\tb \), \( g : \tb \tarrow \T\tc \), \( \tta : \ta \) and \( \tta : \tb \). 
\end{definition} 
The terms \(\monUnit\) and \(\monStar\) and \cref{mon1,mon2} are a straightforward translation of the definition of Kleisli tripe in category theory, an equivalent way to describe a monad\footnotemark. 
\footnotetext{
  This part of the definition follows the one given by Wadler in \cite{wadler92}, with the difference that we replace the term \(\monBind\) with \(\monStar\), where:
  \[ \monBind[\ta,\tb] : \T\ta \tarrow (\ta \tarrow \T\tb) \tarrow \T\tb. \]
  Defining \(\monStar\) and \(\monBind\) in terms of each other is straightforward: 
  \begin{align*}
    \monBind[\ta,\tb] &\equiv 
    \abstr{\mon\tta}{\T\ta} \abstr{f}{\ta \tarrow \T\tb} \monStar f \mon\tta, \\ 
    \monStar[\ta,\tb] &\equiv 
    \abstr{f}{\ta \tarrow \T\tb} \abstr{\mon\tta}{\T\ta} \monBind \mon\tta f. 
  \end{align*}
  The term \(\monStar\) corresponds directly to the operator \(\_^*\) in the definition of Kleisli triple. 
}

Term \(\monMerge\) and \cref{mon_merge1} are connected to the definition of strong monad: \(\monMerge\) is the syntactical counterpart of the natural transformation \(\phi\), induced by the tensorial strength of the monad (see \cite{moggi91} for details). 
While \(\phi\) satisfies several other properties in \cite{moggi91}, \cref{mon_merge1} is the only one we need for our treatment. \note{\(\phi\) may be related to commutative strong monad}

\begin{example} \label{ex:identity_monad1}
  \setmonad\IdM 
  The simplest example of syntactic monad is the \emph{identity monad} \(\IdM\), defined as:  
  \begin{align*}
    \T\ta &\equiv \ta, & 
    \monUnit[\ta] &\equiv \abstr\tta\ta \tta, \\
    \monStar[\ta,\tb] &\equiv \abstr{f}{\ta \tarrow \tb} f, & 
    \monMerge[\ta,\tb] &\equiv \pair[\ta,\tb]. 
  \end{align*}
  This monad cannot describe any additional computational property besides the value a term reduces to. 
\end{example} 
\begin{example} 
  \setmonad\ExM 
  A simple but non-trivial example is the exception monad \(\ExM\).
  It describes computations which may either succeed and yield a (normal) value or fail and yield a description of the failure. 
  Consider the usual predecessor function \(\termname{pred}: \Nat \tarrow \Nat\) on the natural numbers: 
  since zero has no predecessor it is common to define \(\termname{pred} \num 0\) as zero. 
  Instead with \(\ExM\) we could have \(\termname{pred} \num 0\) fail and yield a string\footnote{assuming we had strings in our calculus} saying ``zero has no predecessor''. 

  Let \(\Ex\) be a new ground type and let \(\exmerge : \Ex \tarrow \Ex \tarrow \Ex\) be a new constant term. 
  We think terms of type \(\Ex\) as descriptions of failures and we call them \emph{exceptions}. 
  We think of \(\exmerge\) as an operation that merges the information of multiple exceptions when there are multiple failures in a computations. 
  Now we can define the syntactic monad \(\ExM\) as:
  \begin{align*}
    \T\ta &\equiv \ta + \Ex, \qquad 
    \monUnit[\ta] \equiv \abstr\tta\ta \inl[\ta,\Ex]\tta, \\
    \monStar[\ta,\tb] &\equiv 
    \abstr{f}{\ta \tarrow \tb + \Ex} \abstr\tta\ta \case[\ta,\Ex,\tb+\Ex] \tta f \inr[\tb,\Ex], \\  
    \monMerge[\ta,\tb] &\equiv 
    \abstr{\mon\tta}{\ta+\Ex} \abstr{\mon\ttb}{\tb+\Ex} \case[\ta,\Ex,(\ta\times\tb) + \Ex] \mon\tta \\
    &\phantomrel\equiv (\abstr\tta\ta \case[\tb,\Ex,(\ta\times\tb)+\Ex] \mon\ttb 
    (\abstr\ttb\tb \inl[\ta\times\tb,\Ex] (\pair[\ta,\tb] \tta \ttb)) \inr[\ta\times\tb,\Ex] ) \\
    &\phantomrel\equiv (\abstr{e_1}\Ex \case[\tb,\Ex,(\ta\times\tb)+\Ex] \mon\ttb 
    (\abstr\ttb\tb \inr[\ta\times\tb,\Ex] e_1) (\abstr{e_2}\Ex \inr[\ta\times\tb,\Ex] \exmerge e_1 e_2)) . 
  \end{align*}
  We omit the proof that \(\ExM\) is a syntactic monad. 
  \fixme{give a proof}
\end{example}

A \emph{realizability relation} is a binary relation between terms and closed formulas. 
When a term and a formula are in such a relation we shall say that the term \emph{realizes} the formula or that the term is a \emph{realizer} of the formula. 
The intended meaning is that a realizer of a formula is the computational content of a proof of the formula. 

We proceed towards the definition of a family of realizability relations, which we call \emph{monadic realizability relations}.
Any monadic realizability relation is given with respect to some monad \(\MM\) and determines a particular notion of realizability where realizers have the computational properties described by the monad. 
In the rest of this section we shall assume that \(\MM = (\T, \monUnit, \monStar, \monMerge)\) denotes any fixed syntactic monad. 

We now define the type of the monadic realizers of a formula. 
The idea is to take the standard definition of the type of intuitionistic realizers of a formula \(\fa\) and to apply \(\T\) only to the type \(\ta\) of the whole formula \(\fa\) and to the types appearing in \(\ta\) after an arrow, 
namely the types of consequents \(\fc\) of implication sub-formulas \(\fb \limply \fc\) in \(\fa\) and the types of bodies \(\fb\) of universal quantified sub-formulas \(\qforall\lva \fb\) in \(\fa\). 
This is the standard call-by-value way to treat arrow types in a monadic framework explained in \cite{wadler90}. 

\begin{definition}[Types for Monadic Realizers] \label[definition]{def:mon_types}
  We define two mappings \(\mm\cdot\) and \(\m\cdot\) from formulas to types by simultaneous recursion. 
  The first is the outer or monadic typing of a formula \(\fa\): 
  \[ \mm\fa = \T\m\fa, \] 
  and the latter is the inner typing, defined by induction on the structure of \(\fa\): 
  \begin{align*} 
    \m\lafa & = \Unit, &
    \m{\fb \land \fc} & = \m\fb \times \m\fc, \\ 
    \m{\fb \lor \fc} & = \m\fb + \m\fc, &
    \m{\qexists\lva \fb} & = \Nat \times \m\fb, \\ 
    \m{\fb \limply \fc} & = \m\fb \tarrow \mm\fc, &
    \m{\qforall\lva \fb} & = \Nat \tarrow \mm\fb, 
  \end{align*} 
  where \(\lafa\) is an atomic formula and \(\fa\) and \(\fb\) are any formulas. 
\end{definition} 
Note that, we consider \(\lfalse\) to be atomic and \(\lnot \fa\) to be a notation for \(\fa \limply \lfalse\), so the types of their realizers follow from the previous definition. 

As we defined two types for each formula \(\fa\), each formula has two possible realizers, one of type \(\m\fa\) and one of type \(\mm\fa\). 
The former will follow the BHK interpretation like an ordinary intuitionistic realizer 
while the latter will be able to take advantage of the computational properties given by the syntactic monad \(\MM\). 
A formula (in particular classical principles) may have a realizer of monadic type but no realizer of inner type. 

\begin{thesis}
\begin{remark}
  The definition of \(\mm\cdot\) and \(\m\cdot\) can be derived from the Curry-Howard correspondence between formulas and types and from a call-by-name monadic translation for types. 
  \newcommand\formType[1]{\lvert #1 \rvert}
  We define the standard interpretation \(\formType\cdot\) that maps a formula into the type of its realizers:
  \begin{align*}
    \formType\lafa &= \Unit, & 
    \formType{\fa \land \fb} &= \formType\fa \times \formType\fb, \\ 
    \formType{\fa \lor \fb} &= \formType\fa + \formType\fb, & 
    \formType{\fa \limply \fb} &= \formType\fa \tarrow \formType\fb,\\ 
    \formType{\qforall{x} \fa} &= \Nat \tarrow \formType\fa, & 
    \formType{\qexists{x} \fa} &= \Nat \times \formType\fa. 
  \end{align*}
  Next we define a translation \(\monTrans\cdot\) that lifts types to their monadic counterparts:
  \begin{align*}
    \monTrans{\ta_0} &\equiv \ta_0, &
    \monTrans{\ta\tarrow\tb} &\equiv \monTrans\ta \tarrow \T\monTrans\tb, \\
    \monTrans{\ta\times\tb} &\equiv \monTrans\ta \times \monTrans\tb, & 
    \monTrans{\ta+\tb} &\equiv \monTrans\ta + \monTrans\tb,  
  \end{align*}
  where \(\ta_0\) is a ground type. 
  The first two clauses are taken from \cite{wadler94} and the other ones are a simple extension, based on the idea that products and sums behave like ground types. 

  By composition we can define the types for the monadic realizers of a formula: 
  \[
    \m\fa \equiv \monTrans{\formType\fa}, \qquad
    \mm\fa \equiv \T\m\fa.
  \]
  Expanding the definitions we get :
  \begin{align*}
    \m\lafa &= \Unit, \\
    \m{\fa\land\fb} &= \monTrans{\formType\fa} \times \monTrans{\formType\fb} = \m\fa \times \m\fb, \\
    \m{\fa\lor\fb} &= \monTrans{\formType\fa} + \monTrans{\formType\fb} = \m\fa + \m\fb, \\
    \m{\fa\limply\fb} &= \monTrans{\formType\fa} \tarrow \T\monTrans{\formType\fb} = \m\fa \tarrow \T\m\fb, \\ 
    \m{\qforall{x} \fa} &= \monTrans\Nat \tarrow \T\monTrans{\formType\fb} = \Nat \tarrow \T\m\fa, \\ 
    \m{\qexists{x} \fa} &= \monTrans\Nat \times \monTrans{\formType\fb} = \Nat \times \m\fa. 
  \end{align*}
  This is the same translation we described in \cref{def:mon_types}. 
\end{remark}
\end{thesis}

We shall now state the requirements for a realizability relation to be a monadic realizability relation. 
A realizability relation is to be thought of as the restriction of the realizability semantics to closed formulas, 
that is, a relation between terms of \(T'\) and closed formulas which holds when a term is a realizer of the formula. 
Since a formula can have realizers of inner and outer type, in the following definition two realizability relations will appear: \(\re\) for realizers of inner type, whose definition is modeled after the BHK interpretation and \(\monRe\) for the realizers of outer type, which takes in consideration the computational properties of the monad \(\MM\). 

As a typographical convention we shall use the letters \(\rr\), \(\ra\) and \(\rb\) for terms of type \(\m\fa\). 
Similarly we shall use \(\mon\rr\), \(\mon\ra\) and \(\mon\rb\) for terms of type \(\mm\fa\). 
\begin{definition}[Monadic Realizability Relation] \label[definition]{def:monadic_realizability_relation}
  Let \(\monRe\) be a realizability relation between terms of type \(\mm\fa\) and closed formulas \(\fa\). 
  Let \(\re\) be another realizability relation between terms of type \(\m\fa\) and closed formulas \(\fa\), such that
  \begin{itemize}
    \item \label[property]{re_atomic} \( \rr \re \lafa \) iff 
      \( \rr \leadsto \unit\) and \(\lafa\) is true, 
    \item \label[property]{re_and} \( \rr \re \fb \land \fc \) iff 
      \(\prl \rr \re \fb \) and \(\prr r \re \fc \), 
    \item \label[property]{re_or} \( \rr \re \fb \lor \fc \) iff 
      \(\rr \leadsto \inl a \) and \( a \re \fb \) or 
      \(r \leadsto \inr b \) and \( b \re \fc \), 
    \item \label[property]{re_imply} \( \rr \re \fb \limply \fc \) iff 
      \( \rr \ra \monRe \fc \) for all \(\ra : \m\fb\) such that \( \ra \re \fb \), 
    \item \label[property]{re_forall} \( \rr \re \qforall{x} \fb \) iff 
      \(\rr \num n \monRe \fb\subst{x}{\num n} \) for all natural numbers \(n\), 
    \item \label[property]{re_exists} \( \rr \re \qexists{x} \fb \) iff 
      \(\prr \rr \re \fb\subst{x}{\prl \rr} \), 
  \end{itemize} 
  where \(\lafa\) is a closed atomic formula and \(\fb\) and \(\fc\) are generic formulas. 
  We consider \(\lfalse\) a closed atomic formula which is never true (for instance \(0=1\)).
  We shall say that the pair \( (\monRe, \re) \) is a \emph{monadic realizability relation} if the following properties are satisfied: 
  \begin{enumerate}[label=MR\arabic*]
    \item \label[property]{real1} 
      if \( \rr \re \fa \) then \( \monUnit \rr \monRe \fa \), 
    \item \label[property]{real2} 
      if \( \rr \re \fb \limply \fc \)
      then \( \monStar \rr \mon\ra \monRe \fc \) for all \( \mon\ra : \mm\fb \) such that \( \mon\ra \monRe \fb \), 
    \item \label[property]{real3} 
      if \( \mon\ra \monRe \fb \) and \( \mon\rb \monRe \fc \) 
      then \( \monMerge \mon\ra \mon\rb \monRe \fb \land \fc \). 
  \end{enumerate} 
  We will say that a term \(\rr\) (resp. \(\mon\rr\)) is an \emph{inner} (resp. \emph{outer} or \emph{monadic}) \emph{realizer} of a formula \(\fa\) if \(\rr : \m\fa\) (resp. \(\rr : \mm\fa\)) and \(\rr \re \fa\) (resp. \(\mon\rr \monRe \fa\)). 
\end{definition}
When defining a concrete monadic realizability relation, it is often convenient to define \(\monRe\) in terms of \(\re\) too, that is, the two relations will be defined by simultaneous recursion in terms of each other. 

Note how the properties of the relation \(\re\) resemble the clauses the definition of standard modified realizability. 
The main difference is that in the functional cases, those of implication and universal quantification, \(\re\) is not defined in terms of itself but uses \(\monRe\). 
This makes apparent our claim that the behavior of inner realizers is closely related to the BHK interpretation. 

\Cref{real1} is a constraint on the relationship between \(\monRe\) and \(\re\). 
It requires \(\monUnit\) to transform inner realizers into monadic realizers, 
which can be thought as the fact that realizers satisfying the BHK interpretation are acceptable monadic realizers. 
\Cref{real2} again links \(\re\) and \(\monRe\), this time through \(\monStar\). 
It says that, if we have a term that maps inner realizers into monadic realizers, its lifting by means of \(\monStar\) maps monadic realizers into monadic realizers. 
\Cref{real3} is a compatibility condition between \(\monMerge\) and \(\monRe\). 
These conditions are all we shall need in order to show that any monadic realizability relation determines a sound semantics for \(\HA\). 
Later we shall see how particular instances of monadic realizability 
can produce a sound semantics for more than just \(\HA\). 

\begin{example}
  \setmonad\IdM
  We continue our example with the identity monad \(\IdM\) by defining a monadic realizability relation. 
  We define \(\monRe\) and \(\re\) by simultaneous recursion, with \(\re\) defined in terms of \(\monRe\) as in \cref{def:monadic_realizability_relation} and \(\monRe\) defined as \(\re\), which makes sense since \( \mm\fa = \m\fa\). 
\end{example}

We can now define the monadic realizability semantics for a given monadic realizability relation, that is, we say when a realizer validates a sequent where a formula can be open and depend on assumptions. 
In order to do this we need a notation for a formula in a context, which we call \emph{decorated sequent}. 
A decorated sequent has the form \( \Gamma \monSeq \rr : \fa \) where \(\fa\) is a formula, \(\rr\) is a term of type \(\mm\fa\) and 
\(\Gamma\) is the context, namely, a list of assumptions written as \( \alpha_1 : \fa_1, \dotsc \alpha_k : \fa_k \) where \(\fa_1, \dotsc, \fa_k\) are formulas and \(\alpha_1, \dotsc, \alpha_k\) are proof variables that label each assumption, that is, they are variables of type \( \m{\fa_1}, \dotsc, \m{\fa_k} \). 
As we did with the syntactic monad \(\MM\), in the following we shall assume to be working with a fixed generic monadic realizability relation \(\monRe\). 
\begin{definition}[Monadic Realizability Semantics] \label[definition]{def:mon_sem}
  Consider a decorated sequent: 
  \[ \alpha_1 : \fa_1, \dotsc, \alpha_k : \fa_k \monSeq \mon\rr : \fb, \] 
  such that the free variables of \(\fb\) are \( x_1, \dotsc, x_l \) and the free variables of \(\mon\rr\) are either in \( x_1, \dotsc, x_l \) or in \( \alpha_1, \dotsc, \alpha_k \). 
  We say that the sequent is valid if and only if 
  for all natural numbers \(n_1, \dotsc, n_l\) and for all inner realizers \( \ra_1 : \m{\fa_1}, \dotsc, \ra_k : \m{\fa_k} \) such that 
  \[ \ra_1 \re \fa_1[x_1 \coloneqq \num{n}_1, \dotsc, x_l \coloneqq \num{n}_l] \qquad \dotso \qquad \ra_k \re \fa_k[x_1 \coloneqq \num{n}_1, \dotsc, x_l \coloneqq \num{n}_l], \] 
  we have that 
  \[ \mon\rr[x_1 \coloneqq \num{n}_1, \dotsc, x_l \coloneqq \num{n}_l, \alpha_1 \coloneqq \ra_1, \dotsc, \alpha_k \coloneqq \ra_k] \monRe A[x_1 \coloneqq \num{n}_1, \dotsc, x_l \coloneqq \num{n}_l]. \] 
\end{definition} 

\begin{example}\setmonad\IdM
  From \cref{def:mon_sem}, it follows that the semantics induced by the monadic realizability relation \(\monRe\) is exactly the standard semantics of modified realizability.  
\end{example}


Now that we have defined our semantics, we can illustrate the method to extract monadic realizers from proofs in \(\HA\). Later we shall show how to extend our proof extraction technique to \(\HA+\EM\). 
Since proof in \(\HA\) are constructive, the monadic realizers obtained from them behave much like their counterparts in standard modified realizability and comply with the BHK interpretation. 
In \cref{sec:monadic_interactive_realizability} we shall show how to extend the proof decoration to non constructive proofs by exhibiting a monadic realizer of \(\EM\) that truly takes advantage of monadic realizability since it does not act accordingly to the BHK interpretation. 

In order to build monadic realizers of proofs in \(\HA\) 
we need a generalization of \(\monStar\) that works for functions of more than one argument. 
We can build it using \(\monMerge\) to pack realizers together. 
Thus let
\begin{equation*} 
  \monStarN[\ta_1,\dotsc,\ta_k,\tb]{k} : (\ta_1 \tarrow \dotsb \tarrow \ta_k \tarrow \T\tb) \tarrow (\T\ta_1 \tarrow \dotsb \tarrow \T\ta_k \tarrow \T\tb), 
\end{equation*}
be a family of terms defined by induction on \( k \ge 0 \): 
\begin{gather*}
  \monStarN[\tb]{0}  \equiv \abstr{f}{\T\tb} f, \qquad 
  \monStarN[\ta,\tb]{1}  \equiv \monStar[\ta,\tb], \\ 
  \monStarN{k+2} \equiv \abstr{f}{\ta_1 \tarrow \dotsb \tarrow \ta_{k+2} \tarrow \T\tb} \abstr{x}{\T\ta_1} \abstr{y}{\T\ta_2} \monStarN{k+1} (\abstr{z}{\ta_1 \times \ta_2} f (\prl z) (\prr z) ) (\monMerge x y). 
\end{gather*} 
For instance: 
\[ \monStarN{2} \equiv \abstr{f}{\ta \tarrow \tb \tarrow \T\tc} \abstr{x}{\T\ta} \abstr{y}{\T\tb} \monStar (\abstr{z}{\ta \times \tb} f (\prl z) (\prr z)) (\monMerge x y) \] 
\begin{omitted}
  An alternative definition of \(\monStarN{k}\) that does not use \(\monMerge\): 
  \[ \monStarN{k} \equiv \abstr{f}{\ta_1 \tarrow \dotsb \tarrow \ta_k \tarrow \T\tb}
    \abstr{\mon\tta_1}{\T\ta_1} \dotsc \abstr{\mon\tta_k}{\T\ta_k} 
  \monStar (\abstr{\tta_1}\ta \monStarN{k-1} (f \tta_1) \mon\tta_2 \dotsm \mon\tta_k ) \mon \tta_1\] 
\end{omitted}

Moreover we shall need to ``raise'' the return value of a term \( f : \ta_1 \tarrow \dotsb \tarrow \ta_k \tarrow \tb \) with \(\monUnit\) before we apply \(\monStarN{k}\).
We define the family of terms \(\monRaiseN{k}\) by means of \(\monStarN{k}\), for any \(k \ge 0\):
\begin{align*}
  \monRaiseN{k} &: (\ta_1 \tarrow \dotsb \tarrow \ta_k \tarrow \tb) \tarrow (\T\ta_1 \tarrow \dotsb \tarrow \T\ta_k \tarrow \T\tb) \\ 
  \monRaiseN{k} &\equiv \abstr{f}{\ta_1 \tarrow \dotsb \tarrow \ta_k \tarrow \tc} \monStarN{k} 
  (\abstr{x_1}{\ta_1} \dotsm \abstr{x_k}{\ta_k} \monUnit (f x_1 \dotsm x_k)), 
\end{align*}

Now we can show how to extract a monadic realizer from a proof in \(\HA\). 
Let \(\mathcal{D}\) be a derivation of some formula \(\fa\) in \(\HA\), that is, a derivation ending with \( \Gamma \seq \fa \). 
We produce a decorated derivation by replacing each rule instance in \(\mathcal{D}\) with the suitable instance of the decorated version of the same rule given in \cref{fig:monadic_decorated_rules}. 
These decorated rules differ from the previous version in that they replace sequents with decorated sequents, that is, they bind a term to each formula, 
where the term bound to the conclusion of a rule is build from the terms bound to the premises. 
Thus we have defined a term by structural induction on the derivation: if the conclusion of the decorated derivation is \( \Gamma \monSeq \mon\rr : \fa \) then we set \( \mathcal{D}^* \equiv \mon\rr \). 

\begin{figure}[!ht]
  \begin{mdframed}
  \caption{\(\HA\) rules, decorated with monadic realizers.} \label{fig:monadic_decorated_rules}
  \begin{gather*}
      \renewcommand\arraystretch{2}
      \newcommand\wline[1]{\multicolumn{2}{c}{#1} \\ }
      \PrAx{} 
      \PrLbl{\text{Id}} 
      \PrUn{\Gamma \monSeq \monRaiseN{0} \lva : \fa} 
      \DisplayProof \qquad 
      \PrAx{\Gamma \monSeq \mon\rr_1 : \lafa_1} 
      \PrAx\dotso 
      \PrAx{\Gamma \monSeq \mon\rr_l : \lafa_l} 
      \PrLbl{\text{Atm}}
      \PrTri{\Gamma \monSeq \monRaiseN{l} (\abstr{\gamma_1}\Unit \dotsm \abstr{\gamma_l}\Unit \unit) \mon\rr_1 \dotsm \mon\rr_l : \lafa} 
      \DisplayProof \\ 
        \PrAx{\Gamma \monSeq \mon\rr_1 : \fa} 
        \PrAx{\Gamma \monSeq \mon\rr_2 : \fb} 
        \PrLbl{\RuleNameI\land} 
        \PrBin{\Gamma \monSeq \monRaiseN{2} \pair \mon\rr_1 \mon\rr_2 : \fa \land \fb} 
        \DisplayProof \\
        \PrAx{\Gamma \monSeq \mon\rr : \fa \land \fb} 
        \PrLbl{\RuleName[L]\land{E}}
        \PrUn{\Gamma \monSeq \monRaiseN{1} \prl \mon\rr : \fa} 
        \DisplayProof \qquad 
        \PrAx{\Gamma \monSeq \mon\rr : \fa \land \fb} 
        \PrLbl{\RuleName[R]\land{E}}
        \PrUn{\Gamma \monSeq \monRaiseN{1} \prr \mon\rr : \fb} 
        \DisplayProof \\ 
      \PrAx{\Gamma \monSeq \mon\rr_1 : \fa} 
      \PrLbl{\RuleName[R]\lor{I}} 
      \PrUn{\Gamma \monSeq \monRaiseN{1} \inl \mon\rr_1 : \fa \lor \fb} 
      \DisplayProof \qquad 
      \PrAx{\Gamma \monSeq \mon\rr_2 : \fb} 
      \PrLbl{\RuleName[L]\lor{I}} 
      \PrUn{\Gamma \monSeq \monRaiseN{1} \inr \mon\rr_2 : \fa \lor \fb} 
      \DisplayProof \\ 
        \PrAx{\Gamma \monSeq \mon\rr : \fa \lor \fb } 
        \PrAx{\Gamma, \alpha_{k+1} : \fa \monSeq \mon\ra : \fc} 
        \PrAx{\Gamma, \alpha_{k+1} : \fb \monSeq \mon\rb : \fc} 
        \PrLbl{\RuleNameE\lor} 
        \PrTri{\Gamma \monSeq \monStarN{1} (\abstr\gamma{\m\fa + \m\fb} \case \gamma (\abstr{\alpha_{k+1}}{\m\fa} \mon\ra) (\abstr{\alpha_{k+1}}{\m\fb} \mon\rb)) \mon\rr : \fc} 
        \DisplayProof \\ 
      \PrAx{\Gamma, \alpha_{k+1} : \fa \monSeq \mon\rr : \fb} 
      \PrLbl{\RuleNameI\limply}
      \PrUn{\Gamma \monSeq \monRaiseN{0} (\abstr{\alpha_{k+1}}{\m\fa} \mon\rr) : \fa \limply \fb} 
      \DisplayProof \quad 
      \PrAx{\Gamma \monSeq \mon\rr : \fa \limply \fb} 
      \PrAx{\Gamma \monSeq \mon\ra : \fa} 
      \PrLbl{\RuleNameE\limply}
      \PrBin{\Gamma \monSeq \monStarN{2} (\abstr{\gamma_1}{\m\fa \tarrow \m\fb} \abstr{\gamma_2}{\m\fa} \gamma_1 \gamma_2) \mon\rr \mon\ra : \fb} 
      \DisplayProof \\ 
      \PrAx{\Gamma \monSeq \mon\rr : \fa} 
      \PrLbl{\RuleNameI\forall}
      \PrUn{\Gamma \monSeq \monRaiseN{0} (\abstr\lva\Nat \mon\rr) : \qforall\lva \fa} 
      \DisplayProof \qquad 
      \PrAx{\Gamma \monSeq \mon\rr : \qforall\lva \fa} 
      \PrLbl{\RuleNameE\forall}
      \PrUn{\Gamma \monSeq (\monStarN{1} (\abstr\gamma{\Nat \tarrow \mm\fa} \gamma \lta)) \mon\rr : \fa\subst\lva\lta} 
      \DisplayProof \\ 
        \PrAx{\Gamma \monSeq \mon\rr : \fa\subst\lva\lta} 
        \PrLbl{\RuleNameI\exists}
        \PrUn{\Gamma \monSeq \monRaiseN{1} (\abstr\gamma{\m\fa} \pair\lta\gamma) \mon\rr : \qexists\lva \fa} 
        \DisplayProof \\ 
        \PrAx{\Gamma \monSeq \rr_1 : \qexists\lva \fa} 
        \PrAx{\Gamma, \alpha : \fa\subst\lva\lvb \monSeq \rr_2 : \fc} 
        \PrLbl{\RuleNameE\exists}
        \PrBin{\Gamma \monSeq \monStarN{1} (\abstr\gamma{\Nat \times \m\fa} (\abstr{y}\Nat \abstr\alpha{\m\fa} r_2)(\prl \gamma)(\prr \gamma)) \rr_1 : \fc} 
        \DisplayProof \\ 
        \PrAx{\Gamma, \alpha_{k+1} : \qforall\lvc \lvc<\lvb \limply \fa\subst\lva\lvc \monSeq \rr : \fa\subst\lva\lvb} 
        \PrLbl{\text{Ind}}
        \PrUn{\Gamma \monSeq \monRaiseN{0} (\totRec\infty f) : \qforall\lva \fa} 
        \DisplayProof 
    \end{gather*}
    \begin{legend}
    where all formulas in rule Atm are atomic, \(\lta\) is any term and
    \(f\) is defined as follows: 
    \[ f \equiv \abstr{y}\Nat \abstr\beta{\Nat \tarrow \T\m\fa} 
      (\abstr\alpha{\Nat \tarrow \T(\Unit \tarrow \T\m\fa)} r) 
      (\abstr{z}\Nat \monRaiseN{0} 
    (\abstr{\_}\Unit \beta z)), \] 
    with \(\beta\) not free in \(r\). 
    \end{legend}
  \end{mdframed}
\end{figure}

In \cref{fig:monadic_decorated_rules},
the rule labeled Atm shows how to decorate any atomic rule of \(\HA\). 
By definition unfolding, we may check that an atomic rule is interpreted as a kind of ``merging'' of the information associated to each premise. 
The nature of the merging depends on the monad we choose.

\begin{remark}\allowdisplaybreaks[1]
  In \cref{fig:monadic_decorated_rules}, we wrote all realizers using only \(\monRaiseN k\) and \(\monStarN k\) for the sake of consistency, 
  but note that \(\monRaiseN 0\) could have been replaced by \(\monUnit\) since it reduces to it:
  \begin{align*}
    \monRaiseN{0}
    &\equiv_{\monRaiseN{0}} \abstr{f}\tc \monStarN{0} (\monUnit f) \\ 
    & \equiv_{\monStarN{0}} \abstr{f}\tc (\abstr{f}{\T\tc} f) (\monUnit f) \\ 
    & \reducesto_\beta \abstr{f}\tc \monUnit f \\ 
    & =_\eta \monUnit 
  \end{align*}
  Moreover \(\monRaiseN 2 \pair\) reduces to \(\monMerge\):
  \begin{align*} 
    \monRaiseN{2} \pair &\equiv_{\monRaiseN 2}
    (\abstr{f}{\ta \tarrow \tb \tarrow \ta \times \tb} \monStarN{2} (\abstr\lva\ta \abstr\lvb\tb \monUnit (f \lva \lvb))) \pair \\ 
    &\leadsto_\beta \monStarN{2} (\abstr\lva\ta \abstr\lvb\tb \monUnit (\pair \lva \lvb)) \\ 
    &\equiv_{\monStarN 2} (\abstr{f}{\ta \tarrow \tb \tarrow \T(\ta \times \tb)} \abstr\lva{\T\ta} \abstr\lvb{\T\tb}  \\
    &\phantomrel{\equiv_{\monStarN 2}}
    \monStarN{k} (\abstr\lvc{\ta \times \tb} f (\prl \lvc) (\prr \lvc) ) (\monMerge \lva \lvb)) (\abstr\lva\ta \abstr\lvb\tb \monUnit \pair\lva\lvb) \\ 
    &\leadsto_\beta \abstr\lva{\T\ta} \abstr\lvb{\T\tb} \monStar (\abstr\lvc{\ta \times \tb} (\abstr\lva\ta \abstr\lvb\tb \monUnit \pair\lva\lvb) (\prl \lvc) (\prr \lvc) ) (\monMerge \lva \lvb) \\ 
    &\leadsto_\beta \abstr\lva{\T\ta} \abstr\lvb{\T\tb} \monStar (\abstr\lvc{\ta \times \tb} \monUnit \pair{\prl \lvc}{\prr \lvc}) (\monMerge \lva \lvb) \\ 
    &=_\times \abstr\lva{\T\ta} \abstr\lvb{\T\tb} \monStar (\abstr\lvc{\ta \times \tb} \monUnit \lvc) (\monMerge \lva \lvb) \\ 
    &=_\eta \abstr\lva{\T\ta} \abstr\lvb{\T\tb} \monStar \monUnit (\monMerge \lva \lvb) \\ 
    &\leadsto_{\ref{mon2}} \abstr\lva{\T\ta} \abstr\lvb{\T\tb} \monMerge \lva \lvb \\ 
    &=_\eta \monMerge, 
  \end{align*} 
  so we could replace it in \(\RuleNameI\land\). 
\end{remark}

Note how the monadic realizer of each rule is obtained by lifting the suitable term in the corresponding standard modified realizer with \(\monStarN k\) or \(\monRaiseN k\). 
These monadic realizers do not take advantages of particular monadic features (it cannot be otherwise since we have made no assumption on the syntactic monad or the monadic realizability relation). 
The main difference is that they can act as ``glue'' between ``true'' monadic realizers of non constructive axioms and rules, for instance the one we shall build in \cref{sec:monadic_interactive_realizability}. 

Here we can see that monadic realizability generalizes intuitionistic realizability: decorated rules in \cref{fig:monadic_decorated_rules} reduce to the standard decorated rules for intuitionistic modified realizability in the case of the identity monad \(\IdM\). 

\begin{omitted}
  \begin{proposition}[FIXME] 
    Assume that \( \alpha_1 : A_1, \dotsc, \alpha_k : A_k \re r : \fc \) is provable by the monadic rules and there are terms \( t_1, \dotsc, t_k \) whose types are respectively \( \mm{A_1}, \dotsc, \mm{A_k} \). 
    Then \(r[\alpha_1 \coloneqq t_1, \dotsc, \alpha_k \coloneqq t_k] \) has type \( \mm\fc \). 
  \end{proposition} 
  \begin{proof} 
    \begin{equation*} 
      \monStarN{1} (\overbrace{\abstr\gamma{\Nat \times \m\fa} \underbrace{(\overbrace{\abstr{y}\Nat \abstr\alpha{\m\fa} r_2}^{\Nat \tarrow \m\fa \tarrow \T\m\fc})(\overbrace{\prl \gamma}^\Nat)(\overbrace{\prr \gamma}^{\m\fa})}_{\T\m\fc}}^{\Nat \times \m\fa \tarrow \T\m\fc}) r_1 
    \end{equation*} 
    \[ f \equiv \abstr{y}\Nat \abstr\beta{\Nat \tarrow \T\m\fa} 
      (\abstr\alpha{\Nat \tarrow \T(\Unit \tarrow \T\m\fa)} r) 
      (\abstr{z}\Nat \monRaiseN{0} 
    (\abstr\_\Unit \beta z)) \] 
  \end{proof} 
\end{omitted}

\subsection{The Soundness Theorem}

In this subsection we prove that \(\HA\) is sound with respect to the monadic realizability semantics given in \cref{def:mon_sem}. 
This amounts to say that we can use proof decoration to extract, from any proof in \(\HA\), a monadic realizer that makes its conclusion valid. 
We prove this for a generic monad, which means that the soundness of \(\HA\) does not depend on the special properties of any specific monad. 
The proof only needs the simple properties we have requested in \cref{def:monadic_realizability_relation}. 

Before proving the main result, we need to show that \(\monStarN{k}\) and \(\monRaiseN{k}\) satisfy a generalization of \cref{real2}. 
\begin{proposition}[Monadic Realizability Property for \(\monStarN{k}\)] \label{thm:starN} 
  Let \( \fa_1, \dotsc, \fa_k \) and \(\fb\) be any formulas and let \( \rr : \m{\fa_1} \tarrow \dotsb \tarrow \m{\fa_k} \tarrow \mm\fb \) be a term. 
  Assume that, for all terms \( \ra_1 : \m{\fa_1}, \dotsc, \ra_k : \m{\fa_k} \) such that \( \ra_1 \re \fa_1, \dotsc, \ra_k \re \fa_k \), we have: 
  \[ \rr \ra_1 \dotsm \ra_k \monRe \fb. \] 
  Then, for all terms \( \mon\ra_1 : \mm{\fa_1}, \dotsc, \mon\ra_k : \mm{\fa_k} \) such that \( \mon\ra_1 \monRe \fa_1, \dotsc, \mon\ra_k \monRe \fa_k \), we have: 
  \[ \monStarN{k} \rr \mon\ra_1 \dotsm \mon\ra_k \monRe \fb. \] 
\end{proposition}
\begin{proof}
  By induction on \(k\). 
  For \( k = 0 \) it is trivial and for \( k = 1 \) it follows from \cref{real2} since \( \monStarN{1} \equiv \monStar \). 
  Now we just need to prove that if the statement holds for some \( k \ge 1 \), it holds for \( k+1 \) too. 

  As in the statement we assume that, for all terms \( \ra_1 : \m{\fa_1}, \dotsc, \ra_{k+1} : \m{\fa_{k+1}} \) such that \( \ra_1 \re \fa_1, \dotsc, \ra_{k+1} \re \fa_{k+1} \): 
  \[ \rr \ra_1 \dotsm \ra_{k+1} \monRe \fb, \] 
  and that \( \mon\ra_1 : \mm{\fa_1}, \dotsc, \mon\ra_{k+1} : \mm{\fa_{k+1}} \) are terms such that \( \mon\ra_1 \monRe \fa_1, \dotsc, \mon\ra_{k+1} \monRe \fa_{k+1} \). 
  We need to show that: 
  \[ \monStarN{k+1} \rr \mon\ra_1 \dotsm \mon\ra_{k+1} \monRe \fb. \] 
  Since we know by definition of \(\monStarN{k+1}\) that \( \monStarN{k+1} \rr \mon\ra_1 \dotsm \mon\ra_{k+1} \) reduces to the term: 
  \[ \monStarN{k} (\abstr{z}{\m{\fa_1} \times \m{\fa_2}} \rr (\prl z) (\prr z) ) (\monMerge \mon\ra_1 \mon\ra_2) \mon\ra_3 \dotsm \mon\ra_{k+1}, \]
  and by \cref{real3} that \( \monMerge \mon\ra_1 \mon\ra_2 \monRe \fa_1 \land \fa_2 \),
  we see that we can use the inductive hypothesis on \(k\) to conclude. 
  In order to do so we have to show that the assumption of the inductive hypothesis holds, namely that, for any \( \ra_1 : \m{\fa_1} \times \m{\fa_2} \), \( \ra_3 : \m{\fa_3}, \dotsc, \ra_k : \m{\fa_k} \) such that \( \ra_1 \re \fa_1 \land \fa_2 \), \( \ra_2 \re \fa_2, \dotsc, \ra_k \re \fa_k \) it is the case that: 
  \[ (\abstr{z}{\m{\fa_1} \times \m{\fa_2}} \rr (\prl z) (\prr z) ) \ra_1 \dotsm \ra_k \monRe \fb. \] 
  By reducing the realizer we get that this is equivalent to: 
  \[ \rr (\prl \ra_1) (\prr \ra_1) \ra_2 \dotsm \ra_k \monRe \fb, \] 
  which is true by the assumption on \(\rr\) since \(\ra_1 \re \fa_1 \land \fa_2 \) means that \( \prl \ra_1 \re \fa_1 \) and \( \prr \ra_1 \re \fa_2 \) by definition of \(\re\). 
\end{proof}
We prove a similar property for \(\monRaiseN{k}\). 
\begin{proposition}[Monadic Realizability Property for \(\monRaiseN{k}\)] \label{thm:raiseN} 
  Let \( \fa_1, \dotsc, \fa_k \) and \(\fb\) be any formulas and let \( \rr : \m{\fa_1} \tarrow \dotsb \tarrow \m{\fa_k} \tarrow \m\fb \) be a term. 
  Assume that, for all terms \( \ra_1 : \m{\fa_1}, \dotsc, \ra_k : \m{\fa_k} \) such that \( \ra_1 \re \fa_1, \dotsc, \ra_k \re \fa_k \), it is the case that: 
  \[ \rr \ra_1 \dotsm \ra_k \re \fb. \] 
  Then, for all terms \( \mon\ra_1 : \mm{\fa_1}, \dotsc, \mon\ra_k : \mm{\fa_k} \) such that \( \mon\ra_1 \monRe \fa_1, \dotsc, \mon\ra_k \monRe \fa_k \), we have that: 
  \[ \monRaiseN{k} \rr \mon\ra_1 \dotsm \mon\ra_k \monRe \fb. \] 
\end{proposition}
\begin{proof}
  Assume that, for all terms \( \ra_1 : \m{\fa_1}, \dotsc, \ra_k : \m{\fa_k} \) such that \( \ra_1 \re \fa_1, \dotsc, \ra_k \re \fa_k \), it is the case that: 
  \[ \rr \ra_1 \dotsm \ra_k \re \fb, \] 
  and let \( \mon\ra_1 : \mm{\fa_1}, \dotsc, \mon\ra_k : \mm{\fa_k} \) be terms such that \( \mon\ra_1 \monRe \fa_1, \dotsc, \mon\ra_k \monRe \fa_k \). 
  We want to prove that: 
  \[ \monRaiseN{k} \rr \mon\ra_1 \dotsm \mon\ra_k \monRe \fb. \] 
  By definition of \(\monRaiseN{k}\) this reduces to: 
  \[ \monStarN{k} (\abstr{x_1}{\m{\fa_1}} \dotsm \abstr{x_k}{\m{A_k}} \monUnit (\rr x_1 \dotsm x_k)) \mon\ra_1 \dotsm \mon\ra_k \monRe \fb. \] 
  This follows by \cref{thm:starN} if we can show that, for any \( \ra_1 : \m{\fa_1}, \dotsc, \ra_k : \m{\fa_k} \) such that \( \ra_1 \re \fa_1, \dotsc, \ra_k \re \fa_k \), we have: 
  \[ (\abstr{x_1}{\m{\fa_1}} \dotsm \abstr{x_k}{\m{\fa_k}} \monUnit (\rr x_1 \dotsm x_k)) \ra_1 \dotsm \ra_k \monRe \fb. \]
  Reducing the realizer we get that this is equivalent to: 
  \[ \monUnit (\rr \ra_1 \dotsm \ra_k) \monRe \fb, \]
  and this follows by \cref{real1} and by assumption on \(\rr\). 
\end{proof}
Now we are ready to prove the soundness theorem. 
\begin{theorem}[Soundness of \(\HA\) with respect to the Monadic Realizability Semantics] \label{thm:ha_soundness}
  Let \(\mathcal{D}\) be a derivation of \( \Gamma \seq \fa \) in \(\HA\) and \(\monRe\) a monadic realizability relation. 
  Then \( \Gamma \monSeq \mathcal{D}^* : \fa \) is valid with respect to \(\monRe\). 
\end{theorem}
The proof is long but simple, proceeding by induction on the structure of the decorated version of \(\mathcal{D}\). 
\begin{proof} 
  We proceed by induction on the structure of the decorated version of \(\mathcal{D}\), that is, we assume that the statement holds for all decorated sub-derivations of \(\mathcal{D}\) and we prove that it holds for \(\mathcal{D}\) too. 
  More precisely we have to check the soundness of each decorated rule, showing that the validity of the premises yields the validity of the conclusion. 

  We start with some general notation and observations. 
  Let \( \Gamma \equiv \alpha_1 : \fa_1, \dotsc, \alpha_k : \fa_k \) for some \(k\). 
  Following the notation in \cref{def:mon_sem}, 
  we fix natural numbers \( n_1, \dotsc, n_l \) and terms \( \rr_1 : \fa_1, \dotsc, \rr_k : \fa_k \), 
  we define abbreviations: 
  \begin{align*}
    \Omega &\equiv x_1 \coloneqq \num n_1, \dotsc, x_l \coloneqq \num n_l, \\ 
    \Sigma &\equiv \alpha_1 \coloneqq r_1, \dotsc, \alpha_k \coloneqq r_k, 
  \end{align*}
  and we assume that: 
  \[ \rr_1 \re \fa_1[\Omega] \qquad \dotso \qquad \rr_k \re \fa_k[\Omega]. \]

  Note that if some term \(t : \ta_1 \tarrow \dotsb \tarrow \ta_k \tarrow \tb \) has no free variables then \( (t a_1 \dotsm a_k)[\Omega, \Sigma] \equiv t (a_1[\Omega, \Sigma]) \dotsm (a_k[\Omega, \Sigma]) \). 
  In particular this holds if \(t\) is one of \(\monStarN{k}\), \(\monRaiseN{k}\), \(\pair\), \(\prl\), \(\prr\), \(\case\), \(\inl\), \(\inr\). 
  The same holds for formulas, so \( (\fa \star \fb)[\Omega] \equiv \fa[\Omega] \star \fb[\Omega] \) where \(\star\) is one of \( \land, \lor \) or \( \limply \). 
  Also note that \( \m{\fa[\Omega]} = \m \fa \) since \( \m \cdot \) does not depend on the terms in \(\fa\). 
  In particular the types of the proof variables in \(\Gamma\) do not change, meaning we do not need to perform substitutions in \(\Gamma\).  
  We shall take advantage of these facts without mentioning it. 

  Now we can start showing that the rules are sound. \fixme{is the order of \(\Omega\) and \(\Sigma\) relevant? do we need to substitute \(\Omega\) in realizers?}

  \begin{itemize}
    \item[Id]
      We have to prove that: 
      \[ (\monRaiseN{0} \alpha_i)[\Omega, \Sigma] \monRe \fa[\Omega], \]
      where \( \fa = \fa_i \) for some \( i \in \{1, \dotsc, k\} \). 

      By performing the substitutions, we can rewrite the realizer as \(\monRaiseN{0} \rr_i\) so we need to prove that: 
      \[ \monRaiseN{0} \rr_i \monRe \fa. \]
      This follows by \cref{thm:raiseN} since by assumption \( \rr_i \re \fa_i[\Omega] \).

    \item[Atm]
      We have to prove that: \fixme{\(l\) is already used in this theorem}
      \[ (\monRaiseN{l} (\abstr{\gamma_1}\Unit \dotsm \abstr{\gamma_l}\Unit \unit) \mon\rr_1 \dotsm \mon\rr_l)[\Omega,\Sigma] \monRe \lafa[\Omega]. \]
      By performing the substitutions, we can rewrite the realizer as: 
      \[ \monRaiseN{l} (\abstr{\gamma_1}\Unit \dotsm \abstr{\gamma_l}\Unit \unit) \mon\rr_1[\Omega,\Sigma] \dotsm \mon\rr_l[\Omega,\Sigma]. \]
      By inductive hypothesis we know that 
      \[ \mon\rr_1[\Omega,\Sigma] \monRe \lafa_1[\Omega], \dotsc, \mon\rr_l[\Omega,\Sigma] \monRe \lafa_l[\Omega], \] 
      and thus we can conclude by \cref{thm:raiseN} if we can show that: 
      \[ (\abstr{\gamma_1}\Unit \dotsm \abstr{\gamma_l}\Unit \unit) \rr_1 \dotsm \rr_l \re \lafa[\Omega], \]
      for all \(\rr_1, \dotsc, \rr_l\) that are inner realizers of \(\lafa_1, \dotsc, \lafa_l\) respectively. 
      Since 
      \[ (\abstr{\gamma_1}\Unit \dotsm \abstr{\gamma_l}\Unit \unit) \rr_1 \dotsm \rr_l, \]
      reduces to \(\unit\) and \(\unit \re \lafa[\Omega]\) by definition of \(\re\) we are done. 
  \end{itemize}
  In the following we will apply the substitutions directly without mentioning it. 
  \begin{itemize}
    \item[\(\RuleNameI\land\)] 
      We have to prove that 
      \[ \monRaiseN{2} \pair \mon\ra[\Omega,\Sigma] \mon\rb[\Omega,\Sigma] \monRe \fa[\Omega] \land \fb[\Omega], \] 
      assuming that \( \mon\ra[\Omega, \Sigma] \monRe \fa[\Omega] \) and \( \mon\rb[\Omega, \Sigma] \monRe \fa[\Omega]\). 
      This follows by \cref{thm:raiseN} since
      \[ \pair \ra \rb \re \fa \land \fb, \] 
      for all inner realizers \(\ra\) of \(\fa\) and \(\rb\) of \(\fb\), 
      by definition of \(\re\). 

    \item[\( {\RuleNameE[L]\land} \)]
      We have to prove that 
      \[ (\monRaiseN{1} \prl \mon\rr)[\Omega,\Sigma] \monRe \fa[\Omega], \] 
      assuming that 
      \[ \mon\rr[\Omega,\Sigma] \monRe \fa[\Omega] \land \fb[\Omega]. \] 
      This follows by \cref{thm:raiseN} if 
      \[ \prl \rr \re \fa[\Omega], \] 
      for any inner realizer \(\rr\) of \(\fa[\Omega] \land \fb[\Omega]\).
      This is the case because from \( \rr \re \fa \land \fb \) if and only if \( \prl \rr \re \fa \) by definition of \(\re\). 

    \item[\( {\RuleNameE[R]\land} \)]
      Very similar to the proof for \(\RuleName[L]\land{E}\). 

    \item[\( {\RuleNameI[L]\lor} \)]
      We have to show that: 
      \[ \monRaiseN{1} \inl \mon\ra[\Sigma,\Omega] \monRe \fa[\Omega] \lor \fb[\Omega], \] 
      assuming that: 
      \[ \mon \ra[\Sigma,\Omega] \monRe \fa[\Omega]. \] 
      This follows by \cref{thm:raiseN} if 
      \[ \inl \ra \re \fa[\Omega], \] 
      for any inner realizer \(\ra\) of \(\fa[\Omega]\). 
      This is the case since \( \ra \re \fa[\Omega] \) if and only if \( \inl \ra \re \fa[\Omega] \lor \fb[\Omega] \) by definition of \(\re\). 

    \item[\( {\RuleNameI[R]\lor} \)]
      Very similar to the proof for \(\RuleName[L]\lor{I}\). 

    \item[\( \RuleNameE\lor \)] 
      We have to show that: 
      \[ \monStarN{1} (\abstr\gamma{\m\fa + \m\fb} \case \gamma (\abstr\alpha{\m\fa} \mon\ra[\Omega,\Sigma]) (\abstr\beta{\m\fb} \mon\rb[\Omega,\Sigma])) \mon\rr[\Omega,\Sigma] \monRe \fc[\Omega] \] 
      assuming by inductive hypothesis that: 
      \fixme{\(\alpha\) and \(\beta\) are inconsistent notation}
      \begin{enumerate} 
        \item \( \mon\rr[\Omega,\Sigma] \monRe \fa[\Omega] \lor \fb[\Omega] \), 
        \item \( \mon\ra[\Omega,\Sigma,\alpha \coloneqq \ra]\ \monRe \fc[\Omega] \) for any inner realizer \(\ra\) of \(\fa[\Omega]\), 
        \item \( \mon\rb[\Omega,\Sigma,\beta \coloneqq \rb]\ \monRe \fc[\Omega] \) for any inner realizer \(\rb\) of \(\fb[\Omega]\). 
      \end{enumerate} 
      We can conclude by \cref{thm:starN} if we show that 
      \[
        (\abstr\gamma{\m\fa + \m\fb} \case \gamma (\abstr\alpha{\m\fa} \mon\ra[\Omega,\Sigma]) (\abstr\beta{\m\fb} \mon\rb[\Omega,\Sigma])) \rr,  
      \]
      which \(\beta\)-reduces to
      \begin{equation}\label{eq:or_elim}
        \case \rr (\abstr\alpha{\m\fa} \mon\ra[\Omega,\Sigma]) (\abstr\beta{\m\fb} \mon\rb[\Omega,\Sigma]),  
      \end{equation}
      is a monadic realizer of \(\fc[\Omega]\) 
      for any inner realizer \(\rr\) of \(\fa[\Omega] \lor \fb[\Omega]\). 

      By definition of \(\re\), we know that either 
      \(\rr \leadsto \inl\ra\) where \(\ra\) is an inner realizer of \(\fa[\Omega]\) or 
      \(\rr \leadsto \inr\rb\) where \(\rb\) is an inner realizer of \(\fb[\Omega]\). 
      Assume that we are in the first case (the second case is analogous). 
      Then \eqref{eq:or_elim} becomes: 
      \[ \case (\inl\ra) (\abstr\alpha{\m\fa} \mon\ra[\Omega,\Sigma]) (\abstr\beta{\m\fb} \mon\rb[\Omega,\Sigma]), \]
      which reduces to 
      \[ (\abstr\alpha{\m\fa} \mon\ra[\Omega,\Sigma]) \ra, \]
      and to
      \[ \mon\ra[\Omega,\Sigma,\alpha \coloneqq \ra], \]
      which is a monadic realizer of \(\fc[\Omega]\) by inductive hypothesis. 

    \item[\( \RuleNameI\limply \)]
      We have to show that: 
      \[ \monRaiseN{0} (\abstr{\alpha_{k+1}}{\m\fa} \mon\rr[\Omega,\Sigma]) \monRe \fa[\Omega] \limply \fb[\Omega], \]
      assuming that: 
      \[ \mon \rr[\Omega,\Sigma,\alpha_{k+1} \coloneqq \ra] \monRe \fb[\Omega], \]
      for any inner realizer \(\ra\) of \(\fa[\Omega] \). 
      By \cref{thm:raiseN} it is enough to show that: 
      \[ \abstr{\alpha_{k+1}}{\m\fa} \mon\rr[\Omega,\Sigma] \re \fa[\Omega] \limply \fb[\Omega]. \]
      By definition of \(\re\) this holds if and only if: 
      \[ (\abstr{\alpha_{k+1}}{\m \fa} \mon\rr[\Omega,\Sigma]) \ra \monRe \fb[\Omega], \]
      for any inner realizer \(\ra\) of \(\fa[\Omega] \). 
      Reducing we get: 
      \[ \mon\rr[\Omega,\Sigma][\alpha_{k+1} \coloneqq \ra]) \monRe \fb[\Omega], \]
      and since \( \mon\rr[\Omega,\Sigma][\alpha_{k+1} \coloneqq \ra] \equiv \mon\rr[\Omega,\Sigma,\alpha_{k+1} \coloneqq \ra] \), \fixme{are the substitutions really correct? check this also for or elimination}
      we can conclude by the inductive hypothesis. 

    \item[\( \RuleNameE\limply \)] 
      We have to show that: 
      \[ (\monStarN{2} (\abstr{\gamma_1}{\m\fa \tarrow \m\fb} \abstr{\gamma_2}{\m\fa} \gamma_1 \gamma_2) \mon\rr[\Omega,\Sigma] \mon\ra[\Omega,\Sigma]) \monRe \fb[\Omega], \] 
      assuming by inductive hypothesis that: 
      \begin{enumerate}
        \item \( \mon\rr[\Omega,\Sigma] \monRe \fa[\Omega] \limply \fb[\Omega] \),
        \item \( \mon\ra[\Omega,\Sigma] \monRe \fa[\Omega] \). 
      \end{enumerate}
      This follows by \cref{thm:raiseN} if: 
      \[ (\abstr{\gamma_1}{\m\fa \tarrow \m\fb} \abstr{\gamma_2}{\m\fa} \gamma_1 \gamma_2) \rr \ra, \]
      which \(\beta\)-reduces to 
      \[ \rr \ra, \]
      is a monadic realizer of \( \fb[\Omega]\) 
      for any inner realizers \(\rr\) and \(\ra\) of \(\fa[\Omega] \limply \fb[\Omega]\) and \(\fa[\Omega] \) respectively. 
      This follows immediately by definition of \(\re\). 
  \end{itemize}
  In the following cases we assume that \(\Omega\) does not contain a substitution for the variable \(\lva\)
  and we write it explicitly when it is needed.
  \begin{itemize}
    \item[\( \RuleNameI\forall \)]
      We have to show that: 
      \[ \monRaiseN{0} (\abstr\lva\Nat \mon\rr[\Omega,\Sigma]) \monRe \qforall\lva \fa[\Omega], \]
      assuming by inductive hypothesis that: 
      \[ \mon\rr[\Omega,\lva \coloneqq \num n,\Sigma] \monRe \fa[\Omega,\lva \coloneqq \num n], \]
      for any natural number \(n\). 
      This follows by \cref{thm:raiseN} if: 
      \[ (\abstr\lva\Nat \mon\rr[\Omega,\Sigma]) \re \qforall\lva \fa[\Omega], \]
      which by definition of \(\re\) means that:
      \[ (\abstr\lva\Nat \mon\rr[\Omega,\Sigma]) \num n  \re \fa[\Omega, \lva \coloneqq \num n], \]
      for any natural number \(n\).
      By \(\beta\)-reducing we get:
      \[ \mon\rr[\Omega,\lva \coloneqq \num n,\Sigma]  \re \fa[\Omega, \lva \coloneqq \num n], \]
      which holds by inductive hypothesis. 

    \item[\( \RuleNameE\forall \)]
      We have to show that: 
      \[ (\monStarN{1} (\abstr\gamma{\Nat \tarrow \mm\fa} \gamma (\lta[\Omega]))) \mon\rr[\Omega,\Sigma] \monRe (\fa\subst\lva\lta)[\Omega], \]
      assuming by inductive hypothesis that: 
      \[ \mon\rr[\Omega,\Sigma] \monRe \qforall\lva \fa[\Omega]. \]
      This follows by \cref{thm:starN} if: 
      \[ (\abstr\gamma{\Nat \tarrow \mm\fa} \gamma (\lta[\Omega]))) \rr \leadsto \rr (\lta[\Omega]), \]
      is a monadic realizer of \(\fa[\Omega]\), for any inner realizer \(\rr\) of \(\qforall\lva \fa[\Omega]\).
      This follows by definition of \(\re\) for \(\rr \re \qforall\lva \fa[\Omega]\), 
      since \(\lta[\Omega]\) is closed and thus reduces to a numeral. 
      \fixme{needs normalization of arithmetic terms}

    \item[\( \RuleNameI\exists \)]
      We have to show that: 
      \[ \monRaiseN{1} (\abstr\gamma{\m\fa} \pair\lta[\Omega]\gamma) \mon\rr[\Omega,\Sigma] \monRe \qexists\lva \fa[\Omega], \]
      assuming by inductive hypothesis that: 
      \[ \mon\rr[\Omega,\Sigma] \monRe \fa[\Omega,\lva\coloneqq\lta]. \] 
      This follows by \cref{thm:raiseN} if: 
      \[ (\abstr\gamma{\m\fa} \pair\lta[\Omega]\gamma) \rr \leadsto \pair\lta[\Omega]\rr \]
      is an inner realizer of \( \qexists\lva \fa[\Omega] \), for any inner realizer \(\rr\) of \( \fa[\Omega,\lva\coloneqq\lta] \). 
      This follows by definition of \(\re\). 

    \item[\( \RuleNameE\exists \)] 
      We have to show that: 
      \[ \monStarN{1} (\abstr\gamma{\Nat \times \m\fa} (\abstr\lvb\Nat \abstr\alpha{\m\fa} \mon\rr_2[\Omega,\Sigma])(\prl \gamma)(\prr \gamma)) \mon\rr_1[\Omega,\Sigma] \monRe \fc[\Omega], \]
      assuming by inductive hypothesis that: 
      \begin{enumerate}
        \item \( \mon\rr_1[\Omega,\Sigma] \monRe \qexists\lva \fa[\Omega] \),
        \item \( \mon\rr_2[\Omega,\lvb \coloneqq \num n,\Sigma, \alpha \coloneqq \rr] \monRe \fc[\Omega] \), for any natural number \(n\) and any inner realizer \(\rr\) of \(\fa[\Omega]\).
      \end{enumerate}
      This follows by \cref{thm:starN} and by the inductive hypothesis on \(\mon \rr_1\) if, for any inner realizer \(\rr_1\) of \(\qexists\lva \fa[\Omega]\): 
      \begin{align*}
        &\phantomrel\leadsto (\abstr\gamma{\Nat \times \m\fa} (\abstr\lvb\Nat \abstr\alpha{\m\fa} \mon\rr_2[\Omega,\Sigma])(\prl \gamma)(\prr \gamma)) \rr_1 \leadsto \\ 
        &\leadsto (\abstr\lvb\Nat \abstr\alpha{\m\fa} \mon\rr_2[\Omega,\Sigma])(\prl \rr_1)(\prr \rr_1) \leadsto \\ 
        &\leadsto ((\mon\rr_2[\Omega,\Sigma])\subst\lvb{\prl \rr_1})\subst\alpha{\prr \rr_1} \equiv \\ 
        &\equiv \mon\rr_2[\Omega,\lvb \coloneqq \prl \rr_1,\Sigma, \alpha \coloneqq \prr \rr_1].
      \end{align*}
      is a monadic realizer of \(\fc[\Omega]\).
      By definition of \(\re\) we have that \(\prr \rr_1 \re \fa\subst\lva{\prl \rr_1} \) and thus we can conclude by the inductive hypothesis on \(\mon \rr_2\). 
    \item[Ind] 
      We have to show that: 
      \[ (\monRaiseN{0} (\totRec\infty f))[\Omega, \Sigma] \monRe (\qforall\lva \fa)[\Omega], \] 
      assuming that, for all naturals numbers \(n\) and for all \( \ra : \Nat \tarrow T(\Unit \tarrow T\m\fa) \) such that \( \ra \re \qforall\lvc \lvc < \num n \limply \fa[\lva \coloneqq \lvc] \): 
    \[ \rr[\Omega, \lvb \coloneqq \num n, \Sigma, \alpha_{k+1}] \coloneqq \ra] \monRe \fa[\lva \coloneqq \lvb][\Omega, \lvb \coloneqq \num n]. \] 
    Note that \( \fa[\lva \coloneqq \lvb][\Omega, \lvb \coloneqq \num n] \) is just \( \fa[\Omega, \lva \coloneqq \num n] \). 
    By \cref{thm:raiseN} we get the conclusion if \( \totRec\infty f[\Omega,\Sigma] \re \qforall\lva A[\Omega] \), which by definition of \(\re\) means that 
    \[ \totRec\infty f[\Omega,\Sigma] \num n \monRe A[\Omega, \lva \coloneqq \num n] \] 
    for any natural number \(n\). 
    In order to show this we shall prove that for any natural number \(n\) and any \(\omega \in \mathbb{N} \cup \{\infty\}\) such that either \( \omega = \infty \) or \( \omega > n \), we have:
    \[ \totRec\omega f[\Omega,\Sigma] \num n \monRe A[\Omega, \lva \coloneqq \num n]. \] 
    We proceed by complete induction on \(n\), so we assume that the statement holds for all natural numbers \(m\) such that \( m < n \). 
    We begin by reducing the realizer (in the first step we use the assumption on \(\omega\): 
    \begin{align*} 
      \totRec\omega f[\Omega, \Sigma] \num n 
      &\leadsto f[\Omega, \Sigma] \num n (\totRec{n} f[\Omega, \Sigma]) \\ 
      &\leadsto (\abstr\alpha{} \mon\rr[\Omega, \lvb \coloneqq \num n]) (\abstr\lvc\Nat \monRaiseN{0} (\abstr{\_}\Unit \totRec{n} f[\Omega, \Sigma] \lvc)) \\ 
      &\leadsto \mon\rr[\Omega, \lvb \coloneqq \num n, \Sigma, \alpha \coloneqq \abstr\lvc\Nat \monRaiseN{0} (\abstr{\_}\Unit \totRec{n} f[\Omega, \Sigma] \lvc)] 
    \end{align*} 
    Then we have to show that: 
    \[ \mon\rr[\Omega, \lvb \coloneqq \num n, \Sigma, \alpha \coloneqq \abstr\lvc\Nat \monRaiseN{0} (\abstr{\_}\Unit \totRec{n} f[\Omega, \Sigma] \lvc)] \monRe A[\Omega, \lva \coloneqq \num n]. \] 
    This follows from the inductive hypothesis on the premise of the complete induction rule if we can show that: 
    \[ \abstr\lvc\Nat \monRaiseN{0} (\abstr{\_}\Unit \totRec{n} f[\Omega, \Sigma] \lvc) \re \qforall\lvc \lvc < \num n \limply \fa[\lva \coloneqq \lvc]. \] 
    By definition of \(\re\) this is the case if: 
    \[ \monRaiseN{0} (\abstr{\_}\Unit \totRec{n} f[\Omega, \Sigma] \num m) \monRe \num m < \num n \limply \fa[\lva \coloneqq \num m], \]
    for all  natural numbers \(m\). By \cref{real1} this follows from: 
    \[ \abstr{\_}\Unit \totRec{n} f[\Omega, \Sigma] \num m \re \num m < \num n \limply \fa[\lva \coloneqq \num m]. \]
    Again by definition of \(\re\) this is equivalent to showing that for any \( u : \Unit \) such that \( u \re \num m < \num n \) we have: 
    \[ \totRec{n} f[\Omega, \Sigma] \num m \re \fa[\lva \coloneqq \num m]. \]
    Note that, since \( u : \Unit \), \fixme{by normal form theorem}
    \( u \leadsto \unit \), so there are two possible cases: 
    either \( m < n \) is true and then \( u \re \num m < \re \num n \) for any \( u : \Unit \) or \( m < n \) is false and no \( u : \Unit \) can realize \( \num m < \num n \). 
    In both cases the statement holds: 
    in the former case by inductive hypothesis on \(m\) and in the latter case trivially since the universal quantification on \(u\) is empty. 
\end{itemize}
\end{proof} 

\Cref{thm:ha_soundness} entails that any specific monadic realizability notion is a sound semantics for at least \(\HA\). 
Later, when we prove that \(\HA+\EM\) is sound with respect to interactive realizability semantics, 
we will only need to show that \(\EM\) is sound since the soundness of \(\HA\) derives from \cref{thm:ha_soundness}. 

\begin{wip} 
  While the intended purpose of a monadic realizability notion is realizing proofs employing non constructive axioms and rules, it is important to show
  that, when applied to a constructive proof in \(\HA\), monadic realizers retain the desirable properties we expect from standard modified realizers of constructive proofs. 
  In particular we show that monadic realizers yield effective methods to produce witnesses for existential statements and to decide disjunctions. 

  \begin{proposition}[Monadic Property for \(\monStarN k\)]
    Let \( f : \ta_1 \tarrow \dotsb \tarrow \ta_{k} \tarrow T\tc \) such that for all \( a_1 : \ta_1, \dotsc, a_k : \ta_k \) exists \( b : \tc\) such that \( f a_1 \dotsm a_k \leadsto \monUnit b \).
    Then for all \( a_1 : \ta_1, \dotsc, a_k : \ta_k \) there exists \( b : \tc \) such that 
    \[ \monStarN k f (\monUnit \tta_1) \dotsm (\monUnit \tta_k) \leadsto \monUnit b. \] 
  \end{proposition}
  \begin{proof}
    We proceed by induction on \(k\). For \(k = 0\) it is trivial and for \(k = 1\) 
    \begin{align*}
      \monStarN{1} f (\monUnit \tta) 
      &\equiv \monStar f (\monUnit \tta) \\ 
      &\leadsto f \tta && \text{by \cref{mon2}} \\ 
      &\leadsto \monUnit b && \text{for some \(b:\tc\), by hypothesis on } f.
    \end{align*}
    For \( k > 1 \) we have:
    \begin{align*}
      \monStarN{k} &f (\monUnit \tta_1) \dotsm (\monUnit \tta_k) \\ 
                   &\equiv (\abstr{\tta_1}{T\ta_1} \abstr{\tta_2}{T\ta_2} \monStarN{k-1} (\abstr{z}{\ta_1 \times \ta_2} f (\prl z) (\prr z) ) (\monMerge \tta_1 \tta_2)) (\monUnit \tta_1) \dotsm (\monUnit \tta_k), \\ 
                   &\leadsto \monStarN{k-1} (\abstr{z}{\ta_1 \times \ta_2} f (\prl z) (\prr z) ) (\monMerge (\monUnit \tta_1) (\monUnit \tta_2)) (\monUnit \tta_3) \dotsm (\monUnit \tta_k), \\ 
                   &\leadsto\monStarN{k-1} (\abstr{z}{\ta_1 \times \ta_2} f (\prl z) (\prr z) ) (\monUnit \pair \tta_1 \tta_2) (\monUnit \tta_3) \dotsm (\monUnit \tta_k) && \text{by \cref{mon3}}. \\ 
    \end{align*}
    The inductive hypothesis on \(k-1\) gives the conclusion if we can prove that for all \(z : \ta_1 \times \ta_2, \tta_3 : \ta_3, \dotsc, \tta_k : \ta_k\) there exists \( b : \tc \) such that: 
    \[ (\abstr{z}{\ta_1 \times \ta_2} f (\prl z) (\prr z)) z \tta_3 \dotsm \tta_k \leadsto \monUnit b. \]
    Reducing we get
    \begin{align*}
      &(\abstr{z}{\ta_1 \times \ta_2} f (\prl z) (\prr z)) z \tta_3 \dotsm \tta_k \\ 
      \reducesto_\beta &f (\prl z) (\prr z)  \tta_3 \dotsm \tta_k \\ 
      \leadsto & \monUnit b && \text{ by hypothesis on } f, 
    \end{align*}
    and we can conclude.
  \end{proof}
\end{wip}

\section{Monadic Interactive Realizability} \label{sec:monadic_interactive_realizability} 

\setmonad\IR

\note{Idea: use Ex as the type of false. Maybe exception realizability is for intuitionistic logic. Then prove that only provable formulas have non exceptional realizers (a kind of second validity theorem)}
In this section we define interactive realizability as a particular notion of monadic realizability. 
Thus we show that monadic realizability may realize a sub-classical principle, in this case excluded middle restricted to semi-decidable statements. 

\subsection{A Syntactic Monad for Interactive Realizability}
In order to describe the computational properties of interactive realizability (see \cite{aschieriB10}) we need to define a suitable monad. 
As we said, interactive realizability is based on the idea of learning by trial and error. 
We express the idea of trial and error with an exception monad: a term of intended type \(\ta\) has actual type \(\ta + \Ex\), where \(\Ex\) is the type of exceptions, so that a computation may either return its intended value or an exception. 
The learning part, which is described by the dependency on a knowledge state, fits with a part of the side-effects monad (see \cite{moggi91} for more details): a term of intended type \(\ta\) has actual type \(\State \tarrow \ta\), where \(\State\) is the type of knowledge states, so that the value of a computation may change with the state. 
The syntactic monad we are about to define for interactive realizability combines these two monads. 

We need to extend system \(T'\) with two base types \(\State\) and \(\Ex\) and a term constant that ``merges'' two exceptions into one:
\[ \exmerge : \Ex \tarrow \Ex \tarrow \Ex. \]
We shall avoid defining a specific syntax for terms of type \(\State\) and \(\Ex\). 
\fixme{next sentence is badly phrased}
Instead we exhibit their intended interpretation and, using this interpretation as a guide, we shall require some properties on reductions involving them. 

We write \(\LRel_k\) for the set of symbols of the \(k\)-ary predicates in \(\HA\). 
The intended interpretation of a (knowledge) state \(s\) is a partial function
\[ \interp s : \left( \bigcup_{k=0}^\infty \LRel_{k+1} \times \N^k \right) \rightharpoonup \N, \] 
that sends a \(k+1\)-ary predicate symbol \(P\) and a \(k\)-tuple of parameters \( m_1, \dotsc, m_k \in \N \) to a witness for \(\qexists{x} P (\num m_1, \dotsc, \num m_k, x) \).
We interpret the fact that a state \(s\) is undefined for some \(P,m_1,\dotsc,m_k\) as a lack of knowledge about a suitable witness. 
This is either due to the state being incomplete, meaning that there exists a suitable witness \(m\) we could use to extend the state by setting \( s(P,(m_1,\dotsc,m_k)) = m\), or to the fact that there are no suitable witness, meaning that \( \qforall{x} \lnot P(\num m_1, \dotsc, \num m_k, x) \) holds\footnote{Here we are using \(\EM\) at the metalevel in order to explain the possible situations. Using a principle at the metalevel in order to justify the same principle in the logic is a common practice. In our treatment this is not problematic because we never claim to be able to effectively decide which situation we are in.}. 
We require that \(s\) satisfies two properties. 
The first is for \(s\) to be \emph{sound}, meaning that its values are actually witnesses. More precisely: 
\[ \interp s (P, (m_1, \dotsc, m_k)) = m \text{ entails } P(\num m_1, \dotsc, \num m_k, \num m). \]
The second is that \(s\) is \emph{finite}, namely that the domain of \(s\) (the set of values \(s\) is defined on) is finite. 
This because we want a knowledge state to encode a finite quantity of information. 
Let \(\interp\State\), the set of all finite sound states, be the intended interpretation of the type \(\State\). 
Recall that there is a canonical partial order on states given by the extension relation: we write \(s_1 \leq s_2\) and read ``\(s_2\) \emph{extends} \(s_1\)'' if and only if \(s_2\) is defined whenever \(s_1\) is and with the same value. 

An exception \(e : \Ex \) is produced when we instantiate an assumption of the form \( \qforall{x} \allowbreak \lnot P(\num m_1, \dotsc, \allowbreak \num m_k, x) \) with some \(m\) such that \( \lnot P(\num m_1, \dotsc, \num m_k, \num m) \) does not actually hold (remember that we proceed by trial and error, in particular we may assume things that are actually false). 
This means that \(m\) is a witness for \(\qexists{x} P (\num m_1, \dotsc, \num m_k, x) \), in particular it could be used to extend the knowledge state on values where it was previously undefined. 
The role of exceptions is to encode information about the discovery of new witnesses: since we use this information to extend states the intended interpretation of an exception \(e\) is as a partial function: 
\[ \interp e : \interp\State \rightharpoonup \interp\State. \]
Since \(e\) extends states we require that \( s \leq e(s) \). 
We interpret an exception as a partial function because an exception \(e\) may fail to extend some state \(s\). 
The reason is that \(e\) may contain information about a witness \(m'\) for an existential statement \(\qexists{x} P (\num m_1, \dotsc, \num m_k, x) \) on which \(s\) is already defined as \(m\). 
Note that an existential formula can have more that one witness so two cases may arise: 
either \(m = m'\), meaning that the information of \(e\) is already part of \(s\) or \(m \neq m'\) so that the information of \(e\) is incompatible with the information of the state. 
In the first case \( e(s) = s \), while in the second case \( e(s) \) is not defined. 

Before defining the syntactic monad \(\IR\) for interactive realizability, we need to introduce some terminology on exceptions and states. 
\begin{definition}[Terminology on Exceptions and States]
  We say that a term of type \(\ta + \Ex\) is either a \emph{regular} value \(a\) if it reduces to \(\inl a\) for some term \(a : \ta\) or an \emph{exceptional} value if it reduces to \(\inr e\) for some term \(e : \Ex\). 
  We say that a term of type \(\State \tarrow \ta\) is a \emph{state function}. 
  Finally we say that an exception \(e\) \emph{properly extends} \(s\) if \(e(s)\) is defined and \( s < e(s) \). 
\end{definition}

Note that different exceptions might be used to extend a knowledge state in incompatible ways, that is, by sending the same predicate symbol and the same tuple of parameters into different witnesses. 
The role of the \(\exmerge\) function is to put together the information from two exceptions into a single  exception. 
This means that \(\exmerge\) cannot simply put together all the information from its argument: if such information contains more that one distinct witness for the same existential statement it must choose one in some arbitrary way, for instance the leftmost or the minimum witness.
Many choices for \(\exmerge\) are possible, provided that they satisfy the following property:
\begin{equation}\label[property]{exmerge_prop} \tag{EX}
  \left. \begin{array}{r} 
    e_1 \text{ properly extends } s \\ 
e_2 \text{ properly extends } s \end{array} \right\} 
\text{ entails that } \exmerge e_1 e_2 \text{ properly extends } s,
\end{equation}
for any state \(s\) and exceptions \(e_1, e_2\).
Simple choices for \(\exmerge\) are the projections, always selecting the first or the second argument, or any combination of them using an arbitrary criterion to select which value to return. 
Of course, in general there is no need for \(\exmerge e_1 e_2\) to be \(e_1\) or \(e_2\). 

Before the definition we give an informal description of \(\IR\). 
The monad \(\IR\) maps a type \(\ta\) to \(\State \tarrow (\ta + \Ex)\), that is, values of type \(\ta\) are lifted to state functions that can throw exceptions. 
The term \(\monUnit\) maps a value \(a : \ta\) to a constant state function that returns the regular value \(a\). 
If \(f : \ta \tarrow \T\tb\) then \(\monStar f\) is a function with two arguments, a state \(s\) and a state function \(\mon a : \T\ta\). 
It evaluates \(\mon a\) on \(s\): if this results in a regular value \(a : \ta\) it applies \(f\) to \(a\), otherwise it propagates the exceptional value. 
Lastly, if \(\mon a : \T\ta\) and \(\mon b : \T\tb\) are two state functions, then \(\monMerge \mon a \mon b\) is a state function that evaluates its arguments on its state argument:
when both arguments are regular values it returns their pair, otherwise it propagates the exception(s), using \(\exmerge\) if both arguments are exceptional values. 

We are now ready to give the formal definition of \(\IR\). 
\begin{definition}[Interactive Realizability Monad] \label{def:ir_monad}
  Let \(\IR\) be the tuple \( (\T\), \(\monUnit\), \(\monStar\), \(\monMerge)\), where 
  \begin{align*} 
    \T \ta 
    &= \State \tarrow (\ta + \Ex), \\ 
    \monUnit[\ta] 
    &\equiv \abstr\tta\ta \abstrS{\_} \inl[\ta,\Ex] \tta, \\ 
    \monStar[\ta,\tb] 
    &\equiv \abstr{f}{\ta \tarrow \T\tb} \abstr{\mon\tta}{\T\ta} 
    \abstrS{s} \case[\ta,\Ex,\tb+\Ex] (\mon\tta s) 
    (\abstr\tta\ta f \tta s) \inr[\tb,\Ex], \\  
    \monMerge[\ta,\tb] 
    &\equiv \abstr{\mon\tta}{\T\ta} \abstr{\mon\ttb}{\T\tb} \abstrS{s} \case[\ta,\Ex,(\ta \times \tb)+\Ex] (\mon \tta s) \\ 
    &\phantomrel\equiv (\abstr\tta\ta \case[\tb,\Ex,(\ta \times \tb)+\Ex] (\mon\ttb s) (\abstr\ttb\tb \inl[\ta \times \tb,\Ex] (\pair \tta \ttb)) \inr[\ta \times \tb,\Ex]) \\ 
    &\phantomrel\equiv (\abstr{e_1}\Ex \case[\tb,\Ex,(\ta \times \tb)+\Ex] (\mon\ttb s) (\abstr\_\tb \inr[\ta \times \tb,\Ex] e_1) (\abstr{e_2}\Ex \inr[\ta \times \tb,\Ex] (\exmerge e_1 e_2))),
  \end{align*} 
  for some \(\exmerge\) satisfying \cref{exmerge_prop}. 
\end{definition} 

The term \(\monUnit[\ta]\) takes a value \(a : \ta\) and produces a constant state function that returns the regular (as opposed to exceptional) value \(a\). 
The term \(\monStar[\ta,\tb]\) takes a function \(f : \ta \tarrow \T\tb\) and returns a function \(f'\) which lifts the domain of \(f\) to \(\T\ta\). 
The state function returned by \(f' \) when applied to some \(\mon a : \T\ta\) behaves as follows: it evaluates \(\mon a\) on the state and if \(\mon a s\) is a regular value \(a : \ta\) it returns \(f a\); otherwise if \(\mon a s\) is an exception it simply propagates the exception. 
\begin{thesis}
  \[ 
    \begin{array}{c|c} 
      \mon a s & \mon b s \\ 
      \hline
      a & f a \\ 
      e & e \\ 
    \end{array} 
  \] 
\end{thesis}
The term \(\monMerge[\ta,\tb]\) takes two state functions \(\mon a : \T\ta\) and \(\mon b : \T\tb\) and returns a state function \( \mon c : \T(\ta \times \tb)\). 
When both arguments are regular values it returns their pair, otherwise it propagates the exception(s), using \(\exmerge\) if both arguments are exceptional. 
\begin{thesis}
  \[ \begin{array}{cc|c} 
      \mon a s & \mon b s & \mon c s \\ 
      \hline
      a & b & \pair a b \\ 
      e_1 & b & e_1 \\ 
      a & e_2 & e_2 \\ 
      e_1 & e_2 & \exmerge e_1 e_2 
  \end{array} \] 
\end{thesis}

We still need to check that \cref{def:ir_monad} is correct and that \(\IR\) really is a syntactic monad. 
\begin{proposition}[The Syntactic Monad \(\IR\)]
  \(\IR\) is a syntactic monad. 
\end{proposition} 
\begin{proof}\allowdisplaybreaks[1]
  We just need to check that \(\monUnit\), \(\monStar\) and \(\monMerge\) satisfy all the properties in \cref{def:syntactic_monad}. 
  This amounts to perform some reductions. 
  \begin{itemize}
    \item[\ref{mon1}]
      Given any \(\mon\tta : \T\ta\), we have:
      \begin{align*}
        &\monStar[\ta,\ta] \monUnit[\ta] \mon\tta \\
        &\equiv 
        (\abstr{f}{\ta \tarrow \T\ta} \abstr{\mon \tta}{\T\ta} \abstrS{s} \case[\ta,\Ex,\ta+\Ex](\mon \tta s) (\abstr\tta\ta f \tta s) \inr[\ta,\Ex])
        \monUnit[\ta] \mon\tta \\ 
        &\reducesto_\beta 
        (\abstr{\mon\tta}{\T\ta} \abstrS{s} \case[\ta,\Ex,\ta+\Ex](\mon \tta s) (\abstr\tta\ta \monUnit[\ta] \tta s) \inr[\ta,\Ex])
        \mon\tta \\ 
        &\reducesto_\beta 
        \abstrS{s} \case[\ta,\Ex,\ta+\Ex](\mon\tta s) 
        (\abstr\tta\ta \monUnit[\ta] \tta s) 
        \inr[\ta,\Ex] \\ 
        &\equiv
        \abstrS{s} \case[\ta,\Ex,\ta+\Ex](\mon\tta s) 
        (\abstr\tta\ta (\abstr\tta\ta \abstrS{\_} \inl[\ta,\Ex] \tta) \tta s) 
        \inr[\ta,\Ex] \\ 
        &\reducesto_\beta 
        \abstrS{s} \case[\ta,\Ex,\ta+\Ex](\mon\tta s) 
        (\abstr\tta\ta (\abstrS{\_} \inl[\ta,\Ex] \tta) s) 
        \inr[\ta,\Ex] \\ 
        &\reducesto_\beta 
        \abstrS{s} \case[\ta,\Ex,\ta+\Ex](\mon\tta s) 
        (\abstr\tta\ta \inl[\ta,\Ex] \tta) 
        \inr[\ta,\Ex] \\ 
        &=_\eta
        \abstrS{s} \case[\ta,\Ex,\ta+\Ex](\mon\tta s) 
        \inl[\ta,\Ex] \inr[\ta,\Ex] \\ 
        &=_\times
        \abstrS{s} \mon\tta s \\
        &=_\eta
        \mon\tta, 
      \end{align*}
      as required by \cref{mon1}.
    \item[\ref{mon2}] 
      Given any \( f : \ta \tarrow \T\tb \) and \(\tta : \ta\), we have:
      \begin{align*}
        &\monStar[\ta,\tb] f (\monUnit[\ta] \tta) \\ 
        &\equiv 
        (\abstr{f}{\ta \tarrow \T\tb} \abstr{\mon\tta}{\T\ta} \abstrS{s} \case[\ta,\Ex,\tb+\Ex](\mon\tta s) (\abstr\tta\ta f \tta s) \inr[\tb,\Ex]) f 
        (\monUnit[\ta] \tta) \\ 
        &\reducesto_\beta 
        (\abstr{\mon\tta}{\T\ta} \abstrS{s} \case[\ta,\Ex,\tb+\Ex](\mon\tta s) (\abstr\tta\ta f \tta s) \inr[\tb,\Ex])
        (\monUnit[\ta] \tta) \\ 
        &\reducesto_\beta 
        \abstrS{s} \case[\ta,\Ex,\tb+\Ex] 
        (\monUnit[\ta] \tta s) 
        (\abstr\tta\ta f \tta s) 
        \inr[\tb,\Ex] \\ 
        &\equiv
        \abstrS{s} \case[\ta,\Ex,\tb+\Ex] 
        ((\abstr\tta\ta \abstrS{\_} \inl[\ta,\Ex] \tta) \tta s) 
        (\abstr\tta\ta f \tta s) 
        \inr[\tb,\Ex] \\ 
        &\reducesto_\beta 
        \abstrS{s} \case[\ta,\Ex,\tb+\Ex] 
        (\inl[\ta,\Ex] \tta) 
        (\abstr\tta\ta f \tta s) 
        \inr[\tb,\Ex] \\ 
        &\reducesto_\times 
        \abstrS{s} (\abstr\tta\ta f \tta s) \tta \\
        &\reducesto_\beta 
        \abstrS{s} f \tta s \\
        &=_\eta
        f \tta,
      \end{align*}
      as required by \cref{mon2}.


    \item[\ref{mon_merge1}] 
      Given any \(\tta : \ta \) and \( \ttb : \tb\), we have:
      \begin{align*}
        \monMerge &(\monUnit \tta) (\monUnit \ttb) \\
                  &\equiv (\abstr{\mon \tta}{\T\ta} \abstr{\mon \ttb}{\T\tb} \abstrS{s} \case[\ta,\Ex,(\ta \times \tb)+\Ex] (\mon \tta s)  \\ 
                  &\phantomrel\equiv \quad (\abstr\tta\ta \case[\tb,\Ex,(\ta \times \tb)+\Ex] (\mon \ttb s) (\abstr\ttb\tb \inl[\ta \times \tb,\Ex] (\pair\tta\ttb)) \inr[\ta \times \tb,\Ex]) \\ 
                  &\phantomrel\equiv \quad 
        (\abstr{e_1}\Ex \case[\tb,\Ex,(\ta \times \tb)+\Ex] (\mon\ttb s) (\abstr\tta\tb \inr[\ta \times \tb,\Ex] e_1)(\abstr{e_2}\Ex \inr[\ta \times \tb,\Ex] (\exmerge e_1 e_2) )) ) \\ 
        &\phantomrel\equiv \quad 
        (\monUnit \tta) (\monUnit \ttb) \\  
        &\reducesto_\beta \abstrS s \case[\ta,\Ex,(\ta \times \tb)+\Ex] (\monUnit \tta s)  \\ 
        &\phantomrel{\reducesto_\beta} (\abstr\tta\ta \case[\tb,\Ex,(\ta \times \tb)+\Ex] (\monUnit \ttb s) (\abstr\ttb\tb \inl[\ta \times \tb,\Ex] (\pair\tta\ttb)) \inr[\ta \times \tb,\Ex]) (\dotso) \\ 
        &\equiv \abstrS s \case[\ta,\Ex,(\ta \times \tb)+\Ex] ((\abstr\tta\ta \abstrS{\_} \inl[\ta,\Ex] \tta) \tta s)  \\ 
        &\phantomrel{\reducesto_\beta} (\abstr\tta\ta \case[\tb,\Ex,(\ta \times \tb)+\Ex] (\monUnit \ttb s) (\abstr\ttb\tb \inl[\ta \times \tb,\Ex] (\pair\tta\ttb)) \inr[\ta \times \tb,\Ex]) (\dotso) \\
        &\reducesto_\beta \abstrS s \case[\ta,\Ex,(\ta \times \tb)+\Ex] (\inl[\ta,\Ex] \tta) \\ 
        &\phantomrel{\reducesto_\beta} (\abstr\tta\ta \case[\tb,\Ex,(\ta \times \tb)+\Ex] (\monUnit \ttb s) (\abstr\ttb\tb \inl[\ta \times \tb,\Ex] (\pair\tta\ttb)) \inr[\ta \times \tb,\Ex]) (\dotso) \\
        &\reducesto_+ \abstrS s \case[\tb,\Ex,(\ta \times \tb)+\Ex] (\monUnit \ttb s) (\abstr\ttb\tb \inl[\ta \times \tb,\Ex] (\pair\tta\ttb)) \inr[\ta \times \tb,\Ex] \\
        &\equiv \abstrS s \case[\tb,\Ex,(\ta \times \tb)+\Ex] ((\abstr\ttb\tb \abstrS{\_} \inl[\tb,\Ex] \ttb) \ttb s) (\abstr\ttb\tb \inl[\ta \times \tb,\Ex] (\pair \tta \ttb)) \inr[\ta \times \tb,\Ex] \\
        &\reducesto_\beta \abstrS s \case[\tb,\Ex,(\ta \times \tb)+\Ex] (\inl[\tb,\Ex] \ttb) (\abstr\ttb\tb \inl[\ta \times \tb,\Ex] (\pair \tta \ttb)) \inr[\ta \times \tb,\Ex] \\
        &\reducesto_+ \abstrS s (\abstr\ttb\tb \inl[\ta \times \tb,\Ex] (\pair\tta\ttb)) \ttb \\
        &\reducesto_\beta \abstrS s \inl[\ta \times \tb,\Ex] (\pair \tta \ttb) \\ 
        &\equiv \monUnit[\ta \times \tb] (\pair \tta \ttb), 
      \end{align*}
      as required by \cref{mon_merge1}.
  \end{itemize}
\end{proof}


\subsection{The Interactive Realizability Semantics}

We now define a family of monadic realizability relations, one for each state \(s\), requiring that a realizer, applied to a knowledge state \(s\), 
either realizes a formula in the sense of the BHK semantics or can extend \(s\) with new knowledge.
\begin{definition}[Interactive Realizability Relation] 
  Let \(s\) be a state, \( \mon\rr : \mm\fa \) be a term and \(\fa\) a closed formula. 
  We define two realizability relations \(\monRe^s\) and \(\re^s\) by simultaneous induction on the structure of \(\fa\): 
  \begin{itemize}
    \item \(\mon\rr \monRe^s \fa \) if and only if we have that \(\mon\rr s\) is either a regular value \(\rr\) such that \(\rr \re^s \fa\) or an exceptional value \(e\) such that \(e\) properly extends \(s\), 
    \item \(\re^s\) is defined in terms of \(\monRe^s\) by the clauses in \cref{def:monadic_realizability_relation}. 
  \end{itemize}
  We say that \(\mon\rr\) (resp. \(\rr\)) is a \emph{monadic} (resp. \emph{inner}) \emph{interactive realizer} of \(\fa\) with respect to \(s\) when \(\mon\rr : \mm\fa\) (resp. \(\rr : \m\fa\)) and \(\mon\rr \monRe^s \fa\) (resp. \(\rr \re^s \fa\)).
\end{definition}

In order to show that any interactive realizability relations with respect to a state is a monadic realizability relation we need to verify that is satisfies the required properties. 
\begin{proposition}[The Monadic Realizability Relation \(\monRe^s\)] 
  For any state \(s\), \(\monRe^s\) is a monadic realizability relation. 
\end{proposition} 
\begin{proof} 
  Let \(s\) be any state. 
  We have to show that \(\monRe^s\) satisfies the properties in \cref{def:monadic_realizability_relation}. 

  \begin{description}
    \item[\ref{real1}]
      We begin with \cref{real1}, namely, for any inner interactive realizer \(\rr\) of a formula \(\fa\) with respect to \(s\), we show that: 
      \[ \monUnit \rr \monRe^s \fa. \] 
      By unfolding the definition of \(\monUnit\) we have that:
      \begin{align*}
        \monUnit \rr s &\leadsto (\abstr{\_}\State \inl \rr) s \\ 
                       &\leadsto \inl \rr, 
      \end{align*}
      thus, by definition of \(\monRe^s\), we have to check that \( \rr \re^s \fa \), which holds by assumption. 

    \item[\ref{real2}]
      In order to show \cref{real2}, for any formulas \(\fa\) and \(\fb\), we take an inner interactive realizer \(\rr\) of \(\fa \limply \fb\) with respect to \(s\), 
      that is, a term \( \rr : \m\fa \tarrow \mm\fb \) such that \( \rr \ra \) is a monadic interactive realizer of \(\fb\) with respect to \(s\), for any inner interactive realizer \(\ra\) of \(\fa\) with respect to \(s\). 
      Then we have to show that, given a monadic interactive realizer \( \mon\ra \) of \(\fa\) with respect to \(s\), we have: 
      \[ \monStar \rr \mon\ra \monRe^s \fb. \]
      By definition of \(\monRe^s\) we apply \(s\) to the realizer and 
      by unfolding the definition of \(\monStar\) and reducing we get:
      \begin{equation} \label{eq:mr2r} 
        \monStar \rr \mon\ra s \leadsto 
        \case (\mon\ra s) (\abstr\tta{\m\fa} \rr \tta s) \inr. 
      \end{equation}
      Since \(\mon\ra \monRe^s \fa \), we know that \(\mon\ra s\) reduces to either a regular value \(\inl \ra\), for some inner realizer \(\ra\) of \(\fa\) with respect to \(s\), or an exceptional value \(\inr e\), for some exception \(e\) that properly extends \(s\). 
      \begin{itemize}
        \item In the former case, \eqref{eq:mr2r} reduces to \( \rr \ra s \).
          By the assumptions we made on \(\rr\) and \(\ra\), 
          \(\rr \ra\) is a monadic interactive realizer of \(\fb\) with respect to \(s\), 
          and thus \(\rr \ra s\) reduces to either 
          a regular value which is an inner interactive realizer of \(\fb\) with respect to \(s\) or 
          an exceptional value which properly extends \(s\). 
          Thus \(\monStar \rr \mon\ra\) is a monadic interactive realizer of \(\fb\) with respect to \(s\) as required.
        \item In the latter case, \eqref{eq:mr2r} reduces to \(\inr e\). 
          Since \(e\) properly extends \(s\), \(\monStar \rr \mon\ra\) is again a monadic interactive realizer of \(\fb\) with respect to \(s\) as required.
      \end{itemize}

    \item[\ref{real3}]
      Finally we have to show \cref{real3}. 
      We assume that \(\mon\ra\) and \(\mon\rb\) are monadic interactive realizers of \(\fa\) and \(\fb\) respectively, both with respect to \(s\).
      Then we have to show that:
      \[ \monMerge \mon\ra \mon\rb \monRe^s \fa \land \fb. \] 
      By definition of \(\monRe^s\), this means we have to show that 
      \[ \monMerge \mon\ra \mon\rb s \]  
      reduces to either a regular value which is an inner interactive realizers 
      Since \(\mon\ra\) and \(\mon\rb\) are monadic interactive realizers, \(\mon\ra s\) and \(\mon\rb s\) either reduce to regular values \(\inl \ra\) and \(\inl \rb\), where \(\ra\) and \(\rb\) are inner interactive realizers of respectively \(\fa\) and \(\fb\) with respect to \(s\), or to exceptional values \(\inr e_1\) and \(\inr e_2\), where \(e_1\) and \(e_2\) properly extend \(s\).
      By unfolding the definition of \(\monMerge\) and reducing we get:
      \begin{equation} \label{eq:mr3r}
        \begin{split}
          \monMerge \mon\ra \mon\rb s 
          &\leadsto 
          \case (\mon\ra s) (\abstr\tta{\m\fa} \case (\mon\rb s) (\abstr\ttb{\m\fb} \inl (\pair \tta \ttb)) \inr) \\ 
          &\phantomrel\leadsto 
          (\abstr{e_1}\Ex \case (\mon\rb s) (\abstr\_{\m\fb} \inr e_1)(\abstr{e_2}\Ex \inr (\exmerge e_1 e_2) ))
        \end{split}
      \end{equation}
      We distinguish four cases depending on how \(\mon\ra s\) and \(\mon\rb s\) reduce:
      \begin{description}
        \item[\( \mon\ra s \leadsto \inl \ra \) and \( \mon\rb s \leadsto \inl \rb \)]
          In this case \eqref{eq:mr3r} reduces as follows:
          \begin{align*}
            \monMerge \mon\ra \mon\rb s 
            &\leadsto \case (\mon\rb s) (\abstr\ttb{\m\fb} \inl (\pair \ra \ttb)) \inr \\ 
            &\leadsto \inl (\pair \ra \rb). 
          \end{align*}
          Since it is a regular value, we have to show that \(\pair \ra \rb \re^s \fa \land \fb \). 
          This follows by definition of \(\re^s\) and from the assumption that  \(\ra \re^s \fa\) and \(\rb \re^s \fb\). 
        \item[\( \mon\ra s \leadsto \inl \ra \) and \( \mon\rb s \leadsto \inr e_2 \)]
          In this case \eqref{eq:mr3r} reduces as follows:
          \begin{align*}
            \monMerge \mon\ra \mon\rb s 
            &\leadsto \case (\mon\rb s) (\abstr\ttb{\m\fb} \inl (\pair \ra \ttb)) \inr \\ 
            &\leadsto \inr e_2. 
          \end{align*}
          Since it is an exception value, we have to show that \(e_2\) properly extends \(s\). 
          This follows by the assumption that \(\rb \re^s \fb\). 
        \item[\( \mon\ra s \leadsto \inr e_1 \) and \( \mon\rb s \leadsto \inl \rb \)]
          In this case \eqref{eq:mr3r} reduces as follows:
          \begin{align*}
            \monMerge \mon\ra \mon\rb s 
            &\leadsto \case (\mon\rb s) (\abstr\_{\m\fb} \inr e_1) (\abstr{e_2}\Ex \inr (\exmerge e_1 e_2)) \\ 
            &\leadsto \inr e_1
          \end{align*}
          Since it is an exception value, we have to show that \(e_1\) properly extends \(s\). 
          This follows by the assumption that \(\ra \re^s \fa\). 
        \item[\( \mon\ra s \leadsto \inr e_1 \) and \( \mon\rb s \leadsto \inr e_2 \)]
          \begin{align*}
            \monMerge \mon\ra \mon\rb s 
            &\leadsto \case (\mon\rb s) (\abstr\_{\m\fb} \inr e_1) (\abstr{e_2}\Ex \inr (\exmerge e_1 e_2)) \\ 
            &\leadsto \inr (\exmerge e_1 e_2)
          \end{align*}
          Since it is an exception value, we have to show that \( \exmerge e_1 e_2 \) properly extends \(s\). 
          By \cref{exmerge_prop}, this happens whenever both \(e_1\) and \(e_2\) properly extends \(s\). 
          This is the case by the assumption that \(\ra \re^s \fa\) and \(\rb \re^s \fb\). \qedhere
      \end{description}
  \end{description}
\end{proof} 

Following \cref{def:mon_sem}, for each state \(s\), the monadic realizability relation \(\monRe^s\) induces a monadic realization semantics, which realizes \(\HA\) by \cref{thm:ha_soundness}.  
We employ this family of semantics indexed by a state in order to define another one, which does not depend on a state. 
\begin{definition}[Interactive Realizability Semantics] \label{def:interactive_realizability_semantics}
  We say that the decorated sequent \(\Gamma \monSeq \mon\rr : \fa\) is valid if and only if it is valid with respect to the semantics induced by each \(\monRe^s\) for every state \(s\). 
\end{definition} 
We shall show how we can realize \(\EM\) in this semantics.

\subsection{Realizing the Excluded Middle Axiom}
Interactive realizability aims at producing a realizer of the \(\EM\) axiom, a weakened form of the excluded middle restricted to \(\Sigma^0_1\) formulas. 
A generic instance of \(\EM\) is written as:
\[ 
  \EM(\lafa, t_1, \dotsc, t_k) \equiv (\qforall\lvb \lafa(t_1, \dotsc, t_k, \lvb)) \lor (\qexists\lvb \lnot \lafa(t_1, \dotsc, t_k, \lvb)). 
\]
for any \(k+1\)-ary relation \(\lafa\) and arithmetic terms \(t_1, \dotsc, t_k\). 
We call \emph{universal} (resp. \emph{existential}) \emph{disjunct} the first (resp. the second) disjunct of \(\EM(\lafa, t_1, \dotsc, t_k)\). 
For more information on \(\EM\) see \cite{akamaBHK04}. 

The main hurdle we have to overcome in order to build a realizer of \( \EM(\lafa, t_1, \dotsc, t_k) \) is that, by the well-known undecidability of the halting problem, there is no total recursive function that can choose which one of the disjuncts holds. 
Moreover, if the realizer chooses the existential disjunct, it should also be able to provide a witness. 

As we said before terms of type \(\State\) contain knowledge about witnesses of \(\Sigma^0_1\) formulas. 
In order to query a state \(s\) for a witness \(n\) of 
\(\qexists\lvb \lafa(\num n_1, \dotsc, \num n_k, \lvb)\) for some natural numbers \(n_1, \dotsc, n_k\), we need to extend system \(T'\) with the family of term constants:
\[ \query_\lafa : \State \tarrow \underbrace{\Nat \tarrow \dotsb \tarrow \Nat}_k \tarrow \Unit + \Nat. \] 
indexed by \(\lafa \in \LRel_{k+1} \) (and implicitly by \(k \ge 0\)). 
The value of \( \query_\lafa s \num n_1 \dotsm \num n_k \) should be either \(\unit\) if the \(s\) contains no information about such an \(n\) or a numeral \(\num n\) such that \( \interp\lafa(n_1, \dotsc, n_k, n) \) is true. 
More formally we require that \(\query_\lafa\) satisfies the following syntactic property: 
\begin{equation} \tag{IR1} \label[property]{query_prop}
  \query_\lafa s \num n_1 \dotsm \num n_k \leadsto \inr \num n \text{ entails that } \lafa (\num n_1, \dotsc, \num n_k, \num n) \text{ holds}
\end{equation} 
for all natural numbers \( n_1, \dotsc, n_k \). 
This amounts to require that state do not answer with wrong witnesses and it follows immediately from the intended interpretation if we suitably define \( \query_\lafa s \num n_1 \dotsm \num n_k \) using \( \interp s (\lafa, (n_1, \dotsc, n_k)) \). 

An interactive realizer \(\mon\rr_\lafa\) of \(\EM(\lafa)\) will behave as follows. 
When it needs to choose one of the disjuncts it queries the state. 
If the state answer with a witness, \(\mon\rr_\lafa\) reduces to a realizer \(\mon\rr_\exists\) of the existential disjunct containing the witness given by the state. 
Otherwise we can only assume (since we do not know any witness) that the universal disjunct holds and thus \(\mon\rr_\lafa\) reduces to a realizer \(\mon\rr_\forall\) of the universal disjunct. 
This assumption may be wrong if the state is not big enough. 
When \(\mon\rr_\forall\) is evaluated on numerals (this correspond to the fact that an instance \(\lafa(\num n_1, \dotsc, \num n_k, \num n)\) of the universal disjunct assumption is used in the proof), \(\mon\rr_\forall\) checks whether the instance holds. 
If this is not the case the realizer made a wrong assumption and \(\mon\rr_\forall\) reduces to an exceptional value, with the effect of halting the regular reduction and returning the exceptional value. 
For this we need to extend the system \(T'\) with the last family of terms:
\[ \eval_\lafa : \underbrace{\Nat \tarrow \dotsb \tarrow \Nat}_k \tarrow \Nat \tarrow \Unit + \Ex, \] 
again indexed by \(\lafa \in \LRel_k\). 
We shall need \(\eval_\lafa\) to satisfy the following property: 
\[ \tag{IR2} \label[property]{eval_prop}
  \eval_\lafa \num n_1 \dotsm \num n_k \num n \leadsto \inl \unit \text{ entails that } \lafa (\num n_1, \dotsc, \num n_k, \num n) \text{ does not hold}, 
\] 
for all natural numbers \(n_1, \dotsc, n_k, n \). 
This guarantees that if the universal disjunct instance does not hold \(\eval_\lafa\) reduces to an exceptional value. 
Thus an interactive realizer which uses a false instance of an universal assumption cannot reduce to a regular value. 

The last property we need is that for any state \(s\) and natural numbers \(n_1, \dotsc, n_k\), 
\[ \tag{IR3} \label[property]{query_eval_prop} 
  \left. \begin{array}{r}
    \query_\lafa s \num n_1 \dotsm \num n_k \leadsto \inl \unit \\
    \quad \eval \num n_1 \dotsm \num n_k \leadsto \inr e 
\end{array} \right\}
\text{ entails that \(e\) properly extends \(s\).} 
\]
This condition guarantees that we have no ``lazy'' realizers that throw exceptions encoding witnesses that are already in the state. 



\newtermconstant\emReal{\termname{em}}
Now we can define a realizer for \(\EM(\lafa, t_1, \dotsc, t_k)\) as follows:
\begin{align*}
  \mon\emReal(\lafa, t_1, \dotsc, t_k)
  \equiv \abstr{s}\State \inl 
  (\case &(\query_\lafa s t_1 \dotsm t_k) \\
         &(\abstr\_\Unit \inl (\abstr\lvb\Nat \abstr\_\State \eval_\lafa t_1 \dotsm t_k \lvb)) \\
         &(\abstr\lvb\Nat \inr (\pair \lvb \monUnit))). 
\end{align*} 
Of course we need to check that our definition is correct. 
\begin{proposition}[Interactive Realizer for \(\EM\)] \label{thm:realizer_for_em}
  Given any \(\EM\) instance \(\EM(\lafa, t_1, \dotsc, t_k)\), the decorated sequent:
  \begin{equation} \label{eq:em_seq} 
    \alpha_1 : \fa_1, \dotsc, \alpha_l : \fa_l \monSeq \mon\emReal(\lafa, t_1, \dotsc, t_k) : \EM(\lafa, t_1, \dotsc, t_k), 
  \end{equation}
  is valid with respect to the interactive realizability semantics given in \cref{def:interactive_realizability_semantics}. 
\end{proposition}
\begin{proof}
  Let \(\mon\rr\) and \(\fa\) stand for \(\mon\emReal(\lafa, t_1, \dotsc, t_k)\) and \(\EM(\lafa, t_1, \dotsc, t_k)\) in the following proof. 
  By \cref{def:interactive_realizability_semantics}, we have to prove that \eqref{eq:em_seq} is valid with respect to the semantics induced by \(\monRe^s\) for any given state \(s\). 

  Let the free (arithmetic) variables of \(\fa\) be \(\lva_1, \dotsc, \lva_m\) and let \(\Omega \equiv \lva_1 \coloneqq \num n_1\), \(\dotsc\), \(\lva_m \coloneqq \num n_m\) be a substitution for them. 
  Let \(\Sigma\) be a substitution for the assumption variables in \(\Gamma\). 
  Note that the only free variables in \(\mon\rr\) are arithmetic, thus \(\mon\rr[\Sigma]\) is the same as \(\mon\rr\).

  Thus we have to prove that 
  \[ \mon\rr[\Sigma,\Omega] \monRe^s \fa[\Omega]. \] 
  By definition of \(\monRe^s\), we apply \(s\) and reduce:
  \begin{align*}
    \mon\rr[\Sigma,\Omega] s \leadsto 
    \inl (\case &(\query_\lafa s t_1[\Omega] \dotsm t_k[\Omega]) \\
                &(\abstr\_\Unit \inl (\abstr\lvb\Nat \abstr\_\State \eval_\lafa t_1[\Omega] \dotsm t_k[\Omega] \lvb)) \\
                &(\abstr\lvb\Nat \inr (\pair \lvb \monUnit))), 
  \end{align*}
  and since \(\mon\rr[\Sigma,\Omega] s\) is a regular value, \(\mon\rr[\Sigma,\Omega]\) is a monadic realizer of \(\fa\) if and only if:
  \begin{equation} \label{eq:emr1}
    \begin{split}
      \case &(\query_\lafa s t_1[\Omega] \dotsm t_k[\Omega]) \\
            &(\abstr\_\Unit \inl (\abstr\lvb\Nat \abstr\_\State \eval_\lafa t_1[\Omega] \dotsm t_k[\Omega] \lvb)) \\
            &(\abstr\lvb\Nat \inr (\pair \lvb \monUnit)). 
    \end{split} 
  \end{equation}
  is an inner realizer for \(\fa\). 
  \( \query_\lafa s t_1[\Omega] \dotsm t_k[\Omega] \) reduces either to \(\inl \unit\) or to \(\inr \num n\) for some natural number \(n\). 
  We distinguish the two cases. 
  \begin{itemize}
    \item[\(\inl \unit\)]
      In the first case \eqref{eq:emr1} reduces to:
      \[
        \inl (\abstr\lvb\Nat \abstr\_\State \eval_\lafa t_1[\Omega] \dotsm t_k[\Omega] \lvb). 
      \]
      By definition of \(\re^s\), this is an inner realizer for \(\fa\) if  and only if:
      \[ 
        \rr_\forall \equiv \abstr\lvb\Nat \abstr\_\State \eval_\lafa t_1[\Omega] \dotsm t_k[\Omega] \lvb, 
      \]
      is an inner realizer for \(\qforall\lvb \lafa(t_1[\Omega], \dotsc, t_k[\Omega], \lvb)\). 
      Again by definition of \(\re^s\), this is the case if and only if 
      \[ 
        \rr_\forall \num n \monRe^s \lafa (t_1[\Omega], \dotsc, t_k[\Omega], \num n), 
      \] 
      for any natural number \(n\).
      Following the definition of \(\monRe^s\), we apply \(s\) to \(\rr_\forall \num n\) and reduce:
      \[ \rr_\forall \num n s \leadsto \eval_\lafa t_1[\Omega] \dotsm t_k[\Omega] \num n \]
      Then \( \rr_\forall \num n s \) reduces either to \(\inl \unit\) or to \(\inr e\), for some exception \(e\). 
      \begin{itemize}
        \item[\(\inl \unit\)]
          In the first case, we have to check that: 
          \[ \unit \re^s \lafa (t_1[\Omega], \dotsc, t_k[\Omega], \num n) \]
          By definition of \(\re^s\), this is the case if and only if \(\lafa (t_1[\Omega], \dotsc, t_k[\Omega], \num n) \) and this follows from \cref{eval_prop}. 
        \item[\(\inr e\)]
          In the second case, by definition of \(\monRe^s\), 
          we have to check that \(e\) properly extends \(s\) and 
          this follows from \cref{query_eval_prop}. 
      \end{itemize}
    \item[\(\inr \num n\)]
      In this case, \eqref{eq:emr1} reduces to:
      \[ \inr (\pair \num n \monUnit). \]
      By definition of \(\re^s\), this is an inner realizer for \(\fa\) if and only if 
      \[ \pair \num n \monUnit \]
      is an inner realizer for 
      \[ \qexists\lvb \lnot \lafa(t_1[\Omega], \dotsc, t_k[\Omega], \lvb). \] 
      Again by definition of \(\re^s\), this is the case if and only if 
      \[ \monUnit \re^s \lnot \lafa(t_1[\Omega], \dotsc, t_k[\Omega], \num n). \]
      Since \(\lnot \lafa(t_1[\Omega], \dotsc, t_k[\Omega], \num n) \) is defined as \( \lafa(t_1[\Omega], \dotsc, t_k[\Omega], \num n) \limply \lfalse \), 
      again by definition of \(\re^s\), we have to show that:
      \[ \monUnit u \monRe^s \lfalse, \] 
      for any inner realizer \(u\) of \( \lafa (t_1[\Omega], \dotsm, t_k[\Omega], \num n) \). 
      However, by \cref{query_prop}, \( \lafa (t_1[\Omega]\), \(\dotsm\), \(t_k[\Omega]\), \(\num n) \) does not hold, so there is no such \(u\). 
      Thus
      \[ \monUnit u \monRe^s \lfalse \] 
      holds vacuously. 
  \end{itemize}
\end{proof}

\begin{omitted} 

  Then we just need to write the realizer \(t^P\) of the universal closure of \(F^P_\lor\), that is, \(\EM(P)\): 
  \[ \monRaiseN{0} (\abstr{x_1}\Nat \dotsm \monRaiseN{0} (\abstr{x_k}\Nat t^P_\lor (x_{1}, \dotsc, x_k))) \monRe \qforall{x_1 \dotso x_k} F^P_\lor(x_1, \dotsc, x_k). \]

  Note that the type for an interactive realizer of \(\EM\) is: 
  \[ \mm{\EM(P)} = \T (\Nat \tarrow \dotsb \tarrow \T(\Nat \tarrow \T((\Nat \times \m{P}) + (\Nat \tarrow \T\m{\lnot P})))), \] 
  and since \( \m{P} = \m{\lnot P} = \Unit \), 
  \[ \mm{\EM(P)} = \T (\Nat \tarrow \dotsb \tarrow \T(\Nat \tarrow \T((\Nat \times \Unit) + (\Nat \tarrow \T\Unit)))). \] 
\end{omitted}

Then we can extend our proof decoration for \(\HA\) (see \cref{fig:monadic_decorated_rules}) with the new axiom rule: 
\[ 
  \PrAx{}
  \PrLbl\EM 
  \PrUn{\Gamma \monSeq \mon\emReal(\lafa, t_1, \dotsc, t_k) : \EM(\lafa, t_1, \dotsc, t_k)}
  \DisplayProof 
\]
and show that interactive realizability realizes the whole \(\HA+\EM\). 
\begin{theorem}[Soundness of \(\HA+\EM\) with respect to Interactive Realizability Semantics] 
  Let \(\mathcal{D}\) be a derivation of \( \Gamma \seq \fa \) in \(\HA+\EM\). Then \( \Gamma \monSeq \mathcal{D}^* : \fa \), where \(\mathcal{D}^*\) is the term obtained by decorating \(\mathcal{D}\), is valid with respect to the interactive realizability semantics. 
\end{theorem} 
\begin{proof} 
  By definition of interactive realizability semantics, we have to prove that \(\Gamma \monSeq \mathcal{D}^* : \fa \) is valid with respect to the monadic realizability semantics induced by \(\monRe^s\) for any state \(s\). 
  So we fix a generic state \(s\) and proceed by induction on the structure of the decorated version of \(\mathcal{D}\), exactly as in \cref{thm:ha_soundness}, that is, we prove that each rule whose premisses are valid has a valid conclusion. 
  Since \(\monRe^s\) is a monadic realizability relation, this has already been shown in the proof of \cref{thm:ha_soundness} for all the rules in \(\HA\). 
  We only need to check the \(\EM\) axiom, but we have already done this in \cref{thm:realizer_for_em}. 
\end{proof} 

\section{Conclusions} 

As we mentioned in the introduction, 
interactive realizability describes a learning by trial-and-error process. 
In our presentation we focused on the evaluation of interactive realizers, which corresponds to the trial-and-error part and is but a single step in the learning process.  
For the sake of completeness, we briefly describe the learning process itself. 

We can interpret an interactive realizer \(\mon\rr\) of a formula \(\fa\) as a function \(f\) from states to states. 
Recall that the intended interpretation of a term \(e : \Ex \) is a function that extends states. 
Then we can define \(f\) by means of \(\mon\rr\) as follows:
\[
  f (s) = \begin{cases}
    \interp{e} (s) \qquad & \text{if } \mon\rr \leadsto \inr e, \\ 
    s \qquad & \text{if } \mon\rr \leadsto \inl t \text{ for some t.}
    \end{cases}
\]
Note that by definition of \(\monRe\) we know that in the first case \(\interp{e} s\) properly extends \(s\). 
We can think of \(f\) as a learning function: we start from a knowledge state and try to prove \(\fa\) with \(\mon\rr\). If we fail, we learn some information that was not present in the state and we use it to extend the state. If we succeed then we do not learn anything and we return the input state. 
Thus note that the fixed points of \(f\) are exactly the states containing enough information to prove \(\fa\).

By composing \(f\) with itself we obtain a learning process: we start from some state (for instance the empty one) and we apply \(f\) repeatedly. 
If in this repeated application eventually produces a fixed point, the learning process ends, since we have the required information to prove \(\fa\). 
Otherwise we build an infinite sequence of ever increasing knowledge states whose information is never enough to prove \(\fa\). 
The fact that the learning process described by interactive realizability ends is proved in Theorem 2.15 of \cite{aschieriB10}. 

We wish to point out one of the main differences between our presentation of interactive realizability and the one given in \cite{aschieriB10}. 
In \cite{aschieriB10}, the formula-as-types correspondence is closer to the standard one. 
Exceptions are allowed only at the level of atomic formulas and \(\exmerge\) is only used in atomic rules. 
For instance a realizer for a conjunction \(\fa \land \fb\) could normalize to \(\pair{e_1}{e_2}\). 
In this case, the failure of the realizer is not apparent (at least at the top level) and it is not clear which one of \(e_1\) or \(e_2\) we are supposed to extend the state with. 
In our version exceptions are allowed at the top level of any formula and they ``climb'' upwards whenever possible by means of \(\exmerge\). 


\chapter{A Witness Extraction Technique by Proof Normalization}

\newcommand\foA{\mathfrak{a}}
\newcommand\foB{\mathfrak{b}}
\newcommand\foC{\mathfrak{c}}
\newcommand\brA\zeta
\newcommand\brB\eta
\newcommand\brC\eta
\newcommand\fd{D}
\newcommand\riC\gamma
\newcommand\brApp[2]{#1,#2}
\newcommand\derA\Pi
\newcommand\derB\Sigma

\newcommand\foa{\mathfrak{a}}
\newcommand\fob{\mathfrak{b}}
\newcommand\foc{\mathfrak{c}}
\newcommand\ri[1]{\ifthenelse{\equal{#1}\foA}\riA{%
    \ifthenelse{\equal{#1}\foB}\riB{%
      \ifthenelse{\equal{#1}\foC}\riC{}%
    }%
  }%
}
\newcommand\f[1]{\ifthenelse{\equal{#1}\foA}\fa{%
    \ifthenelse{\equal{#1}\foB}\fb{%
      \ifthenelse{\equal{#1}\foC}\fc{|{#1}|}%
    }%
  }%
}
\newcommand\der[1]{\check{#1}}

\newcommand\RuleNameEM[1][1]{\mathrm{EM}\optionalSubscript{#1}}
\newcommand\RuleNameIndE{\mathrm{Ind}}
\newcommand\RuleNameAtomI{\RuleName{\mathcal A}{I}}
\newcommand\RuleNameAtomE{\RuleName{\mathcal A}{E}}

\newcommand\RedNameProp[1]{\RuleName{{#1}}{-red}}
\newcommand\RedNamePerm[2][]{\RuleName{\ifthenelse{\equal{#1}{}}{#2}{{#2}{/}{#1}}}{-perm}}
\newcommand\RedNameSimpl[1]{\RuleName{{#1}}{-simpl}}

\newcommand\RedNameInd{\RedNameProp\RuleNameIndE}
\newcommand\RedNameWitness{\RedNameProp{\text{Wit}}}

We present a new set of reductions for derivations in natural deduction that can extract witnesses from closed derivations of simply existential formulas in Heyting Arithmetic (\(\HA\)) plus the law of the excluded middle restricted to simply existential formulas (\(\EM\)). 

The reduction we present are inspired by the informal idea of learning by making falsifiable hypothesis and checking them, and by the interactive realizability interpretation. We extract the witnesses directly from derivations in \(\HA+\EM\) by reduction, without encoding derivations by a realizability interpretation.

\section{Introduction}


In proof theory there are reductions that express the computational interpretation we give to logical connectives, quantifiers and, in the case of arithmetic, induction.
Proofs in intuitionistic logic are shown to produce a witness for existential statements:
any proof can be reduced to normal form, in which no more reductions are possible,  
and in a normal proof of an existential statement a witness always appears in a predictable location.
We want to obtain the same result for proofs of semi-decidable statements in intuitionistic logic augmented with \(\EM\) and reduction rules inspired by a trial-and-error interpretation. 

We work in Heyting Arithmetic (\(\HA\)) extended with \(\EM\), which is weaker than classical arithmetic but strong enough to prove non-trivial non-constructive results: 
for instance the fact that every function \(f : \N \to \N\) has a minimum. 
By modifying the standard reductions for Heyting Arithmetic (see \cite{prawitz71}), we show that normal proofs of existential statements in \(\HA+\EM\) produce a witness\footnote{Under suitable assumptions on the proof.}, as they do in the intuitionistic case. 

The fact that classical arithmetic is a conservative extension of \(\HA\) for \(\Pi^0_2\) statements is well known and the fact that we can extract witnesses from classical proofs of \(\Sigma^0_1\) statements follows immediately. 
However proofs of these results usually employ the G\"odel-Gentzen negative translation combined with variants of Kreisel's modified realizability semantics or Friedman's translation. 
Here, by purely proof theoretical means, we prove a slightly weaker result without resorting to negative translations and 
using reductions justified in terms of Interactive Realizability. 

An important remark is that in this chapter we do not prove strong normalization, but a just a result on the form of normal proofs.
A formal type theoretic version of the reductions given in the following and strong normalization proof could not be included in this dissertation for reasons of time, but it will appear in \cite{aschieriBB13}. 
In this section we prove that, \emph{if} we have normalization, then all derivations of simply existential statements compute a witness by a method we describe as trial-and-error. 

\section{A Formal System for Intuitionistic Arithmetic}

\begin{omitted}
  The language of \(\HA\) is a first-order language with connectives \(\land, \lor, \limply\) and quantifiers \(\forall\) and \(\exists\).
  As usual we say that a formula is \emph{closed} if it has no free variables and that a rule is \emph{atomic} if its premisses and its conclusion can only be atomic formulas\footnote{An instance of the induction rule can have atomic premisses and conclusion, but this is not required in general, so the induction rule is not atomic.}.

  The language of \(\HA\) includes: 
  \begin{itemize}
    \item variables \(\lva, \lvb, \lvc, \dotsc\) for natural numbers,
    \item function symbols for all the primitive recursive functions,
    \item the equality binary predicate symbol \(=\), 
    \item arithmetical terms (metavariables \(\lta, \ltb, \ltc\)), atomic formulas (\(\lafa\)) and formulas (\(\fa, \fb, \fc\)), defined inductively as usual.
  \end{itemize}
  In particular we assume that we have the function symbols \(\num 0, \succ \).

  A \emph{numeral} is any term of the form \(\succ^n(\num 0)\), for some \(n \in \N\). 
  We assume having a set of algebraic reduction rules for primitive recursive functions. 
  If, for some \(\lfb : \N^n \to \N\) and \(\lfc : \N^{n+2} \to \N\), \(\lfa\) denotes the primitive recursive \((n{+}1)\)-ary function defined by the equations:
  \begin{align*}
    \lfa(0, \lva_1, \dotsc, \lva_n) &= \lfb(\lva_1, \dotsc, \lva_n), \\ 
    \lfa(\succ(\lva), \lva_1, \dotsc, \lva_n) &= \lfc(\lva, \lfa(\lva, \lva_1, \dotsc, \lva_n), \lva_1, \dotsc, \lva_n),
  \end{align*}
  then we add the reductions:
  \begin{align*}
    \lfa(0, \lva_1, \dotsc, \lva_n) &\reducesto \lfb(\lva_1, \dotsc, \lva_n), \\ 
    \lfa(\succ(\lva), \lva_1, \dotsc, \lva_n) &\reducesto \lfc(\lva, \lfa(\lva, \lva_1, \dotsc, \lva_n), \lva_1, \dotsc, \lva_n).
  \end{align*}
  This reduction system is strongly normalizing and has an unique normal form. 
  Thus any closed normal term is a numeral. 
  We reduce terms inside formulas and we consider two formulas equal when they have the same normal form. 

  We assume that we have a recursive set of atomic rules that is correct (that is, it derives true statements only) and contains:
  \begin{itemize}
    \item the first-order axioms and rules for equality, 
    \item the compatibility rule for equality and functions,
    \item the \emph{Ex Falso Quodlibet} rule for atomic formulas \(\lafa\):
      \[ 
        \PrAx\lfalse 
        \PrLbl{\RuleName\lfalse{E}} 
        \PrUn\lafa 
        \DisplayProof 
      \] 
      where \(\lfalse\) denotes the \(0\)-ary relation that never holds; 
    \item the rules for zero and the successor functions:
      \[
        \PrAx{\succ(\lva) = \num 0}
        \PrUn\lfalse
        \DisplayProof
        \qquad
        \PrAx{\succ(\lva) = \succ(\lvb)}
        \PrUn{\lva = \lvb}
        \DisplayProof
      \]
  \end{itemize}

  The non-atomic inference rules are just those of minimal first-order logic, that is, one elimination and one introduction rule for each of the logical connectives and quantifiers and the induction rule, which we define later. 
  Moreover we consider \(\lnot \fa\) syntactic sugar for \(\fa \limply \lfalse\).

  Note that minimal logic extended with the Ex Falso Quodlibet rule for atomic formulas yields full intuitionistic logic. 
\end{omitted}

As usual we work in \(\HA+\EM\), Heyting Arithmetic extended with the law of the excluded middle for \(\Sigma^0_1\) formulas. 
The full description is in \cref{sec:prelim_logic}.

Since our reduction technique could conceivably be used in other first-order theories, 
we isolate some general assumptions on atomic formulas and rules that we need for our results to hold:
\begin{itemize}
  \item closed atomic formulas are decidable,
  \item any true closed atomic formula has an atomic derivation,
  \item atomic rules do not discharge assumptions,
  \item atomic rules do not bind term variables\footnote{The precise meaning of this will be made precise later.}. 
\end{itemize}
The first two assumption are very reasonable in a constructive setting such as arithmetic where we expect to have decidability at least for atomic formulas\footnote{However they may very well fail in set theory, for instance with the inclusion predicate.}. 
The other two seems also reasonable for any first-order theory. 
These assumption are reasonable in a constructive setting and they are satisfied in \(\HA\).

We assumed that any true closed atomic formula has an atomic derivation. 
For convenience we add atomic rules for proving them in one step. 
Let \(\lafa\) be a closed atomic formula. 
If \(\lafa\) is true then we add the atomic axiom:
\[
  \PrAx{}
  \PrLbl\RuleNameAtomI
  \PrUn\lafa
  \DisplayProof
\]
Otherwise if \(\lafa\) is false we add the atomic rule:
\[
  \PrAx\lafa
  \PrLbl\RuleNameAtomE
  \PrUn\bot
  \DisplayProof
\]

In order to work on the structure of derivations we need suitable notation and terminology. 
We represent derivations as upward growing trees of formulas and we make a distinction between a formula (resp. rule) and its occurrences (resp. instances) in a derivation.

A formula can occur more than once in a derivation. 
While these occurrences are clearly distinct in a tree-like representation, 
in order to avoid confusion when referring to them in the text special care must be taken. 
Thus we make a distinction between formulas and \emph{formula occurrences}, or simply occurrences, which we label with \(\foA, \foB, \foC\). 
In a derivation, formula occurrences are arranged following the patterns given by the inference rules. 
As with formulas, we distinguish between rules and \emph{rule instances}, or simply instances, which we label with \(\riA,\riB, \riC\). 
We write a derivation \(\derA\) as follows:
\[
  \PrAx{}
  \PrInf[\derA_1]
  \PrUn[\foB_1]\fb
  \PrAss\fc\riB
  \PrInf[\derA_2]
  \PrUn[\foB_2]\fb
  \PrLbl[\riA]{\text{rulename}}
  \PrBin[\foA]\fa
  \DisplayProof 
\]
The only occurrence of the formula \(\fa\) is labeled \(\foA\), while \(\fb\) occurs two times, with distinct labels \(\foB_1\) and \(\foB_2\). 
\(\foA\) is the \emph{conclusion} of an instance, labeled \(\riA\), of an inference rule named \(\text{rulename}\). 
\(\foA_1\) and \(\foA_2\) are the \emph{premisses} of \(\riA\). 
We also say that \(\foA\) is the \emph{conclusion} of the whole derivation \(\derA\). 
With \(\derA_1\) and \(\derA_2\) we denote two \emph{subderivations} (as in subtree) of \(\derA\). 
We distinguish subderivations by their conclusion, so we say that \(\derA_i\) is the (sub)derivation of \(\foA_i\) for \(i=1,2\).
By writing \(\PrAss\fc\riB \DisplayProof\) in square brackets above \(\derA_2\), 
we make explicit that \(\derA_2\) may contain occurrences of the assumption \(\fc\), which is discharged by some undisplayed rule instance \(\riB\).

We define \emph{assumptions} and \emph{open assumptions} as usual in natural deduction, see \cite{troelstra00}, page 23.

\section{The Standard Reductions}

In this section we introduce the standard reductions we need for proofs in natural deduction. 

Reductive proof theory stems from the following observation: 
there are derivations that are more complex than they need to be because they have unnecessary detours.
When this occurs, we can produce simpler and more direct derivations with the same conclusion by simple structural manipulations called \emph{reductions}.

In standard reductive proof theory for natural deduction, several reductions are introduced: proper reductions, permutative reductions, immediate simplifications and a reduction for the induction rule (see \cite{prawitz71}).
A derivation is said to be \emph{fully normal} when none of these reductions can be performed on it. 
For our purposes fully normal derivations are not required, so we introduce only the proper reductions and the induction reduction. 

In an instance of an elimination rule, the premiss containing the connective or quantifier that is being eliminated is called the \emph{major} premiss; the other premisses are called the \emph{minor} premisses. 
We always display the major premiss in the leftmost position. 

\subsection{Proper Reductions} 
Consider a derivation in which a formula occurrence \(\foA\) is both the conclusion of an introduction rule instance \(\riA\) and the major premiss of an elimination rule instance \(\riB\). 
Then we can derive the conclusion of \(\riB\) directly by removing \(\riA\) and \(\riB\) and rearranging the derivations of the premisses of \(\riA\) and of the minor premisses of \(\riB\) (if any).
Note that \(\riA\) and \(\riB\) must be instances of an introduction rule and an elimination rule for the same logical connective, since the formula introduced by \(\riA\) is the same formula eliminated by \(\riB\). 
Therefore for each logical connective we have a different type of \emph{proper reduction}. 
They are listed in \Cref{fig:proper_reductions}. 

\begin{figure}[!ht]
  \caption{The proper reductions.}
  \label{fig:proper_reductions}
  \begin{rulelisting}[\reducesto]
    \header{\RedNameProp\land} 
    \PrAx{}\PrInf[\derA_1]
    \PrUn{\fa_1}
    \PrAx{}\PrInf[\derA_2]
    \PrUn{\fa_2}
    \PrLbl[\riA]{\RuleName\land{I}}
    \PrBin{\fa_1 \land \fa_2}
    \PrLbl[\riB]{\RuleName\land{E}}
    \PrUn{\fa_i}
    \DisplayProof 
    \nextrule[\xrightarrow{\RedNameProp\land}]
    \PrAx{}\PrInf[\derA_i]
    \PrUn{\fa_i}
    \DisplayProof
    \newline
    \header{\RedNameProp\lor} 
    \PrAx{}\PrInf[\derA]
    \PrUn{\fa_i}
    \PrLbl{\RuleName\lor{I}}
    \PrUn{\fa_1 \lor \fa_2}
    \PrAss{\fa_1}\riA
    \PrInf[\derA_1]
    \PrUn\fb
    \PrAss{\fa_2}\riA
    \PrInf[\derA_2]
    \PrUn\fb
    \PrLbl[\riA]{\RuleName\lor{E}}
    \PrTri\fb
    \DisplayProof 
    \nextrule[\xrightarrow{\RedNameProp\lor}]
    \PrAx{}\PrInf[\derA]
    \PrUn{\fa_i}
    \PrInf[\derA_i]
    \PrUn\fb
    \DisplayProof
    \newline
    \header{\RedNameProp\limply} 
    \PrAss\fa\riA
    \PrInf[\derA_1]
    \PrUn\fb
    \PrLbl[\riA]{\RuleName\limply{I}}
    \PrUn{\fa \limply \fb}
    \PrAx{}
    \PrInf[\derA_2]
    \PrUn\fa
    \PrLbl{\RuleName\limply{E}}
    \PrBin\fb
    \DisplayProof 
    \nextrule[\xrightarrow{\RedNameProp\limply}]
    \PrAx{}\PrInf[\derA_2]
    \PrUn\fa
    \PrInf[\derA_1]
    \PrUn\fb
    \DisplayProof
    \newline
    \header{\RedNameProp\forall} 
    \PrAx{}
    \PrInf[\derA]
    \PrLbl{\RuleName\forall{I}}
    \PrUn{\qforall\lva\fa}
    \PrLbl{\RuleName\forall{E}}
    \PrUn{\fa\subst\lva\lta}
    \DisplayProof 
    \nextrule[\xrightarrow{\RedNameProp\forall}]
    \PrAx{}
    \PrInf[\derA\subst\lva\lta]
    \PrUn{\fa\subst\lva\lta}
    \DisplayProof
    \newline
    \header{\RedNameProp\exists}
    \PrAx{}
    \PrInf[\derA_1]
    \PrUn{\fa\subst\lva\lta}
    \PrLbl{\RuleName\exists{I}}
    \PrUn{\qexists\lva\fa}
    \PrAss{\fa\subst\lva\lvb}\riA
    \PrInf[\derA_2]
    \PrUn\fb
    \PrLbl[\riA]{\RuleName\exists{E}}
    \PrBin\fb
    \DisplayProof 
    \nextrule[\xrightarrow{\RedNameProp\exists}]
    \PrAx{}
    \PrInf[\derA_1]
    \PrUn{\fa\subst\lva\lta}
    \PrInf[\derA_2\subst\lvb\lta]
    \PrUn\fb
    \DisplayProof
  \end{rulelisting}
\end{figure}

\subsection{Induction Reduction}
Consider the induction rule schema \(\RuleNameIndE\) in the following form:
\[ 
  \PrAx{}
  \PrInf[\derA_1]
  \PrUn{\fa\subst\lva{\num 0}}
  \PrAss\fa\riA
  \PrInf[\derA_2]
  \PrUn{\fa\subst\lva{\succ(\lva)}}
  \PrLbl[\riA]\RuleNameIndE
  \PrBin{\fa\subst\lva\lta}
  \DisplayProof
\] 
We call \(\lta\) the \emph{main term} of the induction. 
\fixme{relate this to the induction definition given in prelims}
An instance \(\riA\) of the \(\RuleNameIndE\) rule can be reduced when the main term \(\lta\) in its conclusion \(\fa\subst\lva\lta\) is either \(\num 0\) or \(\succ(\ltb)\) for some term \(\ltb\). 
Then if \(\lta = \num 0\) we can reduce \(\riA\) to: 
\[ 
  \PrAx{} 
  \PrInf[\derA_1] 
  \PrUn{\fa\subst\lva{\num 0}} 
  \DisplayProof 
\]
and if \(\lta = \succ(\ltb)\) as: 
\[ 
  \PrAx{}
  \PrInf[\derA_1]
  \PrUn{\fa\subst\lva{\num 0}}
  \PrAss\fa\riB
  \PrInf[\derA_2]
  \PrUn{\fa\subst\lva{\succ(\lva)}}
  \PrLbl[\riB]\RuleNameIndE
  \PrBin{\fa\subst\lva\ltb}
  \PrInf[\derA_2\subst\lva\ltb]
  \PrUn{\fa\subst\lva{\succ(\ltb)}}
  \DisplayProof
\]
We call this conditional reduction \(\RedNameInd\). 
It is easy to see that this reduction is ``unraveling'' the induction. 
When \(\ltb\) is a numeral \(\num n\), that is, a term of the form \(\succ^n\), we can apply the \(\RedNameInd\) reduction repeatedly (\(n\) times) until we remove all occurrences of the \(\RuleNameIndE\) rule and get:
\[
  \PrAx{}
  \PrInf[\derA_1]
  \PrUn{\fa\subst\lva{\num 0}}
  \PrInf[\derA_2\subst\lva{\num 0}]
  \PrUn{\fa\subst\lva{\num 1}}
  \PrInf[\derA_2\subst\lva{\num 1}]
  \PrUn{\fa\subst\lva{\num 2}}
  \PrInf
  \PrUn{\fa\subst\lva{\num n}}
  \DisplayProof
\]



\section{The Witness Extracting Reductions}

In this section we introduce an inference rule that is equivalent to 
the restricted excluded middle axiom schema \(\EM\) defined in \Cref{def:excluded_middle}
and two reductions involving this new rule. 
The first one, the \(\RedNameWitness\) reduction, is inspired by Interactive Realizability and it will be instrumental in converting classical derivations into constructive ones.
The second one is a permutative reduction and is needed later for technical reasons. 

\subsection{The $\EM$ Rule}
\begin{omitted}
  The excluded middle axiom schema is the following:
  \[ 
    \PrAx{}
    \PrLbl\EMG
    \PrUn{\fa \lor \lnot \fa} 
    \DisplayProof
  \]
  where \(\fa\) is any formula. 
  In intuitionistic logic with \(\EMG\) we can derive the double negation elimination. 
  This shows that \(\HA\) extended with \(\EMG\) is equivalent to classical (Peano) arithmetic PA.

  We restrict the \(\EMG\) axiom to \(\Sigma^0_1\) formulas
  and we rewrite it in an equivalent way, as in \cite{akamaBHK04}:
  \[ 
    \PrAx{}
    \PrLbl\EM
    \PrUn{\qforall\lva \lafa \lor \qexists\lva \lnot\lafa }
    \DisplayProof
  \]
  where \(\lafa\) is an atomic formula. 
\end{omitted}

For convenience we replace the \(\EM\) axiom schema with the equivalent \(\EM\) rule:
\[
  \PrAss{\qforall\lva\lafa}\riA
  \PrInf
  \PrUn\fa
  \PrAss{\lnot\lafa\subst\lva\lvb}\riA
  \PrInf
  \PrUn\fa
  \PrLbl[\riA]\EM
  \PrBin\fa
  \DisplayProof
\]
where the variable \(\lvb\) does not occur in \(\fa\) nor in any open assumption that \(\fa\) depends on except occurrences of the assumption \(\lnot \lafa\subst\lva\lvb\) (as in the \(\RuleName\exists{E}\) rule). 


The \(\EM\) rule is derived by an \(\RuleName\lor{E}\) rule instance, whose major premiss is an instance of the \(\EM\) axiom and whose rightmost assumption is the major premiss of an \(\RuleName\exists{E}\) instance: 
\[
  \PrAx{}
  \PrLbl\EM
  \PrUn{(\qforall\lva \fa) \lor (\qexists\lva \lnot\fa)}
  \PrAss{\qforall\lva \fa}\riA
  \PrInf
  \PrUn\fc
  \PrAss{\qexists\lva \lnot\fa}\riA
  \PrAss{\lnot\fa\subst\lva\lvb}\riB
  \PrInf
  \PrUn\fc
  \PrLbl[\riB]{\RuleName\exists{E}}
  \PrBin\fc
  \PrLbl[\riA]{\RuleName\lor{E}}
  \PrTri\fc
  \DisplayProof
\]
On the other hand, the \(\EM\) axiom can be derived from the \(\EM\) rule by two \(\RuleName\lor{I}\) instances:
\[
  \PrAss{\qforall\lva \fa}\riA
  \PrLbl{\RuleName\lor{I}}
  \PrUn{(\qforall\lva \fa) \lor (\qexists\lva \lnot \fa)}
  \PrAss{\lnot\fa\subst\lva\lvb}\riA
  \PrLbl{\RuleName\exists{I}}
  \PrUn{\qexists\lva \lnot \fa}
  \PrLbl{\RuleName\lor{I}}
  \PrUn{(\qforall\lva \fa) \lor (\qexists\lva \lnot \fa)}
  \PrLbl[\riA]\EM
  \PrBin{(\qforall\lva \fa) \lor (\qexists\lva \lnot \fa)}
  \DisplayProof
\]

In the following we refer to the assumption \(\qforall\lva\lafa\) in the derivation of the leftmost premiss of the \(\EM\) rule as the \emph{universal assumption} and to the assumption \(\lnot\lafa\subst\lva\lvb\) in the derivation of the rightmost premiss as the \emph{existential assumption}.

We can also write the \(\EM\) rule in sequent style as: 
\[
  \PrAx{\Gamma, \riA : \qforall\lva \lafa \seq \fa}
  \PrAx{\Gamma, \riA : \qexists\lva \lnot \lafa \seq \fa}
  \PrLbl[\riA]\EM
  \PrBin{\Gamma \seq \fa}
  \DisplayProof
\]

The universal assumption \(\qforall\lva\lafa\) is a \(\Pi^0_1\) formula and thus \emph{negatively decidable}, meaning that a finite piece of evidence is enough to prove it false: 
a counterexample, a natural number \(m\) such that \(\lafa\subst\lva{\num m}\) does not hold. 
Moreover, if we know that it is false, then a counterexample exists and we can find it in a finite time, in the worst case by means of a blind search through all the natural numbers. 

On the other hand, in order to prove the universal assumption, we need a possibly infinite evidence, namely, we may need to check \(\lafa\subst\lva{\num m}\) for all natural numbers \(m\) and this cannot be effectively done (at least when we have no information on \(\lafa\)). 

The existential assumption \(\lnot\lafa\subst\lva\lvb\) is not actually a existential formula. 
However it is easy to see that it takes the place of the assumption discharged by the \(\RuleName\exists{E}\) rule. 

We say that we can prove the existential assumption true by showing a witness, namely a number \(m\) such that \(\lnot\lafa\subst\lva{\num m}\). 
Thus the existential assumption behaves as if it were \emph{positively decidable}: 
when it is true, we have a terminating algorithm to find the finite evidence needed to prove it. 
However, when it false, we have no way to effectively decide if it is false. 

Note that a counterexample \(m\) for the universal assumption \(\qforall\lva\lafa\) is a witness for the existential assumption since in that case \(\lnot\lafa\subst\lva{\num m}\) holds.


\subsection{Witness Reduction}

Consider a derivation \(\derA\) ending with an instance \(\riA\) of the \(\EM\) rule for the atomic formula \(\lafa\):
\[
  \PrAss{\qforall\lva\lafa}\riA
  \PrInf[\derA_1]
  \PrUn\fa
  \PrAss{\lnot\lafa\subst\lva\lvb}\riA
  \PrInf[\derA_2]
  \PrUn\fa
  \PrLbl[\riA]\EM
  \PrBin\fa
  \DisplayProof
\]

A priori we do not know any counterexample to the universal assumption (we do not even know whether it holds or not), so we begin by looking at how the assumption is used in \(\derA_1\). 
In \(\derA_1\), consider all the instances \(\riB_1, \dotsc, \riB_n\) of the \(\RuleName\forall{E}\) rule whose premiss is an occurrence of the universal assumption \(\qforall\lva\lafa\) and whose conclusion is the occurrence of a closed (atomic) formula:
\[ 
  \PrAss{\qforall\lva\lafa}\riA
  \PrLbl[\riB_1]{\RuleName\forall{E}}
  \PrUn{\lafa\subst\lva{\lta_1}}
  \PrInf
  \DisplayProof \quad \dotso \quad
  \PrAss{\qforall\lva\lafa}\riA
  \PrLbl[\riB_n]{\RuleName\forall{E}}
  \PrUn{\lafa\subst\lva{\lta_n}}
  \PrInf
  \DisplayProof 
\]
These represent the concrete instances of the universal assumption that are used to derive \(\fa\) in \(\derA_1\). 
Since the conclusions of \( \riB_1, \dotsc, \riB_n\) are closed atomic formulas they are decidable. 
Therefore we can derive the true concrete instances directly with the atomic axiom \(\RuleNameAtomI\) instead of deducing them from the universal assumption. 
We distinguish two cases.
\begin{itemize}
  \item
    If \(\lafa\subst\lva{\lta_i}\) is true for all \(i\) we replace each \(\riB_i\) with the atomic axiom for \(\lafa\subst\lva{\lta_i}\):
    \[ 
      \PrAss{\qforall\lva\lafa}\riA
      \PrLbl[\riB_i]{\RuleName\forall{E}}
      \PrUn{\lafa\subst\lva{\lta_i}}
      \PrInf
      \DisplayProof \quad \leadsto \quad 
      \PrAx{}
      \PrLbl\RuleNameAtomI
      \PrUn{\lafa\subst\lva{\lta_i}}
      \PrInf
      \DisplayProof
    \] 
    We call this new derivation \(\derA_1'\). 

    Now two situations are possible: either \(\derA_1\) needs the universal assumption only to deduce the concrete instances \(\riB_1, \dotsc, \riB_n\) or not. 

    \begin{itemize}
      \item The first case happens when \(\derA_1'\) contains no more occurrences of the universal assumption discharged by \(\riA\),
        that is, the universal assumption only occurs in \(\derA_1\) as the premiss of \(\riB_1, \dotsc, \riB_n\). 
        In this case \(\derA_1'\) is a self-contained derivation of \(\fa\) and we can replace the whole \(\derA\) with \(\derA_1'\).
      \item Otherwise \(\derA_1'\) still contains some occurrence  of the universal assumption.
        Then \(\derA_1'\) does need the universal assumption itself and not just some concrete instances of it. 
        In this case we can only replace \(\derA_1\) with \(\derA_1'\) in \(\derA\), but we cannot eliminate the \(\EM\) rule instance \(\riA\) from the derivation. 
    \end{itemize}

  \item 
    Otherwise there is some \(i\) such that \(\lafa\subst\lva{\lta_i}\) is false. 
    Thus the universal assumption itself is false, since we have found the counterexample \(\lta_i\). 
    Moreover \(\lta_i\) is a witness for the existential assumption, meaning that 
    we can replace \(\lvb\) with \(\lta_i\) in \(\derA_2\) and all the occurrences of the assumption \(\lnot\lafa\subst\lva\lvb\) with a derivation of \(\lnot\lafa\subst\lva\lta\):  
    \[ 
      \PrAss{\lnot\lafa\subst\lva{\lta_i}}\riA
      \PrInf 
      \DisplayProof \quad \leadsto \quad 
      \PrAss{\lafa\subst\lva{\lta_i}}{\riB_i'} 
      \PrLbl\RuleNameAtomE 
      \PrUn\bot 
      \PrLbl[\riB_i']{\RuleName\limply{I}} 
      \PrUn{\lnot\lafa\subst\lva{\lta_i}} 
      \PrInf 
      \DisplayProof 
    \] 
    We call this new derivation \(\derA_2'\). 

    Note that in this case we replace all the occurrences of the existential assumption in \(\derA_2\) and thus \(\derA_2'\) is self-contained derivation of \(\fa\). 
    Therefore we can replace \(\derA\) with \(\derA_2'\).
\end{itemize}
We call this reduction \(\RedNameWitness\).

The gist of the \(\RedNameWitness\) reduction is that we look for counterexamples to the universal assumption in \(\derA_1\). 
If we do not find one then we have checked that all the concrete instances of the universal assumption hold.
Moreover if \(\derA_1\) uses the universal assumption exclusively to deduce these concrete instances, then we get a direct derivation of \(\fa\) without using the \(\EM\) rule.
On the other hand if we find a counterexample then we know that we can put it in \(\derA_2\) and get another direct derivation of \(\fa\).

In some sense we have a procedure to decide which one of the subderivation of the \(\EM\) rule is the effective one, 
Note that this procedure fails when we do not find counterexamples to the universal assumption but we cannot completely eliminate its occurrences from \(\derA_1\). 
Our main result can be thought of as the proof that, when the conclusion of a derivation is simply existential, this 
``failure'' of the procedure
does not happen. 
The whole reduction is summarized in \Cref{fig:em_reduction}.

\begin{figure}[!ht]
  \caption{The \(\RedNameWitness\) reduction possible outcomes.}
  \label{fig:em_reduction}
  \begin{rulelisting}
    \header\RedNameWitness
    \hline
    \wholeline{
      \begin{tikzpicture}
        \node (forall 1) at (0,2) [] {
          \PrAss{\qforall\lva\lafa}\riA
          \PrLbl[\riB_1]{\RuleName\forall{E}}
          \PrUn{\lafa\subst\lva{\lta_1}}
          \DisplayProof
        };
        \node (forall n) at (3,2) [] {
          \PrAss{\qforall\lva\lafa}\riA
          \PrLbl[\riB_n]{\RuleName\forall{E}}
          \PrUn{\lafa\subst\lva{\lta_n}}
          \DisplayProof
        };
        \node (forall x) at (5.5,2) [] {
          \PrAss{\qforall\lva\lafa}\riA
          \DisplayProof
        };
        \node (der1) at (0.93,0.25) [inner sep=0pt] {}; 

        \node (root) at (4,0) [inner sep=0pt] {
          \PrAx{}
          \PrInf[\derA_1]
          \PrUn\fa
          \PrAx{\hspace{4cm}}
          \PrAss{\lafa\subst\lva\lvb}\riA
          \PrInf[\derA_2]
          \PrUn\fa
          \PrLbl[\riA]\EM
          \PrTri\fa
          \DisplayProof
        };
        \draw [dotted] (forall 1) -- (forall n);
        \draw [dotted] (forall n) -- (forall x);
        \draw [dotted] (forall 1) -- (der1);
        \draw [dotted] (forall n) -- (der1);
        \draw [dotted] (forall x) -- (der1);
      \end{tikzpicture}
    } 
    \\
    \header{\text{An derivation ending with an \(\EM\) rule instance reduces to:}}
    \hline
    \wholeline{
      \begin{tikzpicture}
        \node (forall 1) at (0,2) [] {
          \PrAx{}
          \PrLbl\RuleNameAtomI
          \PrUn{\lafa\subst\lva{\lta_1}}
          \DisplayProof
        };
        \node (forall n) at (3,2) [] {
          \PrAx{}
          \PrLbl\RuleNameAtomI
          \PrUn{\lafa\subst\lva{\lta_n}}
          \DisplayProof
        };
        \node (forall x) at (5.5,2) [] {
          \PrAss{\qforall\lva\lafa}\riA
          \DisplayProof
        };
        \node (der1) at (0.93,0.25) [inner sep=0pt] {}; 

        \node (root) at (4,0) [] {
          \PrAx{}
          \PrInf[\derA_1]
          \PrUn\fa
          \PrAx{\hspace{4cm}}
          \PrAss{\lafa\subst\lva\lvb}\riA
          \PrInf[\derA_2]
          \PrUn\fa
          \PrLbl[\riA]\EM
          \PrTri\fa
          \DisplayProof
        };
        \draw [dotted] (forall 1) -- (forall n);
        \draw [dotted] (forall n) -- (forall x);
        \draw [dotted] (forall 1) -- (der1);
        \draw [dotted] (forall n) -- (der1);
        \draw [dotted] (forall x) -- (der1);
      \end{tikzpicture}
    }
    \\
  \header{\begin{tabular}{c}when all \(\lafa\subst\lva{\lta_i}\) hold and some occurrences of the universal assumption remain.\end{tabular}}
    \hline
    \wholeline{
      \begin{tikzpicture}
        \node (forall 1) at (0,2) [] {
          \PrAx{}
          \PrLbl\RuleNameAtomI
          \PrUn{\lafa\subst\lva{\lta_1}}
          \DisplayProof
        };
        \node (forall n) at (3,2) [] {
          \PrAx{}
          \PrLbl\RuleNameAtomI
          \PrUn{\lafa\subst\lva{\lta_n}}
          \DisplayProof
        };
        \node (der1) at (0.93,0.25) [inner sep=0pt] {}; 

        \node (root) at (1,0) [] {
          \PrAx{}
          \PrInf[\derA_1]
          \PrUn\fa
          \DisplayProof
        };
        \draw [dotted] (forall 1) -- (forall n);
        \draw [dotted] (forall 1) -- (root);
        \draw [dotted] (forall n) -- (root);
      \end{tikzpicture}
    }
    \\
  \header{\begin{tabular}{c}when all \(\lafa\subst\lva{\lta_i}\) hold and no occurrence of the universal assumption remains.\end{tabular}}
    \hline
    \wholeline{
      \PrAss{\lafa\subst\lva{\lta_i}}{\riB'} 
      \PrLbl\RuleNameAtomE 
      \PrUn\bot 
      \PrLbl[\riB']{\RuleName\limply{I}} 
      \PrUn{\lnot\lafa\subst\lva{\lta_i}} 
      \PrInf[\derA_2]
      \PrUn\fa
      \DisplayProof
    } \\ 
    \wholeline{\text{when some \(\lafa\subst\lva{\lta_i}\) does not hold.}}
  \end{rulelisting}
\end{figure}

%
%
%

\subsection{Permutative Reduction for $\EM$}

The permutative reduction for \(\EM\) is defined in the same way as the permutative reduction for the \(\RuleName\lor{E}\) rule, 
that is, 
when the conclusion of a \(\EM\) rule instance is the major premiss of an elimination rule instance \(\RuleName\ast{E}\):
\[
  \PrAss{\qforall\lva\lafa}\riC
  \PrInf[\derA_1]
  \PrUn\fa
  \PrAss{\lnot\lafa\subst\lva\lvb}\riC
  \PrInf[\derA_2]
  \PrUn\fa
  \PrLbl[\riC]\EM
  \PrBin\fa
  \PrAx{}
  \PrInf[\bar\derA]
  \PrLbl{\RuleName\ast{E}}
  \PrBin\fb
  \DisplayProof
\]
reduces to:
\[
  \PrAss{\qforall\lva\lafa}\riA
  \PrInf[\derA_1]
  \PrUn\fa
  \PrAx{\bar\derA}
  \PrLbl{\RuleName\ast{E}}
  \PrBin\fb
  \PrAss{\lnot\lafa\subst\lva\lvb}\riA
  \PrInf[\derA_2]
  \PrUn\fa
  \PrAx{\bar\derA}
  \PrLbl{\RuleName\ast{E}}
  \PrBin\fb
  \PrLbl[\riA]\EM
  \PrBin\fb
  \DisplayProof
\]
where \(\bar\derA\) stands for the derivations of the remaining minor premisses of \(\riB\) if any. 
We denote this reduction as \(\RedNamePerm\EM\). 
More explicitly, we can define a permutative reduction for each elimination rule, see \Cref{fig:em_perms_connectives} and \Cref{fig:em_perms_quantifiers}. 
\begin{figure}[!ht]
  \caption{The permutative reductions of the \(\EM\) rule with the \(\RuleName\land{E}, \RuleName\lor{E}\) and \(\RuleName\limply{E}\) rules.}
  \label{fig:em_perms_connectives}
  \begin{rulelisting}[\reducesto]
    \header{\RedNamePerm[\land]\EM}
    \PrAss{\qforall\lva\lafa}\riA
    \PrInf[\derA_1]
    \PrUn{\fa_1 \land \fa_2}
    \PrAss{\lnot\lafa\subst\lva\lvb}\riA
    \PrInf[\derA_2]
    \PrUn{\fa_1 \land \fa_2}
    \PrLbl[\riA]\EM
    \PrBin{\fa_1 \land \fa_2}
    \PrLbl{\RuleName\land{E}}
    \PrUn{\fa_i}
    \DisplayProof
    \nextrule
    \PrAss{\qforall\lva\lafa}\riA
    \PrInf[\derA_1]
    \PrUn{\fa_1 \land \fa_2}
    \PrLbl{\RuleName\land{E}}
    \PrUn{\fa_i}
    \PrAss{\lnot\lafa\subst\lva\lvb}\riA
    \PrInf[\derA_2]
    \PrUn{\fa_1 \land \fa_2}
    \PrLbl{\RuleName\land{E}}
    \PrUn{\fa_i}
    \PrLbl[\riA]\EM
    \PrBin{\fa_i}
    \DisplayProof
    \newline
    \header{\RedNamePerm[\lor]\EM}
    \wholeline{
      \PrAss{\qforall\lva\lafa}\riB
      \PrInf[\derB_1]
      \PrUn{\fa_1 \lor \fa_2}
      \PrAss{\lnot\lafa\subst\lva\lvb}\riB
      \PrInf[\derB_2]
      \PrUn{\fa_1 \lor \fa_2}
      \PrLbl[\riB]\EM
      \PrBin{\fa_1 \lor \fa_2}
      \PrAss{\fa_1}\riA
      \PrInf[\derA_1]
      \PrUn\fb
      \PrAss{\fa_2}\riA
      \PrInf[\derA_2]
      \PrUn\fb
      \PrLbl[\riA]{\RuleName\lor{E}}
      \PrTri\fb
      \DisplayProof
    } \\ 
    \nextrule[\downarrow] \\ 
    \wholeline{
      \PrAss{\qforall\lva\lafa}\riB
      \PrInf[\derB_1]
      \PrUn{\fa_1 \lor \fa_2}
      \PrAss{\fa_1}\riA
      \PrInf[\derA_1]
      \PrUn\fb
      \PrAss{\fa_2}\riA
      \PrInf[\derA_2]
      \PrUn\fb
      \PrLbl[\riA]{\RuleName\lor{E}}
      \PrTri\fb
      \PrAss{\lnot\lafa\subst\lva\lvb}\riB
      \PrInf[\derB_2]
      \PrUn{\fa_1 \lor \fa_2}
      \PrAss{\fa_1}\riA
      \PrInf[\derA_1]
      \PrUn\fb
      \PrAss{\fa_2}\riA
      \PrInf[\derA_2]
      \PrUn\fb
      \PrLbl[\riA]{\RuleName\lor{E}}
      \PrTri\fb
      \PrLbl[\riB]\EM
      \PrBin\fb
      \DisplayProof
    } 
    \newline
    \header{\RedNamePerm[\limply]\EM}
    \wholeline{
      \PrAss{\qforall\lva\lafa}\riA
      \PrInf[\derA_1]
      \PrUn{\fa_1 \limply \fa_2}
      \PrAss{\lnot\lafa\subst\lva\lvb}\riA
      \PrInf[\derA_2]
      \PrUn{\fa_1 \limply \fa_2}
      \PrLbl[\riA]\EM
      \PrBin{\fa_1 \limply \fa_2}
      \PrAx{}
      \PrInf[\derB]
      \PrUn{\fa_1}
      \PrLbl{\RuleName\land{E}}
      \PrBin{\fa_2}
      \DisplayProof
    } \\ 
    \nextrule[\downarrow] \\ 
    \wholeline{
      \PrAss{\qforall\lva\lafa}\riA
      \PrInf[\derA_1]
      \PrUn{\fa_1 \limply \fa_2}
      \PrAx{}
      \PrInf[\derB]
      \PrUn{\fa_1}
      \PrLbl{\RuleName\limply{E}}
      \PrBin{\fa_2}
      \PrAss{\lnot\lafa\subst\lva\lvb}\riA
      \PrInf[\derA_2]
      \PrUn{\fa_1 \limply \fa_2}
      \PrAx{}
      \PrInf[\derB]
      \PrUn{\fa_1}
      \PrLbl{\RuleName\limply{E}}
      \PrBin{\fa_2}
      \PrLbl[\riA]\EM
      \PrBin{\fa_2}
      \DisplayProof
    }
  \end{rulelisting}
\end{figure}
\begin{figure}[!ht]
  \caption{The permutative reductions of the \(\EM\) rule with the \(\RuleName\forall{E}\) and \(\RuleName\exists{E}\) rules.}
  \label{fig:em_perms_quantifiers}
  \begin{rulelisting}[\leadsto]
    \header{\RedNamePerm[\forall]\EM}
    \PrAss{\qforall\lva\lafa}\riA
    \PrInf[\derA_1]
    \PrUn{\qforall\lva \fa}
    \PrAss{\lnot\lafa\subst\lva\lvb}\riA
    \PrInf[\derA_2]
    \PrUn{\qforall\lva \fa}
    \PrLbl[\riA]\EM
    \PrBin{\qforall\lva \fa}
    \PrLbl{\RuleName\forall{E}}
    \PrUn{\fa\subst\lva\lta}
    \DisplayProof
    \nextrule
    \PrAss{\qforall\lva\lafa}\riA
    \PrInf[\derA_1]
    \PrUn{\qforall\lva \fa}
    \PrLbl{\RuleName\forall{E}}
    \PrUn{\fa\subst\lva\lta}
    \PrAss{\lnot\lafa\subst\lva\lvb}\riA
    \PrInf[\derA_2]
    \PrUn{\qforall\lva \fa}
    \PrLbl{\RuleName\forall{E}}
    \PrUn{\fa\subst\lva\lta}
    \PrLbl[\riA]\EM
    \PrBin{\fa\subst\lva\lta}
    \DisplayProof
    \newline
    \header{\RedNamePerm[\forall]\EM}
    \wholeline{
      \PrAss{\qforall\lva\lafa}\riA
      \PrInf[\derA_1]
      \PrUn{\qexists\lva \fa}
      \PrAss{\lnot\lafa\subst\lva\lvb}\riA
      \PrInf[\derA_2]
      \PrUn{\qexists\lva \fa}
      \PrLbl[\riA]\EM
      \PrBin{\qexists\lva \fa}
      \PrAss{\fa\subst\lva\lvb}
      \PrInf[\derB]
      \PrUn\fb
      \PrLbl{\RuleName\exists{E}}
      \PrBin\fb
      \DisplayProof
    } \\ 
    \nextrule[\downarrow] \\ 
    \wholeline{
      \PrAss{\qforall\lva\lafa}\riA
      \PrInf[\derA_1]
      \PrUn{\qexists\lva \fa}
      \PrAss{\fa\subst\lva\lvb}
      \PrInf[\derB]
      \PrUn\fb
      \PrLbl{\RuleName\exists{E}}
      \PrBin\fb
      \PrAss{\lnot\lafa\subst\lva\lvb}\riA
      \PrInf[\derA_2]
      \PrUn{\qexists\lva \fa}
      \PrAss{\fa\subst\lva\lvb}\riB
      \PrInf[\derB]
      \PrUn\fb
      \PrLbl[\riB]{\RuleName\exists{E}}
      \PrBin\fb
      \PrLbl[\riA]\EM
      \PrBin\fb
      \DisplayProof
    }
  \end{rulelisting}
\end{figure}

This reduction moves elimination rule instances from ``outside'' or ``below'' to ``inside'' or ``above'' an \(\EM\) rule instance. 
This is useful because an \(\EM\) rule instance may happen in between an introduction rule instance and an elimination rule instance, preventing a proper reduction from taking place.


In the following we concentrate on proving a result about the form of normal proofs. 
We do not prove here that the reduction process converges. 
In order to do it, the most natural way would be encode our proofs into proof terms in a suitable calculus and show that such calculus is strongly normalizing. 
This has been done in \cite{aschieriBB13}, so we just state the following
\begin{theorem}[Strong Normalization of \(\HA+\EM\)] \label{thm:em_strong_normalization}
All proofs of \(\HA+\EM\) are strongly normalizing under the reductions we described in this section. 
\end{theorem}
A proof can be found in \cite{aschieriBB13}.

\section{Witness Extraction}

In this section we prove the witness extraction theorem, that shows how we can extract witnesses from suitable classical derivations in \(\HA+\EM\), as we can do for intuitionistic derivations in \(\HA\). 

In order to state and prove our results we need to keep track of free term variables in a derivations, 
since both the \(\RedNameInd\) and the \(\RedNameWitness\) reductions can only be performed when certain terms and formulas are closed. 

We need to define when a variable is free in a derivation. 
\begin{definition}[Free term variables]\label{def:free_term_variable}
  We say that a rule instance \(\riA\) \emph{binds} a term variable that occurs free in the derivation \(\derA\) of a premiss of \(\riA\) in the following cases:
  \begin{itemize}
    \item \(\riA\) is an instance of the \(\RuleName\forall{I}\) rule and binds the variable \(\lva\) in the formula occurrences in the derivation of its premiss: 
      \[
        \PrAx{}
        \PrInf[\derA]
        \PrUn\fa
        \PrLbl[\riA]{\RuleName\forall{I}}
        \PrUn{\qforall\lva \fa}
        \DisplayProof
      \]
    \item \(\riA\) is an instance of the \(\RuleName\exists{E}\) rule and binds the variable \(\lvb\) in the formula occurrences in the derivation of its rightmost premiss:
      \[
        \PrAx{\qexists\lva \fa}
        \PrAss{\fa\subst\lva\lvb}\riA
        \PrInf[\derA]
        \PrUn\fb
        \PrLbl[\riA]{\RuleName\exists{E}}
        \PrBin\fb
        \DisplayProof
      \]
    \item \(\riA\) is an instance of the \(\RuleNameIndE\) rule and binds the variable \(\lva\) in the formula occurrences in the derivation of its rightmost premiss:
      \[
        \PrAx{\fa\subst\lva{\num 0}}
        \PrAss\fa\riA
        \PrInf[\derA]
        \PrUn{\fa\subst\lva{\succ(\lva)}}
        \PrLbl[\riA]\RuleNameIndE
        \PrBin{\fa\subst\lva\lta}
        \DisplayProof
      \]
    \item \(\riA\) is an instance of the \(\EM\) rule and binds the variable \(\lvb\) in the formula occurrences in the derivation of its rightmost premiss: 
      \[
        \PrAss{\qforall\lva \lafa}\riA
        \PrInf
        \PrUn\fb
        \PrAss{\lnot\lafa\subst\lva\lvb}\riA
        \PrInf[\derA]
        \PrUn\fb
        \PrLbl[\riA]\EM
        \PrBin\fb
        \DisplayProof
      \]
  \end{itemize}
  We say that a term variable occurrence is \emph{free in a derivation} when the term variable occurs free in a formula occurrence in the derivation and is not bound by any rule instance. 
  A derivation is \emph{closed} if it has no free term variable nor open assumption. 
\end{definition}

Note that no reduction introduces free term variables in a derivation. 

Since a derivation is a tree, it makes sense to give the definition of branch. 
Principal branches are branches of a derivation that contains only major premisses of elimination and \(\EM\) rule instances. 
\begin{definition}[Principal branch]
  A \emph{branch} in a derivation \(\derA\) is a sequence of formula occurrences \(\foA_0, \dotsc, \foA_n\) in \(\derA\) such that:
  \begin{itemize}
    \item \(\foA_0\) is a top formula occurrence, that is, 
      \(\foA_0\) is either an assumption or the conclusion of an atomic axiom; 
    \item \(\foA_i\) and \(\foA_{i+1}\) are respectively a premiss and the conclusion of the same rule instance \(\riA_{i+1}\), for all \(0 \leq i < n\);
    \item \(\foA_n\) is the conclusion of \(\derA\).
  \end{itemize}
  A branch is \emph{principal} if, for all \(0 \leq i < n\) such that \(\riA_i\) is an elimination or \(\EM\) rule instance, \(\foA_i\) is the major (leftmost) premiss of \(\riA_i\). 
\end{definition}
We use the variables \(\brA, \brB\) for branches. 

In order to study the properties of normal proofs we only need to consider the structure of principal branches. 
A head-cut is the lowest point of a principal branch where a reduction is possible. 
\begin{definition}[Head-cut]
  The \emph{head-cut} of a principal branch \(\brA = \foA_1, \dotsc, \foA_n\) is the formula occurrence \(\foA_i\) with the maximum index \(i\) such that one of the following holds:
  \begin{itemize}
    \item 
      \(\foA_i\) is the conclusion of an elimination rule instance \(\riA_i\), 
      \(\foA_{i-1}\) is the major premiss of \(\riA_i\) and the conclusion of an introduction rule instance \(\riA_{i-1}\);
      when \(\riA_{i-1}\) is a \(\RuleName\land{I}\) rule instance 
      we also require that 
      \(\foA_{i-2}\) is an occurrence of the same formula as \(\foA_i\) (proper reductions);
    \item 
      \(\foA_i\) is the conclusion of an \(\RuleNameIndE\) rule instance \(\riA_i\), 
      whose main term is either \(\num 0\) or \(\succ(\ltb)\) for some term \(\ltb\) (\(\RedNameInd\) reduction); 
    \item 
      \(\foA_i\) is the conclusion of an \(\EM\) rule instance \(\riA\) 
      and either
      \(\foA_{i-1}\) is derived without using the assumption discharged by \(\riA\)
      or
      \(\foA_0\) is an occurrence of the universal assumption discharged by \(\riA\)
      and 
      \(\foA_1\) is the occurrence of a closed atomic formula (\(\RedNameWitness\) reduction);
    \item
      \(\foA_i\) is the conclusion of an elimination rule instance \(\riA\) 
      and 
      \(\foA_{i-1}\) is the conclusion of an \(\EM\) rule instance 
      (\(\RedNamePerm\EM\) reductions). 
  \end{itemize}
  If such an \(i\) exists we say that there is a head-cut along the branch \(\brA\). 
\end{definition}
This definition is the result of a analysis of the conditions that must be met in order to perform one of the reductions we have listed. 
In particular note how, in the condition given for the \(\RedNameWitness\) reduction, the fact that \(\foA_1\) is atomic implies that \(\foA_1\) is the conclusion of a \(\RuleName\forall{E}\) rule instance, as we assumed in defining \(\RedNameWitness\). 

We shall show that, with suitable assumptions,
we can perform the \(\RedNameWitness\) reduction as needed in order to extract a witness from a derivation. 
One of these assumptions is that the conclusion of the derivation is ``simple'' enough, as we define next.

\begin{definition}[Simple Formulas]\label{def:simple_formula}
  We say that a formula is \emph{simply existential} (resp. \emph{universal}) when it is \(\qexists\lva \lafa\) (resp. \(\qforall\lva \lafa\)) for some atomic formula \(\lafa\).

  We say that a formula is \emph{simple} when it is closed and atomic or simply existential. 
\end{definition}

In the following we consider the \(\EM\) and \(\RuleNameIndE\) rules to be neither elimination nor introduction rules and we give them special treatment. 

As we shall show later, 
principal branches beginning with an open assumption have particular structure in normal derivations: they begin with a sequence of elimination rule instances, followed by atomic and \(\EM\) rule instances and they end with introduction and \(\EM\) rule instances. Any of these parts may be missing.
\begin{definition}[Open normal form]
  \label{def:open_normal_form}
  A principal branch \(\foA_0, \dotsc, \foA_n\) is said to be in \emph{open normal form} when
  there exist three natural numbers \(n_E, n_A\) and \(n_I\)  such that \(n_E + n_A n_I = n \) and:
  \begin{itemize}
    \item 
      \(\foA_0\) is the occurrence of an open assumption in \(\Pi\),
    \item 
      \(\foA_i\) is the conclusion of an elimination rule instance 
      for \( 0 < i \leq n_E \), 
    \item 
      \(\foA_i\) is the conclusion of an atomic or \(\EM\) rule instance 
      for \( n_E < i \leq n_E + n_A \), 
    \item 
      \(\foA_{n_E + n_A + 1}\) is the conclusion of an introduction rule instance\footnotemark, 
    \item 
      \(\foA_i\) is the conclusion of an introduction or \(\EM\) rule instance 
      for \( n_E + n_A < i \leq n \), 
  \end{itemize}
  \(n_E, n_A\) and \(n_I\) are the number of elimination, atomic or \(\EM\), introduction or \(\EM\) rule instances, respectively.
  \footnotetext{Since \(\EM\) rule instances can appear intermingled with both atomic and introduction rule instances, in the definition we require that \(\foA_{n_E + n_A + 1}\) be the conclusion of an introduction rule, so that \(n_A\) and  \(n_I\) are uniquely determined. } 
\end{definition}

We can now prove our main result: 
closed normal derivations of simply existential formulas in \(\HA+\EM\) can be reduced to derivations ending with an introduction rule instance. 
Derivations in \(\HA\) have a similar property. 
The theorem we are going to prove holds for derivations that are concrete enough, namely they are:
self-contained (without open assumptions), concrete (without open term variables) and with an effective conclusion (a simply existential formula). 
The proof is split into several lemmas.

In the first lemma we show that, in a derivation of a simply existential with no free term variables, a simply universal assumption is followed by a closed atomic formula. 
This will be used later to prove that we can perform the \(\RedNameWitness\) reduction on universal assumption of an \(\EM\) rule instance. 
\begin{lemma}\label{thm:forall_elim_closed}
  Let \(\brA = \foA_0, \dotsc, \foA_n\) be a principal branch in open normal form in a derivation \(\derA\) in \(\HA+\EM\), with \(n_E, n_A\) and \(n_I\) defined as in \Cref{def:open_normal_form}.
  Let \(\fa_0, \dotsc, \fa_n\) be the formulas \(\foA_0, \dotsc, \foA_n\) are occurrences of. 
  Then the following statements hold: 
  \begin{enumerate}
    \item\label{thm:subformula_intro_branch_a}
      \(\fa_i\) is a non-atomic subformula of \(\fa_n\) for all \(n_E + n_A < i \leq n\);
    \item\label{thm:subformula_intro_branch}
      if some \(\foA_i\) is the conclusion of an introduction rule instance, 
      then \(\fa_i\) is a subformula of \(\fa_n\);
  \end{enumerate}
  Moreover assume that \(\fa_n\) is a simple formula. Then:
  \begin{enumerate}[resume]
    \item\label{thm:free_var_branch}
      if a term variable \(\lva\) is free in some \(\fa_i\), 
      then \(\lva\) is free in \(\derA\);
    \item\label{thm:forall_elim_branch}
      if some \(\fa_i\) is simply universal,
      then \(\foA_i\) is the premiss of a \(\RuleName\forall{E}\) rule instance;
    \item\label{thm:cut_branch}
      if \(\derA\) has no free term variables and \(\fa_i\) is simply universal,
      then \(\fa_{i+1}\) is a closed atomic formula. 
  \end{enumerate}
\end{lemma}
\begin{proof} 
  (\ref{thm:subformula_intro_branch}) follows immediately from  (\ref{thm:subformula_intro_branch_a}). 
  We need (\ref{thm:subformula_intro_branch}) to prove (\ref{thm:free_var_branch}) and (\ref{thm:forall_elim_branch}). 
  Then, by (\ref{thm:free_var_branch}) and (\ref{thm:forall_elim_branch}), we prove (\ref{thm:cut_branch}). 
  Here are the proofs. 
  \begin{enumerate}
    \item
      We proceed by induction on \(n_I\). 
      \begin{itemize}
        \item 
          If \(n_I = 0\), the thesis holds vacuously. 
        \item
          If \(n_I = 1\), we need to prove the statement just for \(i = n_E + n_A + 1 = n\) and thus \(\foA_n\) is the conclusion of an introduction rule instance by \Cref{def:open_normal_form}. 
          This means that \(\fa_n\) is not atomic. 
          Obviously it is also a subformula of itself so we are done. 
        \item
          Otherwise, let \(n_I > 1 \). 

          Consider the subderivation \(\derA'\) of \(\derA\) ending with \(\foA_{n-1}\) and its principal branch \(\brA' = \foA_0, \dotsc, \foA_{n - 1}\). 
          \(\brA'\) is in open normal form in \(\derA'\), with \(n_E' = n_E, n_A' = n_A\) and \(n_I' = n_I - 1\).
          Then, by inductive hypothesis, 
          for all \(n_E + n_A < i \leq n-1\),
          \(\fa_i\) is a non-atomic subformula of \(\fa_{n-1}\). 

          By \Cref{def:open_normal_form},
          \(\derA\) ends with an introduction or \(\EM\) rule instance \(\riA\). 
          In both cases \(\fa_{n-1}\) is a subformula of \(\fa_n\), since \(\foA_{n-1}\) is the premiss of \(\riA\) and \(\foA_n\) is its conclusion. 

          Thus for all \(n_E + n_A < i \leq n\), \(\fa_i\) is a subformula of \(\fa_n\). 
          Moreover since \(\fa_{n-1}\) is non-atomic then \(\fa_n\) is too. 
      \end{itemize}
    \item 
      If \(\foA_i\) is the conclusion of an introduction rule instance then \(n_E + n_A < i \leq n \) by \Cref{def:open_normal_form}. 
      Thus we conclude by (\ref{thm:subformula_intro_branch_a}). 
    \item 
      We show that \(\lva\) is not bound by any rule instance and thus is free in \(\derA\).
      The only rule that binds a variable above a major premiss, 
      and thus the only rule that can bind a variable in a principal branch, 
      is the \(\RuleName\forall{I}\) rule.
      Now assume that a \(\RuleName\forall{I}\) rule instance occur along \(\brA\) with conclusion \(\foA_j\). 
      By the previous statement (\ref{thm:subformula_intro_branch}), 
      \(\fa_j\) is a subformula of \(\fa_n\). 
      This yields a contradiction because by assumption \(\fa_n\) is simple and \(\fa_j\) is universally quantified since \(\foA_j\) is the conclusion of a \(\RuleName\forall{I}\) rule instance. 
    \item 
      We show that \(\foA_i\) is the premiss of a \(\RuleName\forall{E}\) rule instance because no other alternative is possible. 
      \begin{itemize}
        \item 
          \(\foA_i\) cannot be the premiss of an atomic rule instance, since we assumed that \(\fa_i\) is simply universal and thus not atomic.
        \item 
          \(\foA_i\) cannot be the premiss of an introduction rule instance, since in that case \(\fa_{i+1}\) is a subformula of \(\fa_n\) by (\ref{thm:subformula_intro_branch}). 
          Therefore a simply universal formula \(\fa_i\) is a subformula of a simple formula \(\fa_n\), which is a contradiction. 
        \item 
          Finally \(\foA_i\) cannot be the premiss of an \(\EM\) rule instance. 
          More precisely assume that \(\foA_i\) is followed by exactly \(j > 0\) instances of the \(\EM\) rule. 
          Then \(\foA_{i+j}\) is the conclusion of the last \(\EM\) rule instance \(\riA\) and \(\fa_{i+j}\) is the same formula as \(\fa_i\), in particular \(\fa_{i+j}\) is simply universal. 
          By definition of open normal form, \(\EM\) rule instances can only be followed by introduction, atomic or \(\EM\) rule instances. 
          Since we assumed that there are exactly \(j\) instances of the \(\EM\) rule, \(\foA_{i+j}\) is the premiss of either an atomic or introduction rule instance. 
          Then we are in one of the previous cases and we have a contradiction.
      \end{itemize}
      Then \(\foA_i\) can only be the premiss of an elimination rule instance and, 
      being \(\fa_i\) simply universal, 
      it must be an instance of the \(\RuleName\forall{E}\) rule. 
    \item 
      By (\ref{thm:forall_elim_branch}) we known that \(\foA_{i+1}\) is the conclusion of a \(\RuleName\forall{E}\) rule instance whose premiss is simply universal.
      Therefore \(\fa_{i+1}\) is an atomic formula.
      If \(\fa_{i+1}\) has a free term variable, \(\derA\) has too by (\ref{thm:free_var_branch}). 
      Since we assumed that \(\derA\) has no free term variable, \(\fa_{i+1}\) must be closed.
      Therefore \(\fa_{i+1}\) is a closed atomic formula. \qed
  \end{enumerate}
\end{proof}

In the following lemma we show how we can apply the \(\RedNameWitness\) reduction. 
\begin{lemma}[\(\EM\) reduction]\label{thm:em_open_branch}
  Let \(\derA\) be a derivation in \(\HA+\EM\) with no free term variables. 
  Assume that \(\derA\) ends with an \(\EM\) rule instance \(\riA\) whose conclusion is an occurrence of a simple formula. 
  Assume that the derivation \(\derA'\) of the leftmost premiss of \(\riA\) has a principal branch \(\brA\) in open normal form.
  Then at least one of the following occurs:
  \begin{enumerate}
    \item \(\derA\) has a head-cut or a non-normal term along a principal branch,
    \item \(\derA\) has a principal branch in open normal form.
  \end{enumerate}
\end{lemma}
\begin{proof} 
  Let \(\brA = \foA_0, \dotsc, \foA_n\) and let \(\fa_0, \dotsc, \fa_n\) be the formulas \(\foA_0, \dotsc, \foA_n\) are occurrences of. 
  Let \(\foA\) be the conclusion of \(\derA\) and of the \(\EM\) rule instance \(\riA\):
  \[
    \PrAss[\foA_0]{\fa_0}{}
    \PrInfBr{\derA'}\brA
    \PrUn[\foA_n]{\fa_n}
    \PrAx{}
    \PrInf
    \PrUn{\fa_n}
    \PrLbl[\riA]\EM
    \PrBin[\foA]{\fa_n}
    \DisplayProof
  \]
  Note that we can extend \(\brA\) to \(\brB = \foA_0, \dotsc, \foA_n, \foA\) and \(\brB\) is a principal branch of \(\derA\).

  If \(\foA_0\) is not discharged by \(\riA\), then \(\brB\) is a principal branch in open normal form and thus we get the statement. 
  Otherwise, \(\foA_0\) is discharged by \(\riA\), meaning that \(\fa_0\) is the universal assumption of the \(\EM\) instance. 
  We can apply (\ref{thm:cut_branch}) of \Cref{thm:forall_elim_closed} to \(\brB\), since \(\foA_0\) is simply universal, \(\foA\) is simply existential and \(\derA\) has no free term variables. 

  Then \(\foA_1\) is a closed atomic formula and 
  we can perform the \(\RedNameWitness\) reduction, that is,
  there is a head-cut along the principal branch \(\brB\) of \(\derA\) and we can conclude. \qed

\end{proof}

The following lemma shows how to handle \(\RuleNameIndE\) rule instances. 
\begin{lemma}[Induction normalization]\label{thm:induction}
  Let \(\derA\) be a derivation in \(\HA+\EM\) ending with an \(\RuleNameIndE\) rule instance. 
  Then at least one of the following holds:
  \begin{enumerate}
    \item \(\derA\) has a head-cut or a non-normal term along a principal branch, 
    \item \(\derA\) contains a free term variable. 
  \end{enumerate}
\end{lemma}
\begin{proof}
  Let \(\riA\) be the \(\RuleNameIndE\) rule instance \(\derA\) ends with
  and
  let its conclusion be an occurrence \(\foA\) of some formula \(\fa\subst\lva\lta\).
  If \(\foA\) has a free term variable \(\lva\) then \(\lva\) is free in \(\derA\) too and we are done. 
  If \(\lta\) is not normal then all principal branches\footnote{Since all branches of \(\derA\) end with its conclusion.} in \(\derA\) have a non-normal term. 
  Otherwise \(\lta\) is a closed normal term and thus 
  it is either \(\num 0\) or \( \succ(\ltb)\) for some term \(\ltb\) and e we can apply the \(\RedNameInd\) reduction, meaning that any principal branch of \(\derA\) has a head cut. \qed
\end{proof}

The following lemma can be thought of as a weak result on the structure of derivations. 
\begin{lemma}[Structure of Normal Form]\label{thm:em_reduction}
  Let \(\derA\) be a derivation in \(\HA+\EM\). Then at least one of the following holds:
  \begin{enumerate}
    \item \(\derA\) has a head-cut or a non-normal term along a principal branch; 
    \item \(\derA\) contains a free term variable; 
    \item \(\derA\) has a principal branch in open normal form; 
    \item \(\derA\) ends with an introduction rule instance;
    \item \(\derA\) is atomic (only atomic formulas occur in \(\derA\)); 
    \item \(\derA\) ends with an \(\EM\) instance and its conclusion is not simple.
  \end{enumerate}
\end{lemma}
\begin{proof} 
  The proof is by induction on the structure of the derivation \(\derA\), that is, we assume that the statement holds for all subderivations of \(\derA\) and we prove that it holds for the whole derivation. 

  Let \(\riA\) be the last rule instance in \(\derA\). 
  If \(\riA\) in an introduction rule instance the statement is satisfied and we are done. 

  If \(\riA\) is an \(\RuleNameIndE\) rule instance then we get the statement by applying \Cref{thm:induction} to \(\derA\).

  Then we only need to understand what happens when \(\riA\) is an elimination, an atomic or an \(\EM\) rule instance. 
  Note that the only case in which \(\riA\) has no premisses is when \(\riA\) is an atomic axiom.
  If this happens then \(\derA\) is atomic (it is just the conclusion of \(\riA\)) and the statement is satisfied.

  Otherwise \(\riA\) has one or more (when \(\riA\) is atomic) major premisses. 
  Let \(\derA'\) be the derivation of any one of the major premisses of \(\riA\).
  Any principal branch \(\brA\) of \(\derA'\) can be extended to a branch \(\brB\) of \(\derA\), by appending the conclusion of \(\derA\). 
  \(\brB\) is principal too because \(\derA'\) is the subderivation of a major premiss of \(\riA\). 
  We shall use this fact often in the following. 


  By inductive hypothesis \(\derA'\) satisfies the statement, 
  so we proceed by considering all the possible cases. 
  \begin{enumerate}
    \item\label{case:cut} 
      \(\derA'\) has  a head-cut or a non-normal term along a principal branch \(\brA\). 
      As we noted \(\brA\) can be extended to a principal branch of \(\derA\) with the same head-cut or non-normal term,
      so \(\derA\) satisfies the statement and we are done. 
    \item\label{case:free_var} 
      \(\derA'\) contains a free term variable. 

      There are four rules that can bind term variables: the \(\RuleName\forall{I}, \RuleName\exists{E}, \RuleNameIndE\) and \(\EM\) rules. 
      Since the cases when \(\riA\) is an introduction or \(\RuleNameIndE\) rule instance have been taken care of already and 
      since the \(\RuleName\exists{E}\) and \(\EM\) rules can only bind term variables in the derivation of its minor premiss, 
      any free term variable in \(\derA'\) is free in \(\derA\) too. 
      Thus \(\derA\) satisfies the statement. 

    \item\label{case:open_ass} 
      \(\derA'\) has a principal branch \(\brA = \foA_0, \dotsc, \foA_n\) in open normal form. 
      Let \(\brB\) be the principal branch of \(\derA\) extending \(\brA\). 
      \[
        \PrAss[\foA_0]{\fa_0}{}
        \PrInfBr{\derA'}\brA
        \PrUn[\foA_n]{\fa_n}{}
        \PrAx\dotso
        \PrLbl[\riA]{\EM|\RuleName\ast{E}|\mathcal{A}}
        \PrBin\fa
        \DisplayProof
      \]
      Note that elimination rule instances do not discharge assumptions in their leftmost subderivation and atomic rule instances do not discharge assumptions in general. 

      Then, when \(\riA\) is either an elimination or atomic rule instance, 
      the assumption \(\foA_0\) is still open in \(\derA\) and 
      we have the following cases depending on how which rule \(\foA_n\) is the conclusion of. 
      Note that since \(\brA\) is in open normal form, \(\foA_n\) cannot be an \(\RuleNameIndE\) instance. Thus we have the following cases: 
      \begin{center}
        \begin{tabular}{cc|c|c|c|c|c|}
          \cline{3-6}
          & &\multicolumn{4}{|c|}{\(\foA_n\) is the conclusion of} \\ 
          \cline{3-6}
          & & ELIM & ATOM & INTRO & \(\EM\) \\ 
          \hline
          \multicolumn{1}{|c|}{\multirow{2}{*}{\(\riA\)}} & ELIM & EXT & NO & CUT & PERM \\
          \cline{2-6}
          \multicolumn{1}{|c|}{} & ATOM & EXT & EXT & NO & NO/EXT \\
          \hline
        \end{tabular}
      \end{center}

      \begin{itemize}
        \item[EXT] 
          \(\brB\) begins with the open assumption \(\foA_0\) followed by elimination and atomic rule instances, so \(\derA\) satisfies the statement; 
        \item[NO] 
          this is never the case since, in a principal branch, 
          an elimination (resp. atomic) rule instance cannot follow an atomic (resp. introduction) rule instance, because
          the major premiss (resp. conclusion) of an elimination (resp. introduction) rule instance is not atomic and thus cannot be the conclusion (resp. premiss) of an atomic rule instance;
        \item[CUT] 
          \(\riA\) is an elimination rule instance and its major premiss is the conclusion of an introduction rule instance, 
          thus \(\brB\) ends with a head-cut and again \(\derA\) satisfies the statement;
        \item[PERM] 
          when a major premiss of \(\riA\) is the conclusion of an \(\EM\) rule instance we can apply the \(\RedNamePerm\EM\) reduction, thus \(\brB\) ends with a head-cut and \(\derA\) satisfies the statement;
        \item[NO/EXT] 
          we have two cases depending on the \(n_I\) of \(\brA\):
          \begin{itemize}
            \item[\(n_I > 0\)] 
              then, by (\ref{thm:subformula_intro_branch_a}) of \Cref{thm:forall_elim_closed}, we have that \(\foA_n\) is an occurrence of a non atomic formula and thus \(\riA\) cannot be an atomic rule instance;
            \item[\(n_I = 0\)]
              in this case \(\brB\) is in open normal form and thus \(\derA\) satisfies the statement.
          \end{itemize}


      \end{itemize}
      On the other hand, if \(\riA\) is an \(\EM\) rule instance then we can apply \Cref{thm:em_open_branch} to \(\derA\) (since we can safely assume that \(\derA\) contains no free term variables) and we get the statement. 
    \item \(\derA'\) ends with an introduction rule instance \(\riB\).
      Then the conclusion \(\foB\) of \(\riB\) cannot be atomic and since it is a premiss of \(\riA\), \(\riA\) cannot be atomic either. 
      Therefore \(\riA\) must be either an elimination or an \(\EM\) rule instance.

      If \(\riA\) is an elimination then there is a head-cut along a principal branch going through \(\foB\) so \(\derA\) satisfies the statement. 
      \[
        \PrAx{}
        \PrInf
        \PrLbl[\riB]{\RuleName{?}{I}}
        \PrUn[\foB]\fb
        \PrAx\dotso
        \PrLbl[\riA]{\RuleName{?}{E}}
        \PrBin[\foA]\fa
        \DisplayProof
      \]
      If \(\riA\) is an \(\EM\) rule instance then \(\foA\) and \(\foB\) are both occurrences of the formula \(\fa\).
      If \(\fa\) is not simple then \(\derA\) satisfies the statement. 
      Otherwise we assume that \(\fa\) is simple: 
      since \(\fa\) cannot be atomic 
      (it occurs as the conclusion of the introduction rule instance \(\riB\))
      \(\fa\) must be \(\qexists\lva \lafa\) for some atomic formula \(\lafa\)
      and
      \(\riB\) must be an \(\RuleName\exists{I}\) rule instance. 
      Then we are in the following situation:
      \[ 
        \PrAx{}
        \PrInf[\derA'']
        \PrUn\fa
        \PrLbl[\riB]{\RuleName\exists{I}}
        \PrUn[\foB]{\qexists\lva\lafa}
        \PrAx{}
        \PrInf
        \PrUn{\qexists\lva\lafa}
        \PrLbl[\riA]\EM
        \PrBin[\foA]{\qexists\lva\lafa}
        \DisplayProof
      \]
      where \(\derA''\) is the derivation of the premiss of \(\riB\). 
      Again note that principal branches of \(\derA''\) extend to principal branches of \(\derA'\) by appending \(\foB\). 

      By inductive hypothesis \(\derA''\) satisfies the statement, 
      so we proceed by considering all the possible cases. 
      \begin{enumerate}
        \item 
          \(\derA''\) has a head-cut or a non-normal term along a principal branch.
          Then \(\derA'\) does too and we are in the previously solved case labeled \ref{case:cut}. 
        \item \(\derA''\) contains a free term variable. 
          Since \(\RuleName\exists{I}\) does not bind free term variables, 
          \(\derA'\) does too and we are in the previously solved case labeled \ref{case:free_var}. 
        \item
          \(\derA''\) has a principal branch \(\brA\) in open normal form. 
          Since \(\RuleName\exists{I}\) does not discharge open assumptions,
          \(\derA'\) does too and we are in the previously solved case labeled \ref{case:open_ass}. 
        \item 
          \(\derA''\) ends with an introduction rule instance. 
          This cannot happen because \(\foB\) is an atomic formula occurrence. 
        \item 
          \(\derA''\) is atomic. 
          The assumption discharged by \(\riA\) from its leftmost subderivation is not atomic, 
          thus it cannot occur in \(\derA''\) since \(\derA''\) is atomic. 
          Therefore we can apply the \(\RedNameWitness\) reduction meaning that 
          there is a head-cut at the end of the principal branches of \(\derA'\)
          and we are in the previously solved case labeled \ref{case:cut}. 
        \item \(\derA''\) ends with an \(\EM\) instance and its conclusion is not simple.
          This cannot happen because we assumed that \(\fa\) is simple and \(\foB\) is an occurrence of \(\fa\).
      \end{enumerate}

    \item \(\derA'\) is atomic. 
      In this case one of major premisses of \(\riA\) is atomic, so \(\riA\) cannot be an elimination rule instance and must be either an \(\EM\) or atomic rule instance. 
      \begin{itemize}
        \item
          If \(\riA\) is an \(\EM\) rule instance then it is redundant: 
          the assumption discharged by \(\riA\) from its leftmost subderivation is not atomic, 
          thus it cannot occur in \(\derA'\) since \(\derA'\) is atomic whose premisses are atomic formula occurrences. 
          Therefore we can apply the \(\RedNameWitness\) reduction to \(\riA\), meaning that 
          there is a head-cut at the end of the principal branches of \(\derA\) and thus \(\derA\) satisfies the statement. 
        \item
          Otherwise, if \(\riA\) is an atomic rule instance, consider the other subderivations of its major premisses. 
          If they are all atomic then \(\derA\) is atomic too and it satisfies the statement. 
          Otherwise there is a major premiss of \(\riA\) with a non atomic derivation \(\derA''\). 
          Then one of the other cases applies with \(\derA''\) in place of \(\derA'\). 
      \end{itemize}

    \item \(\derA'\) ends with an \(\EM\) rule instance \(\riB\) and its conclusion is not simple.
      \(\riA\) cannot be an atomic rule instance since one of its premisses is the conclusion of \(\riB\) which is not simple and thus not atomic. 
      If \(\riA\) is an elimination rule instance we can apply the \(\RedNamePerm\EM\) reduction to \(\riA\) and \(\riB\). 
      Thus there is a head-cut at the end of the principal branches of \(\derA\), 
      and \(\derA\) satisfies the statement. 
      Otherwise, if \(\riA\) is an \(\EM\) rule instance then the conclusions of \(\derA\) and \(\derA'\) are occurrences of the same non-simple formula. 
      Therefore \(\derA\) again satisfies the statement. 
  \end{enumerate}
  Since we exhausted all the possible cases we are done. 
\end{proof}

Our main theorem is now an easy corollary of the previous lemma.
\begin{theorem}[Witness Extraction] \label{thm:witness_extraction}
  Let \(\derA\) be a derivation of a simple formula \(\fa\) in \(\HA+\EM\). Assume that:
  \begin{enumerate}
    \item \(\derA\) has no principal branch with a head-cut or a non-normal term; 
    \item \(\derA\) contains no free term variable; 
    \item \(\derA\) has no open assumptions;
  \end{enumerate}
  Then \(\derA\) is either atomic or ends with a \(\RuleName\exists{I}\) rule instance. 
  In particular, if \(\derA\) is closed, normal and \(\fa\) is simply existential then \(\derA\) ends with an introduction. 
\end{theorem}
\begin{proof}
  The hypotheses rule out most of the cases considered by \Cref{thm:em_reduction}. 
  The only possible cases are:
  \begin{enumerate}
    \item \(\derA\) ends with an introduction rule instance,
    \item \(\derA\) is atomic.
  \end{enumerate}
  Since \(\fa\) is simple it is either an atomic or existentially quantified formula. 
  If \(\fa\) is atomic, \(\derA\) cannot end with an introduction rule instance and thus \(\derA\) must be atomic.
  Otherwise, if \(\fa\) is existentially quantified, \(\derA\) cannot end with an atomic rule instance and thus 
  \(\derA\) must end with an introduction rule instance which can only be  a \(\RuleName\exists{I}\) rule instance. 
\end{proof}
By \cref{thm:em_strong_normalization}, 
the reduction of a derivation halts after a finite number of steps and produces a derivation without head-cuts.
Then, \cref{thm:witness_extraction} shows that our proof reduction can extract a witness from the derivation of a closed formula \(\qexists\lva \lafa\), which can be found in the premise of the \(\RuleName\exists{I}\) rule instance at the end of the normalized derivation, by the definition of the \(\RuleName\exists{I}\) rule.


\begin{omitted}
  \newcommand\pathA{\mathfrak{a}}
  \newcommand\pathB{\mathfrak{b}}
  \begin{definition}[Path]
    A \emph{path} in a derivation \(\Pi\) is a sequence \(\foA_1, \dotsc, \foA_n\) of formula occurrences in \(\Pi\) such that:
    \begin{itemize}
      \item \(\foA_1\) is the conclusion of an atomic axiom instance or an assumption that is not discharged by a \(\RuleName\lor{E}\) or \(\RuleName\exists{E}\) rule instance; 
      \item for all \(i < n\), \(\foA_i\) and \(\foA_{i+1}\) are respectively a premise and the conclusion of the same rule instance, except when \(\foA_i\) is the major premise of an instance \(\riA\) of a \(\RuleName\lor{E}\) or \(\RuleName\exists{E}\) rule: 
        in this case \(\foA_{i+1}\) is an occurrence of an assumption discharged by the rule instance;
      \item \(\foA_n\) is either the conclusion of \(\Pi\) or the minor premise of a \(\RuleName\limply{E}\) and no \(\foA_i\) with \(i < n\) is such a minor premiss.
    \end{itemize}
  \end{definition} 

  \begin{definition}
    A path \(\foA_1, \dotsc, \foA_n\) of occurrences of formulas \(\fa_1, \dotsc,\fa_n\)
    \begin{itemize}
      \item 
        there are \( 1 \leq m_1 \leq m_2 \leq n\) such that 
        \begin{itemize}
          \item 
            \(\foA_i\) is the conclusion of an elimination rule instance
            for all \( 1 < i \leq m_1\),
          \item 
            \(\foA_i\) is the conclusion of an atomic rule instance 
            for all \( m_1 < i \leq m_2\),
          \item 
            \(\foA_i\) is the conclusion of an introduction rule instance 
            for all \( m_2 < i \leq n\),
        \end{itemize}
      \item 
        has the subformula property when \(\fa_i\) is a subformula

    \end{itemize}
  \end{definition}

  \begin{theorem}
    Let \(\derA\) be a derivation in \(\HA\). 
    Then all paths in \(\derA\) have are normal or 
    \(\derA\) is reducible.
  \end{theorem}
  \begin{proof}
    By induction on the structure of \(\Pi\). 

    Consider any path \(\pathA = \foA_1, \dotsc, \foA_n\) of \(\Pi\).
    If \(\foA_n\) is not the conclusion of \(\Pi\) then \(\pathA\) is contained in a subderivation \(\derB\) of \(\derA\). 
    By inductive hypothesis either \(\pathA\) is normal or \(\derB\) is reducible and then so is \(\derA\) and we are done.

    Otherwise \(\foA_n\) is the conclusion of \(\Pi\). 
    Let \(\riA\) be the rule instance \(\foA_n\) is the conclusion of. 
    If \(\riA\) is an atomic axiom then \(n = 1\) and \(\pathA\) is normal.
    Otherwise we can assume that \(n > 1\). 
    Let \(\riB\) be the rule instance \(\foA_{n-1}\) is the conclusion of and let \(\derB\) the derivation of \(\foA_{n-1}\). 
    Let \(\pathB\) be the path \(\foA_1, \dotsc, \foA_{n-1}\) of \(\derB\). 
    \[
      \PrAx\dotso
      \PrAx\dotso
      \PrLbl[\riB]{}
      \PrUn{\fa_{n-1}}
      \PrAx\dotso
      \PrLbl[\riA]{}
      \PrTri{\fa_n}
      \DisplayProof
    \]
    We consider various cases depending on which rule \(\riA\) is an instance of.
    \begin{description}
      \item[atomic] 
        If \(\riA\) is an atomic rule then 
      \item[introduction]
        If \(\riA\) is an introduction rule instance 
      \item[\(\riA\) is atomic]
      \item[\(\riA\) is atomic]
    \end{description}
  \end{proof}

\end{omitted}

\chapter{Interpreting a Geometric Example with Interactive Realizability}

\newcommand\Q{\mathbb Q}
\newcommand\R{\mathbb R}

\newcommand\nb{m}
\newcommand\ia{i}
\newcommand\ib{j}
\newcommand\ic{k}
\renewcommand\id{l}

\newcommand\pa{P}
\newcommand\pb{Q}
\newcommand\pc{R}
\newcommand\pd{S}

\newcommand\pointAt[4]{%
  \coordinate (#1) at (#2);
  \draw[fill] (#1) circle (0.05);
  \node[#4] at (#1) {\(#3\)};
}
\newcommand\lineBetween[2]{%
  \draw [shorten >= -1cm, shorten <=-1cm] (#1)--(#2);
}

In this chapter we show how to extract a monotonic learning algorithm from a classical proof of a geometric statement by interpreting the proof by means of interactive realizability. 

The statement is about the existence of a convex angle including a finite collections of points in the real plane and it is related to the existence of a convex hull. 
We define real numbers as Cauchy sequences of rational numbers, 
therefore equality and ordering are not decidable. 
While the proof looks superficially constructive, 
it employs classical reasoning to handle undecidable comparisons between real numbers, 
making the underlying algorithm non-effective. 

The interactive realizability interpretation transform the non-effective linear algorithm described by the proof into an effective one that uses backtracking to learn from its mistakes. 
The effective algorithm exhibit a ``smart'' behavior, performing comparisons only up to the precision required to prove the final statement. 
This behavior is not explicitly planned but arises from the interactive interpretation of comparisons between Cauchy sequences. 

\section{Introduction}



We study the computational content of the proof of the following geometric statement. 
\begin{theorem}[Convex Angle]
  We have a finite set of at least three points in the real plane \(\R^2\) such that no three points are on the same line. 
  Then there exist distinct points \(\pa, \pb\) and \(\pc\) such that: 
  \begin{itemize}
    \item all other points \(\pd\) are inside \(\widehat{\pb\pa\pc}\), 
    \item the angle \(\widehat{\pb\pa\pc}\) is convex, that is, less than \(\pi\). 
  \end{itemize} 
\end{theorem} 
\begin{figure}[!ht] 
  \begin{center}
  \begin{tikzpicture}[>=latex] 
    \pointAt{A}{0,0}\pa{below}
    \pointAt{B}{45:2}\pb{right}
    \lineBetween{A}{B}
    \pointAt{C}{140:1.5}\pc{left}
    \lineBetween{A}{C}
    \pointAt{A}{0,0}\pa{below}
    \pointAt{D}{60:1.5}{}{}
    \pointAt{D}{70:2}{}{}
    \pointAt{D}{80:1.5}{}{}
    \pointAt{D}{90:1}{}{}
    \pointAt{D}{100:1.5}{}{}
    \pointAt{D}{110:2}{}{}
    \pointAt{D}{120:1}{}{}
  \end{tikzpicture}
  \end{center}
\end{figure}

We choose this particular statement because we have a proof of it that looks algorithmic and can be easily visualized. 
The \gref{thm:bounding_angle} be thought of as weakened version of the existence of the convex hull of a finite set of points. 

As we said proof we choose as example looks constructive, using only decidability of ordering over real numbers. 
However, it is well known that there is no effective ordering on the real numbers. 
In our encoding of the real numbers, totality of the ordering on the recursive reals is equivalent to \(\EM\). 
Since the proof needs the ordering to be total, it needs \(\EM\). 
Due to the low logical complexity of excluded middle which is used, the proof may be interpreted with a simple case of interactive realizability.

We show how interactive realizability can be applied and what it can tell us about the computational content of the proof. 
What we get is an algorithm that, instead of comparing real numbers, makes an arbitrary guess about which one is smaller.
If later it becomes apparent that the guess is wrong the algorithm retracts the choice it made since it can now make an informed decision about that particular comparison. 
Then the algorithm performs comparisons only when needed and only up to the required precision. 

Thus we see how a simple classical proof which performs comparisons between real numbers is interpreted as a learning algorithm which uses ``educated guesses'' 
in order to avoid non effective operations. 
This non-trivial behavior is not explicit in the classical proof, but follows from the definition of ordering on Cauchy sequences by means of the interactive realizability interpretation. 

In this chapter, our main goal is to showcase interactive realizability and the backtracking algorithms it produces through a non-trivial example. 
For this reason, we chose to present interactive realizability as a proof interpretation technique rather than as a realizability semantics, in order to concentrate on the example and its computational interpretation without being bogged down in technical details.  

Note that interactive realizability is by no means the only approach to extract a computational interpretation from our proof. 
It should also be noted that while our proof is classical, it can be seen that our statement admits an intuitionistic proof by the conservativity results in \cite{berger02}. 

\section{Real Numbers}

In this section we present our treatment of real numbers in Heyting Arithmetic. 

\newcommand\eqQ{=_\Q}
\newcommand\ltQ{<_\Q}
\newcommand\gtQ{>_\Q}
\newcommand\leQ{\le_\Q}
\newcommand\geQ{\ge_\Q}
\newcommand\Qplus{+_\Q}
\newcommand\Qminus{-_\Q}
\newcommand\Qprod{\cdot_\Q}
\newcommand\Qzero{0_\Q}
\newcommand\lqa{q}
\newcommand\lqb{p}
There are many ways of encoding integer and rational numbers in \(\HA\) and defining primitive recursive operations and predicates on them. 
In the following we assume that we have any such encoding and that we have decidable equality \(\eqQ\) and ordering \(\ltQ, \leQ\) and effective operations \(\Qplus\, \Qprod\).
We use the variables \(\lqa\) and \(\lqb\) for rationals. 


\subsection{Cauchy Sequences}
There are many equivalent ways of defining the real numbers from the rational numbers. 
The most known are the definition of the reals as equivalence classes of Cauchy sequences and as Dedekind cuts.
We follow the first approach. 

\newcommand\lrea{r}
\newcommand\lreb{s}
\newcommand\lrec{t}
\newcommand\lpa{k}
\newcommand\lpb{l}
A sequence of rationals \( \lrea : \N \to \Q \) is a \emph{Cauchy sequence} if the following holds:
\begin{equation} \label{eq:classical_cauchy}
  \qforall\lpa \qexists{\lpa_0} \qforall{\lpa_1, \lpa_2}
  | \lrea(\lpa_0 + \lpa_2) - \lrea(\lpa_0 + \lpa_1) | < \frac{1}{2^\lpa}. 
\end{equation}
While this sequence approximates a real number, it can do so very slowly. \note{we could cite Cauchy sequences with modulus of convergence}
By means of classical reasoning, 
we can show that, from any Cauchy sequence, we can extract a fast-converging monotone sub-sequence. 
For this reason, instead of general Cauchy sequences, we can consider sequences of nested intervals with rational extremes whose length decreases exponentially. 
An interval is determined by its extremes, so we represent a sequence of intervals as a couple of sequences of rationals \(\lrea^-, \lrea^+\), representing the lower and higher extremes of the intervals respectively.
Then we require that \(\lrea^-\) is increasing and \(\lrea^+\) is decreasing (since the intervals are nested), that \(\lrea^-(\lpa)\) is lesser than or equal to \(\lrea^+(\lpa)\) (since they are the lower and higher extremes of a same interval) and their difference is smaller than \(2^{-\lpa}\). 
More precisely we say that \(\lrea^-\) and \(\lrea^+\) represent a real number when they satisfy the following condition, written as a \(\Pi^0_1\) formula:
\begin{equation} \label{eq:nested_cauchy}
  \begin{split}
    \qforall\lpa &(\lrea^-(\lpa) \leQ \lrea^+(\lpa)) \land (\lrea^-(\lpa) \leQ \lrea^-(\lpa+1)) \land \\ 
                 &\land (\lrea^+(\lpa) \geQ \lrea^+(\lpa+1)) \land (\lrea^+(\lpa) \Qminus \lrea^-(\lpa) \leQ 2^{(-\lpa)}).
  \end{split}
\end{equation} 
While the choice of the specific definition of real number is somewhat arbitrary, it is significant because it affects the logical properties (in particular the degree of undecidability) of the ordering on the reals. 

\subsection{Order Predicate}
\newcommand\op{\text{OP}}
\newcommand\RuleNameOP[1]{\op\text{-{#1}}}

Now we can define an ``order predicate'' \(\op(\lrea, \lreb, \lpa)\), which can be thought of as a family of strict partial orders on the real numbers indexed by natural number \(\lpa\). 
More precisely, it is a formula that determines when the sequence of nested intervals \(\lrea\) is strictly lesser than \(\lreb\), at precision \(\lpa\). 
This happens when, at \(\lpa\), the higher extreme of an interval is strictly greater than the lower extreme of the other. Then, from that point forward, the intervals will be forever disjoint, since we they are nested sequences. 
This allows us to write the order predicate as the formula:
\begin{equation} \label{eq:nested_op}
  \op(\lrea, \lreb, \lpa) \equiv \lrea^+(\lpa) \ltQ \lreb^-(\lpa), 
\end{equation} 
which is decidable in \(\lrea\) and \(\lreb\). 
Note that the definition of \(\op\) depends on that of real number. If we had used the classical definition of Cauchy sequence the order predicate would be the following \(\Pi_1^0\) formula:
\begin{equation} \label{eq:cauchy_op}
  \op'(\lrea, \lreb, \lpa) \equiv \qforall{\lpb} \lpb \ge \lpa \limply \lrea(\lpb) \ltQ \lrea(\lpb).  
\end{equation} 
This is very significant for our purposes: 
the order predicate in \eqref{eq:nested_op} is decidable in \(\lrea\) and \(\lreb\) (since the order on the rationals is), while in \eqref{eq:cauchy_op} it is only \emph{negatively decidable}. 
This means that we have an effective method to decide \eqref{eq:cauchy_op} when it is false, but not when it is true. \fixme{explain better}

We need \(\op\) to satisfy some properties, written as rules in \cref{fig:op_properties}. 
\begin{figure}[!ht] 
  \caption{Rules for \(\op\).}\label{fig:op_properties}
  \renewcommand\arraystretch{2}
  \[ 
    \begin{array}{|cc|}
      \hline
      \PrAx{\op(\lrea, \lreb, \lpa)}
      \PrLbl{\RuleNameOP{mon}}
      \PrUn{\op(\lrea, \lreb, \lpa+1)}
      \DisplayProof
      &
      \PrAx{\op(\lrea, \lrea, \lpa)}
      \PrLbl{\RuleNameOP{irrefl}}
      \PrUn\lfalse
      \DisplayProof
      \\ 
      \PrAx{\op(\lrea, \lreb, \lpa)}
      \PrAx{\op(\lreb, \lrea, \lpb)}
      \PrLbl{\RuleNameOP{asym}}
      \PrBin\lfalse
      \DisplayProof \quad
      &
      \PrAx{\op(\lrea, \lreb, \lpa)}
      \PrAx{\op(\lreb, \lrec, \lpb)}
      \PrLbl{\RuleNameOP{trans}}
      \PrBin{\op(\lrea, \lrec, \max(\lpa,\lpb))}
      \DisplayProof \\[2mm]
      \hline
    \end{array}
  \]
\end{figure}
The \(\RuleNameOP{mon}\) rule expresses a monotonicity property: when an comparison at a given precision can distinguish two approximations, then comparisons at greater precision should too. 
The other rules correspond to the standard axioms for a strict partial order: irreflexivity, asymmetry and transitivity. \note{asymmetry may be useless, maybe we should remove it}

We verify that our definition of \(\op\) satisfies these properties. 
\begin{lemma} \label{thm:op_properties}
  The order predicate \(\op\) defined by \eqref{eq:nested_op} satisfies the properties given in \cref{fig:op_properties}.
\end{lemma}
\begin{proof}
  We show that the properties follow directly from the definition of \(\op\) as \eqref{eq:nested_op} and from our representation of real number as sequences of nested intervals \eqref{eq:nested_cauchy}. 
  \begin{description}
    \item[Monotonicity]
      We want to prove that 
      \begin{align*}
        \op(\lrea, \lreb, \lpa+1) &\equiv \lrea^+(\lpa+1) \ltQ \lreb^-(\lpa+1),
        \intertext{assuming that: }
        \op(\lrea, \lreb, \lpa) &\equiv \lrea^+(\lpa) \ltQ \lreb^-(\lpa).
      \end{align*}
      This follows by applying the transitive property of the order on the rationals to the following chain of inequalities: 
      \begin{align*}
        \lrea^+(\lpa+1) 
        &\leQ \lrea^+(\lpa) && \text{since \(\lrea^+\) is weakly decreasing}, \\
        &\ltQ \lreb^-(\lpa) && \text{by assumption}, \\ 
        &\leQ \lreb^-(\lpa+1) && \text{since \(\lreb^+\) is weakly increasing}. 
      \end{align*}
    \item[Reflexivity]
      We have to prove that 
      \[
        \op(\lrea, \lrea, \lpa) \equiv \lrea^+(\lpa) \ltQ \lrea^-(\lpa),
      \]
      yields a contradiction. 
      This is a consequence of the fact that \(\lrea^-(\lpa)\) and \(\lrea^+(\lpa)\) are respectively the lower and higher extremes of the same interval. 
    \item[Asymmetry]
      The \(\RuleNameOP{asym}\) is actually derivable by monotonicity and transitivity: 
      \[
        \PrAx{\op(\lrea, \lreb, \lpa)}
        \PrAx{\op(\lreb, \lrea, \lpb)}
        \PrLbl{\RuleNameOP{trans}}
        \PrBin{\op(\lrea, \lrea, \lpa)}
        \PrLbl{\RuleNameOP{irrefl}}
        \PrUn\lfalse
        \DisplayProof
      \]
    \item[Transitivity]
      We have to prove that 
      \begin{align*}
        \op(\lrea, \lrec, \lpa) &\equiv \lrea^+(\lpa) \ltQ \lrec^-(\max(\lpa,\lpb)),
        \intertext{follows from the assumptions:}
        \op(\lrea, \lreb, \lpa) &\equiv \lrea^+(\lpa) \ltQ \lreb^-(\lpa), \\
        \op(\lreb, \lrec, \lpb) &\equiv \lreb^+(\lpb) \ltQ \lrec^-(\lpb).
      \end{align*}
      We have two cases depending on whether \(\max(\lpa,\lpb)\) is \(\lpa\) or \(\lpb\).
      Since the two cases are very similar, we only show the proof of the first. 
      Thus we assume that \(\max(\lpa,\lpb) = \lpa \), which means that \(\lpa \ge \lpa\). 

      Again this follows applying the transitive property of the order on the rationals to the following chain of inequalities: \note{we could add a nice picture}
      \begin{align*} 
        \lrea^+(\lpa) &\ltQ \lreb^-(\lpa) 
        && \text{ by the first assumption}, \\
        &\leQ \lreb^+(\lpa) 
        && \text{ since \( [\lreb^-(\lpa)\), \(\lreb^-(\lpa)] \) is an interval}, \\ 
        &\leQ \lreb^+(\lpb) 
        && \text{ since \(\lreb^+\) is weakly decreasing and \(\lpa \ge \lpb\)}, \\ 
        &\ltQ \lrec^-(\lpb) 
        && \text{ by the second assumption}, \\
        &\leQ \lrec^-(\lpa) 
        && \text{ since \(\lrec^-\) is weakly increasing and \(\lpa \ge \lpb\)}.
      \end{align*}
      Thus, we have:
      \[
        \lrea^+(\lpa) \ltQ \lrec^-(\max(\lpa,\lpb)). \qedhere
      \]
  \end{description} 
\end{proof}

\subsection{Order and Equality on the Real Numbers}

\newcommand\eqR{=_\R}
\newcommand\neqR{\neq_\R}
\newcommand\ltR{<_\R}
\newcommand\gtR{>_\R}
\newcommand\leR{\le_\R}
\newcommand\geR{\ge_\R}

We can now defined order and equality on the reals.
It is noteworthy that, while we define order and equality in terms of \(\op\),  we never use the definition of \(\op\) itself in proving their properties. 
We only need the properties of \(\op\) we proved in \cref{thm:op_properties}, 
thus we could proceed in the same way even if we had defined \(\op\) differently, as long as \cref{thm:op_properties} holds. 


They are defined as follows:
\begin{align*}
  \lrea \ltR \lreb &\equiv \qexists\lpa \op(\lrea, \lreb, \lpa), \\
  \lrea \leR \lreb &\equiv \qforall\lpa \lnot \op(\lreb, \lrea, \lpa), \\ 
  \lrea \neqR \lreb &\equiv \qexists\lpa \op(\lrea, \lreb, \lpa) \lor \op(\lreb, \lrea, \lpa), \\
  \lrea \eqR \lreb &\equiv \qforall\lpa \lnot \op(\lrea, \lreb, \lpa) \land \lnot \op(\lreb, \lrea, \lpa). 
\end{align*}
Note that \(\ltR\) and \(\neqR\) are \(\Sigma^0_1\) formulas and \(\leR\) and \(\eqR\) are \(\Pi^0_1\) formulas. 
Moreover \(\leR\) and \(\eqR\) are the dual formulas of \(\ltR\) and \(\neqR\) respectively, as defined in \cref{sec:em}. 

In order to prove the \gref{thm:least_element}, which is needed in the proof of the \gref{thm:bounding_angle}, 
we need to show some of the properties of the order \(\leR\). 
\begin{lemma}[Reflexivity, Semi-Transitivity and Totality of \(\leR\)]
  The following properties hold:
  \begin{align}
    \tag{reflexivity}\label[property]{prop:Rrefl} 
    \lrea \leR \lrea \\ 
    \tag{semi-transitivity}\label[property]{prop:Rtrans} 
    \lrea \ltR \lreb \land \lreb \leR \lrec \limply \lrea \leR \lrec, \\
    \tag{totality}\label[property]{prop:Rtot}
    \lrea \leR \lreb \lor \lreb \ltR \lrea. 
  \end{align}
\end{lemma}
\begin{proof}
  The first two properties follows from the corresponding properties of \(\op\). 
  The last is a classical tautology. 
  \begin{itemize}
    \item
      In order to prove reflexivity we have to show that:
      \[
        \lrea \leR \lrea \equiv \qforall\lpa \lnot \op(\lrea, \lrea, \lpa). 
      \]
      This follows by the \RuleNameOP{irrefl} rule: 
      \[
        \PrAss{\op(\lrea, \lrea, \lpa)}\riA
        \PrLbl{\RuleNameOP{irrefl}}
        \PrUn\lfalse
        \PrImplyI[\riA]{\lnot \op(\lrea, \lrea, \lpa)}
        \PrForallI{\qforall\lpa \lnot \op(\lrea, \lrea, \lpa)}
        \DisplayProof
      \]
    \item
      In order to prove this transitive property for mixed \(\ltR\) and \(\leR\) we have to show that:
      \begin{align*}
        \lrea \leR \lrec &\equiv \qforall\lpa \lnot \op(\lrec, \lrea, \lpa), \\ 
        \intertext{assuming that:}
        \lrea \ltR \lreb &\equiv \qexists\lpa \op(\lrea, \lreb, \lpa), \\
        \lreb \leR \lrec &\equiv \qforall\lpa \lnot \op(\lrec, \lreb, \lpa).
      \end{align*}
      This follows by means of the \RuleNameOP{trans} rule: 
      \[
        \PrAx{\qexists\lpa \op(\lrea, \lreb, \lpa)}
        \PrAx{\qforall\lpa \lnot \op(\lrec, \lreb, \lpa)}
        \PrForallE{\lnot \op(\lrec, \lreb, \max(\lpa,\lpb))}
        \PrAss{\op(\lrec, \lrea, \lpa)}{1}
        \PrAss{\op(\lrea, \lreb, \bar\lpa)}{2}
        \PrLbl{\RuleNameOP{trans}}
        \PrBin{\op(\lrec, \lreb, \max(\lpa,\lpb))}
        \PrLbl{\RuleNameOP{asym}}
        \PrImplyE\lfalse
        \PrExistsE[2]\lfalse
        \PrImplyI[1]{\lnot \op(\lrec, \lrea, \lpa)}
        \PrForallI{\qforall\lpa \lnot \op(\lrec, \lrea, \lpa)}
        \DisplayProof
      \]
    \item
      We have to show that:
      \[ \lrea \leR \lreb \lor \lreb \ltR \lrea \equiv
        \qforall\lpa \lnot \op(\lreb, \lrea, \lpa) \lor \qexists\lpa \op(\lrea, \lreb, \lpa),
      \]
      which is an instance of \(\EM\) when \(\lrea\) and \(\lreb\) denote recursive real numbers. \qedhere
  \end{itemize} 
\end{proof}
The proof is constructive apart from the last point, where we show that totality is actually an instance of \(\EM\). 
Note that only the reflexivity property is stated in the standard way, while transitivity and totality are written in non-standard forms.  
We chose these forms for two reasons: they are easier to prove and they are the exact form we need in the proof of the \gref{thm:least_element}. \note{i could not find an constructive proof of the standard transitivity}

\subsection{Variables for Real Numbers}
Until now we have used \(\lrea, \lreb\) and \(\lrec\) as metavariables for real numbers in an informal way. 
However, since we are working in the first-order language of arithmetic, our variables range only on natural numbers and not on functions. 
For our example we only need to address a finite but arbitrary number of real numbers, that is, we only need a countable quantity of them. 
Thus we can assume that we have a countable set of function symbols indexed by the natural numbers. 
These function symbols represent sequences of rational numbers satisfying some notion of convergence. 

In the case we are considering, where real numbers are represented with sequences of nested intervals satisfying the convergence condition \eqref{eq:nested_cauchy}, we proceed as follows.  
We assume that we have two sequences of indexed function symbols:
\[ 
  \lfa^+_0, \dotsc, \lfa^+_n, \dotsc 
  \quad \text{and} \quad 
  \lfa^-_0, \dotsc, \lfa^-_n, \dotsc 
\] 
such that, for any index \(n\), \(\lfa^+_n\) and \(\lfa^-_n\) satisfy the convergence condition \eqref{eq:nested_cauchy}.
Then we can formally define the order predicate as:
\[
  \op(\ia, \ib, \lpa) \equiv \lfa^+_\ia(\lpa) \ltQ \lfa^-_\ib(\lpa), 
\]
where \(\ia\) and \(\ib\) are metavariables for arithmetic terms. 
Thus, when we write \( \ia < \ib \), we mean that \(\ia\) is smaller than \(\lreb\) as indexes, that is, as natural numbers; 
on the other hand, when we write \( \ia \ltR \ib \), we mean that the real number indexed by \(\ia\) is smaller than the one indexed by \(\ib\). 

However this notation, while formally correct, is hard to read:
\( \ia \le \ib \) and \( \ia \leR \ib \) look confusingly similar while having unrelated meaning. 
In order to avoid confusion and hurting the eyes of mathematicians, we sugar coat our syntax. 
We write \(\lrea_\ia\) instead of \(\ia\) when thinking of \(\ia\) as a real number. 
For instance we write \(\lrea_\ia \leR \lrea_\ib \) instead of \(\ia \leR \ib\). 
The last one is a much more intuitive than the unsugared version.

\subsection{The Least Element Lemma}

Now we can reason about finite sets of real numbers as sets of indexes. 
In the next lemma, we shall work with the sets of real numbers indexed by initial segments of the natural numbers. 
We show the existence of a least element in each of these sets. 
The least element is actually a minimum, that is, the unique least element of the set. 
However, in order to prove the \gref{thm:bounding_angle}
we do not need to show its uniqueness, just its existence. 
\begin{lemma}[Least Element] \label{thm:least_element}
  For any \(\na\), the real numbers \(\lrea_0, \dotsc, \lrea_\na\) have a least element with respect to \(\leR\). 
  More precisely: 
  \footnote{
    We use the standard compact notation for bounded quantifications: 
    \begin{gather*}
      \qforall{\ib \le \na} \fa \text{ stands for } \qforall\ib \ib \le \na \limply \fa, \\ 
      \qexists{\ib \le \na} \fa \text{ stands for } \qexists\ib \ib \le \na \land \fa. 
    \end{gather*}
  }
  \[
    \qforall\na \qexists{\ia \le \na} \qforall{\ib \le \na} \lrea_\ia \leR \lrea_\ib. 
  \]
\end{lemma}
\begin{proof}
  We proceed by induction on \(\na\). 
  \begin{description}
    \item[Zero case] 
      In the base case \(\na = 0\) and we have to prove that: 
      \[
        \qexists{\ia \le 0} \qforall{\ib \le 0} \lrea_\ia \leR \lrea_\ib. 
      \]
      Both \(\ia\) and \(\ib\) can only be \(0\); 
      thus we just have to check the condition \( \lrea_0 \leR \lrea_0 \), 
      which holds by reflexivity of \(\leR\). 
    \item[Successor case]
      In the inductive case we have to prove that: 
      \[
        \qexists{\ia \le \na+1} \qforall{\ib \le \na+1} \lrea_\ia \leR \lrea_\ib, 
      \]
      from the inductive hypothesis:
      \[
        \qexists{\ia \le \na} \qforall{\ib \le \na} \lrea_\ia \leR \lrea_\ib. 
      \]
      By the inductive hypothesis, let \( \bar\ia \le \na \) be the index of the least element in \(\lrea_0, \dotsc, \lrea_\na\). 
      By totality of \(\leR\) we have two cases.
      \begin{description}
        \item[\( \lrea_{\bar\ia} \leR \lrea_{\na+1} \)]
          Then \(\bar\ia\) is the index of a least element in \( \lrea_0, \dotsc, \lrea_{\na+1} \),
          since \( \lrea_{\bar\ia} \leR \lrea_\ib \) when \(\ib = \na+1\) (since we are considering this case) and when \(\ib \le \na\) by inductive hypothesis. 
        \item[\( \lrea_{\na+1} \ltR \lrea_{\bar\ia} \)]
          Then \(\na+1\) is the index of a least element in \(\lrea_0, \dotsc, \lrea_{\na+1} \), 
          since \( \lrea_{\na+1} \leR \lrea_\ib\) when \(\ib = \na+1\) by reflexivity of \(\leR\) and when \(\ib \le \na\) by transitivity of \(\ltR\) and \(\leR\), since: 
          \[
            \lrea_{\na+1} \ltR \lrea_{\bar\ia} \leR \lrea_\ib
          \]
          by inductive hypothesis. \qedhere
      \end{description}
  \end{description}
\end{proof}

\begin{omitted}
  \begin{landscape}
    \small
    \begin{gather*} 
      \PrAx{\eqref{zero}}
      \PrUn{\qexists{\lrea \leq 0} \qforall{\lreb \leq 0} \lrea \leR \lreb}
      \PrAss{\qexists{\lrea \leq \na} \qforall{\lreb \leq \na} \lrea \leR \lreb}{1}
      \PrAx{\lrea' \leR \lreb \lor \neg \lrea' \leR \lreb}
      \PrAx{\eqref{oldmin}}
      \PrUn{\qexists{\lrea \leq \na+1} \qforall{\lreb \leq \na+1} \lrea \leR \lreb}
      \PrAx{\eqref{newmin}}
      \PrUn{\qexists{\lrea \leq \na+1} \qforall{\lreb \leq \na+1} \lrea \leR \lreb}
      \PrLbl[3]{\RuleNameE\lor}
      \PrTri{\qexists{\lrea \leq \na+1} \qforall{\lreb \leq \na+1} \lrea \leR \lreb}
      \PrLbl[2]{\RuleNameE\exists} 
      \PrBin{\qexists{\lrea \leq \na+1} \qforall{\lreb \leq \na+1} \lrea \leR \lreb}
      \PrLbl[1]{Ind} 
      \PrBin{\qforall\na \qexists{\lrea \leq \na} \qforall{\lreb \leq \na} \lrea \leR \lreb}
      \DisplayProof \\[5mm] 
      \label{zero} \tag{zero}
      \PrAx{0 \leq 0} 
      \PrAx{0 \leR 0}
      \PrAss{\lreb \leq 0}{4}
      \PrUn{\lreb = 0}
      \PrBin{0 \leR \lreb}
      \PrLbl[4]{\RuleNameI\limply}
      \PrUn{\lreb \leq 0 \limply 0 \leR \lreb}
      \PrLbl{\RuleNameI\forall}
      \PrUn{\qforall{\lreb \leq 0} 0 \leR \lreb}
      \PrLbl{\RuleNameI\land}
      \PrBin{0 \leq 0 \land (\qforall{\lreb \leq 0} 0 \leR \lreb)}
      \PrLbl{\RuleNameI\exists}
      \PrUn{\qexists{\lrea \leq 0} \qforall{\lreb \leq 0} \lrea \leR \lreb}
      \DisplayProof \\[5mm] 
      \label{oldmin} \tag{oldmin}
      \PrAss{\lrea' \leq \na-1 \land (\qforall{\lreb \leq \na-1} \lrea' \leR \lreb)}{2}
      \PrLbl{\RuleNameE\land}
      \PrUn{\lrea' \leq \na-1}
      \PrUn{\lrea' \leq \na}
      \PrAx{\lreb \leq \na-1 \lor \lreb = \na}
      \PrAss{\lrea' \leq \na \land (\qforall{\lreb \leq \na - 1} \lrea' \leR \lreb)}{2}
      \PrLbl{\RuleNameE\land}
      \PrUn{\qforall{\lreb \leq \na-1} \lrea' \leR \lreb}
      \PrLbl{\RuleNameE\forall}
      \PrUn{\lreb \leq \na-1 \limply \lrea' \leR \lreb}
      \PrAss{\lreb \leq \na-1}{5}
      \PrLbl{\RuleNameE\limply}
      \PrBin{\lrea' \leR \lreb}
      \PrAss{\lreb = \na}{5}
      \PrAss{i' \leR \na}{3}
      \PrBin{\lrea' \leR \lreb}
      \PrLbl[5]{\RuleNameE\lor}
      \PrTri{\lrea' \leR \lreb}
      \PrLbl[4]{\RuleNameI\limply}
      \PrUn{\lreb \leq \na \limply \lrea' \leR \lreb}
      \PrLbl{\RuleNameI\forall}
      \PrUn{\qforall{\lreb \leq \na} \lrea' \leR \lreb}
      \PrLbl{\RuleNameI\land}
      \PrBin{\lrea' \leq \na \land (\qforall{\lreb \leq \na} \lrea' \leR \lreb)}
      \PrLbl{\RuleNameI\exists}
      \PrUn{\qexists{\lrea \leq \na} \qforall{\lreb \leq \na} \lrea \leR \lreb}
      \DisplayProof \\[5mm]
      \label{newmin} \tag{newmin}
      \PrAx{\na \leq \na \hspace{-5cm}}
      \PrAx{\lreb \leq \na-1 \lor \lreb = \na}
      \PrAss{\na \leR \lrea'}{3}
      \PrUn{\na \leR \lrea'}
      \PrAss{\lrea' \leq \na \land (\qforall{\lreb \leq \na - 1} \lrea' \leR \lreb)}{2}
      \PrLbl{\RuleNameE\land}
      \PrUn{\qforall{\lreb \leq \na-1} \lrea' \leR \lreb)}
      \PrAss{\lreb \leq \na-1}{5}
      \PrLbl{\RuleNameE\limply}
      \PrBin{\lrea' \leR \lreb}
      \PrBin{\na \leR \lreb}
      \PrAss{\lreb = \na}{5}
      \PrAx{\na \leR \na}
      \PrBin{\na \leR \lreb}
      \PrLbl[5]{\RuleNameE\lor}
      \PrTri{\na \leR \lreb}
      \PrLbl[4]{\RuleNameI\limply}
      \PrUn{\lreb \leq \na \limply \na \leR \lreb}
      \PrUn{\qforall{\lreb \leq \na} \na \leR \lreb}
      \PrLbl{\RuleNameI\land}
      \PrBin{\na \leq \na \land (\qforall{\lreb \leq \na} \na \leR \lreb)}
      \PrLbl{\RuleNameI\exists}
      \PrUn{\qexists{\lrea \leq \na} \qforall{\lreb \leq \na} \lrea \leR \lreb}
      \DisplayProof 
    \end{gather*}
  \end{landscape}

  \newcommand\name[1]{{\texttt{\bf #1 }}}
  \newcommand\dname[1]{\texttt{ #1 }}
  \begin{align*}
    \dname{min} 0 &= \lrea_0, \\ 
    \dname{min} \na+1 &= \name{if} \dname{min} \na \leR \lrea_{\na+1} \\
                         &\qquad \name{then} \dname{min} \na, \\
                         &\qquad \name{else} \lrea_{\na+1}. 
  \end{align*}

\end{omitted}

\newcommand\prog[1]{``\lstinline"#1"''}
The proof looks constructive: its computational interpretation is the usual algorithm that finds the least element in a vector, by a simple recursion or by looping on its elements.
We can write it as a recursive function \prog{rmin} in Haskell: 
\begin{lstlisting}[caption=The Least Element Program,label=prog:min]
rmin 0   = 0
rmin n   = if rle (rmin (n-1)) n
  then rmin (n-1)
  else n
\end{lstlisting}
where \prog{rle} is a boolean function that stands for \(\leR\), that is, it compares the reals indexed by its arguments. 
The problem is that this is not a good program, because we are unable to write \prog{rle} as a terminating program.
The closest approximation would be the following unfounded recursion:
\begin{lstlisting}[caption=The Lesser or Equal Program,label=prog:rle]
rle i j        = rle_urec 0 i j
rle_urec k i j = if op j i k  
  then False
  else rle_urec (k+1) i j
\end{lstlisting}
where \prog{op} is a total boolean function that stands for the order predicate \(\op\). 
We can assume that  \prog{op} terminates for any input since \(\op\) is decidable. 
The problem is that \(\leR\) is total only classically. 
More precisely, totality is an instance of \(\EM\) because \(\leR\) is a \(\Pi^0_1\) formula and thus negatively decidable.
This can bee seen concretely in the program for \prog{rle}: 
\prog{rle i j} only halts (returning \prog{False}) if \prog{op j i k} is true for some \(k\), that is, if and only if \(\lrea_\ia \leq \lrea_\ib\) is false. 
On the other hand, when \(\lrea_\ia \leq \lrea_\ib\) is true there is no such \(k\) and the evaluation of \prog{rle i j} will never halt. 
Note that \prog{True} does not even occur in the program, 
so its is clear that \prog{rle i j} never returns \prog{True}. 
This is the general behavior of an algorithm that computes a negatively decidable predicate: when the predicate is false it halt with the correct answer and when the predicate is true it does not halt. 

For positively decidable predicates we have the dual behavior. 
For instance, in the case of \(\ltR\) which is defined by a \(\Sigma^0_1\) formula and thus positively decidable, the decision procedure can be written as: 
\begin{lstlisting}[caption=The Lesser Than Program,label=prog:rlt] 
rlt i j        = rlt_urec 0 i j
rlt_urec k i j = if op i j k  
  then True
  else rlt_urec (k+1) i j
\end{lstlisting}
The program is very similar to the previous one, the only noteworthy changes are the order of the argument given to \prog{op} and the fact that the only possible return value is \prog{True} instead of \prog{False}. 
It only halts (returning \prog{True}) if \prog{op i j k} is true for some \(k\), that is, when \(\lrea_\ia \leq \lrea_\ib\) is true. 

\begin{remark} \label{rmk:program_shorter_than_proof}
  Note how \cref{prog:min} is much shorter than the proof of the \gref{thm:least_element}. 
  This difference would be even bigger if the proof was written in a completely formal language, for instance in a proof assistant. 
  The reason for this discrepancy is that the proof contains both the algorithm written in \cref{prog:min} and the evidence for the correctness of the algorithm. 
  This last part is missing from \cref{prog:min}, thus explaining the difference in length. 
\end{remark}

\section{The Interactive Interpretation of the Least Element Lemma} 

We have seen why the naive way of extracting a program from proofs fails in the case of the \gref{thm:least_element}. 
Now we give the interactive interpretation of the \gref{thm:least_element}. 
\begin{changed}
  Since we are working in \(\HA+\EM\), any proof can be thought of as a constructive proof with open assumptions, that are the instances of \(\EM\) that are used in the proof.
  The interactive realizability interpretation follows the standard BHK interpretation for the constructive parts, so we will concentrate on the interpretation of the \(\EM\) instances. 
\end{changed}

The only instances of \(\EM\) in the proof are those used to deduce the totality property:
\begin{equation} \label{eq:em_tot_instance} 
  \lrea_\ia \leR \lrea_\ib \lor \lrea_\ib \ltR \lrea_\ia. 
\end{equation}
The left disjunct, which we call the \emph{universal disjunct}, is \(\Pi^0_1\) and negatively decidable, while the right one, the \emph{existential disjunct} is \(\Sigma^0_1\) and positively decidable. 
Moreover universal disjunct and negation of the existential disjunct are classically equivalent. 
We say that a formula is \emph{concrete} when it is closed and all its arithmetical terms are normal. 

In order to motivate the interactive realizability interpretation we show why a naive attempt to give a computational behavior an \(\EM\) instance fails. 
This can be seen concretely, by recalling the semi-effective procedures \cref{prog:rle,prog:rlt} that decide the disjuncts of the instances of \(\EM\) representing the totality, but the general argument is the same. 
The universal disjunct is negatively decidable, that is, its deciding program halts if and only if it is false; the existential disjunct is positively decidable, that is, its deciding program halts if and only if it is true. 
What happens if we run these two decision programs in parallel? 
Can one give the answer whenever the other fails?
In recursion theory, this method is used for instance to prove that if the complement of a recursively enumerable set is recursively enumerable then the set is recursive. 
Unfortunately this does not work in our case. We have two scenarios:
\begin{itemize}
  \item 
    If the existential disjunct is true, its decision procedure halts and returns true. 
    Since in this case the universal disjunct is false, its decision procedure also halts and returns false. 
  \item 
    If the universal disjunct is true, its decision procedure does not halt. 
    Since in this case the existential disjunct is false, its decision procedure does not halt either. 
\end{itemize}
The problem lies in the fact that the disjuncts are dual and 
their decision procedures describe basically the same algorithm with minor variations. In particular they halt or fail to halt on the same inputs. 
This is evident when considering the programs given in \cref{prog:rle,prog:rlt}. 

Interactive realizability proposes a way to side-step the problem evidenced above. 
This is possible since it is not true that the computational interpretation of a proof using instances of \(\EM\) necessarily needs to decide these instances. 
Consider the case of totality of the order on the real numbers. 
The universal disjunct is: 
\[ 
  \lrea_\ia \leR \lrea_\ib \equiv \qforall\lpa \lnot \op(\lrea_\ib, \lrea_\ia, \lpa). 
\] 
Being an universally quantified statement, it proves infinite instances \(\lnot \op(\lrea_\ib, \lrea_\ia, \lpa)\), one for each natural number \(\lpa\). 
A proof that uses totality may need all this infinite information or (for example, when proving a simply existential statement) may only need a finite quantity of these instances. 
In the second case, we can avoid the problem of effectively deciding the \(\EM\) instance. 
We only need to decide those instances that are actually used in the proof. 
This is possible, since each instance is decidable (being a quantifier free formula) and we assumed there is a finite quantity of them. 
Interactive realizability takes advantage of this fact and gives a procedure to determine which instances of the universal disjunct are needed and to iteratively decide them. 

The interactive interpretation is a ``relaxation'' of the BHK interpretation. 
In the BHK interpretation the decision of a disjunction effectively selects a true disjunct, in the interactive case instead of a decision we have a sort of ``educated guess''. 
Therefore, while \(\EM\) cannot be realized by the BHK interpretation since there is no effective procedure to decide it, the interactive interpretation can because it yields a weaker semantics, which produces a sure result only when the goal is simply existential. 

Interactive realizability revolves around the concept of \emph{knowledge state}. 
A knowledge state, or simply state, is a finite object that stores information about the \(\EM\) instances we use in the proof. 
The purpose of this information is help us decide the \(\EM\) instances, that is, help us in choosing which disjunct holds. 
Moreover, whenever the state chooses the existential disjunct, it should also produce a witness, like in the BHK interpretation. 

We can represent a state as a finite partial function\footnote{By finite partial function, we mean a partial function whose domain (the set of elements where it is defined) is finite.} that maps a concrete instance of \(\EM\) into a witness of its existential disjunct. 
Such a function decides or guesses a concrete instance \(\fa\) of \(\EM\): if it is undefined on \(\fa\), then we choose the universal disjunct; if it is defined we chose the existential disjunct with the returned witness. 
We are only interested to the instances appearing in the proof, namely, those of the form \eqref{eq:em_tot_instance} when \(\ia, \ib\) are numerals. 
Thus an instance is determined by two natural numbers; since witnesses are natural numbers too, a state can be concretely defined as a finite partial function from \(\N \times \N\) to \(\N\). 

For instance, consider the case of the \(\EM\) instances used in the proof of the \gref{thm:least_element}. 
When we have to decide \eqref{eq:em_tot_instance}, we check the state on the pair \((\ia,\ib)\). 
At first, let us assume that the state is undefined on \((\ia,\ib)\). 
This means we have no knowledge about the universal disjunct \(\lrea_\ia \leR \lrea_\ib\).
Since we cannot effectively check that the universal disjunct holds, 
we make an educated guess and assume that \(\lrea_\ia \leR \lrea_\ib\) is true. 
Clearly this assumption could very well be wrong, which may or may not become apparent later in the proof. 
Keeping track of this assumption, we carry on with the proof. 
Every time we use this assumption to prove a decidable instance of its we check if the instance holds. 
More concretely, if later in the proof we use the assumption \(\lrea_\ia \leR \lrea_\ib\) to deduce that \(\lnot \op (\ib, \ia, \lpa)\) for some \(\lpa\), we check that \(\lnot \op (\ib, \ia, \lpa)\) holds. 
If this is the case, we carry on with the proof: \(\lrea_\ia \leR \lrea_\ib\) could still be false, but at least the particular instance we are using is true. 
If this is not the case, we have found a counterexample to the assumption \(\lrea_\ia \leR \lrea_\ib\): 
being negatively decidable, the counterexample is enough to effectively decide that it is false. 
Therefore we stop following the proof because we have chosen the wrong disjunct in the \(\EM\) instance \eqref{eq:em_tot_instance}. 

Moreover, a counterexample to \(\lrea_\ia \leR \lrea_\ib\) is a natural number \(\lpa\) such that \(\op (\ib, \ia, \lpa)\). 
Therefore \(\lpa\) is a witness for the existential disjunct \(\lrea_\ib \ltR \lrea_\ia\). 
We can use this new knowledge to add \((\ia,\ib)\) to the domain of the state with value \(\lpa\). Remember that we assumed the state to be undefined on \((\ia,\ib)\), which is why we assumed the universal disjunct to be true in the first place. 

At this point, we forget what we did after guessing (wrongly) that the universal disjunct was true and start again, 
More precisely, we need to backtrack to a computation state \emph{before} we decided the \(\EM\) instance in question and repeat our decision with the extended state. 
Since the extended state is defined on \((\ia,\ib)\) and yields \(\lpa\), this time we decide the \(\EM\) instance differently: we choose the existential disjunct \(\lrea_\ib \ltR \lrea_\ia\) with \(\lpa\) as witness. 
Now we are sure that our choice is the correct one and not a guess, since we have effectively decided that the existential disjunct holds (we can since it is positively decidable). 

The exact point we need to backtrack to is not relevant, as long as it is before the decision of the \(\EM\) instance. 
A simple choice would be the very beginning, in which case we do not need to keep track of where we decided the \(\EM\) instance. 
A more efficient choice is right before the decision point, so that we do not need to repeat the computations before it, which do not change. 

In order for the interactive interpretation to produce correct results,
we need to assume that the state is sound, that is, when it is defined, the witness it yields is actually a witness. 
More formally, a state \(s\) is sound if, for any pair \((\ia,\ib)\), we have that \(\op(\ib, \ia, s (\ia,\ib))\) holds. 
This assumption is not problematic: the empty state, namely the state that is always undefined, satisfies it vacuously. 
Moreover, note that in the interactive interpretation we outlined above, we only extend a state with an actual witness. 
In other words, the extension preserves the soundness property. 
\fixme{remark that we work in HA+EM, thus we already know how to interpret everything except EM}

To summarize, the general procedure is the following: 
\begin{enumerate} 
  \item we start from any sound state (usually the empty state),
  \item\label{step2} we follow the proof choosing any \(\EM\) instance according to the state,
  \item if we discover that we wrongly assumed the universal disjunct of an \(\EM\) instance: 
    \begin{enumerate}
      \item we extend the state with the counterexample we found,
      \item we backtrack to a point before the \(\EM\) instance we guessed wrong,
      \item we proceed as in step \ref{step2},
    \end{enumerate}
  \item if we never discover that we wrongly assumed an universal disjunct we carry on until the end of the proof and we are done. 
\end{enumerate}

Interactive realizability can be thought as a ``smart'', albeit ``partial'', decision algorithm for negatively decidable statements. 
This can be seen comparing it with the naive algorithm given in \cref{prog:rle}. 
It is partial because a real decision is impossible, so it only considers a finite number of instances, unlike the unbounded recursion employed by \cref{prog:rle}. 
It is smart because it does not perform a blind search, trying in order all the natural numbers. 
Instead it uses the proof itself to find the counterexamples. 
There is a reasonable expectation that the ideas underlying the proof provide a more focused way of selecting counterexamples than a blind search (this of course depends on the proof itself). 

Until now we considered a single instance of the \(\EM\) axiom, but little changes if there is more than one. We will return to this point later. 
In the proof of the \gref{thm:least_element}, one instance of \(\EM\) is used for each inductive step in the proof. 
When we interpret the proof with the empty state, for each of these instances we assume that the universal disjunct holds. 
Therefore the proof is interpreted as follows. 
In the base step we choose \(\lrea_0\). 
In the first inductive step, we have to decide the \(\EM\) instance: 
\[
  \lrea_0 \leR \lrea_1 \lor \lrea_1 \ltR \lrea_0. 
\] 
Since the state is empty, we assume that \(\lrea_0 \leR \lrea_1\). 
Thus we keep \(\lrea_0\) as the least element of \(\lrea_0, \lrea_1\). 
In the second inductive step, we have to decide the \(\EM\) instance: 
\[
  \lrea_0 \leR \lrea_2 \lor \lrea_2 \ltR \lrea_0. 
\] 
Since the state is empty, we again assume that \(\lrea_0 \leR \lrea_2\). 
Thus we keep \(\lrea_0\) as the least element of \(\lrea_0, \lrea_1, \lrea_2\). 
At the end of the proof, we have assumed the following universal disjuncts:
\begin{equation}\label{eq:univ_instances}
  \lrea_0 \leR \lrea_1, \lrea_0 \leR \lrea_2, \dotsc, \lrea_0 \leR \lrea_\na. 
\end{equation}
Under these assumptions, we have found that the least element is \(\lrea_0\). 
Rather disappointing, isn't it?

The reason for this is that the universal disjuncts \(\lrea_\ia \leR \lrea_\ib\) are never instanced, so we have neither opportunity or reason to falsify one of them. 
However this may change if the \gref{thm:least_element} is used inside a bigger proof. 
This will happen later in the proof of the \gref{thm:bounding_angle}. 
In this case the outer proof might instance these assumptions and discover them wrong, in which case we have to backtrack to the proof of the \gref{thm:least_element}. 

Let us see how the \gref{thm:least_element} behaves when its conclusion is used to deduce decidable instances. 
Assume that \(\na = 5\). 
If the state is empty, then the \gref{thm:least_element} tells us that \(\lrea_0\) is a least element. 
This means that \(\lrea_0 \leR \lrea_\ia\) for any \(\ia\).
Imagine that we use the \gref{thm:least_element} in a bigger proof to prove that \(\lrea_0 \leR \lrea_3\). 
This is one of the \(\EM\) instances we assumed in \eqref{eq:univ_instances}. 
Moreover, imagine that, using this assumption, we discover that \(\lrea_0 \leR \lrea_3\) does not hold at precision \(33\). 
Then we have to extend the domain of the state to \((0,3)\) with value \(33\). 
At this point we backtrack, say at the beginning of the proof of the \gref{thm:least_element}. 

We again start from \(\lrea_0\) and proceed like before. 
The first and second inductive steps again select \(\lrea_0\) as the least element, assuming that \(\lrea_0 \leR \lrea_1\) and \(\lrea_0 \leR \lrea_2\). 
Things change at the third inductive step when we have to decide \(\lrea_0 \leR \lrea_3 \lor \lrea_3 \ltR \lrea_0\). 
Since know the state has a relevant witness, this time we choose the existential disjunct with witness \(33\), thus selecting \(\lrea_3\) as the new least element. 
In the next inductive steps we again assume the universal disjuncts \(\lrea_3 \leR \lrea_4\) and \(\lrea_3 \leR \lrea_5\), since the state has no information on them. 
Thus our the least element is \(\lrea_3\). 
A summary of our decisions is represented in \cref{fig:least3}.
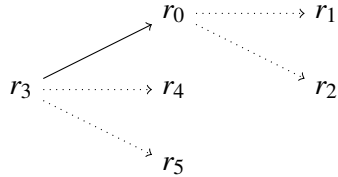
\begin{figure}[!ht]
  \caption{A graph representing the result of the least element computation. Full arrows represent information provided by the state, dotted arrows ``guessed'' information the state knows nothing about.} \label{fig:least3}
  \begin{center}
    \begin{tikzpicture}
      \node (r3) at (0,0) {\(\lrea_3\)};
      \node (r0) at (2,1) {\(\lrea_0\)};
      \node (r4) at (2,0) {\(\lrea_4\)};
      \node (r5) at (2,-1) {\(\lrea_5\)};
      \node (r1) at (4,1) {\(\lrea_1\)};
      \node (r2) at (4,0) {\(\lrea_2\)};
      \draw [->] (r3) -- (r0);
      \draw [dotted,->] (r3) -- (r4);
      \draw [dotted,->] (r3) -- (r5);
      \draw [dotted,->] (r0) -- (r1);
      \draw [dotted,->] (r0) -- (r2);
    \end{tikzpicture}
  \end{center}
\end{figure}
Imagine that we were to discover a counterexample to \(\lrea_3 \leR \lrea_2\), say at precision \(25\). 
This statement is not one of the universal disjuncts that we assumed. By looking at the proof or at \cref{fig:least3}, we can see that it has been deduced by the semi-transitivity property from \(\lrea_3 \ltR \lrea_0\) and \(\lrea_0 \leR \lrea_2\). 
The first is the existential disjunct for which we found a witness, so we are sure that it holds. 
Thus the the wrong assumption is \(\lrea_0 \leR \lrea_2\). 
By checking the proof of semi-transitivity we can see that the counterexample for \(\lrea_0 \leR \lrea_2\) is \(\max(25,33)\), thus \(33\) again. 
We extend the state accordingly and repeat the least element computation, which results in new least element \(\lrea_2\).
In \cref{fig:interactive_least_element} we summarize the iterations we saw until now and add some more, as an example. 
\newcommand\stack[2][c]{{
    \renewcommand\arraystretch{0.6}
    \begin{array}{#1} 
      #2 
    \end{array}
}}
\begin{figure}[!ht]
  \caption{An example of evaluations of the interactive interpretation of the \gref{thm:least_element} with state extensions.} \label{fig:interactive_least_element}
  \[
    \begin{array}{|c|c|c|c|c|c|}
      \mathbf{Iter} &
      \mathbf{State} &
      \mathbf{Least\ element} & 
      \mathbf{Used} & 
      \stack{\mathbf{Deduced} \\ \mathbf{from}} & 
      \mathbf{Discovered} \\ \hline
      1st & & 
      \begin{tikzpicture}[scale=0.7,baseline=0]
        \node (r0) at (0,0) {\(\lrea_0\)};
        \node (r1) at (2,0) {\(\lrea_{1,\dotsc,5}\)};
        \draw [dotted] (r0) -- (r1);
      \end{tikzpicture}
      & 
      \lrea_0 \leR \lrea_3 & \lrea_0 \leR \lrea_3 & \lrea_3 \ltR \lrea_0 
      \\ \hline
      2nd & \lrea_3 \ltR \lrea_0
      &
      \begin{tikzpicture}[scale=0.7,baseline=-0.5]
        \node (r3) at (0,0) {\(\lrea_3\)};
        \node (r0) at (2,0.5) {\(\lrea_0\)};
        \node (r1) at (4,0.5) {\(\lrea_{1,2}\)};
        \node (r4) at (2,-0.5) {\(\lrea_{4,5}\)};
        \draw (r3) -- (r0);
        \draw [dotted] (r0) -- (r1);
        \draw [dotted] (r3) -- (r4);
      \end{tikzpicture}
      & \lrea_3 \leR \lrea_2 
      & \stack{\lrea_3 \ltR \lrea_0 \\ \lrea_0 \leR \lrea_2} 
      & \lrea_2 \ltR \lrea_0 
      \\ \hline
      3rd & \stack{\lrea_2 \ltR \lrea_0 \\ \lrea_3 \ltR \lrea_0} 
          &\begin{tikzpicture}[scale=0.7,baseline=-0.5]
      \node (r2) at (0,0) {\(\lrea_2\)};
      \node (r0) at (2,0.5) {\(\lrea_0\)};
      \node (r1) at (4,0.5) {\(\lrea_1\)};
      \node (r3) at (2,-0.5) {\(\lrea_{3,4,5}\)};
      \draw (r2) -- (r0);
      \draw [dotted] (r0) -- (r1);
      \draw [dotted] (r2) -- (r3);
    \end{tikzpicture}
    & \lrea_2 \leR \lrea_3 
    & \lrea_2 \leR \lrea_3 
    & \lrea_3 \ltR \lrea_2
    \\ \hline
    4th & \stack{\lrea_2 \ltR \lrea_0 \\ \lrea_3 \ltR \lrea_0 \\ \lrea_3 \ltR \lrea_2} 
        & \begin{tikzpicture}[scale=0.6,baseline=-0.5]
    \node (r3) at (-2,0) {\(\lrea_3\)};
    \node (r2) at (0,0.5) {\(\lrea_2\)};
    \node (r0) at (2,0.5) {\(\lrea_0\)};
    \node (r1) at (4,0.5) {\(\lrea_1\)};
    \node (r4) at (0,-0.5) {\(\lrea_{4,5}\)};
    \draw (r3) -- (r2);
    \draw (r2) -- (r0);
    \draw [dotted] (r0) -- (r1);
    \draw [dotted] (r3) -- (r4);
  \end{tikzpicture}
  & \lrea_3 \leR \lrea_1 
  & \stack{\lrea_3 \ltR \lrea_2 \\ \lrea_2 \ltR \lrea_0 \\ \lrea_0 \leR \lrea_1} 
  & \lrea_1 \ltR \lrea_0
  \\ \hline
  5th & \stack{\lrea_1 \ltR \lrea_0 \\ \lrea_2 \ltR \lrea_0 \\ \lrea_3 \ltR \lrea_0 \\ \lrea_4 \ltR \lrea_2} 
      & 
  \newcommand\safechild[1]{
      node {\(\lrea_{#1}\)}
      edge from parent[solid]
  }
  \newcommand\unsafechild[1]{
      node {\(\lrea_{#1}\)}
      edge from parent[dotted]
  }
  \begin{tikzpicture}[baseline=0,grow'=right,sibling distance=5mm,fill=white]
    \node {\(\lrea_1\)}
    child {
      node {\(\lrea_{0}\)}
      edge from parent[solid]
    }
child {
      node {\(\lrea_{2,3,4,5}\)}
      edge from parent[dotted]
    };
  \end{tikzpicture}
  & \lrea_1 \leR \lrea_4 
  & \lrea_1 \leR \lrea_4 
  & \lrea_4 \ltR \lrea_1 
  \\ \hline
  6th & \stack{\lrea_1 \ltR \lrea_0 \\ \lrea_2 \ltR \lrea_0 \\ \lrea_3 \ltR \lrea_0 \\ \lrea_4 \ltR \lrea_2 \\ \lrea_4 \ltR \lrea_1} 
      & 
  \begin{tikzpicture}[baseline=0,grow'=right,sibling distance=5mm,fill=white]
    \node {\(\lrea_4\)}
    child {
      node {\(\lrea_1\)}
      child {
        node {\(\lrea_0\)}
        edge from parent[solid]
      }
      child {
        node {\(\lrea_{2,3}\)}
        edge from parent[dotted]
      }
      edge from parent[solid]
    }
    child {
      node {\(\lrea_5\)}
      edge from parent[dotted]
    };
  \end{tikzpicture}
  & \dotso
  & \dotso
  & \dotso
  \\ \hline
\end{array}
  \]
  \begin{legend}
      {\bf Iter:} the iteration represented by the current row; 
      {\bf State:} the existential disjuncts witnessed by the state; 
      {\bf Least element:} the least element yielded by the \gref{thm:least_element}; 
      {\bf Used:} a falsifiable consequence of the \gref{thm:least_element} used in the proof; 
      {\bf Deduced from:} what we deduced the falsifiable consequence from; 
      {\bf Discovered:} the existential assumption we found a witness of. 
  \end{legend}
\end{figure}

\begin{changed}
  \subsection{Backtracking, Termination and Complexity}
In the iterations listed in \cref{fig:interactive_least_element}, we compute the following sequence of least element candidates: 
\[ \lrea_0, \lrea_3, \lrea_2, \lrea_3, \lrea_1, \lrea_4. \]  
The fact that \(\lrea_3\) appears two times may cause doubts regarding the termination of the backtracking algorithm. 
The termination of the backtracking algorithms in interactive realizability has been proven in general, see Theorem 2.15 in \cite{aschieriB10}. 

In this particular case we can understand why \(\lrea_3\) is computed two times by taking a closer look at the tree of the possible computations of the least element, 
which is shown in \cref{fig:least_element_computation_tree}. 
For reasons of space, we only show the tree for \(\na = 3\), which is enough to see what happens up to the fifth iteration in \cref{fig:interactive_least_element}.
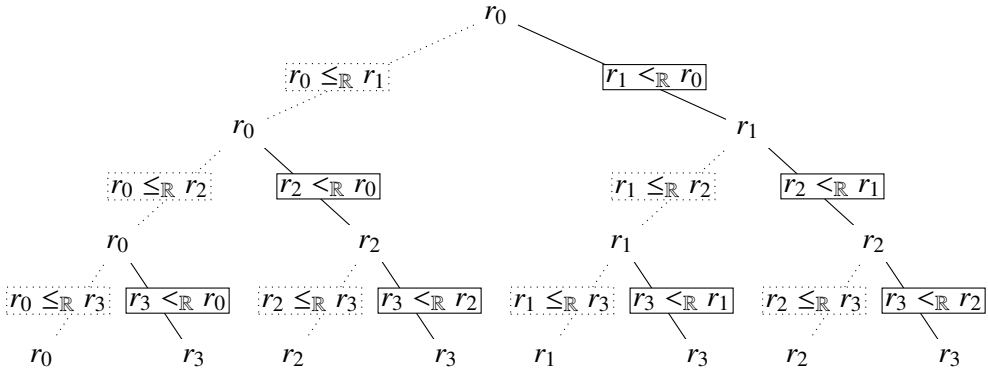
\begin{figure}[!ht]
  \caption{The computation tree of the least element for \(\na = 3\)}. 
  \label{fig:least_element_computation_tree}
  \begin{center}
    \begin{tikzpicture}
      \newcommand\leftchoice[2]{
        node {\(\lrea_{#1}\)} 
        edge from parent[dotted] 
        node[draw,left=-4mm,inner sep=1pt] 
        {\(\lrea_{#1} \leR \lrea_{#2}\)}
      }
      \newcommand\rightchoice[2]{
        node {\(\lrea_{#2}\)} 
        edge from parent[solid] 
        node[draw,right=-4mm,inner sep=1pt] 
        {\(\lrea_{#2} \ltR \lrea_{#1}\)}
      }
      \tikzstyle{level 1}=[sibling distance=6.6cm]
      \tikzstyle{level 2}=[sibling distance=3.3cm]
      \tikzstyle{level 3}=[sibling distance=2cm]
      \tikzstyle{every node}=[fill=white]
      \node {\(\lrea_0\)}
      child {
        child {
          child {
            \leftchoice{0}{3}
          }
          child {
            \rightchoice{0}{3}
          }
          \leftchoice{0}{2}
        }
        child {
          child {
            \leftchoice{2}{3}
          }
          child {
            \rightchoice{2}{3}
          }
          \rightchoice{0}{2}
        }
        \leftchoice{0}{1}
      }
      child {
        child {
          child {
            \leftchoice{1}{3}
          }
          child {
            \rightchoice{1}{3}
          }
          \leftchoice{1}{2}
        }
        child {
          child {
            \leftchoice{2}{3}
          }
          child {
            \rightchoice{2}{3}
          }
          \rightchoice{1}{2}
        }
        \rightchoice{0}{1}
      };
    \end{tikzpicture}
  \end{center}
  \begin{legend}
    Each path represents a possible computation, proceeding from root to leaf, where
    non-leaf nodes are the current least element candidates and the leaf is the final result. 
    Each branching corresponds to an \(\EM\) instance, where the left branch is taken when we guess that the universal disjunct holds for lack of information and the right branch is taken when the state contains the relevant witness. 
  \end{legend}
\end{figure}
We can see that the first five iterations in \cref{fig:interactive_least_element} correspond to the computation paths ending with the first five leaves from the left in \cref{fig:least_element_computation_tree}, in order. 

Moreover, from the computation tree we can see that we never perform the same computation more than once. 
Indeed, assume we have just followed a particular computation path. 
When we backtrack we increment the state adding a witness of one of the \(\EM\) instances we encountered along the path, an instance we did not have a witness for. 
This means that in the next computation, when we arrive at the node corresponding to that \(\EM\) instance, instead of taking the left path as we did previously (since the state did not have a witness for that instance), we take the right path, because this time we do have a witness (since we just extended the state with it). 
Therefore, each time we backtrack, the computation path ends with a leaf that is more to the right in \cref{fig:least_element_computation_tree}. 
This gives a bound to the number of backtrackings, namely \(2^\na - 1\).

This is very different from what one could expect by a superficial look at the proof of the \gref{thm:least_element}. 
Indeed, if we ignore the undecidability of the order on the reals, this simple and very natural proof seems to be quite efficient, since its complexity is linear in \(\na\). 
However, its interactive interpretation has exponential complexity. 
This can be seen in the computation tree too: a single computation correspond to a path and paths have length \(\na\). 
On the other hand, since we have backtracking, in the worst case we may have to perform every possible computation. 
Naturally, the real situation is different since the order on the reals \emph{is} undecidable and thus an actual comparison is impossible. 

Moreover, while in the worst case the interactive interpretation needs a time that is exponential in \(\na\), in general it is hard to estimate the amount of backtracking that will be actually performed, for two different reasons. 
\begin{itemize}
  \item
    The first one is that the actual order of \(\lrea_0, \dotsc, \lrea_\na\) affects heavily the operation of the algorithm. 
Indeed, assume that \(\lrea_0\) is the least element: the interactive interpretation only performs \(\na\) dummy comparisons and immediately returns a least element candidate that, in this case, is the actual least element, so no backtracking can ensue later. 
\item
The second reason is that the backtracking is controlled by how the least element candidate returned by the interactive interpretation is used. 
It is possible for the interactive interpretation to return a candidate that is not a least element, but such that its use in an outer proof is does not cause backtracking. 
In other words, we only need to compute a least element candidate that is good enough instead of the correct one and this can translate to a faster computation, again depending on the situation.
\end{itemize}
In a sense, the second reason explains also how the interactive interpretation is effective even if an certainly correct least element cannot be found effectively.

\subsection{The Whole Proof Is Relevant}

  In \cref{rmk:program_shorter_than_proof}, we said that proofs contains both an algorithm (which may be trivial if no information is being computed) and the proof of its correctness. 
  This is also the case when we consider the computational content of a proof in the BHK interpretation: 
  we can separate the part that computes values and such (the informative computation) from the part that computes the evidence showing that the values are correct (the correctness computation). 
  The correctness computation does not affect the result of the informative computation and can be safely discarded when we are only interested in algorithm extraction. 

  This is not the case for the computational content in the interactive interpretation. 
  Here the correctness part of the computation affects the backtracking, which affects the state, which in turn affects the informative part of computation and thus the computed values. 
  Therefore, in interactive realizability both parts of the proof interact to produce the final result. 

  We have already seen an example of this interaction. 
  In the second iteration we chose \(\lrea_3\) as the least element and then we tried to instance \(\lrea_3 \leR \lrea_2\). 
  Then we realized that \(\lrea_3 \leR \lrea_2\) is false and that we had made a wrong assumption somewhere. 
  However \(\lrea_3 \leR \lrea_2\) is not one of the universal disjunct which we assumed by lack of information. 
  Therefore we have to look at the proof in order to find out which universal disjuncts we needed to deduce \(\lrea_3 \leR \lrea_2\) and to compute the witness which we need to extend the state. 
  This shows that in the interactive interpretation we cannot forget how we proved the correctness of our computations. 
\end{changed}

\begin{omitted}
  \begin{remark}
    In order to effectively check an instance of \(\op (\ib, \ia, \lpa)\), we need for \(\ib, \ia, \lpa\) to reduce to numerals. 
  \end{remark}

  \begin{lstlisting}[caption=The Least Element Program,label=prog:ir_min]
  data Answer = NoWitness | Witness Int
  type State = Int -> Int -> Answer

  rmin 0   = 0
  rmin n   = if s (rmin (n-1)) n == NoWitness
    then rmin (n-1)
    else n
  \end{lstlisting}

  \begin{remark}
    As it is common in logic, in order to provide justification for an axiom, we appeal to the informal principle formalized by the axiom, in this case the law of the excluded middle. 
    This may seem like cheating to an intuitionist, but we would like argue otherwise. 
    First of all a trivial remark: \(\EM\) really does not hold in intuitionistic logic, so it would be impossible to justify it appealing only to intuitionistic reasoning.
    Most importantly, we argue that our use of classical reasoning is sound: we do say that one of the disjunct holds, but we do not claim to know which one. 
    If we were do deduce a disjunction by classical means and then use its BHK interpretation, this would be unsound, since classical truth carries less information than intuitionistic truth. 
    In interactive realizability we do not presume to be able to decide \(\EM\) instances: when we choose a disjunct we make an educated guess based on the finite knowledge we have available at the moment, in a completely effective fashion. 
  \end{remark}

\end{omitted}

\section{The Real Plane}

In this section we introduce the real plane, points, lines and some relations between them. 
We use elementary analytic geometry: points are represented by coordinates, lines by equations and proofs are mostly computations with real numbers. 

We represent a point as a pair of real numbers, its coordinates. 
Formally we can say that a point is just a natural number \(\ia\) and that there is a primitive recursive function mapping indexes into pairs of real numbers. 
As we did for real numbers, in order to improve readability we add some sugar to the notation and 
use the metavariables \(\pa, \pb, \pc, \pd\) for arithmetic terms used as indexes of points. 
When we use the index of a point both as a number and as a point, we write it as \(\ia\) in the first case and as \(\pa_\ia\) in the second. 
We write the coordinates of a point \(\pa\) as \((x_\pa, y_\pa)\) and of a point \(\pa_\ia\) as \((x_\ia, y_\ia)\). 

A line passing through two points \(\pa\pb\) is written as \(\pa\pb\). 
The order of the points induces an orientation on the line. 

Before proceeding we need to introduce further infrastructure for the real numbers. 

\subsection{Operations on Real Numbers}

\newcommand\Rzero{0_\R}
\newcommand\Rplus{+_\R}
\newcommand\Rminus{-_\R}
\newcommand\Rprod{\cdot_\R}
Any rational number \(\lqa\) can be embedded in our coding of the real numbers: 
indeed we can represent \(\lqa\) as a real number by taking the nested interval sequence with the lower and higher extremes constantly equal to \(\lqa\).
In particular we assume that there is an index \(\Rzero\) such that \(\lfa^+\) and \(\lfa^-\) are constantly zero. 

We need to introduce the addition, subtraction and multiplication operations on the reals. 
In order to do this formally, we need to assume that for each pair of indexes \(\ia\) and \(\ib\) of real numbers, there is an index \(\ic\) which correspond to the nested interval sequence that is the result of their sum, difference or product. 
Again, instead of writing the index \(\ic\), we use the usual syntax \(\lrea_\ia \Rplus \lrea_\ib\) for the sum, \(\lrea_\ia \Rminus \lrea_\ib\) for the difference and \(\lrea_\ia \Rprod \lrea_\ib\) for the product. 

Now we define the actual sequences that represent the result of each operation and show that they satisfy the real number condition \eqref{eq:nested_cauchy}. 

We define addition on the nested interval sequences as:
\begin{align*}
  (\lrea_\ia \Rplus \lrea_\ib)^+(\lpa) &\equiv \lrea^+_\ia(\lpa+1) \Qplus \lrea^+_\ib(\lpa+1), \\ 
  (\lrea_\ia \Rplus \lrea_\ib)^-(\lpa) &\equiv \lrea^-_\ia(\lpa+1) \Qplus \lrea^-_\ib(\lpa+1). 
\end{align*}
It is immediate to check that the the sequences are a sequence of nested intervals; we only check that they converge with the required speed, which is the condition that requires the use of \(\lpa+1\) in the previous definition: 
\begin{align*}
  &\phantomrel{\eqQ} 
  (\lrea_\ia \Rplus \lrea_\ib)^+(\lpa) \Qminus (\lrea_\ia \Rplus \lrea_\ib)^-(\lpa) \eqQ \\
  &\eqQ (\lrea^+_\ia(\lpa+1) \Qplus \lrea^+_\ib(\lpa+1)) \Qminus 
  (\lrea^-_\ia(\lpa+1) \Qplus \lrea^-_\ib(\lpa+1)) \eqQ \\ 
  &\eqQ (\lrea^+_\ia(\lpa+1) \Qminus \lrea^-_\ia(\lpa+1))\Qplus  
  (\lrea^+_\ib(\lpa+1) \Qminus \lrea^-_\ib(\lpa+1)) \leQ \\ 
  &\leQ 2^{-\lpa+1} \Qplus 2^{-\lpa+1} \eqQ 2^{-\lpa}.
\end{align*}
We can define the difference by combining the sum and the opposite, which is defined as:
\begin{align*} 
  (-\lrea)^+(\lpa) &\equiv -_\Q \lrea^-(\lpa), \\ 
  (-\lrea)^-(\lpa) &\equiv -_\Q \lrea^+(\lpa).
\end{align*}
Defining the product is slightly more complicated. For simplicity we only show the case when the extreme of the intervals are always positive. 

So let \(\lrea^+_\ia\), \(\lrea^-_\ia\), \(\lrea^+_\ib\) and \(\lrea^-_\ib\) be sequences of positive rational numbers. 
We define their product as:
\begin{align*}
  (\lrea_\ia \Rprod \lrea_\ib)^+(\lpa) &\equiv \lrea^+_\ia(\lpb) \Qprod \lrea^+_\ib(\lpb), \\ 
  (\lrea_\ia \Rprod \lrea_\ib)^-(\lpa) &\equiv \lrea^-_\ia(\lpb) \Qprod \lrea^-_\ib(\lpb), 
\end{align*}
where \(\lpb\) depends on \(\lpa\). 
In order to determine \(\lpb\) we consider the convergence condition
and look for the smallest \(\lpb\) that satisfies it: 
\[ 
  (\lrea_\ia \Rprod \lrea_\ib)^+(\lpa) \Qminus (\lrea_\ia \Rprod \lrea_\ib)^-(\lpa) \leQ 2^{-\lpa}.
\]
We begin by finding a simple upper bound for the left-hand side: 
\begin{align*}\renewcommand\Qprod{}
  &\phantomrel{\eqQ} 
  (\lrea_\ia \Rprod \lrea_\ib)^+(\lpa) \Qminus (\lrea_\ia \Rprod \lrea_\ib)^-(\lpa) \eqQ \\ 
  &\eqQ
  \lrea^+_\ia(\lpb) \Qprod \lrea^+_\ib(\lpb) \Qminus \lrea^-_\ia(\lpb) \Qprod \lrea^-_\ib(\lpb) \eqQ \\ 
  &\eqQ
  \lrea^+_\ia(\lpb) \Qprod \lrea^+_\ib(\lpb) \Qminus \lrea^-_\ia(\lpb) \Qprod \lrea^-_\ib(\lpb) \Qminus \lrea^+_\ia(\lpb) \Qprod \lrea^-_\ib(\lpb) \Qplus \lrea^+_\ia(\lpb) \Qprod \lrea^-_\ib(\lpb) \eqQ \\ 
  &\eqQ
  \lrea^+_\ia(\lpb) \Qprod (\lrea^+_\ib(\lpb) \Qminus \lrea^-_\ib(\lpb)) \Qplus \lrea^-_\ib(\lpb) \Qprod (\lrea^+_\ib(\lpb) \Qminus \lrea^-_\ib(\lpb)) \leQ \\ 
  &\leQ
  \lrea^+_\ia(\lpb) 2^{-\lpb} \Qplus \lrea^-_\ib(\lpb) 2^{-\lpb} \eqQ (\lrea^+_\ia(\lpb) \Qplus \lrea^-_\ib(\lpb)) 2^{-\lpb}. 
\end{align*}
Thus the convergence condition is satisfied when:
\[
  (\lrea^+_\ia(\lpb) \Qplus \lrea^-_\ib(\lpb)) 2^{-\lpb} \leQ 2^{-\lpa}. 
\]
We define \(\lpa\) as the smallest natural number satisfying the previous inequality. 

\subsection{The Left and Right Predicates} 

\newcommand\leftP{\operatorname{\sf left}}
\newcommand\rightP{\operatorname{\sf right}}
In order to write the formal statement of the \gref{thm:bounding_angle}, we need a way to determine the position of a point with respect to a line. 

First of all consider two points \(\pa\) and \(\pb\). 
We can write the equation that a point \(\pc\) has to satisfy to be on the line going through them:
\begin{equation} \label{eq:line}
  (x_\pb - x_\pa)(y_\pc - y_\pa) - (x_\pc - x_\pa)(y_\pb - y_\pa) \eqR \Rzero. 
\end{equation}
If the left-hand side is zero then \(\pc\) is on the same line with \(\pa\) and \(\pb\). 
When left-hand side is not zero, we can use its sign to distinguish which side of \(\pa\pb\) \(\pc\) is on. 
We call these sides left and right. 
We write \( \leftP (\pa, \pb, \pc) \) (resp. \( \rightP (\pa, \pb, \pc) \)) and we say that \(\pc\) is to the \emph{left} (resp. \emph{right}) of the line passing through the points \(\pa\) and \(\pb\) when 
\begin{gather*}
  \leftP(\pa, \pb, \pc) \equiv (x_\pb - x_\pa)(y_\pc - y_\pa) - (x_\pc - x_\pa)(y_\pb - y_\pa) \gtR \Rzero, \\
  \rightP(\pa, \pb, \pc) \equiv (x_\pb - x_\pa)(y_\pc - y_\pa) - (x_\pc - x_\pa)(y_\pb - y_\pa) \ltR \Rzero,
\end{gather*}
as seen in \cref{fig:left_right}. 
\begin{figure}[!ht]
  \caption{\(\pc\) is to the left of \(\pa\pb\).}
  \label{fig:left_right}
  \begin{center}
    \begin{tikzpicture}[>=latex] 
      \draw[pattern=north west lines,pattern color=red,draw=white] (45:2.5) arc (45:225:2.5);
      \draw[pattern=north east lines,pattern color=green,draw=white] (225:2.5) arc (225:405:2.5);
      \draw[fill] (-135:1) circle (0.05);
      \node[left] at (-135:1) {\(\pa\)};
      \draw[fill] (45:2) circle (0.05);
      \node[left] at (45:2) {\(\pb\)};
      \draw (225:3) -- (45:3);
      \node[fill=white] at (135:1.5) {Left};
      \node[fill=white] at (-45:1.5) {Right};
      \draw[fill] (90:1) circle (0.05);
      \node[left] at (90:1) {\(\pc\)};
    \end{tikzpicture}
  \end{center}
\end{figure}
A few remarks on this definition:
\begin{itemize}
  \item
    \(\leftP\) and \(\rightP\) are positively decidable, since they are defined by means of \(\ltR\); 
  \item 
    since the definitions of \(\leftP\) and \(\rightP\) are almost the same and only the direction of the inequality changes, 
    \(\pc\) is to the left of \(\pa \pb\) if and only if \(\pb\) is to the right of \(\pa\pb\); 
  \item 
    the left side of \(\pa\pb\) corresponds to the right side of \(\pb\pa\) and the other way around, so the order of the points is significant; 
  \item
    the left-hand side of \eqref{eq:line} can also be thought as the scalar product of \((-(y_\pb - y_\pa), x_\pb - x_\pa)\), the orthogonal of the vector from \(\pa\) to \(\pb\), and \((x_\pc - x_\pa, y_\pc - y_\pa)\), the vector from \(\pa\) to \(\pc\). 
\end{itemize}

We say that \(\pa\) is \emph{above} \(\pb\) if \( y_\pa \geR y_\pb \) and that \(\pc\) is \emph{below} \(\pb\) when \( y_\pc \leR y_\pb \).

\section{The Geometric Part of the Proof}

Now we are ready to present the rest of the proof of the main statement. 
We divide the proof in two parts, the first given as a lemma. 
Since these proofs are more complex, for reason of readability and space 
we will not be as formal as we have we have been until now. 

From this point onward we assume that no three points are on the same line, formally:
\begin{equation} \label{eq:no_three_aligned}
  \qforall{\pa,\pb,\pc} \leftP (\pa, \pb, \pc) \lor \rightP (\pa, \pb, \pc). 
\end{equation}
This a strong assumption, even more so because we require this to hold constructively: since \(\leftP\) and \(\rightP\) are \(\Sigma^0_1\) formulas defined with \(\leR\), we assume that we have an effective map that given three points yields the precision we need to reach in order to check that \(\pc\) is not on the line \(\pa\pb\). 
In other words, we are assuming that we have a procedure that effectively decides instances of the \(\leftP\) and \(\rightP\) predicates. 
The effective computation we extract uses this procedure as a parameter. 

A further consequence is that all points must be distinct: when \(x_\pa \eqR x_\pb\) and \(y_\pa \eqR y_\pb\), the left-hand side in \eqref{eq:line} is always zero for any \(\pc\).

In the next lemma the points \(\pb_0, \pb_1, \pb_2\) are three generic points, 
that is, \(\pb_\ia\) is not necessarily the point indexed by the natural number \(\ia\). 
Moreover we assume that the index \(\ia\) in \(\pb_\ia\) is interpreted up to congruence modulo 3 and thus always falls in \(\{0,1,2\}\). 
For instance, when we write \(\pb_4\), we actually mean \(\pb_1\). 
We write the coordinates of \(\pb_\ia\) as \((x_\ia,y_\ia)\), with the same conventions for the index. 
We prove that when three points are one to the left of the other with respect to a central one, one of them is necessarily lower than the central point, as shown in \cref{fig:three_points}. 
\begin{lemma}[Three points] \label{thm:three_points}
  Assume \eqref{eq:no_three_aligned} and let \(\pa, \pb_0, \pb_1\) and \(\pb_2\) be four points in the real plane such that \(\pb_{\ia+1}\) is to the left (resp. right) of \(\pa \pb_\ia\) for any \( \ia < 3\).
  Then at least one of \(\pb_0, \pb_1, \pb_2\) is strictly below \(\pa\). 
  Formally:
  \[
    \qforall{\pa, \pb_0, \pb_1, \pb_2} 
    (\qforall{\ia < 3} \leftP(\pa,\pb_\ia,\pb_{\ia+1}))
    \limply
    \qexists{\ia < 3} y_\ia \ltR y_\pa.
  \]
\end{lemma}
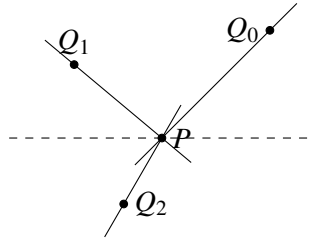
\begin{figure}[!ht]
  \caption{The three points lemma when \(\pb_2\) is the point below \(\pa\).}
  \label{fig:three_points}
  \begin{center}
    \begin{tikzpicture}[>=latex] 
      \draw[fill] (0,0) circle (0.05);
      \node[right] at (0,0) {\(\pa\)};
      \draw[fill] (45:2) circle (0.05);
      \node[left] at (45:2) {\(\pb_0\)};
      \draw (-135:0.5) -- (45:2.5);
      \draw[fill] (140:1.5) circle (0.05);
      \node[above] at (140:1.5) {\(\pb_1\)};
      \draw (-40:0.5) -- (140:2);
      \draw[fill] (240:1) circle (0.05);
      \node[right] at (240:1) {\(\pb_2\)};
      \draw (60:0.5) -- (240:1.5);
      \draw[dashed] (-2,0) -- (2,0);
    \end{tikzpicture}
  \end{center}
\end{figure}

\begin{proof}[Classical proof]
  Without loss of generality we can assume that the coordinates of \(\pa\) are \((\Rzero,\Rzero)\). 
  Then, unfolding the definition of \(\leftP\), the hypothesis on the points can be written as:
  \[
    \qforall{\ia < 3} x_\ia y_{\ia+1} - x_{\ia+1} y_\ia \gtR \Rzero,
  \]
  The first step is showing that there at least two points whose vertical coordinate is not zero. 
  This follows from the fact that 
  if, for some \(\ia < 3\), \( y_\ia = \Rzero \) then \( y_{\ia+1} \neq \Rzero \) and \( y_{\ia-1} \neq \Rzero \). 
  This is the case since if \( y_\ia = \Rzero \) then
  \begin{gather*}
    x_{\ia-1} y_\ia - x_\ia y_{\ia-1} = - x_\ia y_{\ia-1} \gtR \Rzero, \\
    x_\ia y_{\ia+1} - x_{\ia+1} y_\ia = x_\ia y_{\ia+1} \gtR \Rzero,
  \end{gather*}
  and then \( y_{\ia+1} \neq \Rzero\) and \( y_{\ia-1} \neq \Rzero \). 
  Then we can assume that \(y_{\ia+1}\) and \(y_{\ia-1}\) are not zero. 
  If either of them is negative we can conclude. 

  Otherwise they are both positive. We show that, in this case, \(y_\ia\) is negative. 
  By hypothesis we know that: 
  \begin{gather*}
    x_{\ia-1} y_\ia - x_\ia y_{\ia-1} \gtR \Rzero, \\
    x_\ia y_{\ia+1} - x_{\ia+1} y_\ia \gtR \Rzero. 
  \end{gather*}
  Since \( y_{\ia+1} \) and \( y_{\ia-1} \) are positive we can multiply the previous inequalities:
  \begin{gather*}
    y_{\ia+1} (x_{\ia-1} y_\ia - x_\ia y_{\ia-1}) \gtR \Rzero, \\
    y_{\ia-1} (x_\ia y_{\ia+1} - x_{\ia+1} y_\ia) \gtR \Rzero.  
  \end{gather*}
  By adding them together we get:
  \[ 
    y_\ia (x_{\ia-1} y_{\ia+1} - x_{\ia+1} y_{\ia-1}) > 0. 
  \]
  Since the term in parenthesis is negative by hypothesis, \(y_\ia\) must be too. 
  Since \(\pb\), \(\pc\) and \(\pd\) are to the right of each other with respect to \(\pa\) if and only if \(\pd\), \(\pc\) and \(\pb\) are to the left of each other with respect to \(\pa\), the proof for \(\rightP\) is basically the same. 
\end{proof}
\begin{changed}
The previous proof can be made constructive. 
Since the proofs are very similar, we give the intuitionistic proof without explaining the how we obtained it from the classical one. 
The main difference is that, in the intuitionistic proof, we work directly on the rational intervals approximating the coordinates of the points and thus we use rational arithmetic which is decidable. 
For any precision\(\lpa\), let \(X_\ia(\lpa)\) and \(Y_\ia(\lpa)\) be the closed rational intervals \([x^-_\ia(\lpa),x^+_\ia]\) and \([y^-_\ia(\lpa),y^+_\ia]\) respectively. 
We write \( \lqa_\ia \in X_\ia(\lpa) \) as a compact notation for \( x^-_\ia \leQ x_\ia \land x_\ia \leQ x^+_\ia \). 
\begin{proof}[Intuitionistic proof]
  Without loss of generality we can assume that the coordinates of \(\pa\) are \((\Rzero,\Rzero)\).
  Then, unfolding the definition of \(\leftP\), the hypothesis on the points can be written as:
  \[
    \qforall{\ia < 3} x_\ia y_{\ia+1} - x_{\ia+1} y_\ia \gtR \Rzero. 
  \]
  By unfolding the definition of real numbers as sequences of nested intervals and these of the operations on real numbers, we can compute some precision \(\lpa\) such that:
  \begin{equation} \label{eq:three_point_assumption}
    \qforall{\ia < 3} \qforall{\lqa_\ia \in X_\ia(\lpa), y_\ia \in Y_\ia(\lpa), \lqa_{\ia+1} \in X_{\ia+1}(\lpa), y_{\ia+1} \in Y_{\ia+1}(\lpa)} \lqa_\ia y_{\ia+1} - \lqa_{\ia+1} y_\ia \gtQ \Qzero.
  \end{equation}
  The first step is showing that there at least two points whose vertical coordinate is not zero. 
  In order to show this, assume that for some \(\ia < 3\) we have \(\Qzero \in Y_\ia(\lpa)\). 
  Then we can take \(\lqb_\ia = \Qzero\) in \eqref{eq:three_point_assumption}, for \(\ia-1\) and \(\ia\):
  \begin{gather*}
    \qforall{\lqa_{\ia-1} \in X_{\ia-1}(\lpa), \lqa_\ia \in X_\ia(\lpa), \lqb_{\ia-1} \in Y_{\ia-1}(\lpa)} \lqa_{\ia-1} \lqb_\ia - \lqa_\ia \lqb_{\ia-1} = - \lqa_\ia \lqb_{\ia-1} \gtQ \Qzero, \\
    \qforall{\lqa_\ia \in X_\ia(\lpa), \lqa_{\ia+1} \in X_{\ia+1}(\lpa), \lqb_{\ia+1} \in Y_{\ia+1}(\lpa)} \lqa_{\ia-1} \lqb_\ia - \lqa_\ia \lqb_{\ia-1} \lqa_\ia \lqb_{\ia+1} - \lqa_{\ia+1} \lqb_\ia = \lqa_\ia \lqb_{\ia+1} \gtQ \Qzero. 
  \end{gather*}
  Therefore, for \( \ib \in \{\ia-1,\ia+1\} \), \(\Qzero \not\in Y_\ib(\lpa)\) and, since \(Y_\ib(\lpa)\) is an interval, it must be either completely positive or completely negative, namely, either \(x^-_\ib \gtQ \Qzero\) or \(x^+_\ib \ltQ \Qzero\). 
  If either one is completely negative then we have the conclusion. 
  
  Otherwise they are both completely positive and we show that, in this case, \( Y_{\ia}(\lpa) \) is completely negative. 
  For all \( \lqa_\ib \in X_\ib(\lpa)\) and all \(\lqb_\ib \in Y_\ib(\lpa)\) with \(\ib \in \{ 0,1,2\}\); we know by hypothesis that: 
  \begin{gather*}
    \lqa_{\ia-1} \lqb_\ia - \lqa_\ia \lqb_{\ia-1} \gtQ \Qzero, \\
    \lqa_\ia \lqb_{\ia+1} - \lqa_{\ia+1} \lqb_\ia \gtQ \Qzero.
  \end{gather*}
  Since \( \lqb_{\ia+1} \) and \( \lqb_{\ia-1} \) are positive we can multiply the previous inequalities:
  \begin{gather*}
    \lqb_{\ia+1} (\lqa_{\ia-1} \lqb_\ia - \lqa_\ia \lqb_{\ia-1}) \gtQ \Qzero, \\
    \lqb_{\ia-1} (\lqa_\ia \lqb_{\ia+1} - \lqa_{\ia+1} \lqb_\ia) \gtQ \Qzero.
  \end{gather*}
  By adding them together we get:
  \[
    \lqb_\ia (\lqa_{\ia-1} \lqb_{\ia+1} - \lqa_{\ia+1} \lqb_{\ia-1}) \gtQ \Qzero.
  \]
  Since the term in parenthesis is negative by hypothesis, \(\lqb_\ia\) must be too, for all \( \lqb_\ia \in Y_{\ia}\).

  Since \(\pb\), \(\pc\) and \(\pd\) are to the right of each other with respect to \(\pa\) if and only if \(\pd\), \(\pc\) and \(\pb\) are to the left of each other with respect to \(\pa\), the proof for \(\rightP\) is basically the same. 
\end{proof}
\end{changed}

We can now prove the main statement. 
\begin{theorem}[Convex Angle] \label{thm:bounding_angle}
  Assume \eqref{eq:no_three_aligned}. For any \(\na \ge 2\), we can select three points \(\pa\), \(\pb\) and \(\pc\) from \(\{\pa_0, \dotsc, \pa_\na\}\) 
  such that all the remaining points fall in the angle \(\widehat{\pb\pa\pc}\), that is, all points are to the left of \(\pa\pb\) and to the right of \(\pa\pc\). 
  \begin{multline*}
    \qforall{\na\ge 2} \qexists{\ia,\ib,\ic \le \na} 
    \qforall{\id \le \na} \id \neq \ia \land 
    (\id \neq \ib \limply \leftP (\pa_\ia, \pa_\ib, \pa_\id)) \land
    (\id \neq \ic \limply \rightP (\pa_\ia, \pa_\ic, \pa_\id)). 
  \end{multline*}
\end{theorem}
\begin{proof}[Classical proof]
  Let \(\pa\) be the point with the least vertical coordinate and choose other two points \(\pb'\) and \(\pc'\), which are our candidates for \(\pb\) and \(\pc\) respectively. 
  We want all points except \(\pa\) to be to the left of \(\pa\pb\) and to the right of \(\pa\pc\). 
  If \(\pb'\) is to the left of \(\pa\pc'\), we swap \(\pb'\) and \(\pc'\). 
  Thus we know that \(\pb'\) is to the right of \(\pa\pc'\) and \(\pc'\) is to the left of \(\pa\pb'\). 

  Now consider any point \(\pd\) except \(\pa\), \(\pb'\) and \(\pc'\). 
  We have four cases:
  \begin{itemize}
    \item 
      If \(\pd\) is to the left of \(\pa \pb'\) and if it is to the right of \(\pa \pc'\), then we keep \(\pb'\) and \(\pc'\) as candidates for \(\pb\) and \(\pc\). 
    \item 
      If \(\pd\) is to the right of \(\pa \pb'\), then we choose \(\pd\) as the new candidate for \(\pb\).
      \begin{center}
        \begin{tikzpicture}[>=latex] 
          \pointAt{A}{0,0}\pa{below}
          \pointAt{B}{60:1}{\pb'}{right}
          \lineBetween{A}{B}
          \pointAt{C}{-2,1}{\pc'}{left}
          \lineBetween{A}{C}
          \pointAt{D}{1,0.5}{\pd}{right}
          \lineBetween{A}{D}
          \pointAt{D'}{-1,1}{\pd'}{right}
          \draw[dashed] (-2,0) -- (2,0);
        \end{tikzpicture}
      \end{center}
      Clearly \(\pb'\) is to the left of \(\pa \pd\). 
      Moreover, any other point \(\pd'\), which we already checked to be to the left of \(\pa \pb'\), is to the left of \(\pa \pd\) too. 
      This is a consequence of \eqref{eq:no_three_aligned} and \cref{thm:three_points}. 
      
      Indeed, from \eqref{eq:no_three_aligned}, we know that \(\pd'\) is either to the left or to the right of \(\pa\pd\). 
      We already know that \(\pd\) is to the right of \(\pa \pb'\) and \(\pb'\) is to the right of \(\pa \pd'\).
      If \(\pd'\) were to the right of \(\pa \pd\), then by \cref{thm:three_points}, one of \(\pb'\), \(\pd\) or \(\pd'\) would have be strictly lower than \(\pa\) which would be a contradiction, since \(\pa\) is the lowest point.
      Thus \(\pd'\) is to the left of \(\pa \pd\). 
    \item 
      Symmetrically, if \(\pd\) is to the left of \(\pa \pc'\), then we choose \(\pd\) as the new candidate for \(\pc\). 
    \item 
      We shot that \(\pd\) cannot be to the right of \(\pa \pb'\) and to the left of \(\pa \pc'\): 
      \begin{center}
        \begin{tikzpicture}[>=latex] 
          \pointAt{A}{0,0}\pa{below}
          \pointAt{B}{30:1}{\pb'}{above}
          \lineBetween{A}{B}
          \pointAt{C}{160:2}{\pc'}{above}
          \lineBetween{A}{C}
          \pointAt{D}{-40:1.5}{\pd}{above}
          \lineBetween{A}{D}
          \draw[dashed] (-2,0) -- (2,0);
        \end{tikzpicture}
      \end{center}
      If this were the case, \(\pb'\) would be to the left of \(\pa \pd\) and \(\pd\) would be to the left of \(\pa \pc'\). 
      Since we know that \(\pc'\) is to the left of \(\pa\pb'\), 
      by \cref{thm:three_points}, one of \(\pd\), \(\pb'\) or \(\pc'\) would be strictly lower than \(\pa\). 
      This is a contradiction, since \(\pa\) is the lowest point by the \gref{thm:least_element}. 
  \end{itemize}
  We repeat this procedure for all the points except \(\pa\), \(\pb'\) and \(\pc'\) and we find the required points \(\pb\) and \(\pc\). 
\end{proof}

For convenience we have written the proof as an iterative algorithm. 
The proof is actually by induction on a slightly stronger version of the final statement, that adds the requirement for \(\pa\) to be lower than all the other points. 

\section{The Interactive Interpretation}

Before studying the interactive interpretation of the whole proof of the \gref{thm:bounding_angle} along with its lemmas, we need understand their computational significance. 
Thus we stop for a moment and recall some general considerations on the computational meaning of formulas in the BHK interpretation and, more specifically, in the Curry-Howard correspondence. 

As a consequence of the proof-as-programs and formulas-as-types interpretation, 
the conclusion of a proof (that is, the statement it proves) can be thought of as the specification of the program representing the proof. 

\begin{omitted}
\subsection{Informative Computations and Harrop Formulas}

From a logical point of view atomic formulas are just what their name suggests, that is, formulas that cannot be decomposed and analyzed further. 
The only relevant property we are interested in is their truth, which we assume to be able to check with an effective computation (which is reasonable, at least in an arithmetic context). 
Thus any proof of an atomic formula is represented by the same token program, which performs no computation and whose only significance is witnessing that there exists a proof. 
Essentially the proof of an atomic formula carries no computational content, at least from the point of view we follow here.

The situation does not change if we add conjunctions: if we pack together a pair of programs encoding trivial computations, we still get a program that performs a trivial computation. 

Things change slightly if we consider implication and universal quantification: we have a program that encodes a function. 
When this program is given in input another program in the case of the implication or a natural number in the case of the universal quantification, it performs a computation. 
If this computation is trivial so is the original program. 

These considerations can be summarized in the following definition. 
\begin{definition}[Harrop Formulas]
  We inductively define the class of \emph{Harrop formulas} as follows:
  \begin{itemize}
    \item atomic formulas are Harrop,
    \item \(\fa \land \fb\) is Harrop if \(\fa\) and \(\fb\) are,
    \item \(\fa \limply \fb\) is Harrop if \(\fb\) is,
    \item \(\qforall\lva \fa\) is Harrop if \(\fa\) is,
  \end{itemize}
\end{definition}
The program corresponding to the proof of an Harrop formula describes a trivial computation: its existence is significant but it produces no information. 

The real difference happens when we consider the informative connectives, namely disjunction and existential quantification. 
Then the proof encodes a non trivial computation, one that decides which disjunct holds or gives a witness. 
Therefore, a disjunction specifies a program describing the computation of a boolean value, either ``left'' or ``right''. 
Similarly, an existential quantified formula specifies a program describing the computation of a natural number (at least in arithmetic). 
Actually the both compute a pair, where the first entry is the value and the second entry is a (program describing a) computation that checks that the value is correct. 
As we noted in \cref{rmk:program_shorter_than_proof}, a proof can be ideally divided in two parts: one that describes the algorithm that decides disjunctions and computes witnesses and another one that proves its correctness. 
In order to distinguish them, we call the former the \emph{informative computation} and the second the \emph{correctness computation}
Of course these parts are usually intertwined in the proof. 
In the case of the proof of an Harrop formula, there is no informative computation and the proof describes only a correctness computation. 

When we are interested in extracting the computational content from a constructive proof, we are usually interested in the informative computation, which computes values. 
We know that the values are correct because the correctness computation exists, but usually we are not interested in performing the correctness computation itself. 
If the proof we consider is not constructive and we consider its interactive interpretation, the correctness computation starts affecting the informative computation and thus cannot be ignored anymore. 

Consider the statement of the \gref{thm:least_element} and note how the last part, the formula stating that \(\ia\) is a least element, is Harrop:
\[
  \qforall\na \qexists{\ia \le \na} \qforall{\ib \le \na} \lrea_\ia \leR \lrea_\ib. 
\]
Thus the informative computation of the proof is specified by the first two quantifiers, as evidenced by \cref{prog:min}. 
We have seen that in the interactive interpretation, the value of \(\ia\) we compute may very well be wrong. 
This fact can only become manifest if and when we look at the correctness computation. 
When the \gref{thm:least_element} is used in a bigger proof, the value it produces will be used along with its correctness computation.
When the correctness computation is performed, even partially, we may discover that the value was wrong. 
Then we have to backtrack, find out which wrong assumption caused the correctness computation to fail, extend the state and resume the computation from a suitable point. 
All of this behavior is governed by the correctness computation, which has the fundamental task of computing the correct extension to the state. 
In order to do this, it has to remember each ``educated guess'' we made about each \(\EM\) instance. 
Thus the correctness computation becomes essential because it controls the backtracking. 
At this point the distinction between the informative and correctness computations blurs, since the value of \(\ia\) in the interactive interpretation of the \gref{thm:least_element} can only be computed by the interaction of the informative and correctness computations through backtracking. 
This is a fundamental difference between the standard BHK interpretation and the interactive interpretation. 
\end{omitted}

\subsection{Subroutines, arguments and effective computations}

In order to understand how the interactive interpretation works, it is important to distinguish computations that can be carried out effectively from those that cannot. 
Consider a proof of a statement of the form:
\begin{equation} \label{eq:simple_function}
  \qforall\lva \qexists\lvb \fa. 
\end{equation}
If we read the previous formula as a specification, it calls for a program that describes a function, a subroutine.
It takes a natural number as an argument named \(\lva\) and returns a pair containing a natural number \(\lvb\) and a program/proof that \(\lvb\) satisfies \(\fa\). 
More generally, statements in mathematics have the following form:
\[
  \qforall{\lva_1, \dotsc, \lva_\na} \fa_1 \land \dotsb \land \fa_\nb \limply \fa. 
\]
This can be again seen as the specification of a subroutine, taking \(\na\) natural numbers and \(\nb\) programs as arguments. 
From this point of view, it becomes clear that such a program is not computing anything, at least in itself. 
An effective computation can only start once the subroutine is applied to an argument. 
\begin{omitted}
Moreover, if an argument is really needed in the subroutine, then it must be computed itself. 
In general, we can only effectively compute something if it contains no unspecified parameters.
Thus, when we think of a proof as a computation, we need to be aware that the computation can only be effectively carried out once we have instanced the outer universally quantified variables with concrete values. 

This is another way to explain why universally quantified atomic formulas are negatively decidable: they can be thought of subroutines with missing arguments, thus in a state of ``suspended computation''. 
Moreover since the possible arguments are all the natural numbers, it is unfeasible to verify them. 
This is not the case when the quantification is bounded: \(\qforall{\ia \le \na} \fa \) is effectively computable if \(\fa\) is. 
From a logical point of view, it behaves like the conjunction \( \fa\subst\ia{0} \land \dotsm \land \fa\subst\ia\na \). 

This is very important in the case of correctness computations. 
In \eqref{eq:simple_function}, we may be able to effectively check that \(\lvb\) satisfies the specification if the condition \(\fa\) is effectively computable, for instance when \(\fa\) is positively decidable. 
\end{omitted}

All of our theorems begin with universal quantifications and implications, that is, they are specification for programs that code functions with arguments. 
Thus, in order to have an actual computation we have to provide the program with the required arguments. 

\subsection{The Interactive Interpretation of the Whole Proof}

We can now explain the interactive interpretation of the whole proof, composed of the two lemmas and the final algorithmic proof. 
We focus on the interaction between these parts without analyzing each part in detail as we have done for the \gref{thm:least_element}. 

\begin{wip}
  Proof irrelevant formulas: (these are simply Harrop formulas)
  \begin{itemize}
    \item an atomic formula \(\lafa\) is p.i.,
    \item \(\fa \land \fb\) is p.i. if \(\fa\) and \(\fb\) are,
    \item \(\fa \limply \fb\) is p.i. if \(\fb\) is,
    \item \(\qforall\lva \fa \) is p.i. if \(\fa\) is,
  \end{itemize}
  Effectively computable formulas:
  \begin{itemize}
    \item an atomic formula \(\lafa\) is e.c.,
    \item \(\fa \land \fb\) and \(\fa \lor \fb\) are e.c. whenever \(\fa\) and \(\fb\) are,
    \item \(\lafa \limply \fa\) if \(\fa\) is,
    \item \(\fa \limply \fb\) if \(\fb\) is e.c. and \(\fa\) is Harrop?,
    \item \(\qexists\lva \fa\) if \(\fa\) is, 
    \item \(\qforall{\lva \le \na} \fa\) if \(\fa\) is. 
  \end{itemize}
\end{wip}

We start by considering the statement of the \gref{thm:bounding_angle}. 
Assume that we are given a natural number \(\na\). 
In the proof we work with the first \(\na+1\) points of the enumeration. 

The proof is an iterative procedure to select \(\pa\), \(\pb\) and \(\pc\) satisfying the following \emph{bounding condition}: 
\begin{equation} \label{eq:bounding_condition}
  \qforall{\id \le \na} \id \neq \ia \land 
  (\id \neq \ib \limply \leftP (\pa_\ia, \pa_\ib, \pa_\id)) \land
  (\id \neq \ic \limply \rightP (\pa_\ia, \pa_\ic, \pa_\id)). 
\end{equation}
The bounding condition specifies an informative computation, since \(\leftP\) and \(\rightP\) are defined by means of \(\ltR\), which is an existential quantification. 
Thus its proofs computes some witnesses, namely the precision of the comparisons we need to check that the bounding condition holds. 
While are mainly interested in the choice of the points \(\pa, \pb\) and \(\pc\) and not in the information needed to prove the bounding condition itself, 
the precision of the computation provided by \eqref{eq:bounding_condition} is actually used in interactive interpretation since it can cause backtracking. 

We claim that this bounding condition specifies an effective computation. 
First of all, the outer universal quantification is bounded, thus, in order to compute the condition, we have to compute the body of the quantification \(\na+1\) times. 
The same holds for the conjunctions. 
Thus the effectiveness of the whole condition follows from the effectiveness of the conjuncts. 
The implications are effective: 
their only argument, the proof of the antecedent, is arithmetical atomic, hence irrelevant, thus the computations they specify must be constant functions. 
Therefore, we can effectively compute them by applying them to any single argument. 
Finally their consequents specify effective computations, thanks to \eqref{eq:no_three_aligned}, the assumption that no three points are on the same line. 
Thus, proofs of the bounding condition describe effective computations.


Now we can start following the proof. 
In the beginning, the lowest point \(\pa\) is selected using the \gref{thm:least_element} on the vertical coordinate. 
Consider the statement of the \gref{thm:least_element}:
\[
  \qforall\na \qexists{\ia \le \na} \qforall{\ib \le \na} \lrea_\ia \leR \lrea_\ib. 
\]
As a specification, it calls for a program that, given \(\na\), yields the value \(\ia\) and the correctness computation that checks that \(\ia\) is the least element. 
Since the correctness computation cannot be carried out effectively (it is negatively decidable),
the interactive interpretation computes a trivial least element the first time.
If later in the proof we happen to partially compute the correctness computation, then we may discover new information and backtrack again to the least element computation. 
Since the \gref{thm:least_element} does not necessarily return a least element, but only a least element candidate, \(\pa\) is not the lowest point either, but just a lowest point candidate. 

The function of \cref{thm:three_points} is to prove that some point is strictly lower than \(\pa\), thus producing a contradiction. 
In the classical proof this ensures that undesirable situations never happen. 
In the interactive interpretation however, since \(\pa\) is not necessarily the lowest point, no contradiction occurs. 
Instead, what happens is that we actually are in one of the cases we had excluded in the classical proof. 
At this point, in order to deduce the contradictory statement, we have partially computed the correctness computation returned by the \gref{thm:least_element} and thus discovered which assumption was incorrectly guessed. 
We compute the relevant witness and extend the state accordingly. 
Then we compute a new lowest point candidate and continue again following the proof of the \gref{thm:bounding_angle} until either we can satisfy its conclusion or we backtrack again. 

We use \cref{thm:three_points} in two places in the proof of the \gref{thm:bounding_angle}. 
The first use takes place when, while iterating on the points, we discover that the bounding condition fails for some \(\pd\) and we choose \(\pd\) as the new candidate for \(\pb\) or \(\pc\). 
We use \cref{thm:three_points} to show that this choice satisfies the bounding condition for all the previous points we iterated over until now. 
More precisely we use \cref{thm:three_points} to prove that, if the bounding conditions fails for \(\pd\), then one of \(\pb\), \(\pc\) or \(\pd\) is strictly lower than \(\pa\).
As we described previously, this in turn starts the backtracking. 
Here is an example illustrating an interesting situation:
\begin{center}
  \begin{tikzpicture}[>=latex] 
    \pointAt{A}{0,0}\pa{below}
    \pointAt{B}{1,-0.5}\pb{above}
    \pointAt{C}{-0.2,1}\pc{left}
    \pointAt{D}{1,-1}\pd{below}
    \pointAt{D1}{0.8,1}{}{}
    \pointAt{D2}{1.3,0.5}{}{}
    \lineBetween{A}{B}
    \lineBetween{A}{C}
    \draw[dashed] (-2,0) -- (2,0);
  \end{tikzpicture}
\end{center}
Here \(\pd\) is to the right of \(\pa\pb\), so we replace \(\pb\) with \(\pd\).
From the picture we see that \(\pd\) is actually strictly lower than \(\pa\). 
On the other hand the bounding condition is satisfied by taking \(\pd\) as \(\pa\). 
In this situation, do we backtrack or not? 
We know that we can find a better candidate for the lowest point, but \(\pa\) seems to be good enough already, so there is no real need for a better candidate. 
Both options are sound and the choice depends on the exact formalization of the proof and the exact sequences of rationals representing the vertical coordinate of the points in question.

We also use \cref{thm:three_points} to claim that the bounding condition cannot fail because \(\pd\) is both to the right of \(\pa\pb\) and to the left of \(\pa\pc\): 
\begin{center}
  \begin{tikzpicture}[>=latex]
    \pointAt{A}{0,0}\pa{below}
    \pointAt{B}{1,0.5}\pb{above}
    \pointAt{C}{-0.5,1}\pc{left}
    \pointAt{D}{-1,-1}\pd{below}
    \pointAt{D1}{0.8,1}{}{}
    \pointAt{D2}{1.3,1.5}{}{}
    \lineBetween{A}{B}
    \lineBetween{A}{C}
    \draw[dashed] (-2,0) -- (2,0);
  \end{tikzpicture}
\end{center}
This case was excluded completely in the classical proof, since it always leads to contradiction. 
When it occurs in the interactive interpretation, we backtrack for sure since the bounding condition cannot be satisfied. 
More precisely, in this case \cref{thm:three_points} proves that one of \(\pb\), \(\pc\) or \(\pd\) is strictly lower than \(\pa\). 
Therefore, in order to get the contradiction, we instance the assumptions \(y_\pa \leR y_\pb\), \(y_\pa \leR y_\pc\) and \(y_\pa \leR y_\pd\) with enough precision to falsify at least one of them. 

As a last example, consider a situation where the state is empty and thus \(\pa\) is simply the first point in the enumeration. Assume that the points are arranged as shown:
\begin{center}
  \begin{tikzpicture}[>=latex]
    \pointAt{A}{0,0}\pa{below}
    \pointAt{B}{-30:1}\pb{above}
    \pointAt{C}{-140:2}\pc{above}
    \pointAt{D}{-60:1.5}\pd{below}
    \pointAt{D1}{-100:1}{}{}
    \pointAt{D2}{-130:1.2}{}{}
    \lineBetween{A}{B}
    \lineBetween{A}{C}
    \draw[dashed] (-2,0) -- (2,0);
  \end{tikzpicture}
\end{center}
Since the bounding condition is satisfied immediately, we never need to use \cref{thm:three_points}. Thus backtracking never ensues. 
This mean that \(\pa\), while certainly not the lowest point, is a good enough candidate and we do not need another one. 
This is one of the cases we mentioned where the interactive interpretation produces a fast computation, since the lowest point is only computed once and the proof ends with no backtracking. 
This shows how the behavior of interactive interpretation of the \gref{thm:least_element} depends heavily on the final statement of the proof. 



\begin{omitted}

The result produced by the interactive interpretation of a lemma are heavily dependent on the final conclusion the lemma 
The computations performed in the interactive interpretation are 


\begin{figure}[!ht]
  \caption{The interactive interpretation is guided by the last conclusion} \label{fig:example_three_points}
  \begin{center}
    \begin{tikzpicture}[>=latex]
      \pointAt{A}{0,0}\pa{below}
      \pointAt{B}{-1,-0.5}\pb{above}
      \pointAt{C}{0.5,-1}\pc{left}
      \pointAt{D1}{-0.8,-1}{}{}
      \pointAt{D2}{.3,-1.5}{}{}
      \lineBetween{A}{B}
      \lineBetween{A}{C}
      \draw[dashed] (-2,0) -- (2,0);
    \end{tikzpicture}
  \end{center}
  \begin{legend}
    In this case the assumption that \(\pa\) is the lowest point is clearly false, since \(\pa\) is actually the highest point. 
    However the conclusion of the \gref{thm:bounding_angle} holds and thus \(\pa\) is acceptable.

  \end{legend}
\end{figure}

\end{omitted}

\begin{omitted}
\section{Appendix}

The main inductive lemma
\begin{proof}
  We prove this by induction on the number of points \(\na\). 
  The minimum \(\na\) for which the statement is not trivial is \(3\). 
  \begin{description}
    \item[Base case]
      When \(\na\) 
    \item[Successor case]
      Assume that 
  \end{description}

\end{proof}

\end{omitted}

\begin{omitted}
\subsection{Cauchy sequence and equivalent definitions}
We start from a Cauchy sequence:
\[ 
  \qforall\lpa \qexists{\lpa_0} \qforall{\lpa_1, \lpa_2}
  | \lrea(\lpa_0 + \lpa_2) - \lrea(\lpa_0 + \lpa_1) | < \frac{1}{2^\lpa}. 
\]
If we had countable choice,  we could extract the modulus of convergence of \(\lrea\), that is, a function \(f : \N \to \N\) such that:
\[ 
  \qforall{\lpa, \lpa_1, \lpa_2}
  | \lrea(f(\lpa) + \lpa_2) - \lrea(f(\lpa) + \lpa_1) | < \frac{1}{2^\lpa}. 
\] 
Then, instead of \(\lrea\) we can consider \(\lrea' \equiv \lrea \circ f\), a sub-sequence of \(\lrea\) that converges much faster: 
\[ 
  \qforall{\lpa, \lpa_1, \lpa_2}
  | \lrea'(\lpa + \lpa_2) - \lrea'(\lpa + \lpa_1) | < \frac{1}{2^\lpa}. 
\]

Given rational sequence \(q : \N \to \Q\) we can extract a monotone sub-sequence. 
That is, there is a function \(f : \N \to \N\) such that it is strictly increasing: 
\[ \qforall\na f \na < f \na+1, \] 
and such that \(q \circ f\) is either weakly increasing:
\[ \qforall\na q (f \na) \leq q (f \na+1), \] 
or weakly decreasing:
\[ \qforall\na q (f \na) \geq q (f \na+1). \] 

\end{omitted}

\appendix

\chapter{Appendix}

\begin{wip}
  \section{Realizers for the Contrapositive of the Axiom of Choice}
  We use the letters \(x, y, z\) as variables for natural numbers and \(f, g, h\) for functions from natural numbers to natural numbers. 
  We write the axiom of choice in the language of \(\HA^\omega\):
  \[ \tag{AC} \label[axiom]{ac}
  \forall x \exists y.\ \fc(x,y) \limply \exists f \forall x.\ \fc(x,f(x)). \] 
  In intuitionistic logic this entails its contrapositive, that is:
  \[ \neg\exists f \forall x.\ \fc(x,f(x)) \limply \neg\forall x \exists y.\ \fc(x,y), \] 
  which is classically (but not intuitionistically) equivalent to:
  \[ \tag{CAC} \label[axiom]{cac}
  \forall f \exists x.\ \fc(x,f(x)) \limply \exists x \forall y.\ \fc(x,y). \] 
  We want to realize this classical principle. 

  We can read \cref{cac} by means of the BHK interpretation. 
  We begin from the antecedent of the implication, \( \forall f \exists x.\ \fc(x,f(x)) \), whose interpretation is a (higher order) function \(F\) that from a function \(f\) calculates a witness \(x\) satisfying \( \fc(x, f(x)) \). 
  The consequent interpretation should be a couple of a witness \(x\) and a function that for any \(y\) produces a proof of \( \fc(x,y) \). 
  The interpretation of the whole \cref{cac} should be a function mapping \(F\) into a witness \(x\) for the consequent. 
  The problem is that the witnesses given by \(F\) are partial, that is, they are required to satisfy \(\fc(x,y)\) only when \( y = f(x) \). 

  An interesting fact is that there is a function \(f\) such that \(F(f)\) is a witness for the consequent. 
  It is enough to require that, for any \(x\) which is no witness for the consequent, \(f(x)\) is a witness for \(\exists y. \ \neg \fc(x,y)\). 
  More formally the existence of such a function follows from \cref{ac} instanced for one of the following formulas\footnotemark: 
  \begin{align*}
    D_1(x,y) &\equiv (\forall z.\ \fc(x,z)) \lor \neg \fc(x,y), \\ 
    D_2(x,y) &\equiv (\fc(x,y) \limply \forall z.\ \fc(x,z)), \\ 
    D_3(x,y) &\equiv (\exists z.\ \neg \fc(x,z)) \limply \neg \fc(x,y), 
  \end{align*}
  that state that if there are counterexamples to \(\fc(x, \cdot)\) then \(y\) is one. 
  Then \cref{ac} proves that exists a function \(f\) that evaluates to a counterexample whenever possible. 
  We say that such an \(f\) is a worst counterexample. 
  Of course the antecedent is generally not provable constructively and therefore we cannot take advantage of \cref{ac} to produce a worst counterexample. 
  \footnotetext{While these are classically equivalent, intuitionistically \(D_1\) entails \(D_2\) and \(D_2\) entails \(D_3\).}

  We can use the same idea used in interactive realizability to craft a realizer of \cref{cac}. 
  We begin with an informal sketch. 

  Consider a proof in \(\HA^\omega+\)CAC and assume that in order to take advantage of an instance \cref{cac} the proof contains a derivation of its antecedent. 
  The derivation gives us a function \(F\) as above. 
  We describe now an iterative procedure to build approximations of a worst counterexample. 
  Given any \(f\) we take \(F(f)\) as a tentative witness for \(\exists x \forall y. \fc(x,y) \), that is we make the assumption that \( \forall y. \fc(F(f),y) \). 
  Note that we only know this to be true if \(f\) is a worst counterexample. 
  If in the proof this assumption is instanced for some natural number \(n\) such that \( \neg \fc(F(f),\num n) \), then our witness was wrong and we need a better witness to conclude the proof. 
  In order to get a better witness from \(F\) we can update \(f\): 
  we know that \( \fc(F(f),f(F(f))) \), but we also know that \( \neg \fc(F(f),\num n) \). 
  Then from \(f\) we can define \(f'\) as: 
  \[ f'(x) = \begin{cases}
      n &\text{ if } x = F(f), \\
      f(x) &\text{ otherwise}.
  \end{cases} \] 
  We can then repeat the same procedure on \(f'\). 
  If the witness \(F(f')\) is good enough (in general it does not need to be an actual witness of \( \exists x \forall y. \fc(x,y) \)) we then stop, otherwise we carry on with still another function. 

  Note that since \(f'\) gives the counterexample for \(F(f)\), necessarily \( F(f') \neq F(f) \) and since we update a function \(f\) only on the value of \(F(f)\), we are sure that we shall never update a value again. 
  Therefore there are no loops.

  We would like to know more about this computational interpretation of \cref{cac}, in particular whether this iteration eventually halts or not. 
  We can try to build an example where we never reach 
  \begin{align*}
    C(x,y) &\equiv x = 0 \lor y > 0, \\ 
    F(f) &\equiv \begin{cases}
    \text{the smallest } x > 0 \text{ such that } f(x) > 0 &\text{ if } \exists x.\ x > 0 \land f(x) > 0, \\ 
  0 &\text{ otherwise}. \end{cases} \\ 
    f_0 (x) &= 1 &&\text{forall } x
  \end{align*}
  In the proof: 
  \[ 
    \PrAss{\forall y.\ x = 0 \lor y > 0}{}
    \PrUn{x = 0 \lor 0 > 0}
    \PrAss{x = 0}{}
    \PrAss{0 > 0}{}
    \PrUn{x = 0}
    \PrTri{x = 0}
    \PrAx{\exists x \forall y.\ x = 0 \lor y > 0}
    \PrBin{x=0} 
    \DisplayProof
  \]

  Now we can proceed to formalize our reasoning by introducing a suitable monadic realizability notion. 
  We use a variation of the interactive realizability monad, with the state replaced by a function \( \Nat \tarrow \Nat \): 
  \[ M \ta \equiv (\Nat \tarrow \Nat) \tarrow \ta + \Ex. \]
  The type of the realizer is: 
  \[ \mm{\neg AC} = 
    M (((\Nat \tarrow \Nat) \tarrow \Nat \times \mm\fc) 
    \tarrow 
  M (\Nat \times (\Nat \tarrow \mm\fc))) \]

  The realizer is given as input a realizer \( \rr_1 : (\Nat \tarrow \Nat) \tarrow \Nat \times \mm\fc \) of \( \forall f \exists x.\ \fc(x,f(x)) \), that is, for any \( f : \Nat \tarrow \Nat \), if \( x = \prl (\rr_1 f) \) then \( \fc(x,f(x)) \). 
  The idea is that \(  x = \prl \rr_1 f \) is our tentative witness for \( \exists x \forall y \fc(x,y) \). 
  \[ \assert_\fc \num m \num n \equiv \case (\eval_\fc \num m \num n) (\inl \unit) (\inr ) \]
  \[ f' \equiv \abstr{x}\Nat \case (x = F(f)) \] 

  \[ \abstr{\rr_1}{(\Nat \tarrow \Nat) \tarrow \Nat \times \mm\fc}  \pair (\prl (\rr_1 f)) (\abstr{y}\Nat \assert_\fc(\prl (\rr_1 f)) y) \]

\end{wip}

\begin{wip}
  \newcommand\ff\fa
  \newcommand\eb{?}
  \section{Stateless Exception Semantics for $\EM$}
  We add a new ground type \(\Ex\). Its purpose is to store information about a witness for a \(\Sigma^0_1\) formula of the form \(\exists x. P(y_1, \dotsc, y_k, x)\).  
  We call terms of type \(\Ex\) \emph{exceptions}. 

  We add two families of terms indexed by a \(k+1\)-ary predicate symbol \(P\): 
  \begin{align*} 
    \throw_P &: \Nat^{k+1} \tarrow \Ex, \\
    \catch_P &: \Nat^k \tarrow \Ex \tarrow \Nat + \Ex. 
  \end{align*} 
  The first family are exception constructors: 
  we think of an exception \( \throw_P \num m_1 \dotsm \num m_k \num n \) as encoding the information that \(n\) is a witness for \(\exists x. P(\num m_1, \dotsc, \num m_k, x)\). 
  The second family are conditional exception destructors: \(\catch\) extracts the witness \(n\) from an exception \(e\) created by \(\throw\) if and only if they have the same index and the same parameters. 
  Formally we assume that:
  \begin{align*} 
    \catch_P \num m_1 \dotsm \num m_k (\throw_P \num m_1 \dotsm \num m_k \num n) &\leadsto \inl \num n \\ 
    \catch_P \num m_1 \dotsm \num m_k (\throw_Q \num m_1' \dotsm \num m_k' \num n) &\leadsto \inr (\throw_Q \num m_1' \dotsm \num m_k' \num n) 
  \end{align*} 
  \[ \catch_P \bar{\num m} (\throw_Q \bar{\num m'} \num n) \leadsto 
    \begin{cases}
      \inl \num n  & \text{ if } P \equiv Q \text{ and } \bar{\num m} = \bar{\num m'}, \\ 
      \inr (\throw_{Q \bar s} \num n) & \text{ otherwise}. 
  \end{cases} \] 
  \[ \catch_P \num m_1 \dotsm \num m_k e \leadsto 
    \begin{cases}
      \inl \num n  & \text{ if } e \leadsto \throw_P \num m_1 \dotsm \num m_k \num n, \\ 
      \inr e & \text{ otherwise}. 
  \end{cases} \] 
  for any natural numbers \(m_1, \dotsc, m_k, n\).

  \subsection{Exception Monad}
  \setmonad{\mathfrak{E}}
  \newcommand\ee{e}
  We define the syntactic monad for exceptions. 
  We define the type operator \(\T\ta = \ta + \Ex\) and two families of terms: 
  \begin{align*}
    \reg[\ta] &\equiv \inl[\ta,\Ex] : \ta \tarrow \T\ta, \\ 
    \exc[\ta] &\equiv \inr[\ta,\Ex] : \Ex \tarrow \T\ta. 
  \end{align*}
  \begin{align*} 
    \monUnit[\ta] \tta &\equiv \reg[\ta] \tta, \\ 
    \monStar[\ta,\tb] f (\reg[\ta] \tta) &\equiv f \tta, \\ 
    \monStar[\ta,\tb] f (\exc[\ta] \ee) &\equiv \exc[\tb] \ee, \\ 
    \monMerge[\ta,\tb] (\reg[\ta] \tta) (\reg[\tb] \ttb) &\equiv \reg[\ta\times\tb] (\pair[\ta,\tb] \tta \ttb), \\ 
    \monMerge[\ta,\tb] (\reg[\ta] \tta) (\exc[\ta] \ee) &\equiv \exc[\ta\times\tb] \eb, \\ 
    \monMerge[\ta,\tb] (\exc[\ta] \ee) (\reg[\tb] \ttb) &\equiv \exc \ee, \\ 
    \monMerge[\ta,\tb] (\exc[\ta] \ee_1) (\exc[\tb] \ee_2) &\equiv \exc[\ta \times \tb] (\exmerge \ee_1 \ee_2). 
  \end{align*}

  \begin{align*} 
    \monUnit[\ta] &\equiv \reg[\ta], \\ 
    \monStar[\ta,\tb] &\equiv \abstr{f}{\ta \tarrow \T\tb} \abstr{\mon\tta}{\T\ta} \case[\ta,\Ex,\tb+\Ex] \mon\tta f \exc[\tb], \\  
    \monMerge[\ta,\tb] &\equiv \abstr{\mon\tta}{\T\ta} \abstr{\mon\ttb}{\T\tb} \case[\ta,\Ex,(\ta \times \tb)+\Ex] \mon\tta \\ 
                       &\phantomrel\equiv (\abstr\tta\ta \case[\tb,\Ex,(\ta \times \tb)+\Ex] \mon\ttb (\abstr\ttb\tb \reg[\ta \times \tb] (\pair \tta \ttb)) \exc[\ta \times \tb]) \\ 
                       &\phantomrel\equiv (\abstr{e_1}\Ex \case[\tb,\Ex,(\ta \times \tb)+\Ex] \mon\ttb (\abstr\_\tb \exc[\ta \times \tb] e_1)(\abstr{e_2}\Ex \exc[\ta \times \tb] (\exmerge e_1 e_2) )). 
  \end{align*}

  The realizability relation: \( \mon\rr \monRe \ff \) if \( \mon\rr \leadsto \inl \rr \) and \( \rr \re \ff \) or \( \mon\rr \leadsto \inr n \). 

  We define some low level terms.
  The first is a family of terms indexed by a \((k+1)\)-ary predicate symbol \(P\) which enable us to evaluate primitive recursive predicates: 
  \begin{align*} 
    \eval_P &: \Nat^{k+1} \tarrow \Unit + \Unit \\ 
    \eval_P \num m_1 \dotsm \num m_k \num n &\leadsto \begin{cases}
    \inl \unit &\text{ if } P(\num m_1, \dotsc, \num m_k, \num n), \\ 
  \inr \unit &\text{ otherwise.} \end{cases}
  \end{align*} 

  We define a realizer of \( \Pi^0_1 \) formulas, called \(\assert\): 
  \begin{align*} 
    \assert_P &: \Nat^{k+1} \tarrow \Unit + \Ex \\ 
    \assert_P \num m_1 \dotsm \num m_k \num n &\equiv 
    \case (\eval_P \num m_1 \dotsm \num m_k \num n) \reg (\abstr\_\Unit \throw_P \num m_1 \dotsm \num m_k \num n) 
  \end{align*} 
  Equivalently we can just require:
  \[ \assert_P \num m_1 \dotsm \num m_k \num n \leadsto 
    \begin{cases}
      \inl \unit &\text{ if } P \num m_1 \dotsm \num m_k \num n \\
      \inr \throw_P \num m_1 \dotsm \num m_k \num n &\text{ otherwise} 
    \end{cases} 
  \] 

  In the same way we define:
  \begin{align*} 
    \handle_{P \bar t} &: ((\Nat \tarrow \Unit + \Ex) \tarrow \ta) \tarrow (\Nat \times \Unit \tarrow \ta) \tarrow \ta + \Ex \\ 
    \handle_{P \bar t} &\equiv \abstr{\rr_1}{\Nat \tarrow \Unit + \Ex} \abstr{\rr_2}{\Nat \times \Unit} 
    \case (\rr_1 \assert_{P \bar t}) \inl (\abstr{e}\Ex \\
    &\phantomrel\equiv \qquad \case (\catch_{P \bar t} e) (\abstr{x}\Nat \rr_2 (\pair x \unit) ) \inr)  
  \end{align*} 
  Or we can just require that: 
  \[ \handle_{P \bar t} \rr_1 \rr_2 \leadsto 
    \begin{cases} 
      \inl x &\text{ if } \rr_1 \assert_{P \bar t} \leadsto \inl x \\ 
      \inl (\rr_2 (\pair x \unit)) &\text{ if } \rr_1 \assert_{P \bar t} \leadsto e_{P \bar t} x \\ 
      \inr (e_{Q \bar s} x) &\text{ if } \rr_1 \assert_{P \bar t} \leadsto e_{Q \bar s} x 
    \end{cases} 
  \] 
  where \( P \neq Q \) or \( \bar t \neq \bar s \).

  The \(\EM\) rule is derived from the disjunction elimination rule: 
  \[
    \PrAx{\Gamma, \alpha_{k+1} : \forall x P \bar t x \monSeq \mon\rr_1 : \ff} 
    \PrAx{\Gamma, \alpha_{k+1} : \exists x \neg P \bar t x \monSeq \mon\rr_2 : \ff} 
    \PrLbl{\RuleName\EM{}} 
    \PrBin{\Gamma \monSeq \handle_{P \bar t} (\abstr{\alpha_{k+1}}{\Nat \tarrow \Unit} \mon\rr_1) (\abstr{\alpha_{k+1}}{\Nat \times \Unit} \mon\rr_2) : \ff}
    \DisplayProof 
  \] 
  Soundness:

\end{wip}

\begin{wip}
  \section{Monads in Category Theory}


  Let \(\ta\) and \(\tb\) be objects and \( f : \ta \tarrow \tb \) and \( g : \tb \tarrow \tc \) be morphisms. 

  Kleisli triple: 
  \begin{align*}
    \eta_\ta &: \ta \tarrow T\ta \\ 
    \_^{*_{\ta,\tb}} &: (\ta \tarrow T\tb) \tarrow T\ta \tarrow T\tb 
  \end{align*}
  Corresponding monad: 
  \begin{align*}
    T_{\ta,\tb} &\equiv f \mapsto (f ; \eta_\tb)^{*_{\ta,\tb}} &&: (\ta \tarrow \tb) \tarrow T\ta \tarrow T\tb \\ 
    \mu_\ta &\equiv i_{T\ta}^{*_{T\ta,\ta}} &&: TT\ta \tarrow T\ta
  \end{align*}

  Monad properties:
  \begin{align}
    \label[property]{cat_mon1} \eta_\ta^{*_\ta,\ta} = i_{T\ta} \\ 
    \label[property]{cat_mon2} \eta_\ta ; f^{*_\ta,\tb} = f \\ 
    \label[property]{cat_mon3} f^{*_\ta,\tb} ; g^{*_\tb,\tc}  = (f ; g^{*_\tb,\tc})^{*_\ta,\tc} 
  \end{align}

  \newcommand\cprod[2]{\langle #1, #2 \rangle}


  Product equalities and related definitions and their properties: 
  \begin{align}
    \cprod{f_1}{f_2} ; \pi^i &= f_i \\ 
    \cprod{f ; \pi^1}{f ; \pi^2} &= f \\
    f ; \cprod{g_1}{g_2} &= \cprod{f ; g_1}{f ; g_2} \\ 
    f_1 \times f_2 &\equiv \cprod{\pi^1 ; f_1}{\pi^2 ; f_2} \\ 
    \label{cross1} \cprod{f_1}{f_2} ; g_1 \times g_2 &= \cprod{f_1 ; g_1}{f_2 ; g_2} \\ 
    \label{cross2} f_1 \times f_2 ; g_1 \times g_2 &= (f_1 ; g_1) \times (f_2 ; g_2) \\ 
    c_{\ta,\tb} &: \ta \times \tb \tarrow \tb \times \ta \\ 
    c_{\ta,\tb} &\equiv \cprod{\pi^2}{\pi^1} \\ 
    \label{swap1} c_{\ta,\tb} ; f_1 \times f_2 &= f_2 \times f_1 ; c_{\ta,\tb} \\ 
    \label{swap2} c_{\ta,\tb} ; c_{\tb,\ta} &= i_{\ta \times \tb} \\ 
    \alpha_{\ta,\tb,\tc} &: (\ta \times \tb) \times \tc \tarrow \ta \times (\tb \times \tc) \\ 
    \alpha_{\ta,\tb,\tc} &\equiv \cprod{\pi^1_{\ta \times \tb,\tc} ; \pi^1_{\ta,\tb}}{\cprod{\pi^1_{\ta \times \tb,\tc} ; \pi^2_{\ta,\tb}}{\pi^2_{\ta \times \tb,\tc}}} \\ 
    \label{alpha1} (f_1 \times f_2) \times f_3 ; \alpha &= \alpha ; f_1 \times (f_2 \times f_3) 
  \end{align}

  \subsection{Strong monads}

  Terms: 
  \begin{align*}
    \tau_{\ta,\tb} &: \ta \times T\tb \tarrow T(\ta \times \tb) \\ 
    \phi_{\ta,\tb} & : T\ta \times T\tb \tarrow T(\ta \times \tb)
  \end{align*}

  Strong monad properties for \(\tau\): 
  \begin{align}
    \label{tau1} \tau_{1,\ta} ; T_{1 \times \ta, \ta} \pi^2_{1,\ta} &= \pi^2_{1,T\ta} \\ 
    \label{tau2} \tau_{\ta \times \tb, \tc} ; T_{(\ta \times \tb) \times \tc,\ta \times (\tb \times \tc)} \alpha_{\ta,\tb,\tc} &= 
    \alpha_{\ta,\tb,T\tc} ; i_\ta \times \tau_{\tb,\tc} ; \tau_{\ta,\tb \times \tc} \\ 
    \label{tau3} i_\ta \times \eta_\tb ; \tau_{\ta,\tb} &= \eta_{\ta \times \tb} \\ 
    \label{tau4} i_\ta \times \mu_\tb ; \tau_{\ta,\tb} &= 
    \tau_{\ta,T\tb} ; T_{\ta \times T\tb,T(\ta \times \tb)} \tau_{\ta,\tb} ; \mu_{\ta \times \tb} 
  \end{align}

  The morphism \(\tau\) and \(\phi\) can be defined in terms of each other: 
  \begin{align}
    \label{tau2phi} \tau_{\ta,\tb} &\equiv \eta_\ta \times i_{T\tb} ; \phi_{\ta,\tb} \\ 
    \label{phi2tau} \phi_{\ta,\tb} &\equiv c_{T\ta,T\tb} ; \tau_{T\tb,\ta} ; (c_{T\tb,\ta} ; \tau_{\ta,\tb})* 
  \end{align}
  We can verify this by the following equality: 
  \begin{align*}
    \tau_{\ta,\tb} &= \eta_\ta \times i_{T\tb} ; \phi_{\ta,\tb} \\ 
                   &= \eta_\ta \times i_{T\tb} ; c_{T\ta,T\tb} ; \tau_{T\tb,\ta} ; (c_{T\tb,\ta} ; \tau_{\ta,\tb})* \\ 
                   &= c_{T\ta,T\tb} ; i_{T\tb} \times \eta_\ta ; \tau_{T\tb,\ta} ; (c_{T\tb,\ta} ; \tau_{\ta,\tb})* \\ 
                   &= c_{T\ta,T\tb} ; \eta_{T\tb \times \ta} ; (c_{T\tb,\ta} ; \tau_{\ta,\tb})* && \text{by \cref{tau3}}\\ 
                   &= c_{T\ta,T\tb} ; c_{T\tb,\ta} ; \tau_{\ta,\tb} && \text{by \cref{mon2}}\\ 
                   &= \tau_{\ta,\tb} && \text{by \cref{swap2}}
  \end{align*}

  Now we can translate the strong monad properties for \(\tau\) in properties for \(\phi\) by replacing \(\tau\) as in \cref{tau2phi}. 
  We start with \cref{tau1}:
  \begin{align}
    \notag \tau_{1,\ta} ; T_{1 \times \ta, \ta} \pi^2_{1,\ta} &= \pi^2_{1,T\ta} \\ 
    \notag \eta_1 \times i_{T\ta} ; \phi_{1,\ta} ; T_{1 \times \ta, \ta} \pi^2_{1,\ta} &= \pi^2_{1,T\ta} \\ 
    \label{phi1} \eta_1 \times i_{T\ta} ; \phi_{1,\ta} ; (\pi^2_{1,\ta} ; \eta_\ta)^{*_{1 \times \ta,\ta}} &= \pi^2_{1,T\ta}
  \end{align}

  Now \cref{tau2} left-hand side:
  \begin{align*}
    \tau_{\ta \times \tb, \tc} ; T \alpha_{\ta,\tb,\tc} &= 
    \tau_{\ta \times \tb, \tc} ; T \alpha_{\ta,\tb,\tc} \\ 
    &= \eta_{\ta \times \tb} \times i_{T\tc} ; \phi_{\ta \times \tb,\tc} ; T_{(\ta \times \tb) \times \tc,\ta \times (\tb \times \tc)} \alpha_{\ta,\tb,\tc} \\ 
    &= (\eta_\ta \times \eta_\tb ; \phi_{\ta,\tb}) \times i_{T\tc} ; \phi_{\ta \times \tb,\tc} ; T_{(\ta \times \tb) \times \tc,\ta \times (\tb \times \tc)} \alpha_{\ta,\tb,\tc} && \text{by \cref{phi3}} \\ 
    &= (\eta_\ta \times \eta_\tb) \times i_{T\tc} ; \phi_{\ta,\tb} \times i_{T\tc} ; \phi_{\ta \times \tb,\tc} ; T_{(\ta \times \tb) \times \tc,\ta \times (\tb \times \tc)} \alpha_{\ta,\tb,\tc} && \text{by \cref{}} \\ 
  \end{align*}
  Now \cref{tau2} right-hand side:
  \begin{align*}
    \alpha_{\ta,\tb,T\tc} &; i_\ta \times \tau_{\tb,\tc} ; \tau_{\ta,\tb \times \tc} \\ 
                          &= \alpha_{\ta,\tb,T\tc} ; i_\ta \times (\eta_\tb \times i_{T\tc} ; \phi_{\tb,\tc}) ; \eta_\ta \times i_{T(\tb \times \tc)} ; \phi_{\ta,\tb \times \tc} \\ 
                          &= \alpha_{\ta,\tb,T\tc} ; (i_\ta  ; \eta_\ta) \times (\eta_\tb \times i_{T\tc} ; \phi_{\tb,\tc} ; i_{T(\tb \times \tc)}) ; \phi_{\ta,\tb \times \tc} && \text{by \cref{cross2}}\\ 
                          &= \alpha_{\ta,\tb,T\tc} ; \eta_\ta \times (\eta_\tb \times i_{T\tc} ; \phi_{\tb,\tc}) ; \phi_{\ta,\tb \times \tc} \\ 
                          &= \alpha_{\ta,\tb,T\tc} ; \eta_\ta \times (\eta_\tb \times i_{T\tc}) ; i_{T\ta} \times \phi_{\tb,\tc} ; \phi_{\ta,\tb \times \tc} \\ 
                          &= (\eta_\ta \times \eta_\tb) \times i_{T\tc} ; \alpha_{T\ta,T\tb,T\tc} ; i_{T\ta} \times \phi_{\tb,\tc} ; \phi_{\ta,\tb \times \tc} && \text{by \cref{alpha1}}\\ 
  \end{align*}
  The equality follows directly if we assume that the following stronger property holds: 
  \begin{equation*}
    \phi_{\ta,\tb} \times i_{T\tc} ; \phi_{\ta\times\tb,\tc} ; T \alpha_{\ta,\tb,\tc} = \alpha_{T\ta,T\tb,T\tc} ; i_{T\ta} \times \phi_{\tb\times\tc} ; \phi_{\ta,\tb\times\tc}. 
  \end{equation*}

  Now \cref{tau3}:
  \begin{align}
    \notag i_\ta \times \eta_\tb ; \tau_{\ta,\tb} &= \eta_{\ta \times \tb} \\ 
    \notag i_\ta \times \eta_\tb ; \eta_\ta \times i_{T\tb} ; \phi_{\ta,\tb} &= \eta_{\ta \times \tb} && \text{by \cref{tau2phi}}\\ 
    \notag (i_\ta  ; \eta_\ta) \times (\eta_\tb ; i_{T\tb}) ; \phi_{\ta,\tb} &= \eta_{\ta \times \tb} && \text{by \cref{cross2}} \\ 
    \label{phi3} \eta_\ta \times \eta_\tb ; \phi_{\ta,\tb} &= \eta_{\ta \times \tb} 
  \end{align}
  Now \cref{tau4} left-hand side:
  \begin{align*}
    i_\ta \times \mu_\tb ; \tau_{\ta,\tb} &= 
    i_\ta \times \mu_\tb ; \eta_\ta \times i_{T\tb} ; \phi_{\ta,\tb} && \text{by \cref{tau2phi}} \\ 
                                                                     &= (i_\ta  ; \eta_\ta) \times (\mu_\tb ; i_{T\tb}) ; \phi_{\ta,\tb} && \text{by \cref{cross2}} \\ 
                                                                     &= \eta_\ta \times \mu_\tb ; \phi_{\ta,\tb} \\ 
                                                                     &= \eta_\ta \times i_{T\tb}^{*_{T\tb,\tb}} ; \phi_{\ta,\tb} \\ 
  \end{align*}
  Now \cref{tau4} right-hand side:
  \begin{align*}
    \tau_{\ta,T\tb} &; T_{\ta \times T\tb,T(\ta \times \tb)} \tau_{\ta,\tb} ; \mu_{\ta \times \tb} \\ 
                    &= \eta_\ta \times i_{TT\tb} ; \phi_{\ta,T\tb} ; T_{\ta \times T\tb,T(\ta \times \tb)} (\eta_\ta \times i_{T\tb} ; \phi_{\ta,\tb}) ; \mu_{\ta \times \tb} \\ 
                    &= \eta_\ta \times i_{TT\tb} ; \phi_{\ta,T\tb} ; (\eta_\ta \times i_{T\tb} ; \phi_{\ta,\tb} ; \eta_{T(\ta \times \tb)})^{*_{\ta \times T\tb,T(\ta \times \tb)}} ; \mu_{\ta \times \tb} \\ 
                    &= \eta_\ta \times i_{TT\tb} ; \phi_{\ta,T\tb} ; (\eta_\ta \times i_{T\tb} ; \phi_{\ta,\tb} ; \eta_{T(\ta \times \tb)})^{*_{\ta \times T\tb,T(\ta \times \tb)}} ; i_{T (\ta \times \tb)}^{*_{T (\ta \times \tb),\ta \times \tb}} \\ 
                    &= \eta_\ta \times i_{TT\tb} ; \phi_{\ta,T\tb} ; (\eta_\ta \times i_{T\tb} ; \phi_{\ta,\tb} ; \eta_{T(\ta \times \tb)} ; i_{T (\ta \times \tb)}^{*_{T (\ta \times \tb),\ta \times \tb}})^{*_{\ta \times T\tb,\ta \times \tb}} && \text{by \cref{mon3}} \\ 
                    &= \eta_\ta \times i_{TT\tb} ; \phi_{\ta,T\tb} ; (\eta_\ta \times i_{T\tb} ; \phi_{\ta,\tb} ; i_{T (\ta \times \tb)})^{*_{\ta \times T\tb,\ta \times \tb}} && \text{by \cref{mon2}} \\ 
                    &= \eta_\ta \times i_{TT\tb} ; \phi_{\ta,T\tb} ; (\eta_\ta \times i_{T\tb} ; \phi_{\ta,\tb})^{*_{\ta \times T\tb,\ta \times \tb}} \\ 
  \end{align*}
\end{wip}

\section{Additional Reductions}
In this section we list the standard reductions (given in \cite{prawitz71}) that we did not need in order to prove \cref{thm:em_reduction}. 
The main reason is that in the proof we are only concerned with principal branches and some reductions only affect non-principal branches. 
We list them for completeness. 

\subsection{Permutative Reductions}\label{sec:permutative_reductions} 
We saw how the proper reductions are performed when the conclusion of an introduction rule instance is the major premiss of a elimination rule instance. 
The situation becomes less straightforward when the conclusion of an introduction rule instance \(\riA\) is a minor premiss of an \(\RuleName\lor{E}\) or \(\RuleName\exists{E}\) rule instance \(\riC\) whose conclusion is in turn the major premiss of an elimination rule instance \(\riB\). 
Also in this case the formula introduced by \(\riA\) is eliminated by \(\riB\), but we cannot apply a proper reduction since \(\riC\) is in the way. 
What we can do is to rearrange the derivation by moving \(\riB\) above (or ``inside'') \(\riC\), so that \(\riB\) is immediately below \(\riA\) and we can apply the suitable proper reduction. 
Therefore we have two \emph{permutative reductions}, depending on whether \(\riC\) is an instance of the \(\RuleName\lor{E}\) or the \(\RuleName\exists{E}\) rule.
Note that repeated application of the permutative reductions allows us to apply a proper reduction even when there is more than one instance of the \(\RuleName\lor{E}\) or the \(\RuleName\exists{E}\) between \(\riA\) and \(\riB\).  
Thus they can be thought of as auxiliary reductions that can eventually enable a suitable proper reduction. 
They are listed in \Cref{fig:permutative_reductions}. 

\begin{figure}[!ht]
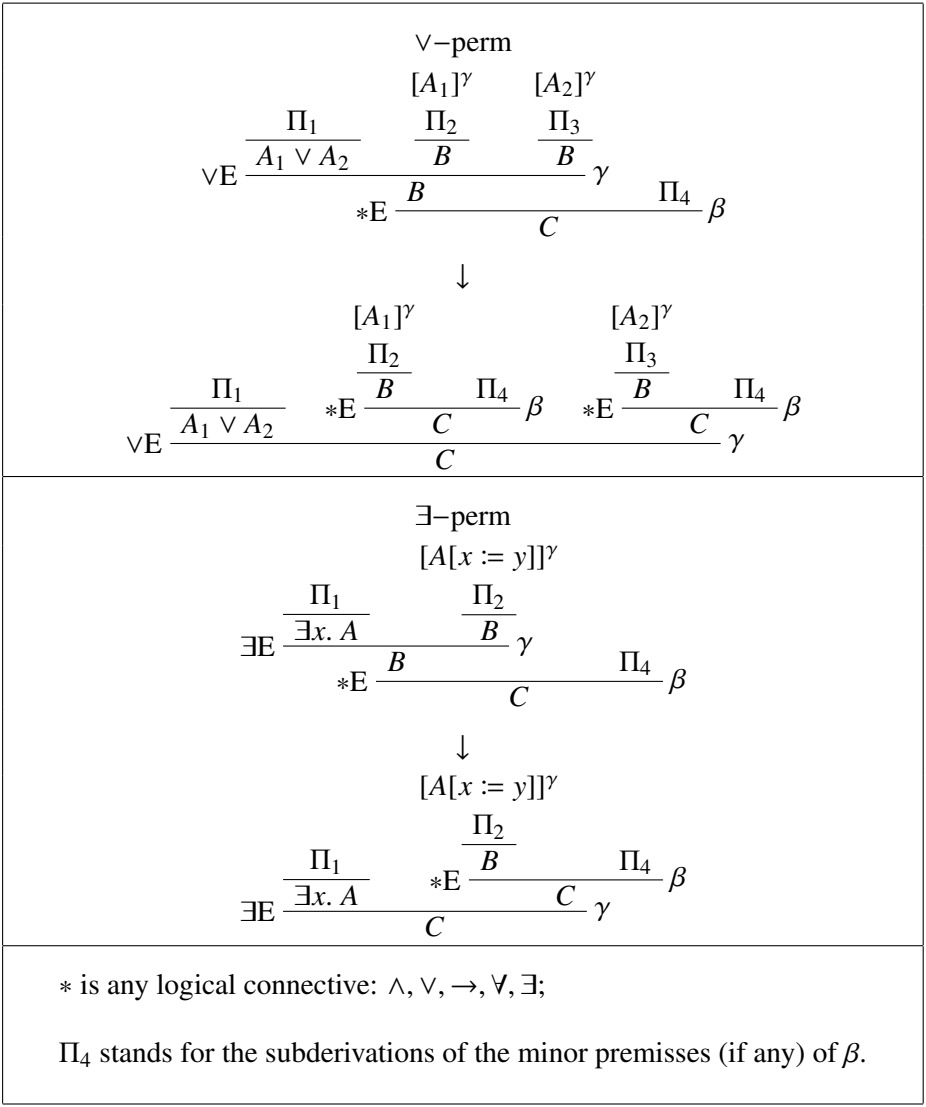

  \caption{The permutative reductions.}
  \label{fig:permutative_reductions}
  \begin{rulelisting}[]
    \header{\RedNamePerm\lor}
    &
    \PrAx{}
    \PrInf[\derA_1]
    \PrUn{\fa_1 \lor \fa_2}
    \PrAss{\fa_1}\riC
    \PrInf[\derA_2]
    \PrUn\fb
    \PrAss{\fa_2}\riC
    \PrInf[\derA_3]
    \PrUn\fb
    \PrLbl[\riC]{\RuleName\lor{E}}
    \PrTri\fb
    \PrAx{}
    \PrInf[\derA_4]
    \PrLbl[\riB]{\RuleName\ast{E}}
    \PrBin\fc
    \DisplayProof 
    &
    \\
    & \downarrow &
    \\
    &
    \PrAx{}
    \PrInf[\derA_1]
    \PrUn{\fa_1 \lor \fa_2}
    \PrAss{\fa_1}\riC
    \PrInf[\derA_2]
    \PrUn\fb
    \PrAx{}
    \PrInf[\derA_4]
    \PrLbl[\riB]{\RuleName\ast{E}}
    \PrBin\fc
    \PrAss{\fa_2}\riC
    \PrInf[\derA_3]
    \PrUn\fb
    \PrAx{}
    \PrInf[\derA_4]
    \PrLbl[\riB]{\RuleName\ast{E}}
    \PrBin\fc
    \PrLbl[\riC]{\RuleName\lor{E}}
    \PrTri\fc
    \DisplayProof
    &
    \newline
    \header{\RedNamePerm\exists} 
    &
    \PrAx{}
    \PrInf[\derA_1]
    \PrUn{\qexists\lva\fa}
    \PrAss{\fa\subst\lva\lvb}\riC
    \PrInf[\derA_2]
    \PrUn\fb
    \PrLbl[\riC]{\RuleName\exists{E}}
    \PrBin\fb
    \PrAx{}
    \PrInf[\derA_4]
    \PrLbl[\riB]{\RuleName\ast{E}}
    \PrBin\fc
    \DisplayProof 
    &
    \\ & \downarrow & \\
       &
    \PrAx{}
    \PrInf[\derA_1]
    \PrUn{\qexists\lva\fa}
    \PrAss{\fa\subst\lva\lvb}\riC
    \PrInf[\derA_2]
    \PrUn\fb
    \PrAx{}
    \PrInf[\derA_4]
    \PrLbl[\riB]{\RuleName\ast{E}}
    \PrBin\fc
    \PrLbl[\riC]{\RuleName\exists{E}}
    \PrBin\fc
    \DisplayProof
    &
    \newline
    &
    \begin{tabular}{l}
      \(\ast\) is any logical connective: \(\land, \lor, \limply, \forall, \exists\); \\
      \(\derA_4\) stands for the subderivations of the minor premisses (if any) of \(\riB\).
    \end{tabular}
    &
  \end{rulelisting}

\end{figure}

\subsection{Immediate Simplifications} 
Consider a different kind of avoidable complexity in derivation: 
an instance \(\riA\) of the \(\RuleName\lor{E}\) or the \(\RuleName\exists{E}\) rule such that the one of its minor premisses \(\foA\) is derived without using the assumption discharged by \(\riA\). 
More precisely this means that no occurrence of the assumption discharged by \(\riA\) appears in the subderivation of \(\foA\). 
Whenever this is the case we say that \(\riA\) is \emph{redundant} since we do not need the assumptions it provides in order to prove its conclusion. 
There are two reductions, called \emph{immediate simplifications}, depending on whether \(\riA\) is an instance of the \(\RuleName\lor{E}\) or the \(\RuleName\exists{E}\) rule. 
They are listed in \Cref{fig:immediate_simplifications}. 

\begin{figure}[!ht]
  \caption{The immediate simplifications.}
  \label{fig:immediate_simplifications}
  \begin{rulelisting}[\leadsto]
    \header{\RedNameSimpl\lor}
    \PrAx{\fa_1 \lor \fa_2}
    \PrAss{\fa_1}\riA
    \PrInf[\derA_1]
    \PrUn\fc
    \PrAss{\fa_2}\riA
    \PrInf[\derA_2]
    \PrUn\fc
    \PrLbl[\riA]{\RuleName\lor{E}}
    \PrTri\fc
    \DisplayProof
    \nextrule
    \PrAx{}
    \PrInf[\derA_i]
    \PrUn\fc
    \DisplayProof 
    \\
    \header{\text{ if \(\derA_i\) does not depend on \(\fa_i\)}}
    \hline
    \header{\RedNameSimpl\exists} 
    \PrAx{\qexists\lva \fa}
    \PrAss{\fa\subst\lva\lvb}\riA
    \PrInf[\derA]
    \PrUn\fb
    \PrLbl[\riA]{\RuleName\exists{E}}
    \PrBin\fb
    \DisplayProof
    \nextrule
    \PrAx{}
    \PrInf[\derA]
    \PrUn\fc
    \DisplayProof 
  \end{rulelisting}
\end{figure}

\begin{omitted}
  \subsection{Normal Derivations}
  When a derivation cannot be reduced by any of the reductions we listed, it is said to be a \emph{normal} derivation. 
  If we make the ``cannot be reduced'' part explicit, we get the following definition. 
  \begin{definition}
    A derivation is \emph{directly reducible} when one of the following holds: 
    \begin{itemize}
      \item it ends with an elimination rule instance whose major premiss is the conclusion of an introduction rule instance (proper reduction),
      \item it ends with an elimination rule instance whose major premiss is the conclusion of an \(\RuleName\lor{E}\) or \(\RuleName\lor{E}\) rule instance (permutative reduction), 
      \item it ends with an \(\RuleNameIndE\) rule instance whose main\fixme{define} term is either \(\num 0\) or \(\succ(\ltb)\) for some \(\ltb\) (\(\RedNameInd\) reduction), 
      \item it ends with an \(\EM\) rule instance \(\riA\), 
        in the subderivation of the leftmost premiss of \(\riA\)
        one occurrence of the assumption discharged by \(\riA\) is the premiss of a \(\RuleName\forall{E}\) rule instance \(\riB\) and
        the conclusion of \(\riB\) is an occurrence of a closed formula (\(\RedNameWitness\)), 
      \item it ends with an elimination or atomic rule instance \(\riA\) and one premiss of \(\riA\) is the conclusion of a \(\EM\) rule instance (\(\EM\) permutative reduction), 
      \item it ends with an \(\EM\) rule instance \(\riA\) and one of the premisses of \(\riA\) is derived without using the assumption discharged by \(\riA\) (\(\EM\) immediate simplification). 
    \end{itemize}
    A derivation is \emph{reducible} when it or one of its subderivations is directly reducible.
    A derivation is \emph{normal} when it is not reducible. 
  \end{definition}
  Since the reductions are meant to remove detours and simplify a derivation, we expect a normal derivation to have a particularly simple structure. 

\end{omitted}

\section{Witness Reduction in Two Steps} \fixme{rewrite}
Instead of giving the single reduction \(\RedNameWitness\), 
we can split it in two distinct reductions, 
one that looks for counterexamples and eliminates occurrences of the open assumptions of the \(\EM\) rule and 
one that eliminates instances of the \(\EM\) rule when their conclusion can be derived without the universal or existential assumption. 


\subsection{A Lighter Witness Extracting Reduction}
More precisely consider an instance \(\riA\) of the \(\EM\) rule for the quantifier-free formula \(\fa\):
\[
  \PrAss{\qforall\lva\fa}\riA
  \PrInf[\derA_1]
  \PrUn\fc
  \PrAss{\lnot\fa\subst\lva\lvb}\riA
  \PrInf[\derA_2]
  \PrUn\fc
  \PrLbl[\riA]\EM
  \PrBin\fc
  \DisplayProof
\]
Let \(\foA\) be an occurrence of the assumption \(\qforall\lva\fa\) discharged by \(\riA\) in \(\derA_1\)
In order to be able to perform the reduction we assume the following:
\begin{itemize}
  \item \(\foA\) is the premiss of a \(\RuleName\forall{E}\) instance \(\riB\), 
  \item the conclusion of \(\riB\) is the occurrence of a closed formula. 
\end{itemize}
Let \(\foB\) be the conclusion of \(\riB\). 
\(\foB\) is an occurrence of the closed quantifier-free formula \(\fa\subst\lva\lta\) for some term \(\lta\). 
Since closed quantifier-free formula are decidable, 
in the reduction we can distinguish two cases,  
depending on whether \(\fa\subst\lva\lta\) is true or false. 
\begin{itemize}
  \item \(\fa\subst\lva\lta\) is true. 

    Let \(\derA_1\) be a derivation of \(\fa\subst\lva\lta\) and 
    replace \(\riB\) with \(\derA_1\):
    \[ 
      \PrAss{\qforall\lva\fa}\riA
      \PrLbl{\RuleName\forall{E}}
      \PrUn{\fa\subst\lva\lta}
      \PrInf
      \DisplayProof \quad \leadsto \quad 
      \PrAx{}
      \PrInf[\derA_1]
      \PrUn{\fa\subst\lva\lta}
      \PrInf
      \DisplayProof
    \] 

  \item otherwise \(\lnot\fa\subst\lva\lta\) is true.

    Let \(\derA_2\) be a derivation of \(\lnot\fa\subst\lva\lta\) and 
    replace all the occurrences of the assumption \(\lnot\fa\subst\lva\lta\) discharged by \(\riA\) in the derivation of its rightmost premiss with \(\derA_2\):
    \[ 
      \PrAss{\lnot\fa\subst\lva\lvb}\riA
      \PrInf
      \DisplayProof \quad \leadsto \quad
      \PrAx{}
      \PrInf[\derA_2]
      \PrUn{\lnot\fa\subst\lva\lta}
      \PrInf
      \DisplayProof
    \] 
\end{itemize}
We denote this reduction as \(\RedNameWitness\). 

Whenever this reduction can be applied it removes one or more occurrences of one of the assumptions discharged by \(\riA\). 
If there are no more occurrences of such assumptions in either \(\derA_1\) or \(\derA_2\) then \(\riA\) is redundant and can be deleted by \(\RedNameSimpl\EM\). 

\subsection{Immediate Simplification}
Redundant \(\EM\) rule instances can be defined and reduced in the same way as redundant \(\RuleName\lor{E}\) instances. 
Consider an \(\EM\) rule instance \(\riA\) such that one of its premisses is derived without using the assumption discharged by \(\riA\).
Then we can reduce as follows:
\[
  \PrAss{\qforall\lva \fa}\riA
  \PrInf[\derA_1]
  \PrUn\fc
  \PrAss{\lnot \fa\subst\lva\lvb}\riA
  \PrInf[\derA_2]
  \PrUn\fc
  \PrLbl[\riA]\EM
  \PrBin\fc
  \DisplayProof
  \leadsto 
  \PrAx{} 
  \PrInf[\Sigma_i] 
  \PrUn\fc 
  \DisplayProof 
\]
depending on whether it is \(\derA_1\) or \(\derA_2\) that contains no occurrence of the assumption discharged by \(\riA\). 
We denote this reduction as \(\RedNameSimpl\EM\).

\begin{omitted}

\section{$\HA+\EM$ in Type Theory}

\renewcommand\EM { {\mathsf{EM}} }
\newcommand\Language {\mathcal{L}}
\newcommand\E[1]{{\mathsf{em}_{#1}}}
\newcommand\proj{ {\mathsf{p}} }
\newcommand\inj\upiota
\newcommand\emp[1]{\mathsf{P}_{#1}}
\newcommand\hyp {{\mathtt{H}_{\emp{}}}}
\newcommand\wit {{\mathtt{W}_{\emp{}}}}
\newcommand\red {\;\mathsf{r}\;}
\newcommand\cruno {{\textbf{(CR1)}}}
\newcommand\crdue {{\textbf{(CR2)}}}
\newcommand\crtre {{\textbf{(CR3)}}}
\renewcommand\num[1]{\overline{#1}}
\newcommand\nSystemT{\mathsf{T}^\star}
\newcommand\nlambda{\mathsf\Lambda^\star}
\newcommand\SystemT{\mathcal{T}}
\newcommand\SystemTG {\mathsf{T}}
\newcommand\ifn{{\mathsf{if}}}
\newcommand\ifthen [3]{ {\mathsf{if}\ {#1}\ \mathsf{then}\ {#2}\ \mathsf{else}\ {#3} } }
\renewcommand\Bool { {\tt Bool} }
\newcommand\sn{\mathsf{SN}}
\newcommand\rec {{\mathsf{R}}}
\newcommand\True { {\tt{True}} }
\newcommand\False { {\tt{False}} }
\newcommand\suc{\mathsf{S}}
\newcommand\itr{\mathsf{It}}
\renewcommand\pair[2]{\langle #1,#2\rangle}
\newcommand\prj[2]{{[#2]}\pi_{#1}}
\newcommand\const{\mathsf{c}}
\newcommand\stred {{\,\overset{\mathsf{st}}\mapsto\,}}
\newcommand\streds {{\,\overset{\mathsf{st}*}\mapsto\,}}
\newcommand\trans[1] {{#1}^{*}}
\newcommand\redcbn {\mapsto_{\mathsf{cbn}}}
\newcommand\sncbn {\mathsf{SN}_{\mathsf{cbn}}}
\newcommand\elcbn {\mathsf{E}_{\mathsf{cbn}}}
\newcommand\nf{\mathsf{NF}}
\newcommand\gn{\mathsf{GN}}
\newcommand\prel{\,\mathscr{P}\,}
\newcommand\beq{=_\beta}
\newcommand\seqt[3]{#1_{1}^{#3}\ldots #1_{#2}^{#3}}
\newcommand\redcbv {\mapsto_{\mathsf{cbv}}}
\newcommand\linea{\leavevmode\hrule\mbox{}}
\newcommand\dlinea{\leavevmode\hrule\vspace{1pt}\hrule\mbox{}}

$\HA+\EM_1$  is formally described in figure \ref{fig:D}. 
We define a standard natural deduction system for $\HA+\EM_1$ 
together with a term assignment  in the spirit of Curry-Howard correspondence for classical logic. 
We employ three distinct classes of variables in the proof terms: 
one for proof terms, denoted usually as $x, y,\ldots$; 
one for quantified variables of the formula language $\Language$ of $\HA+\EM_{1}$, denoted usually as $\alpha, \beta, \ldots$; 
one for hypotheses bound by $\EM_{1}$, denoted usually as $a, b, \ldots$.  
We assume atomic predicates to be denoted by $\emp{}, \emp{0}, \emp{1}, \ldots$. In the term $\E{a}\, u\, v $ all the occurrences of $a$ in $u$ and $v$ are bound. 

\begin{description}
  \item[Grammar of Untyped Terms]
    \[t,u, v::=\ c\ |\ x\  |\ tu\ |\ tm\ |\ \lambda x u\  |\ \lambda \alpha u\ |\ \langle t, u\rangle\ |\ \pi_0u\ |\ \pi_{1} u\ |\ \inj_{0}(u)\ |\ \inj_{1}(u)\ |\ (m, t)\ | t[x.u, y.v]\ |\ t[(\alpha, x). u]\]
    \[|\ [a]\hyp\ |\ [a]\wit\ |\ \True \ |\ \rec u v m \ |\ \mathsf{r}t_{1}\ldots t_{n}\]
    where $m$ ranges over terms of $\Language$. 
  \item[Contexts] 
    With $\Gamma$ we denote contexts of the form $e_1:A_1, \ldots, e_n:A_n$, where $e_{i}$ is either a proof-term variable $x, y, z\ldots$ or a $\EM_{1}$ hypothesis variable $a, b, \ldots$ 

  \item[Axioms] 
    $\begin{array}{c}   \Gamma, x:{A}\vdash x: A 
    \end{array}\ \ \ \ $
    $\begin{array}{c}   \Gamma, a:{\forall \alpha^{\Nat} \emp{}\beta \alpha}\vdash  [a] \hyp \beta:  \forall\alpha^{\Nat} \emp{}\beta \alpha
    \end{array}\ \ \ \ $
    $\begin{array}{c}   \Gamma, a:{\exists \alpha^{\Nat} \lnot\emp{}}\beta\alpha\vdash [a]\wit\beta:  \exists\alpha^{\Nat} \lnot \emp{}\beta\alpha
    \end{array}$\\

  \item[Conjunction] 
    $\begin{array}{c}  \Gamma \vdash u:  A\ \ \ \Gamma\vdash t: B\\ \hline \Gamma\vdash \langle
      u,t\rangle:
      A\wedge B
    \end{array}\ \ \ \ $
    $\begin{array}{c} \Gamma \vdash u: A\wedge B\\ \hline\Gamma \vdash\pi_0u: A
    \end{array}\ \ \ \ $
    $\begin{array}{c}  \Gamma \vdash u: A\wedge B\\ \hline \Gamma\vdash\pi_1 u: B
    \end{array}$\\\\

  \item[Implication] 
    $\begin{array}{c}  \Gamma\vdash t: A\rightarrow B\ \ \ \Gamma\vdash u:A \\ \hline
      \Gamma\vdash tu:B
    \end{array}\ \ \ \ $
    $\begin{array}{c}  \Gamma, x:A \vdash u: B\\ \hline \Gamma\vdash \lambda x u:
      A\rightarrow B
    \end{array}$\\\\
  \item[Disjunction Intro.] 
    $\begin{array}{c}  \Gamma \vdash u: A\\ \hline \Gamma\vdash \inj_{0}(u): A\vee B
    \end{array}\ \ \ \ $
    $\begin{array}{c}  \Gamma \vdash u: A\\ \hline \Gamma\vdash\inj_{1}(u): A\vee B
    \end{array}$\\\\

  \item[Disjunction Elim.] $\begin{array}{c} \Gamma\vdash u: A\vee B\ \ \ \Gamma, x: A \vdash w_1: C\ \ \ \Gamma, x:B\vdash w_2:
      C\\ \hline \Gamma\vdash  u [x.w_{1}, x.w_{2}]: C
    \end{array}$\\\\

  \item[Universal Quantification] 
    $\begin{array}{c} \Gamma \vdash u:\forall \alpha^{\Nat} A\\ \hline  \Gamma\vdash ut: A[t/\alpha]
    \end{array} $
    $\begin{array}{c}  \Gamma \vdash u: A\\ \hline \Gamma\vdash \lambda \alpha u:
      \forall \alpha^{\Nat} A
    \end{array}$\\

    where $t$ is a term of  the language $\Language$ and $\alpha$ does not occur
    free in any formula $B$ occurring in $\Gamma$.\\

  \item[Existential Quantification] 
    $\begin{array}{c}\Gamma\vdash  u: A[t/\alpha]\\ \hline \Gamma\vdash (
      t,u):
      \exists
      \alpha^\Nat. A
    \end{array}$ \ \ \ \
    $\begin{array}{c} \Gamma\vdash u: \exists \alpha^\Nat A\ \ \ \Gamma, x: A \vdash t:C\\
      \hline
      \Gamma\vdash u [(\alpha, x). t]: C
    \end{array} $\\
    where $\alpha$ is not free in $C$
    nor in any formula $B$ occurring in $\Gamma$.\\

  \item[Induction] 
    $\begin{array}{c} \Gamma\vdash u: A(0)\ \ \ \Gamma\vdash v:\forall \alpha^{\Nat}.
      A(\alpha)\rightarrow A(\suc(\alpha))\\ \hline \Gamma\vdash\lambda \alpha^{\Nat} \rec uv\alpha :
      \forall
      \alpha^{\Nat} A
    \end{array}\ \ \ \ $\\\\

  \item[Post Rules] 
    $\begin{array}{c}  \Gamma\vdash u_1: A_1\ \Gamma\vdash u_2: A_2\ \cdots \ \Gamma\vdash u_n:
      A_n\\ \hline\Gamma\vdash u: A
    \end{array}$

    where $A_1,A_2,\ldots,A_n,A$ are atomic
    formulas of  and the rule is a Post rule for equality, for a Peano axiom or for a classical propositional
    tautology or for booleans and if $n>0$, $u=\mathsf{r} u_{1}\ldots u_{n}$, otherwise $u=\True$.  \\

  \item[\(\EM_1\)]
    $\begin{array}{c} \Gamma, a: \forall \alpha^{\Nat} \emp{}\beta\alpha \vdash w_1: C\ \ \ \Gamma, a: \exists \alpha^{\Nat}\lnot \emp{}\beta \alpha \vdash w_2:
      C\\ \hline \Gamma\vdash  \E{a}\, w_{1}\, w_{2} : C
    \end{array}$\\\\

\end{description}
\end{omitted}
\bibliographystyle{plain} 
\bibliography{realizability} 
\end{document}

%% file: notation.tex
\usepackage{txfonts}
\usepackage{ifthen}

\newcommand\phantomrel[1]{\mathrel{\phantom{#1}}} 
\newcommand\reducesto\to

\newcommand\N{\mathbb{N}}

\newcommand\num[1]{\boldsymbol{#1}} 
\newcommand\tarrow\to 

\newcommand\seq\vdash 
\newcommand\monSeq\Vdash 
\newcommand\mon\mathfrak 

\newcommand\ifNotEmpty[2]{\ifthenelse{\equal{#1}{}}{}{#2}}
\newcommand\optionalSubscript[1]{\ifthenelse{\equal{#1}{}}{}{_{#1}}}
\newcommand\optionalSuperscript[1]{\ifthenelse{\equal{#1}{}}{}{^{#1}}}

\newcommand\MonadStyle\mathfrak

\newcommand\id[1]{\mathrm{id}_{#1}} 

\newcommand\Unit{\mathrm{Unit}} 
\newcommand\Bool{\mathrm{Bool}} 
\newcommand\Nat{\mathrm{Nat}} 
 
\newcommand\State{\mathrm{State}} 
\newcommand\Ex{\mathrm{Ex}} 

\newcommand\abstr[2]{\lambda #1^{#2}.}

\newcommand\abstrS[1]{\abstr{#1}\State}
\DeclareMathOperator\unit\ast

\newcommand\newtermconstant[3]{\newcommand{#1}[1][]{#2
\ifthenelse{\equal{##1}{}}{}{^{##1}}%
\ifthenelse{\equal{#3}{}}{}{_\textrm{#3}}}%
}
\newcommand\newtermconstantN[2]{\newcommand{#1}[2][]{#2
\ifthenelse{\equal{##1}{}}{}{^{##1}}%
\ifthenelse{\equal{##2}{}}{}{_{##2}}}%
}

\newcommand\termname[1]{\operatorname{\textsf{\small#1}}}
\newtermconstant\pair{\termname{pair}}{}
\newtermconstant\prr{\termname{pr}}{R}
\newtermconstant\prl{\termname{pr}}{L}
\newtermconstant\inr{\termname{in}}{R}
\newtermconstant\inl{\termname{in}}{L}
\newtermconstant\case{\termname{case}}{}
\newtermconstantN\totRec{\termname{crec}}
\newtermconstantN\monStarN{\operatorname{\mon{star}}}
\newtermconstantN\monRaiseN{\operatorname{\mon{raise}}}
\newtermconstant\monBind{\operatorname{\mon{bind}}}{}
\newcommand\exmerge{\termname{merge}}
\newcommand\eval{\termname{eval}}
\newcommand\query{\termname{query}}
\newcommand\throw{\termname{throw}}
\newcommand\catch{\termname{catch}}
\newcommand\assert{\termname{assert}}
\newcommand\handle{\termname{handle}}
\newtermconstant\reg{\termname{reg}}{}
\newtermconstant\exc{\termname{ex}}{}
\newtermconstant\ite{\termname{ite}}{}
\newtermconstant\true{\termname{true}}{}
\newtermconstant\false{\termname{false}}{}
\newtermconstant\dummy{\termname{dummy}}{}
\newtermconstant\zero{\termname{zero}}{}
\let\succ\undefined
\newtermconstant\succ{\termname{succ}}{}

\newcommand\interp[1]{\llbracket #1 \rrbracket}

%% file: gdraft.tex
\usepackage{comment} 
\usepackage{ifdraft}
\usepackage{xcolor}

\newcommand\newmarkedenvironment[2]{%
  \newenvironment{#1}{\noindent\textcolor{red}{*** BEGIN: {#2} ***}\\}{\textcolor{red}{\\ *** END: {#2} ***}}
}

\ifoptiondraft{
  \newmarkedenvironment{wip}{Work in progress}
  \newmarkedenvironment{omitted}{Omitted}
}{
  \excludecomment{wip}
  \excludecomment{omitted}
}

\ifoptionfinal{
  \includecomment{changed}
  \newcommand\fixme[1]{}
  \newcommand\note[1]{}
}{
  \newmarkedenvironment{changed}{Changed from last version}
  \newcommand\fixme[1]{\textcolor{red}{ (X)}\marginpar{\textcolor{red}{#1}}}
  \newcommand\note[1]{\textcolor{green}{ (X)}\marginpar{\textcolor{green}{#1}}}
}

%% file: gproof.tex
\usepackage{bussproofs}
\usepackage{txfonts}
\usepackage{ifthen}
\usepackage{tikz}
\usetikzlibrary{patterns}
\usetikzlibrary{decorations.pathmorphing}

\providecommand\optionalSubscript[1]{\ifthenelse{\equal{#1}{}}{}{_{#1}}}
\providecommand\optionalSuperscript[1]{\ifthenelse{\equal{#1}{}}{}{^{#1}}}

\newcommand\fa{A}
\newcommand\fb{B}
\newcommand\fc{C}

\newcommand\ltrue\top
\newcommand\lfalse\bot

\newcommand\limply\to 
\newcommand\subst[2]{[ #1 \coloneqq #2 ]}
\newcommand\quant[2]{#1 #2.\ }
\newcommand\qlambda[1]{\quant\lambda{#1}}
\newcommand\qforall{\quant\forall}
\newcommand\qexists{\quant\exists}

\tikzstyle{label} = [fill=white,inner sep=1pt]
\newcommand\IfNotEmpty[2]{\ifthenelse{\equal{#1}{}}{}{#2}}

\newcommand\PrDer[3]{
\begin{tikzpicture}[baseline=0.3cm,scale=1]
  \IfNotEmpty{#3}{ 
    \draw decorate [decoration={snake,amplitude=2pt,segment length=14pt}] {(0,0.2) --  node[above,label] {\(#3\)} (0,1.25)};
  }
  \draw (-0.5,1) -- (-0.1,0) -- (0.1,0) -- (0.5,1);
  \IfNotEmpty{#1}{ 
    \node[above,label] () at (0,0.2) {\(#1\)}; 
  } 
    
\end{tikzpicture}
}
\newcommand\PrInfBr[2]{\def\extraVskip{0pt}\noLine\UnaryInfC{\PrDer{#1}{}{#2}}\noLine\def\extraVskip{2pt}}

\newcommand\PrAss[3][]{\AxiomC{\([#2\optionalSuperscript{#1}]^{#3}\)}}
\newcommand\PrAx[2][]{\AxiomC{\(#2\optionalSuperscript{#1}\)}}
\newcommand\PrUn[2][]{\UnaryInfC{\(#2\optionalSuperscript{#1}\)}}
\newcommand\PrBin[2][]{\BinaryInfC{\(#2\optionalSuperscript{#1}\)}}
\newcommand\PrTri[2][]{\TrinaryInfC{\(#2\optionalSuperscript{#1}\)}}
\newcommand\PrLbl[2][]{\LeftLabel{\(#2\)}\ifthenelse{\equal{#1}{}}{}{\RightLabel{\(#1\)}}}

\newcommand\PrInf[1][]{\ifthenelse{\equal{#1}{}}{%
\def\extraVskip{-2pt}\noLine\UnaryInfC\vdots\noLine\def\extraVskip{2pt}}{%
\noLine\UnaryInfC{\(#1\)}}}
\newcommand\RuleName[3][]{#2\mathrm{#3}\IfNotEmpty{#1}{_\mathrm{#1}}}
\newcommand\RuleNameI[2][]{\RuleName[#1]{#2}{I}}
\newcommand\RuleNameE[2][]{\RuleName[#1]{#2}{E}}

\newcommand\PrImplyI[2][]{%
  \PrLbl[#1]{\RuleName\limply{I}}%
  \PrUn{#2}%
}
\newcommand\PrImplyE[2][]{%
  \PrLbl[#1]{\RuleName\limply{E}}%
  \PrBin{#2}%
}

\newcommand\PrExistsE[2][]{%
  \PrLbl[#1]{\RuleName\exists{E}}%
  \PrBin{#2}%
}
\newcommand\PrForallI[2][]{%
  \PrLbl[#1]{\RuleName\forall{I}}%
  \PrUn{#2}%
}
\newcommand\PrForallE[2][]{%
  \PrLbl[#1]{\RuleName\forall{E}}%
  \PrUn{#2}%
}

\newenvironment{rulelisting}[1][\quad]{
\renewcommand\newline{\\[8mm] \hline}
\newcommand\nextrule[1][]{& \ifthenelse{\equal{##1}{}}{#1}{##1} &}
\newcommand\wholeline[1]{\multicolumn{3}{|c|}{##1}}
\newcommand\header[1]{\wholeline{##1} \\ }
\renewcommand\arraystretch{1.5}
\begin{displaymath}
\begin{array}{|ccc|}
\hline
}{
\newline
\end{array}
\end{displaymath}
}